\documentclass[a4paper,english,10pt]{article}
\usepackage[left=.9in, right=.9in, bottom=1.2in, top=1in,dvips]{geometry}
\makeatletter

\AtBeginDocument{\setlength\abovedisplayskip{5pt plus 2pt minus 1pt}
\setlength\belowdisplayskip{5pt plus 2pt minus 1pt}}
\usepackage[utf8]{inputenc}
\usepackage{lmodern}
\usepackage{indentfirst}
\usepackage[affil-it]{authblk}
\usepackage{rotating}
\usepackage{tocloft}
\usepackage{titlesec}
 \titleformat{\subparagraph}[hang]{\normalfont}{\thesubparagraph}{0pt}{\underline}
 \titleformat{\paragraph}[hang]{\normalfont}{\theparagraph}{0pt}{\myuline}
\usepackage{titletoc}

\usepackage[textsize=scriptsize]{todonotes} 
\makeatletter
\renewcommand{\todo}[2][]{\tikzexternaldisable\@todo[#1]{#2}\tikzexternalenable}
\makeatother

\usepackage{standalone}
\usepackage{subfiles}

\usepackage[hidelinks]{hyperref}
\usepackage{enumitem}
\setlist[enumerate]{itemsep=2.0pt plus 1.0 pt minus 0.5pt, topsep=4.0pt plus 2.0 pt minus 1.0pt}
\setlist[itemize]{itemsep=2.0pt plus 1.0 pt minus 0.5pt, topsep=4.0pt plus 2.0 pt minus 1.0pt}
\usepackage{array}
\newcolumntype{M}[1]{>{\centering\arraybackslash}m{#1}}
\usepackage{multirow}
\usepackage{tikz}
\usetikzlibrary{external}
\usepackage{calc}
\usepackage{pgfplots}
\pgfplotsset{compat=1.16}
\usepackage{graphicx}
\usepackage{subcaption}
\usepackage[labelfont=bf]{caption}
\usepackage{epsfig}
\usepackage{color, soul}

\usepackage{contour}
\usepackage[normalem]{ulem}

\contourlength{1pt}

\newcommand{\myuline}[1]{%
  \uline{\phantom{#1}}%
  \llap{\contour{white}{#1}}%
}

\usepackage{marginnote}

\usepackage{color,soulutf8}

\usepackage{amsmath, amsthm, slashed}
\usepackage[capitalize]{cleveref}
\usepackage{amssymb}
\usepackage{bbm}
\usepackage{bm}
\usepackage{stackengine}
\usepackage{marvosym}
\usepackage{mathtools}
\usepackage{mathrsfs}
\usepackage{upgreek}
\usepackage{accents}
\usepackage{textcomp}
\usepackage{slashed} %for gradient on spheres
\usepackage{centernot}
\usepackage{cancel}
\usepackage{harmony} %for musical symbols
\usepackage{wasysym} %for clock
\wasyfamily
\DeclareFontShape{U}{wasy}{b}{n}{ <-10> ssub * wasy/m/n
 <10> <10.95> <12> <14.4> <17.28> <20.74> <24.88>wasyb10 }{}
\DeclareFontShape{U}{wasy}{b}{it}{ <-10> ssub * wasy/m/n
 <10> <10.95> <12> <14.4> <17.28> <20.74> <24.88>wasyb10 }{}
\DeclareMathAlphabet\mathbfcal{OMS}{cmsy}{b}{n}
\newcommand\numberthis{\addtocounter{equation}{1}\tag{\theequation}}
\renewcommand{\Re}{\operatorname{Re}}
\renewcommand{\Im}{\operatorname{Im}}
\DeclareMathOperator{\sign}{sign}
\allowdisplaybreaks

\numberwithin{equation}{section}

\theoremstyle{plain}

\newtheorem{alphtheorem}{Theorem}

\newtheorem{arabictheorem}{Theorem}
\newtheorem{theorem}{Theorem}[section]
\newtheorem*{theorem*}{Theorem}
 
\newtheorem*{conjecture*}{Conjecture}
\newtheorem{corollary}{Corollary}[section]
\newtheorem*{corollary*}{Corollary}
\newtheorem{proposition}{Proposition}[section]
\newtheorem{lemma}[proposition]{Lemma}

\theoremstyle{definition}
\newtheorem{definition}{Definition}[subsection]
\newtheorem{remark}{Remark}[subsection]
\newtheorem*{remark*}{Remark}

\usepackage{csquotes}
\newcommand{\mb}[1]{\mathbb{#1}}

\newcommand{\mc}[1]{\mathcal{#1}}
\newcommand{\mr}[1]{\mathrm{#1}}

\newcommand{\p}{\partial}
\newcommand{\lp}{\left}
\newcommand{\rp}{\right}
\newcommand\supp{\mathrm{supp}}
\newcommand{\dbtilde}[1]{\widetilde{\raisebox{0pt}[0.85\height]{$\widetilde{#1}$}}}

\newlength{\dhatheight}

\newcommand{\uL}{\underline{L}}

\newcommand{\swei}[2]{#1^{[{#2}]}}
\newcommand{\sweie}[2]{#1^{[{#2}],\,\epsilon}}

\usepackage{pifont}
\newcommand{\cut}{\text{\ding{33}}}
\newcommand{\cutt}{\text{\ding{35}}}
\newcommand{\cutr}{\text{\rotatebox[origin=c]{90}{\ding{34}}}}

\newcommand{\sml}[2]{#1_{ml}^{[{#2}],\,a,\,\omega}}
\newcommand{\smlk}[3]{#1_{({#3}),\,ml}^{[{#2}],\,a,\,\omega}}
\newcommand{\smlambda}[2]{#1_{m\Lambda}^{[{#2}],\,a,\,\omega}}
\newcommand{\smlambdak}[3]{#1_{({#3}),\,m\Lambda}^{[{#2}],\,a,\,\omega}}

\usepackage{totcount}

\newtotcounter{citnum} %From the package documentation
\begin{document}
 \title{\LARGE \textbf{Boundedness and decay for the Teukolsky equation \\ on Kerr in the full subextremal range $|a|<M$:\\ physical space analysis}}

\author[1,2]{{\Large Yakov Shlapentokh-Rothman}}
\author[3,4,5]{{\Large  Rita \mbox{Teixeira da Costa}\vspace{0.4cm}}}

\affil[1]{\small  University of Toronto, Department of Mathematics, 40 St. George Street, Toronto, ON, Canada  }
\affil[2]{\small  University of Toronto Mississauga, Department of Mathematical and Computational Sciences, 3359 Mississauga Road, Mississauga, ON, Canada \vspace{0.2cm} }

\affil[3]{\small  Princeton University, Department of Mathematics, Washington Road, Princeton, NJ 08544,
United States }

\affil[4]{\small  Princeton Gravity Initiative, Jadwin Hall, Washington Road, Princeton, NJ 08544,
United States }

\affil[5]{\small 
University of Cambridge, Department of Pure Mathematics and Mathematical Statistics, Wilberforce Road, Cambridge CB3 0WA, United Kingdom}

\date{\today}

\maketitle

\begin{abstract}
This paper concludes the study, initiated by the authors in \cite{SRTdC2020}, of the Teukolsky equation of spin $\pm 1$ and spin $\pm 2$ on Kerr backgrounds in the full subextremal range of parameters $|a| < M$. In our previous \cite{SRTdC2020}, we obtained uniform-in-frequency estimates for the ODEs governing separable solutions to the Teukolsky equation. In this paper, by adapting the techniques developed by the first author with Dafermos and Rodnianski for scalar waves, we show that our ODE estimates can be upgraded to estimates for the Teukolsky PDE. In particular, we conclude the proof that solutions of the Teukolsky equation on subextremal Kerr arising from regular initial data remain bounded and decay in time. 
\end{abstract}

\bigskip\bigskip

\tableofcontents

\section{Introduction}

Kerr black holes are a family of 4-dimensional Lorentzian manifolds, solving the Einstein vacuum equations
\begin{equation}
\mathrm{Ric}(g)=0\,, \label{eq:EVE-intro}
\end{equation}
which is parametrized by $M>0$ and $|a|\leq M$. Establishing the stability of the subextremal $|a|<M$ subfamily to perturbations under \eqref{eq:EVE-intro} has been a longstanding open question which, in recent years, has seen a great deal of progress. Following the linear stability analysis of \cite{Dafermos2016a}, Dafermos, Holzegel, Rodnianski and Taylor proved  in \cite{Dafermos2021} the nonlinear stability of the Schwarzschild, $a=0$, subfamily in full codimension, i.e.\ for perturbations which eventually radiate away all angular momentum but are not assumed to verify any (axi)symmetry conditions. Soon afterwards, Giorgi, Klainerman and Szeftel \cite{Klainerman2019,Giorgi2020b,Klainerman2021,Klainerman2022,Giorgi2022} established the nonlinear stability of Kerr for $|a|\ll M$. Their work uses results of Shen \cite{Shen2022}, nonlinear and geometric insights from Klainerman and Szeftel's earlier proof of nonlinear stability of Schwarzschild for polarized axisymmetric perturbations \cite{Klainerman2017}, and builds on linear insights in \cite{Andersson2015,Dafermos2017, Dafermos2016a, Ma2017a}.  These milestones come after several decades of intense work in the problem, and we refer the reader to the introduction of our previous \cite{SRTdC2020} for a more complete overview of the literature.

A common theme to the two aforementioned works is that they crucially rely on a solid understanding of the Teukolsky equation \cite{Teukolsky1973}
\begin{align}
\begin{split}
\bigg[\Box_g &+\frac{2s}{\rho^2}(r-M)\p_r +\frac{2s}{\rho^2}\lp(\frac{a(r-M)}{\Delta}+i\frac{\cos\theta}{\sin^2\theta}\rp)\p_\phi\\
&+\frac{2s}{\rho^2}\lp(\frac{M(r^2-a^2)}{\Delta}-r -ia\cos\theta\rp)\p_t +\frac{1}{\rho^2}\lp(s-s^2\cot^2\theta\rp)\bigg]\upalpha^{[s]}  =0\,,
\end{split}\label{eq:Teukolsky-intro}
\end{align}
which, for spin  $s=\pm 2$, describes the behavior of the extreme curvature components of perturbations of Kerr black holes under the \textit{linearization} of \eqref{eq:EVE-intro} in the algebraically special frame associated to Kerr spacetimes. In \eqref{eq:Teukolsky-intro}, $\Box_g$ is the covariant scalar wave operator, $\swei{\upalpha}{s}$ is an $s$-spin weighted function (see already Definition~\ref{def:smooth-spin-weighted}) and we have used Boyer--Lindquist coordinates $(t,r,\theta,\phi)\in\mathbb{R}\times(r_+,\infty)\times\mathbb{S}^2$ with the usual conventions 
\begin{align*}
\Delta:=(r-r_+)(r-r_-)\,, \quad r_\pm :=M\pm \sqrt{M^2-a^2}\,,\qquad \rho^2=r^2+a^2\cos^2\theta\,.
\end{align*}
The importance of \eqref{eq:Teukolsky-intro} to the previous works on stability of the $a=0$ and $|a|\ll M$ family is easily understood: while the system of linearized gravity around Kerr contains many coupled equations and many different geometric unknowns, the curvature components described by $\swei{\upalpha}{\pm 2}$ have the particular feature of being gauge-invariant and solving the decoupled equations \eqref{eq:Teukolsky-intro}. Thus, the first step towards establishing stability of Kerr outside the $|a|\ll M$ perturbative regime is to establish that, for $|a|<M$, solutions to \eqref{eq:Teukolsky-intro} arising from suitably regular initial data enjoy boundedness and decay properties which are compatible with the nonlinear stability of the $|a|<M$ Kerr family. In the present paper we complete the proof, initiated in our previous \cite{SRTdC2020}, of precisely this result:
\begin{alphtheorem} \label{thm:bddness-decay-big} Fix $s\in\{0,\pm 1,\pm 2\}$, $M>0$ and $|a|<M$. On the exterior of a subextremal Kerr black hole spacetime with parameters $(a,M)$, general solutions to the Teukolsky equation \eqref{eq:Teukolsky-intro} arising from sufficiently regular initial data on a Cauchy surface 
\begin{itemize}[noitemsep]
\item have uniformly bounded energy fluxes through a suitable spacelike foliation of the black hole exterior, through the future event horizon, $\mc{H}^+$, and through future null infinity, $\mc{I}^+$, in terms of an energy flux of initial data \underline{at the same level of regularity}; 
\item satisfy a suitable version of ``integrated local energy decay'' with loss of derivatives at trapping;
\item and satisfy similar statements for higher-order energies.
\end{itemize}
Moreover, solutions to \eqref{eq:Teukolsky-intro} remain uniformly bounded pointwise and, in fact, decay at a suitable inverse polynomial rate, with estimates which depend only on $M$, $|a|$, $s$ and a suitable, higher order, initial data quantity.
\end{alphtheorem}

In the above theorem, $s=0,\pm 1$ represent, respectively, linear scalar and electromagnetic waves on a fixed Kerr background. The statements are novel for $s=\pm 1,\pm 2$, but not in the case of $s=0$ which had been previously understood by the first author with Dafermos and Rodnianski in \cite{Dafermos2016b}, see also \cite{Holzegel2023}. We note also the very recent work \cite{Millet2023} by Millet on the Teukolsky equation in the full subextremal range which establishes detailed asymptotics for solutions but does not provide quantitative energy boundedness statements. Theorem~\ref{thm:bddness-decay-big} should not be seen as a statement of full linear stability for the subextremal Kerr family, as such a statement would require additionally establishing energy flux boundedness  \textit{without derivative loss}, and decay, for the many remaining geometric quantities in the linearization  of \eqref{eq:EVE-intro} (after subtracting the Kerr values) in an appropriate gauge. We direct the reader to the thesis \cite{Benomio2020} and recent preprint  \cite{Benomio2022} by Benomio for a proposal for such a gauge, and to \cite{Benomio2022a} for a proof of concept in a linear stability problem. A method for establishing decay for the system of linearized gravity around Kerr  from the estimates given in Theorem~\ref{thm:bddness-decay-big} has also been provided by Andersson, B\"ackdahl, Blue and Ma \cite{Andersson2019,Andersson2022}; see also \cite{Andersson2022b,Hafner2019} for a different approach. In the nonlinear setting, in the aforementioned works, Giorgi, Klainerman and Szeftel also proposed a method for analyzing the stability of the Kerr family, see \cite{Giorgi2020b} and references therein, which they implemented successfully for $|a|\ll M$. While many different gauges might in principle allow one to design a scheme to control the gauge-\textit{dependent} quantities in linearized gravity, closing the estimates requires, as an input, having precise control of the gauge-\textit{invariant} Teukolsky equations as given in Theorem~\ref{thm:bddness-decay-big}. We highlight in particular that  not losing derivatives in our energy boundedness statement is essential for  eventually proving an orbital linear stability statement for the subextremal Kerr black hole family.\footnote{We note that in view of the Nash--Moser inverse function theorem, it is, in principle, possible to establish certain nonlinear stability results for the Kerr spacetime even if the linear theory loses derivatives. For this type of argument to work, it is crucial to obtain the corresponding linear estimates on all spacetimes sufficiently close to the Kerr spacetime. We direct the interested readers to the works~\cite{Hintz2016,Hintz2020} where such a strategy is successfully carried out in the context of slowly rotating Kerr-de Sitter spacetimes and for Minkowski space.}

In analyzing \eqref{eq:Teukolsky-intro}, it is convenient to pass to an auxiliary system of equations introduced by Dafermos, Holzegel and Rodnianski in \cite{Dafermos2016a,Dafermos2017} in the $|s|=2$ case, after a frequency space transformation of Chandrasekhar \cite{Chandrasekhar1975a}. For $s\in\mathbb{Z}$ and $0\leq k\leq |s|$, let $\swei{\upphi}{s}_{(k)}$ be obtained by taking $k$ appropriately $r$-weighted null derivatives of $\swei{\upalpha}{s}$, i.e.\
\begin{align}
\swei{\upphi}{s}_{0}=j_{s,|s|}^{-1}(r)\swei{\upalpha}{s}\,;\qquad \swei{\upphi}{s}_{k+1} = j_{s,|s|-k}^{-1}(r)\mc{L} \swei{\upphi}{s}_{k}\,,\,\,\, k=0,...,|s|-1\,, \label{eq:def-upppsi-k-intro-basic}
\end{align}
where $\mc{L}$ denotes a null vector field and $j_{s,k}$ are appropriate functions of $r$. Then, for $k<|s|$, $\swei{\upphi}{s}_{k}$ satisfies
\begin{equation}
\swei{\mathfrak{R}}{s}_k \swei{\upphi}{s}_k = W(r)\swei{\upphi}{s}_{k+1}+\sum_{j=0}^{k-1}\swei{\mathfrak{J}}{s}_{k,j}\lp[\swei{\upphi}{s}_{j}, \p_\phi\swei{\upphi}{s}_{j} \rp]+\swei{\mathfrak H}{s}_k\,,\label{eq:transformed-equation-bottom-intro-basic}
\end{equation} 
where the last term is absent if $k=0$, and $\swei{\Phi}{s}:=\swei{\upphi}{s}_{|s|}$ satisfies
\begin{equation}
\swei{\mathfrak{R}}{s} \swei{\Phi}{s} = a \sum_{k=0}^{|s|-1}\swei{\mathfrak{J}}{s}_{k}\lp[\swei{\upphi}{s}_{k}, \p_\phi\swei{\upphi}{s}_{k} \rp]+\swei{\mathfrak H}{s}\,,\label{eq:transformed-equation-top-intro-basic}
\end{equation}
where $W(r)$ is a function of $r$, $\swei{\mathfrak{J}}{s}_{k}$ and $\swei{\mathfrak{J}}{s}_{k,j}$ are linear operators depending only on the coordinate $r$, $\swei{\mathfrak{R}}{s}$ can be interpreted as a wave operator for spin-weighted functions and $\swei{\mathfrak{R}}{s}_k$ differs from it only through a  $\p_\phi$ term and a  zeroth order term with good $r$ decay, and we have added inhomogeneities $\swei{\mathfrak H}{s}_k$, $\swei{\mathfrak H}{s}:=\swei{\mathfrak H}{s}_{|s|}$ for generality. In the $a=0$ and perturbative $|a|\ll M$ settings, the direct analysis of the Teukolsky equation~\eqref{eq:Teukolsky-intro} is commonly replaced by the analysis of the mixed hyperbolic-transport system comprising the single wave equation \eqref{eq:transformed-equation-top-intro-basic} and the $|s|$ transport equations \eqref{eq:def-upppsi-k-intro-basic}. In broad strokes, one applies scalar wave techniques to  \eqref{eq:transformed-equation-top-intro-basic} to obtain estimates for $\swei{\Phi}{s}$ up to small $O(a)$ coupling errors to $\swei{\upphi}{s}_{k}$, $k<|s|$; the latter are treated with transport estimates on \eqref{eq:def-upppsi-k-intro-basic}. Once estimates for the the top level variable $\swei{\Phi}{s}$ have been closed, estimates for $\swei{\alpha}{s}$ are obtained by integrating the transport equations \eqref{eq:def-upppsi-k-intro-basic}. Thus, in the $|a|\ll M$ case, one can completely bypass a direct analysis of \eqref{eq:Teukolsky-intro} or \eqref{eq:transformed-equation-bottom-intro-basic}. In the general subextremal $|a|<M$ case, as the coupling on the right hand side of \eqref{eq:transformed-equation-top-intro-basic} is no longer weak, it is not clear how to avoid the latter equations. Instead, our approach is to replace---or, rather, supplement---\eqref{eq:Teukolsky-intro} with the mixed hyperbolic-transport system comprising \textit{all} the $|s|$ wave equations \eqref{eq:transformed-equation-bottom-intro-basic} and \eqref{eq:transformed-equation-top-intro-basic} together with the $|s|$ transport equations \eqref{eq:def-upppsi-k-intro-basic}. The upshot is that we naturally do not lose control over any derivatives of $\swei{\upalpha}{s}$, as one expects to with transport estimates: in particular, the first result of Theorem~\ref{thm:bddness-decay-big} can be more rigorously stated as control over a (weighted) full $H^{|s|+1}$ norm of $\swei{\upalpha}{s}$ on the hypersurfaces of a hyperboloidal foliation in terms of the same (weighted) full $H^{|s|+1}$ norm of Cauchy data for $\swei{\upalpha}{s}$, see already Theorem~\ref{thm:main}.

In our analysis of the wave equations \eqref{eq:transformed-equation-bottom-intro-basic}  and \eqref{eq:transformed-equation-top-intro-basic}, we combine frequency and physical space methods. Our previous work on the subject considered separable solutions of \eqref{eq:transformed-equation-bottom-intro-basic}--\eqref{eq:transformed-equation-top-intro-basic}: we took
\begin{align*}
    \mathfrak{H}^{[s],\,\omega}_{k, \, ml}(t,r,\theta,\phi)&= e^{-i\omega t}e^{im\phi}S_{ml}^{[s],\,a\omega}(\theta)\smlk{\mathfrak G}{s}{k}(r)\,, \numberthis\label{eq:mode-ansatz-H-k}\\
\upphi^{[s]\,,\omega}_{k, \, ml}(t,r,\theta,\phi)&=e^{-i\omega t}e^{im\phi}S_{ml}^{[s],\,a\omega}(\theta)\smlk{\uppsi}{s}{k}(r)\,, \numberthis\label{eq:mode-ansatz-upphi-k}
\end{align*}
with $\sml{\mathfrak G}{s}:=\smlk{\mathfrak G}{s}{|s|}$, $\sml{\Psi}{s}:=\smlk{\uppsi}{s}{|s|}$ and $\sml{\upalpha}{s}:=j_{s,|s|}(r)\smlk{\uppsi}{s}{0}$. Here, for each $l$, $e^{im\phi}S_{ml}^{[s],\,a\omega}$ are the so-called spin-weighted spheroidal harmonics.  The functions $\smlk{\uppsi}{s}{k}(r)$ solve linear second order ODEs, the radial ODEs
\begin{align*}
\begin{split}
\lp(\smlk{\uppsi}{s}{k}\rp)''+(\omega^2-\swei{\mc V}{s}_k)\smlk{\uppsi}{s}{k}&= W(r)\smlk{\uppsi}{s}{k+1}+\sum_{j=0}^{k-1}\swei{\mathfrak{J}}{s}_{k,j}\lp[\smlk{\uppsi}{s}{j}, im\smlk{\uppsi}{s}{j} \rp]+\smlk{\mathfrak{G}}{s}{k}\,,\\
\lp(\sml{\Psi}{s}\rp)''+(\omega^2-\swei{\mc V}{s})\sml{\Psi}{s}&=a\sum_{k=0}^{|s|-1}\swei{\mathfrak{J}}{s}_{k,j}\lp[\smlk{\uppsi}{s}{k}, im\smlk{\uppsi}{s}{k} \rp]+\sml{\mathfrak{G}}{s}\,,
\end{split} \numberthis\label{eq:radial-ODEs-intro}
\end{align*}
where $'$ denotes a derivative with respect to a rescaled radial coordinate, $\swei{\mc V}{s}_k$, $\swei{\mc V}{s}$ are short range potentials. They are subject to so-called \textit{outgoing boundary conditions}, which are conditions at $r=r_+$ and $r=\infty$ which ensure that \eqref{eq:mode-ansatz-upphi-k} are regular at the future event horizon and have finite energy on a hyperboloidal spacelike hypersurface. 

In order for us to discuss the inhomogeneous \eqref{eq:radial-ODEs-intro}, it is useful to begin by considering their homogeneous versions, where $\smlk{\mathfrak G}{s}{k}\equiv 0$. For $\omega\neq 0$, denote by $\uppsi^{[s]}_{(0),\, ml,\, \mc H^+}$ and $\uppsi^{[s]}_{(0),\, ml,\, \mc I^+}$ normalized solutions of the \textit{homogeneous} radial ODE \eqref{eq:radial-ODEs-intro} for which the corresponding solutions to \eqref{eq:Teukolsky-intro} are, respectively, regular at the future event horizon and regular at the future null infinity. Building on the seminal result of Whiting \cite{Whiting1989}, it was shown in \cite{Shlapentokh-Rothman2015,Andersson2017} that these solutions are linearly independent. Thus, their Wronskian is bounded away from zero in any compact range of frequencies excluding $\omega\neq 0$. The two authors have shown that this bound can be made \textit{explicit} in terms of the black hole and frequency parameters:
\begin{arabictheorem}[\cite{Shlapentokh-Rothman2015,TeixeiradaCosta2019}] Let $\mc{A}$ be any set of real frequencies  with $\sup_{\mathcal{A}} \lp(|\omega|+|\omega|^{-1}+|m|+|\Lambda|\rp)<\infty$.
Then, for any $|a|\leq a_0<M$, the Wronskian of $\uppsi^{[s]}_{(0),\, ml,\, \mc H^+}$ and $\uppsi^{[s]}_{(0),\, ml,\, \mc I^+}$ can be bounded away from zero with a constant which depends \textbf{explicitly} on $s$, $M$, $a_0$ and $\mc A$. 
\label{thm:mode-stability-intro} 
\end{arabictheorem}

Theorem~\ref{thm:mode-stability-intro} can be directly applied to the quantitative analysis of the inhomogenous \eqref{eq:radial-ODEs-intro}, as it allows us to define a Green's function for the latter which depends explicitly on the frequency and black hole parameters, see already Section~\ref{sec:bulk-terms}.

In our previous joint manuscript \cite{SRTdC2020}, we considered the complement of $\mc A$. Our main result was
\begin{arabictheorem}[\cite{SRTdC2020}] There exists a choice of ${\mc A}$ such that, for any $|a|\leq a_0<M$, $\smlk{\uppsi}{s}{k}$ verify $L^2_{r\in (r_+,\infty)}$ quantitative estimates with loss of derivatives at trapping, i.e.\ at $r=r_{\rm trap}(\omega,m,l)$, which are  \textbf{uniform} in $(\omega,m,l)\in\mc A^c$ and depend explicitly on $|s|$, $M$,  $a_0$ and the inhomogeneities.
\label{thm:high-low-freqs-intro}
\end{arabictheorem}
Theorem~\ref{thm:high-low-freqs-intro} was obtained in \cite{SRTdC2020} by a careful analysis of the structure of the radial ODEs \eqref{eq:radial-ODEs-intro}. Our proof combined the methods developed by the first author with Dafermos and Rodnianski for the case $s=0$  in \cite{Dafermos2016a} with several new insights for the difficulties associated with $s\neq 0$, most notably associated with the Teukolsky--Starobinsky identities and an ellipticity property, for $k<|s|$, at trapping. We direct the reader to the introduction of \cite{SRTdC2020} for a detailed explanation. Let us note, however, that in addition to the derivative loss problem that occurs for very large frequencies at $r=r_{\rm trap}$ mentioned above, an additional major difficulty that of both Theorems~\ref{thm:mode-stability-intro} and \ref{thm:high-low-freqs-intro} must overcome is the existence of superradiant frequencies
\begin{align*}
\omega\lp(\omega-\frac{am}{2M(M+\sqrt{M^2-a^2})}\rp)<0
\end{align*}
for which the conservation laws for the radial ODEs~\eqref{eq:radial-ODEs-intro} induced by the Killing vector fields cease to be coercive. This is the frequency space manifestation of the physical space \textit{ergoregion}, a region of size $O(|a|)$ in the exterior of rotating Kerr black holes where the stationary Killing field becomes spacelike.

Theorems~\ref{thm:mode-stability-intro} and \ref{thm:high-low-freqs-intro} have several possible applications, as outlined in our previous \cite{SRTdC2020}. In the present paper we have focused on the application to obtaining boundedness and decay in time in the forward evolution problem, i.e.\ the application to establishing Theorem~\ref{thm:bddness-decay-big}. To understand the role of Theorems~\ref{thm:mode-stability-intro} and \ref{thm:high-low-freqs-intro} in achieving that goal, recall that if a general solution to the transformed system \eqref{eq:def-upppsi-k-intro-basic}--\eqref{eq:transformed-equation-top-intro-basic} arising from suitably regular data could be expressed as an infinite superposition of separable solutions,
\begin{align*}
\swei{\upphi}{s}_k(t,r,\theta,\phi)=\int\sum_{ml}\uppsi^{[s],\,\omega}_{k, \, ml}(t,r,\theta,\phi) d\omega\,, \numberthis \label{eq:infinite-superposition}
\end{align*} 
and the same is true for the inhomogeneities 
\begin{align*}
\swei{\mathfrak H}{s}_k(t,r,\theta,\phi)=\int\sum_{ml}\mathfrak{H}^{[s],\,\omega}_{k, \, ml}(t,r,\theta,\phi) d\omega\,,\numberthis \label{eq:infinite-superposition-2}
\end{align*}
then Theorems \ref{thm:mode-stability-intro} and \ref{thm:high-low-freqs-intro} would immediately imply, by Plancherel's theorem and the completeness of $e^{im\phi}S_{ml}^{[s],\,a\omega}(\theta)$ as a basis for spin-weighted functions on the sphere, an integrated local energy decay estimate, as in Theorem~\ref{thm:bddness-decay-big}, but in terms of an infinite superposition of contributions from the inhomogeneities $\smlk{\mathfrak G}{s}{k}$ rather than initial data. A more involved argument would show that an analogous statement holds for the energy fluxes. In fact, the more precise versions of Theorem~\ref{thm:high-low-freqs-intro} in our previous \cite{SRTdC2020} in fact implies that this contribution can be rewritten in terms of (derivatives of) the physical space inhomogeneities $\swei{\mathfrak H}{s}_k$.  

The goal of the present paper is to construct appropriate inhomogeneous versions of the system \eqref{eq:def-upppsi-k-intro-basic}--\eqref{eq:transformed-equation-top-intro-basic} so that not only \eqref{eq:infinite-superposition} and \eqref{eq:infinite-superposition-2} hold but also the aforementioned contributions from $\swei{\mathfrak H}{s}_k$ can be controlled in terms of a suitable norm on the Cauchy data. Thus, Theorem~\ref{thm:high-low-freqs-intro}, the main result from our previous \cite{SRTdC2020}, is the heart of the proof of Theorem~\ref{thm:bddness-decay-big} in the present paper. While Theorem~\ref{thm:high-low-freqs-intro} required some genuinely new insights, the methods in this work to pass from these frequency space estimates to physical space estimates follow from those of \cite{Dafermos2016a} after minor, if technically challenging, adaptations which we now describe.

To start, it is useful to identify the obstacles to the representation \eqref{eq:infinite-superposition}: \textit{a priori}, solutions to the homogeneous versions of \eqref{eq:transformed-equation-bottom-intro-basic} and \eqref{eq:transformed-equation-top-intro-basic} which arise from smooth Cauchy data do not possess enough integrability in time to admit Fourier transforms, neither to the past of the initial Cauchy hypersurface nor, unless $a=0$, due to the existence of a ergoregion where the conserved energy ceases to be coercive, to its the future. This is an issue that we can easily correct: by applying suitable time cutoffs to $\swei{\upphi}{s}_k$, we can force them to be $L^2$ in time so that \eqref{eq:infinite-superposition} holds; the cutoffs in time produce {inhomogeneities} $\swei{\mathfrak{H}}{s}_k$ in the transformed wave equations~\eqref{eq:transformed-equation-bottom-intro-basic}--\eqref{eq:transformed-equation-top-intro-basic} and, thus, {inhomogeneities} $\smlk{\mathfrak{G}}{s}{k}$ in the radial ODEs~\eqref{eq:radial-ODEs-intro}. A bit more concretely, we will consider cutoffs which vanish in the future of some hypersurface $\Sigma_\tau$  and  to the past of another hypersurface $\Sigma_0$, and whose gradient is supported around these two hypersurfaces, see Figure~\ref{fig:hyp-folliations}. We will show that, after Plancherel and summing in $(m,l)$,
\begin{align*}
I^{\rm deg}(0,\tau)\lesssim E(\tau)+E(0)\,,\numberthis\label{eq:non-future-int-intro}
\end{align*}
where $I^{\rm deg}$  and $E$ are suitably weighted Sobolev norms for $\swei{\upphi}{s}_k$, $k=0,\dots, |s|$: the former is a spacetime norm which experiences loss of one derivative at trapping, and the latter denotes a flux across a spacelike hypersurface as pictured in Figure~\ref{fig:hyp-folliations}. Comparing to the $s=0$ case in \cite{Dafermos2016b}, the key differences in our proof of \eqref{eq:non-future-int-intro} are that we consider hyperboloidal folliations instead of asymptotically flat, and that the complications introduced by having a \textit{system} of wave equations connected by transport equations~\eqref{eq:def-upppsi-k-intro-basic} leads us to  consider a somewhat nonlocal construction of some of our cutoff variables. We direct the reader to Section~\ref{sec:physical-space-estimates}, where \eqref{eq:non-future-int-intro} is shown, for further comments.

\begin{figure}[htbp]
\centering
\includegraphics[scale=1]{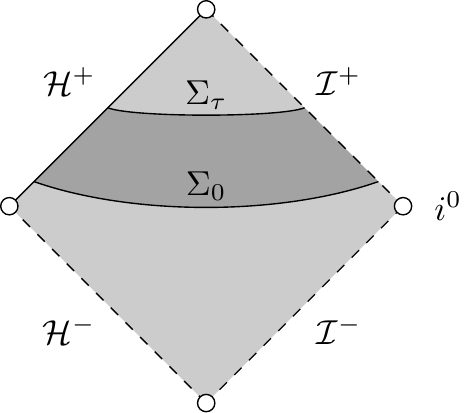}
\caption{Hyperboloidal folliations for cutoffs to the past and to the future.}
\label{fig:hyp-folliations}
\end{figure}

If we had an \textit{a priori} energy flux boundedness statement for \eqref{eq:transformed-equation-bottom-intro-basic}--\eqref{eq:transformed-equation-top-intro-basic}, estimate \eqref{eq:non-future-int-intro} would imply the integrated local energy decay statement of Theorem~\ref{thm:bddness-decay-big}. Standard arguments for obtaining higher order estimates and, from there, decay estimates by the $r^p$ method of Dafermos and Rodnianski \cite{Dafermos2010a} (see also \cite{Moschidis2016}) would then conclude the proof of Theorem~\ref{thm:bddness-decay-big}. However, we can \textit{a priori} hope for no better than
\begin{align*}
E(\tau) \lesssim E(0) + |a| I^{\rm deg}(0,\tau)\numberthis\label{eq:non-energy-bddness-intro}
\end{align*} 
from energy estimates, and in fact the second term will not in general have a degeneration at trapping unless $|a|\ll M$ is sufficiently small. If $|a|\ll M$,  we can treat the last term as an error which can be absorbed by \eqref{eq:non-future-int-intro}, thus closing the flux and, \textit{a posteriori}, integrated estimates of Theorem~\ref{thm:bddness-decay-big} and, by the aforementioned standard estimates, establishing the decay properties of solutions to \eqref{eq:Teukolsky-intro}. We direct to the original proofs in \cite{Dafermos2010,Dafermos2016a,Dafermos2017,Ma2017,Ma2017a} for details.

In the full subextremal case $|a|<M$, the bulk term contribution to the right hand side of \eqref{eq:non-energy-bddness-intro} is not, in any sense, \textit{a priori} small. From \eqref{eq:non-future-int-intro}, we can nevertheless softly deduce an interesting fact. If we choose to restrict to solutions of the transformed system \eqref{eq:def-upppsi-k-intro-basic}--\eqref{eq:transformed-equation-top-intro-basic} which are assumed to be $L^2_{\tau\geq 0}$, i.e.\ which have enough integrability in time to the future of the Cauchy hypersurface on which we pose data, then we may take $\tau\to \infty$ in \eqref{eq:non-future-int-intro} and drop the second term on the right hand side. This establishes integrated local energy decay as in Theorem~\ref{thm:bddness-decay-big}. After having established this, and only after, we can revisit our energy estimates leading to \eqref{eq:non-energy-bddness-intro} and show, by a careful argument crucially exploiting the fact that our Theorem~\ref{thm:high-low-freqs-intro} in \cite{SRTdC2020} shows that trapping is very well-localized in frequency space, the boundedness of energy fluxes. It is worth emphasizing again how the logic of the proof differs from the case $a=0$ and the perturbative $|a|\ll M$ regime: there, as we have outlined above, one first shows an \textit{a priori} (almost) energy boundedness statement \eqref{eq:non-energy-bddness-intro} first, and local energy decay statement later by \eqref{eq:non-energy-bddness-intro}; here, in the full $|a|<M$ range, this order is inverted. We also note this is already the case in the $s=0$ setting of \cite{Dafermos2016b}, and the results for $s\neq 0$ here follow in the same fashion with little conceptual differences. The upshot is that we now have:
\begin{arabictheorem} \label{thm:bddness-decay-future-int-intro} Theorem~\ref{thm:bddness-decay-big} holds for solutions of the transformed system \eqref{eq:def-upppsi-k-intro-basic}--\eqref{eq:transformed-equation-top-intro-basic} which are assumed to be \textit{a priori} $L^2_{\tau\in(0,\infty)}$, where $\tau$ is the continuous time parameter identifying hypersurfaces in the spacelike folliation. It also holds unconditionally if we assume $|a|\ll M$ is small enough.
\end{arabictheorem}

The final stage of our proof, as in the case $s=0$ in \cite{Dafermos2016b} is to show that the restriction in integrability we have made in Theorem~\ref{thm:bddness-decay-future-int-intro} is, in fact, no restriction at all:
\begin{arabictheorem} \label{thm:continuity-arg-intro}
Solutions to the transformed system \eqref{eq:def-upppsi-k-intro-basic}--\eqref{eq:transformed-equation-top-intro-basic} arising from  Cauchy data for $\swei{\upalpha}{s}$ with the level of regularity given in   Theorem~\ref{thm:bddness-decay-big} are, in fact, necessarily $L^2_{\tau\in(0,\infty)}$ where $\tau$ is the continuous time parameter identifying hypersurfaces in the spacelike folliation.
\end{arabictheorem}
Let us note that Theorem~\ref{thm:bddness-decay-future-int-intro} already implies Theorem~\ref{thm:continuity-arg-intro} for the Schwarzschild $a=0$ black hole background or, indeed, on  very slowly rotating $|a|\ll M$ Kerr exterior spacetimes. Thus, to obtain Theorem~\ref{thm:continuity-arg-intro}  we need only show that for each $M>0$, the set
\begin{align*}
\mathbb A :=\{|a|\in[0,M)\colon &\text{all solutions to \eqref{eq:def-upppsi-k-intro-basic}--\eqref{eq:transformed-equation-top-intro-basic} on Kerr exterior with parameters $(a,M)$ }\\
& \text{arising from regular data are $L^2_{\tau\in(0,\infty)}$}\}
\end{align*}
is open and closed. It is not hard to see from \eqref{eq:non-future-int-intro} that the future integrability properties we seek are ensured if we can prove $\sup_{\tau>0} E(\tau)<\infty$. Thus, closedness of $\mathbb{A}$ follows easily from our energy flux boundedness in Theorem~\ref{thm:bddness-decay-future-int-intro}. For openness, as in \cite{Dafermos2016b}, we find it convenient to use the completeness of the azimuthal modes to reduce the proof to the fixed azimuthal number $m$ setting. For finite $m$, revisiting the energy estimates leading to \eqref{eq:non-energy-bddness-intro} yet again, we can replace the right hand side by a bulk term still, but one which is lower order in the number of derivatives (though now the constant will, of course, depend on $m$). This derivative gain argument is similar to that in \cite{Dafermos2016b}, with a slightly different critical insight: while \cite{Dafermos2016b} appeals to the $r$-localization of trapping region for scalar waves with fixed $m$ to conclude, we rely on the more general fact, established in the precise version of  Theorem~\ref{thm:high-low-freqs-intro} from \cite{SRTdC2020}, that trapping need not cost us control over one full derivative but only half of a derivative (in fact, even this loss is not sharp). Afterwards, as in \cite{Dafermos2016b}  for $s=0$, we fix $M>0$ and study solutions to a system of the type of \eqref{eq:transformed-equation-bottom-intro-basic}--\eqref{eq:transformed-equation-top-intro-basic} where the operators which interpolate between the differential operators of \eqref{eq:transformed-equation-bottom-intro-basic}--\eqref{eq:transformed-equation-top-intro-basic} for Kerr parameter $\mathring{a}\in\mathbb A$ to the future of $\Sigma_\tau$ and another Kerr parameter $a$ such that $|a-\mathring{a}|\ll 1$ in the slab between $\Sigma_0$ and $\Sigma_{\tau-\delta_0}$.  This allows us to derive an integrated and energy flux estimate for the solutions to the system with parameter $a$ where the right hand side ``sees'' the better estimates we have for parameter $\mathring{a}$; making $\delta_0$ sufficiently small and then $|a-\mathring{a}|$ very small as well, our improved \eqref{eq:non-energy-bddness-intro} allows us to conclude that $\mathring{a}\in\mathbb{A}$ implies $a\in\mathbb{A}$, so $\mathbb{A}$ is open.

We conclude this introduction with a brief outline of the paper:
\begin{itemize}
\item Sections \ref{sec:kerr-prelims}, \ref{sec:pdes-big}, \ref{sec:toolbox-physical-space}, \ref{sec:precise-statement} and  \ref{sec:odes-big} are preliminary sections which establish some basic results and notation:  Section~\ref{sec:kerr-prelims} introduces the Kerr spacetime and relevant folliations;  Sections~\ref{sec:pdes-big} and  \ref{sec:odes-big} introduce the main differential equations of study in this paper, namely the PDEs and the corresponding (radial) ODEs, respectively; Section~\ref{sec:toolbox-physical-space} contains some preliminary energy flux and Morawetz estimates as well as a general procedure for obtaining higher order estimates. We also point out that Section~\ref{sec:precise-statement} contains a precise formulation of Theorem~\ref{thm:bddness-decay-big}, using the energy norms defined in Section~\ref{sec:template-energy-norms}.
\item Section~\ref{sec:frequency-estimates-review} is a brief section which reviews the frequency estimates obtained in our previous \cite{SRTdC2020}, i.e.\ gives a more precise version of Theorem~\ref{thm:high-low-freqs-intro}. These estimates are the core of our proof.
\item Sections~\ref{sec:physical-space-estimates} is concerned with passing from the frequency space estimates of Theorems~\ref{thm:mode-stability-intro} and \ref{thm:high-low-freqs-intro} to physical space estimates when the transformed variables have been cutoff using, respectively, hyperboloidal and radial cutoffs. In particular, we prove a precise version of \eqref{eq:non-future-int-intro} with and without the $L^2_{\tau>0}$ assumption on solutions. Comparing with the $s=0$ case of \cite{Dafermos2016b}, this is the section of the present paper where one finds the most novel, though technical, ideas.
\item Section~\ref{sec:proof-main-thm} contains the proofs Theorems~\ref{thm:bddness-decay-future-int-intro} and \ref{thm:continuity-arg-intro}. These proofs follow closely those of \cite{Dafermos2016b}, and rely on the higher order estimates in Section~\ref{sec:toolbox-physical-space}  as well the $r^p$ method of \cite{Dafermos2010a} to obtain the decay statements. 
\end{itemize}

\vspace{\baselineskip}

\noindent\textbf{Acknowledgments.} During the period in which this research was carried out, Y.\ S.\ was supported by NSF grant DMS-1900288, by an Alfred P.\ Sloan Fellowship in Mathematics, and by NSERC discovery grants RGPIN-2021-02562 and DGECR-2021-00093. R.\ TdC.\ acknowledges support through the EPSRC doctoral grant EP/L016516/01 and the NSF award DMS-2103173, as well as the hospitality of Princeton University in Fall 2018 and 2019 through their VSRC program. Both authors thank Mihalis Dafermos for many important discussions.

%---------------
%\part{Preliminaries}
%---------------

\section{The Kerr spacetime}
\label{sec:kerr-prelims}

In this section, we introduce some basic notation to discuss analysis on the Kerr manifold: in  \ref{sec:geometry}, we describe the Kerr family as family of metrics on a single manifold-with-boundary; in Section~\ref{sec:relevant-vector fields} we set the notation for some of the relevant vector fields on Kerr; and in Section~\ref{sec:foliations} we introduce the family of foliations generated by the flow of the stationary Killing field which we will consider in our study of PDEs on Kerr. We conclude the section with a short list of conventions and notation which will be used throughout.

\subsection{Manifold, coordinate systems and metric}
\label{sec:geometry}

Following the notation of \cite{Dafermos2010,Dafermos2016b}, we first define the Kerr exterior manifold as the manifold-with-boundary $$\mc R:=\mathbb{R}^+_0\times\mathbb{R}\times \mathbb{S}^2$$ covered by coordinates $y^*\in\mathbb{R}^+_0$, $t^*\in\mathbb{R}$ and $(\theta^*,\phi^*)\in\mathbb{S}^2$ the standard spherical coordinates. The coordinate system $(y^*,t^*,\theta^*,\phi^*)$  is global apart for the usual degeneration of spherical coordinates. The event horizon is the boundary $$\mc H^+:=\partial\mc R=\{y^*=0\}.$$ Note that the manifold is, thus, independent of the parameters $(a,M)$.

Next, we introduce Kerr-star coordinates, $(r,t^*,\theta^*,\phi^*)$, on $\mc R$ by letting $r=y^*+r_+$, where
$$r_\pm :=M\pm\sqrt{M^2-a^2}.$$
These coordinates now depend on the parameters $(a,M)$. 

We may also introduce Boyer--Lindquist coordinates  $(r,t,\theta,\phi)$ on $\mathrm{int}(\mc R)$ by letting $\theta=\theta^*$, and $t=t^*-\overline{t}(r)$, $\phi=\phi^*-\overline{\phi}(r)$, where $\overline{t}(r)$ and $\overline{\phi}(r)$ are smooth functions for $r\in(r_+,\infty)$. One possible choice, consistent with \cite{Dafermos2010,Dafermos2016b}, is to make  $\overline{t}$, $\overline{\phi}$ interpolate between the zero function for $r\geq 9M/4$, and the functions
\begin{align*}
\frac{d\overline{t}}{dr}=\frac{2M r}{(r-r_+)(r-r_-)}&\Rightarrow \overline{t}(r)=
\begin{dcases}
\frac{2Mr_+}{r_+-r_-}\log(r-r_+)-\frac{2Mr_-}{r_+-r_-}\log(r-r_-)+C,&|a|<M\\
-\frac{2M^2}{r-M}+2M\log(r-M)+C,\ |a|=M
\end{dcases}\,,\\
\frac{d\overline{\phi}}{dr}=\frac{a}{(r-r_+)(r-r_-)}&\Rightarrow \overline{\phi}(r)=
\begin{dcases}
\frac{a}{r_+-r_-}\log\lp(\frac{r-r_+}{r-r_-}\rp)+C,&|a|<M\\
-\frac{a}{r-M}+C, &|a|=M
\end{dcases}\,,
\end{align*}
holding for $r\leq 15M/8$ with a suitable choice of integration constant $C$. Let us also introduce  the tortoise coordinate $r^*=r^*(r)$ as an alternative to the Boyer--Lindquist $r$ which is defined by 
\begin{align*}
\frac{dr^*}{dr}=\frac{r^2+a^2}{(r-r_+)(r-r_-)}\,,\quad r^*(3M)=0\,.
\end{align*}

In Boyer--Lindquist coordinates, the Kerr family is a two parameter family of metrics on $\mathrm{int}(\mc R)$: for parameters $(a,M)$, it is given by 
\begin{align}
g_{a,M} &= -\frac{\Delta -a^2\sin^2\theta}{\rho^2}dt^2- \frac{4Mar\sin^2\theta}{\rho^2} dtd\phi \nonumber \\ 
&\qquad+\left(\frac{(r^2+a^2)^2-\Delta a^2 \sin^2\theta}{\rho^2}\right)\sin^2\theta d\phi^2 +\frac{\rho^2}{\Delta}dr^2+\rho^2 d\theta^2\,, \label{eq:kerr-metric-BL}
\end{align}
with inverse
\begin{align}
g_{a,M}^{-1} &= -\frac{(r^2+a^2)^2-a^2\Delta \sin^2\theta}{\rho^2\Delta}\partial_t^2-\frac{4Mar}{\rho^2\Delta}\partial_t\partial_\phi  +\frac{\Delta-a^2\sin^2\theta}{\rho^2\Delta\sin^2\theta}\partial_\phi^2+\frac{1}{\rho^2}\partial_\theta^2 +\frac{\Delta}{\rho^2}\partial_r^2\,. \label{eq:kerr-metric-inverse-BL}
\end{align}
where we have used the notation 
\begin{gather*}
\rho^2 := r^2+a^2\cos^2\theta, \quad\quad \Delta:=r^2-2Mr+a^2=(r-r_+)(r-r_-) \,. \numberthis\label{eq:rho-delta}
\end{gather*}
Changing to Kerr-star coordinates, we obtain a metric on the entire $\mc R$; for instance, if $r_+\leq r\leq 15M/8$, then we can rewrite the metric
\begin{align}
g_{a,M} &= -\lp(1-\frac{2Mr}{\rho^2}\rp)(dt^*)^2+\frac{4Mr}{\rho^2}dt^*dr + \lp(1+\frac{2Mr}{\rho^2}\rp)dr^2  -\frac{2a(2Mr+\rho^2)\sin^2\theta}{\rho^2}drd\phi^*\nonumber\\ 
&\qquad- \frac{4aMr\sin^2\theta}{\rho^2}dt^*d\phi^* +\frac{(r^2+a^2)^2-\Delta a^2\sin^2\theta}{\rho^2}\sin^2\theta(d\phi^*)^2\,, \label{eq:kerr-metric-star}
\end{align}
and its inverse
\begin{align}\label{eq:kerr-metric-star-inverse}
g_{a,M}^{-1} &=-\lp(1+\frac{2Mr}{\rho^2}\rp)\p_{t^*}^2+\frac{4Mr}{\rho^2}\p_{t^*}\p_{r^*}+\frac{\Delta}{\rho^2}\p_r^2+\frac{2a}{\rho^2}\p_r\p_{\phi^*}+\frac{1}{\rho^2}\p_\theta^2+\frac{1}{\rho^2\sin^2\theta}\p_{\phi^*}^2\,.
\end{align}

We also note the volume form in various different coordinates: we write
\begin{align*}
d\sigma = \sin\theta d \theta d\phi
\end{align*}
to be the usual volume form on $\mathbb{S}^2$, making the volume form of $g_{a,M}$
\begin{align*}
dV=\rho^2dtdrd\sigma=\frac{\rho^2~\Delta}{r^2+a^2}dtdr^*d\sigma = \rho^2 dt^*drd\sigma.
\end{align*}

\subsection{Relevant vector fields}
\label{sec:relevant-vector fields}

In $\mathrm{int}(\mc R)$, we define the vector field $X$ by the formula $X:=\p_r$, in Boyer--Lindquist coordinates. We have 
\begin{align*}
\p_{r^*}=\frac{\Delta}{r^2+a^2}X\,,
\end{align*}
and we let $'$ be a shorthand for a $\p_{r^*}$ derivative.

In $\mc{R}$, we define the vector fields $T, X^*,\Phi$ by the following formulas in Kerr-star coordinates:
\begin{align*}
T:=\p_{t^*}\,, \qquad Z := \p_{\phi^*}\,,\qquad X^* := \p_r\,,
\end{align*}
and we note that the first two are the same as in the $(a,M)$-independent coordinate chart. The restriction of these vector fields to $\mathrm{int}(\mc R)$ is given, in terms of Boyer--Lindquist coordinates, by
\begin{align*}
T=\p_{t}\,, \qquad Z = \p_{\phi}\,,\qquad X^* = X-\frac{d\overline{t}}{dr}T-\frac{d\overline{\phi}}{dr}Z\,,
\end{align*}
where $\overline{t}(r)$ and $\overline{\phi}(r)$ are identified in the previous section.

In $\mc{R}$, we define the vector fields
\begin{align*}
L&:=\frac{\Delta}{r^2+a^2}X^*+\lp(1+\frac{d\overline{t}}{dr*}\rp)T+\lp(\frac{a}{r^2+a^2} +\frac{d\overline{\phi}}{dr*}\rp)Z\,, \\
\uL &:= -\frac{\Delta}{r^2+a^2}X^*+\lp(1-\frac{d\overline{t}}{dr*}\rp)T+\lp(\frac{a}{r^2+a^2} -\frac{d\overline{\phi}}{dr*}\rp)Z\,,
\end{align*}
with respect to Kerr-star coordinates. In $\mathrm{int}(\mc R)$, using Boyer-Lindquist coordinates, they take on the more familiar form
\begin{align*}
L=\frac{\Delta}{r^2+a^2}\p_r+T+\frac{a}{r^2+a^2}Z\,, \qquad \uL=-\frac{\Delta}{r^2+a^2}\p_r+T+\frac{a}{r^2+a^2}Z\,,
\end{align*}
or, using the $r^*$ variable in lieu of the Boyer--Lindquist $r$,
\begin{align*}
L=\p_{r^*}+T+\frac{a}{r^2+a^2}Z\,, \qquad \uL=-\p_{r^*}+T+\frac{a}{r^2+a^2}Z\,.
\end{align*}
Let us also note the identities
\begin{align*}
\nabla_a\lp(L^a\frac{1}{\rho^2}\frac{r^2+a^2}{\Delta}\rp)=\nabla_a\lp(\uL^a\frac{1}{\rho^2}\frac{r^2+a^2}{\Delta}\rp)=0\,. \numberthis\label{eq:commutation-mcL-volume-form}
\end{align*}

\subsection{Spacelike foliations}
\label{sec:foliations}

\subsubsection{Hyperboloidal foliations}
\label{sec:hyp-folliation}

In Kerr-star coordinates $(r,t^*,\theta,\phi^*)$, let $\tilde{\zeta}_1=1$ for $r_+\leq r \leq r_++M/2$ and 0 for $r\geq 3M$ and $\tilde{\zeta}_2=0$ for $r \leq 3M$ and 1 for $r\geq 4M$. Set 
\begin{align*}
\tilde{t}^*= t^*+\tilde{\zeta}_1(r)\lp(r-\frac{M}{2}\log r\rp)-\tilde{\zeta}_2(r)\lp(r+2M\log r-\frac{3M^2}{r}\rp)
\end{align*}
and define
\begin{align*}
\Sigma_\tau:=\{\tilde{t}^*=\tau\}\,.
\end{align*}
We note that, since $t^*$ is independent of the Kerr parameters and $\tilde{t}^*$ depends only on the mass, $M$, of the black hole, the hypersurfaces $\Sigma_\tau$ are the same for the entire family $|a|\leq M$.

\begin{lemma} For $\tau\geq 0$, the $\Sigma_\tau$ are spacelike hypersurfaces which pierce $\mc H^+$ and are asymptotically null for large $r$. Moreover,
\begin{align*}
dV_{\Sigma_\tau}=v(r,\theta)drd\sigma\quad \text{hence}\quad dV=\rho^2dV_{\Sigma_\tau}d\tau
\end{align*}
for some smooth $v(r,\theta)$ which is bounded above and below by constants depending only on $M$.
\end{lemma}
\begin{proof} Let $f=f(r)$, and write $n_{\Sigma_\tau}=\nabla(t^*-f)$. Then, simple calculations give 
\begin{align*}
g\lp(n_{\Sigma_\tau},n_{\Sigma_\tau}\rp)&= 
\begin{dcases}
-1-\frac{2Mr(r^2+a^2)}{\Delta\rho^2}+\frac{\Delta}{\rho^2}\lp(\frac{df}{dr}\rp)^2\,, \quad&r\geq 9M/4\\
-1-\frac{2Mr}{\rho^2}-\frac{4Mr}{\rho^2}\frac{df}{dr}+\frac{\Delta}{\rho^2}\lp(\frac{df}{dr}\rp)^2\,, \quad&r\leq 15M/8\\
\end{dcases}\,,\\
g\lp(n_{\Sigma_\tau},L\rp)&=
\begin{dcases}
1-\frac{\Delta}{r^2+a^2}\frac{df}{dr}\,, \quad &r\geq 9M/4\,,\\
1+\frac{2M r}{r^2+a^2}-\frac{\Delta}{r^2+a^2}\frac{df}{dr}\,, \quad &r\leq 15M/4\,,
\end{dcases}\,,\\
g\lp(n_{\Sigma_\tau},\frac{r^2+a^2}{\Delta}\uL\rp)&=\begin{dcases}
\frac{r^2+a^2}{\Delta}+\frac{df}{dr}\,, \qquad &r\geq 9M/4\,,\\
1+\frac{df}{dr}\,, \quad &r\leq 15M/8\,,
\end{dcases}\,,\\
\det g_{\Sigma_\tau}&=\rho^2 \sin^2\theta\begin{dcases}
w^{-1}-\Delta\lp(\frac{df}{dr}\rp)^2 -a^2\sin^2\theta\,, \quad &r\geq 9M/4\,,\\
\rho^2+2Mr\lp(1+2\frac{df}{dr}\rp)-\Delta\lp(\frac{df}{dr}\rp)^2\,, \quad &r\leq 15M/8\,,
\end{dcases}\,.
\end{align*}

In the limit $r\to \infty$, since $f(r)=r+2M\log r+\frac{-3M^2}{r}$,  one obtains
\begin{gather*}
g\lp(n_{\Sigma_\tau},n_{\Sigma_\tau}\rp)=-\frac{M^2(2-\frac{a^2}{M^2}\sin^2\theta)}{r^2} +O(r^{-3})\,, \quad g\lp(n_{\Sigma_\tau},L\rp)=\frac{M^2}{r^2} +O(r^{-3})\,, \\
\det g_{\Sigma_\tau} = \rho^2\sin^2\theta\lp[2M^2-a^2\sin^2\theta + O(r^{-1})\rp]\,.
\end{gather*}
On the other hand, in the limit $r\to r_+$, since $f(r)=-r+\frac{M}{2}\log r$, one obtains 
\begin{gather*}
g\lp(n_{\Sigma_\tau},n_{\Sigma_\tau}\rp)=-\frac{2M^2-\sin^2\theta}{2Mr_+-a^2\sin^2\theta}+O(r-r_+)\,, \qquad g\lp(n_{\Sigma_\tau},\frac{r^2+a^2}{\Delta}\uL\rp)=\frac{M}{2r_+}+O(r-r_+)\,, \\
\det g_{\Sigma_\tau} = \rho^2\sin^2\theta\frac{4 a^2 r^2 \cos ^2\theta +M r \left(2 M^2-M r+4 r^2\right)-a^2 (3r-M) (r-M)}{4r^2}\,.
\end{gather*}

This concludes the proof.
\end{proof}

\subsubsection{Asymptotically flat foliations}

Though our proof relies mainly on the hyperboloidal hypersurfaces introduced in the previous section, we will also find it convenient to sometimes compare with asymptotically flat foliations, which pierce $\mc H^+$ but  terminate at $i^0$. With this in mind, we define 
\begin{align*}
\tilde{\Sigma}_\tau:=\{t^*=\tau\}\,.
\end{align*}

\begin{remark} In other works, $\tilde{\Sigma}_\tau$ typically denote hyperboloidal folliations, and ${\Sigma}_\tau$ denote asymptotically flat folliations. However, as we will more often use the former, we have decided to reverse the notation.
\end{remark}

\subsection{Parameters and conventions}

Throughout the paper, we rely on the notation
\begin{align}
w:=\frac{\Delta}{(r^2+a^2)^2}\,,
\end{align}
for an $r$-weight that will be heavily used; recall $\Delta$ is given in \eqref{eq:rho-delta}.

In our estimates, we use $B$ to denote possibly large positive constants and $b$ to denote possibly small positive constants depending only on $M>0$. Whenever the constant depends, additionally, on another parameter that has not yet been fixed, say $x$, we write $B(x)$ or $b(x)$; we revert to $B$ and $b$ once it has been fixed. We also note the algebra of constants
$$B+B=BB=B\,, \qquad b+b=bb=b\,, B+b=B\,, \qquad Bb =B\,, \qquad b^{-1}=B\,,\text{~etc}.$$

\section{The Teukolsky equation and the transformed system of PDEs}
\label{sec:pdes-big}

In this section, we introduce the main partial differential equations of study in the present paper. We start by introducing the spin-weighted formalism, in Section~\ref{sec:spin-weighted-formalism} which allows us to make sense of the PDEs we will introduce in Section~\ref{sec:PDEs}. We then introduce the weighted Sobolev norms we will use to study them in Section~\ref{sec:template-energy-norms}.

\subsection{Spin-weighted formalism}
\label{sec:spin-weighted-formalism}

\subsubsection{Smooth spin-weighted functions}

Let $(\theta,\phi^*)$ denote standard spherical coordinates in the sphere $\mathbb{S}^2$, and consider the vector fields 
\begin{gather}
\begin{gathered} \label{eq:Zi}
\tilde{Z}_1 = -\sin \phi^* \partial_\theta + \cos \phi^*  \left( -is \csc \theta - \cot \theta \partial_{\phi^*}\right) \,, \\
\tilde{Z}_2 = - \cos \phi^* \partial_\theta - \sin \phi^*  \left( -is \csc \theta - \cot \theta \partial_{\phi^*}\right) \,, \quad 
\tilde{Z}_3 = \partial_{\phi^*} \, .
\end{gathered}
\end{gather}

We recall for the reader the notion of spin-weighted function, see for instance \cite[Section 2.2.1]{Dafermos2017}.
\begin{definition}[Smooth spin-weighted functions] \label{def:smooth-spin-weighted}
Fix $s\in\frac12\mathbb{Z}$.

\begin{enumerate}[label=(\roman*)]
\item Let $f\colon \mathbb{S}^2\to \mathbb{C}$. We say $f$ is a {\normalfont smooth $s$-spin-weighted function on $S^2$}, i.e.\ $f\in \swei{\mathscr{S}}{s}_\infty$,  if, for any $k_1,k_2,k_3 \in \mathbb{N}_0$, 
$$(\tilde{Z}_1)^{k_1} (\tilde{Z}_2)^{k_2} (\tilde{Z}_3)^{k_3} f\colon \mathbb{S}^2\to \mathbb{C}$$ 
is a function of $(\theta,\phi^*)$ which is smooth for $\theta\neq 0,\pi$ and 
$$e^{is\phi}(\tilde{Z}_1)^{k_1} (\tilde{Z}_2)^{k_2} (\tilde{Z}_3)^{k_3} f \text{~and~}e^{-is\phi}(\tilde{Z}_1)^{k_1} (\tilde{Z}_2)^{k_2} (\tilde{Z}_3)^{k_3} f$$ 
extend continuously to, respectively, $\theta=0$ and $\theta=\pi$.  \label{it:smooth-spin-weighted-S2}

\item Now, fix $M>0$ and $|a|\leq M$. Consider a function $f\colon \mc R\to \mathbb{C}^2$. We say $f$ is a {\normalfont smooth $s$-spin-weighted function on $\mc{R}$}, i.e.\ $f\in\mathscr{S}_\infty^{[s]}(\mathcal{R})$, if, for any $k_4,k_5 \in \mathbb{N}_0$, the function
$$(\p_{t^*})^{k_4} (\p_{r})^{k_5} f\Big|_{(t^*,r)=(T^*,R)}$$
is a smooth $s$-spin-weighted function on $\mathbb{S}^2$ for any $T^*\in\mathbb{R}$ and $R\in[r_+,\infty)$. \label{it:smooth-spin-weighted-R}
\end{enumerate}
\end{definition}

\subsubsection{Spin-weighted spherical laplacian and spinorial gradient}

For any $s\in\frac{1}{2}\mb{Z}$, the operator 
\begin{align} \label{eq:spin-weighted-laplacian}
\mathring{\slashed\triangle}^{[s]}&=  -\frac{1}{\sin \theta} \frac{\partial}{\partial \theta} \left(\sin \theta \frac{\partial}{\partial \theta}\right) - \frac{1}{\sin^2 \theta} \partial_\phi^2 - 2s i\frac{ \cos \theta}{\sin^2 \theta} \partial_\phi + s^2 \cot^2 \theta \\
&=-\tilde{Z}_1^2-\tilde{Z}_2^2-\tilde{Z}_3^2-s-s^2 \nonumber
\end{align}
is a smooth operator on $\mathscr{S}_\infty^{[s]}$, which we call the spin-weighted laplacian. It induces a representation of every function in $\mathscr{S}_\infty^{[s]}$ by spin-weighted spherical harmonics
\begin{lemma}[Smooth spin-weighted spherical harmonics] \label{lemma:spherical-angular-ode}
Fix $s\in\frac12\mathbb{Z}$. On $\mathscr{S}_\infty^{[s]}$, the operator $\mathring{\slashed\triangle}^{[s]}$ has a complete countable set of eigenfunctions, $e^{im\phi}S_{ml}^{[s],\,(0)}$, with corresponding eigenvalues $\bm\uplambda_{ml}^{[s],\,(0)}$, which are indexed by parameters $m$ and $l$ chosen to satisfy $m-s\in\mathbb{Z}$ and  $l\geq \max\{|s|,|m|\}$, and $\bm\uplambda^{[-s],\,(0)}_{ml}=l(l+1)-s^2$. The eigenvalues also have the following symmetry
\begin{align}
\bm\uplambda^{[s],\,(0)}_{ml}=\bm\uplambda^{[-s],\,(0)}_{ml}=l(l+1)-s^2\,. \label{eq:lambda-symmetries-0}
\end{align}
The functions $e^{im\phi}S_{ml}^{[s],\,(\nu)}$, also referred to as spin-weighted spherical harmonics, form a complete orthogonal basis of $\mathscr{S}_\infty^{[s]}$.
\end{lemma}

Let us now introduce the spinorial gradient
\begin{align} \label{eq:spin-gradient}
\mathring{\slashed{\nabla}}^{[s]}=( \partial_\theta\,,~ 
\partial_\phi + is \cos \theta)\,,
\end{align}
we can easily derive, for $\Xi\,,\Pi\in\mathscr{S}_\infty^{[s]}$,
\begin{equation}
\int_0^\pi\int_0^{2\pi}d\phi d\theta \sin\theta \lp(\swei{\mathring{\slashed{\triangle}}}{s}\Xi\,\rp)\overline{\Pi} =\int_0^\pi\int_0^{2\pi}d\phi d\theta \sin\theta \lp[\swei{\mathring{\slashed{\nabla}}}{s}\Xi\cdot \overline{\swei{\mathring{\slashed{\nabla}}}{s}\Pi}\rp]_{\mb{S}^2}\, \label{eq:IBP-spin-weighted-laplacian}
\end{equation}
In light of Lemma~\ref{lemma:spherical-angular-ode}, we have the following properties for $\mathring{\slashed{\nabla}}^{[s]}$:
\begin{lemma}[Spinorial gradient] Let $s\in\mathbb{Z}$ and $\swei{\mathring{\slashed{\nabla}}}{s}$ be as defined in \eqref{eq:spin-gradient}. Then, for any $\Xi\in \mathscr{S}_\infty^{[s]}$, one has the Poincaré inequality
\begin{equation}
\int_0^\pi\int_0^{2\pi} \lp|\swei{\mathring{\slashed{\nabla}}}{s}\Xi\rp|^2\sin\theta  d\theta d\phi \geq |s| \int_0^\pi\int_0^{2\pi} \lp|\Xi\rp|^2\sin\theta  d\theta d\phi \label{eq:Poincare}
\end{equation}
and, moreover,
\begin{equation}
\int_0^\pi\int_0^{2\pi} \lp|\swei{\mathring{\slashed{\nabla}}}{s}\Xi\rp|^2\sin\theta  d\theta d\phi \geq \min\lp\{1,\frac{1}{|s|}\rp\} \int_0^\pi\int_0^{2\pi} \lp|\Phi\Xi\rp|^2\sin\theta  d\theta d\phi\,. \label{eq:Poincare-2}
\end{equation}
\end{lemma}

\begin{proof} Since spin-weighted spheroidal harmonics form a complete basis of  $\mathscr{S}_\infty^{[s]}$, see Lemma~\ref{lemma:spherical-angular-ode}, if $\Xi\in \mathscr{S}_\infty^{[s]}$, there are some real numbers $\{c_{ml}\}_{ml}$ such that $\Xi(\theta,\phi)=\sum_{ml}c_{ml}S_{ml}^{[s]}(\theta)e^{im\phi}$. Using \eqref{eq:IBP-spin-weighted-laplacian}, we obtain
\begin{align*}
\int_0^\pi\int_0^{2\pi} \lp|\swei{\mathring{\slashed{\nabla}}}{s}\Xi\rp|^2\sin\theta  d\theta d\phi&=\int_0^\pi\int_0^{2\pi} \lp(\swei{\mathring{\slashed{\triangle}}}{s}\Xi\rp)\overline{\Xi}\sin\theta  d\theta d\phi \\
&=\int_0^\pi\int_0^{2\pi} \lp(\sum_{ml}\bm\uplambda_{ml}^{[s]}c_{ml}S_{ml}^{[s]}(\theta)e^{im\phi}\rp)\overline{\sum_{m'l'}c_{m'l'}S_{m'l'}^{[s]}(\theta)e^{im'\phi}}\sin\theta  d\theta d\phi \,. \numberthis\label{eq:Poincare-intermediate}
\end{align*}
Using the orthogonality of the spin-weighted spherical harmonics and  the trivial inequalities
\begin{align*}
\bm\uplambda_{ml}^{[s]}\geq |s|\,, \text{~~and, if $m\neq0$,~~}\bm\uplambda_{ml}^{[s]}\geq m^2\lp\{\begin{array}{ll}
1\,, &s=0\\
|s|^{-1}\,, &s\neq 0
\end{array}\rp\}\,,
\end{align*}
 we can now complete the proof.
\end{proof}

\subsubsection{Spin-weighted spheroidal laplacian}

Rather than using \eqref{eq:spin-weighted-laplacian}, we will often find it more convenient to rely on a modified version of the spin-weighted laplacian, given by
\begin{gather}
\mathring{\slashed{\triangle}}^{[s],\,\nu}=\mathring{\slashed{\triangle}}^{[s]}-\nu^2\cos^2\theta+2\nu s \cos\theta\,,
 \label{eq:spin-weighted-laplacian-nu}
\end{gather}
which is also a smooth operator on $\swei{\mathscr{S}}{s}_\infty(\mc{R})$. The generalization of Lemma~\ref{lemma:spherical-angular-ode} to this case is given next.

\begin{lemma}[Smooth spin-weighted spheroidal harmonics] \label{lemma:spheroidal-angular-ode}
Fix $s\in\frac12\mathbb{Z}$. For each $\nu\in\mathbb{R}$, on $\mathscr{S}_\infty^{[s]}$, the operator $\mathring{\slashed{\triangle}}^{[s]}_\nu$ has a countable set of eigenfunctions, $e^{im\phi}S_{ml}^{[s],\,(\nu)}$, with corresponding eigenvalues $\bm\uplambda_{ml}^{[s],\,(\nu)}$. The eigenvalues are smooth in $\nu$ and they, together with the corresponding eigenfunctions, are indexed by parameters $m$ and $l$ chosen to satisfy $m-s\in\mathbb{Z}$,  $l\geq \max\{|s|,|m|\}$, and  $\bm\uplambda_{ml}^{[s],\,(0)}=l(l+1)-s^2$. We also not the symmetries
\begin{align}
\bm\uplambda^{[s],\,(\nu)}_{ml}=\bm\uplambda^{[-s],\,(\nu)}_{ml}=\bm\uplambda^{[-s],\,(-\nu)}_{-m,\,l}\,. \label{eq:lambda-symmetries}
\end{align}
The functions $e^{im\phi}S_{ml}^{[s],\,(\nu)}$, also referred to as spin-weighted spheroidal harmonics, form a complete orthogonal basis of $\mathscr{S}_\infty^{[s]}$.
\end{lemma}

An additional remarkable property of spin-weighted spheroidal harmonics is that they verify so-called Teukolsky--Starobinsky identities \cite{Teukolsky1974,Starobinsky1974,Kalnins1989}.

\begin{lemma}[Angular TS]\label{prop:TS-angular-constant} Fix $s\in\{0,\frac12,1,\frac32,2\}$. Let $(m,l)$ be a pair with the constraints  in Lemma~\ref{lemma:spheroidal-angular-ode} and $\nu\in\mathbb{R}$.  Define
\begin{align*}
\hat{\mc{L}}^{\pm}_n &:= \frac{d}{d\theta}\pm\lp(\frac{m}{\sin\theta}-\nu\cos\theta\rp)+n\cot\theta\,.
\end{align*}
The spin-weighted spheroidal harmonics of Lemma with spin $\pm s$ are eigenfunctions of the operator 
\begin{align*}
\lp(\prod_{j=0}^{2s-1}\hat{\mc{L}}_{s-j}^\mp\rp)\lp(\prod_{k=0}^{2s-1}\hat{\mc{L}}_{s-k}^\pm\rp) \equiv \lp(\sin\theta\rp)^{2s}\lp(\frac{\hat{\mc{L}}_{s}^\mp}{\sin\theta}\rp)^{2s}\lp(\sin\theta\rp)^{2s}\lp(\frac{\hat{\mc{L}}_{s}^\pm}{\sin\theta}\rp)^{2s}\,, 
\end{align*}
with indices $j,k$ increasing from right to left on the product, and the latter being replaced by the identity if $s=0$, for the same eigenvalue. This eigenvalue, $\mathfrak B_s=\mathfrak B_s(|s|,\nu,m,l)$, called the {\normalfont angular Teukolsky--Starobinsky constant}, satisfies $(-1)^{2s}\mathfrak B_s>0$, and can be computed explicitly: for instance,
\begin{align}
 \label{eq:TS-angular-constants}
\begin{split}
\mathfrak B_{1}&=(\bm\Lambda_{ml}^{[s],\,\nu}-2m\nu+1)^2+4m\nu-4\nu^2\,,\\
\mathfrak B_2&= \lp[(\bm\Lambda_{ml}^{[s],\,\nu}-2m\nu+2)(\bm\Lambda_{ml}^{[s],\,\nu}-2m\nu+4)\rp]^2+40\nu(\bm\Lambda_{ml}^{[s],\,\nu}-2m\nu+2)^2(m-\nu)\\
&\qquad+48\nu(\bm\Lambda_{ml}^{[s],\,\nu}-2\nu+2)(m+\nu)\,.
\end{split}
\end{align}
\end{lemma}

To conclude, let us note that $\mathring{\slashed{\triangle}}^{[s],\,\nu}$ introduced in \eqref{eq:spin-weighted-laplacian-nu} can be obtained from the operator
\begin{equation}
\mathring{\slashed\triangle}^{[s]}+a^2\cos^2\theta TT -2ias\cos\theta T = \mathring{\slashed\triangle}^{[s]}_T +a^2TT+2a TZ
\end{equation}
on $\swei{\mathscr{S}}{s}_\infty(\mc{R})$ by replacing $-iaT$ with $\nu$, where we have introduced yet another operator on $\swei{\mathscr{S}}{s}_\infty(\mc{R})$ given by
\begin{equation}
\label{eq:spin-weighted-laplacian-modified}
\begin{split}
\mathring{\slashed{\triangle}}^{[s]}_T&= \mathring{\slashed\triangle}^{[s]} -\lp(2aTZ +a^2\sin^2\theta TT -2ias\cos\theta T\rp)\,.\\
%&=\mathring{\slashed\triangle}^{[s]}+\frac{a^2(\rho^2+r^2+a^2)}{(r^2+a^2)^2}Z^2-\frac{a\rho^2}{r^2+a^2}(L+\uL)Z-\frac{1}{4}a^2\sin^2\theta(L+\uL)^2+ias \cos\theta(L+\uL)-\frac{2ia^2s}{r^2+a^2}\cos\theta Z\,.
\end{split}
\end{equation}

\subsection{The relevant PDEs}
\label{sec:PDEs}

The present paper will be, for the most part, concerned with the analysis of a mixed transport-hyperbolic system:
\begin{definition}[Transformed system] \label{def:transformed-system}
Fix $s\in\mathbb{Z}$. Write $\mathcal{L}=L$ if $s<0$, $\mathcal{L}=\underline{L}$ if $s>0$. We say that, for $k\in\{0,\dots |s|\}$, the functions $\swei{\upphi}{s}_{k}$  are solutions to the transformed system with inhomogeneity $\swei{\mathfrak H}{s}_k$ if $\Delta^{|s|-k}\swei{\upphi}{s}_{k}\in\mathscr S_\infty^{[s]}(\mc R)$, $\Delta^{|s|-k-1}\swei{\mathfrak{H}}{s}_{k}\in S_\infty^{[s]}(\mc R)$, and they satisfy the following PDEs.
\begin{itemize}
\item \textit{Transport equations}: if $s\neq 0$, 
\begin{align}
\swei{\upphi}{s}_{k}=\frac{1}{w}\mc{L}\swei{\upphi}{s}_{k-1}\,, \label{eq:transformed-transport}\\
\swei{\mathfrak{H}}{s}_{k}=\mc{L}\lp(\frac{\swei{\mathfrak H}{s}_{k-1}}{w}\rp)\,. \label{eq:transformed-transport-inhom}
\end{align}
for $k=1,\dots |s|$.

\item \textit{Hyperbolic equations}: 
\begin{equation}
\swei{\mathfrak{R}}{s}_{k}\swei{\upphi}{s}_{k}=\swei{\mathfrak{H}}{s}_{k}+\sign s (|s|-k)\frac{w'}{w}\swei{\upphi}{s}_{k+1}+a\sum_{i=0}^{k-1}\swei{\mathfrak{J}}{s}_{k,i} \,, \label{eq:transformed-k}
\end{equation}
 where we have defined 
\begin{align}
\swei{\mathfrak{R}}{s}_{k}&:=\frac12 \lp(L\underline{L}+\underline{L}L\rp)+w\lp[\mathring{\slashed{\triangle}}^{[s]}_T+|s|+k(2|s|-k-1)\rp] \nonumber\\
&\qquad-\frac{4arw}{r^2+a^2}\sign{s}\lp(|s|-k\rp)Z
+ \frac{a^2\Delta^2}{(r^2+a^2)^4}\lp[1-2|s|-2k(2|s|-k-1)\rp] \label{eq:transformed-R-k}\\
&\qquad+\frac{2Mr(r^2-a^2)\Delta}{(r^2+a^2)^4}\lp[1-3|s|+2s^2-3k(2|s|-k-1)\rp]\,,\nonumber
\\
\swei{\mathfrak{J}}{s}_{k,i}&:=w\lp(c_{s,\,k,\,i}^{Z}(r)Z+ c_{s,\,k,\,i}^{\mr{id}}(r)\rp)\swei{\upphi}{s}_{i}\,.\label{eq:transformed-J-k}
\end{align}
In the above, $c_{s,\,k,\,i}^\Phi$ and $c_{s,\,k,\,i}^\mr{id}$, $i=0,...,|s|$ are explicit functions of $r$ which have the following properties:
\begin{itemize}
\item  $a c_{s,\,k,\,i}^\mr{id}$ should be replaced by $c_{s,\,k,\,i}^\mr{id}$ if and only if $|s|\neq 1$, $k=1$ and $i=0$;
\item $c_{s,\,k,\,i}^\mr{id}$ and $c_{s,\,k,\,i}^Z$  are, at most, of $O(1)$ as $r^*\to\pm\infty$, their derivatives with respect to $r^*$ are, at most, of $O(w)$ as $r^*\to\pm\infty$, and $\lp(c_{s,\,1,\,0}^\mr{id}\rp)'=a\times O(w)$ as $r^*\to\pm\infty$;
\item $c_{s,\,|s|,\,|s|-1}^\mr{id}$ and $c_{s,\,|s|,|s|-2j}^Z$, for $j=1,\dots,\lfloor|s|/2\rfloor$ have better decay, as they are $O(r^{-1})$ as $r^*\to\pm\infty$.
\end{itemize}
In light of \eqref{eq:transformed-transport} and \eqref{eq:transformed-k}, we can recast the latter as the transport equation, thus obtaining
\begin{align*}
&\underline{\mc L} \swei{\upphi}{s}_{k+1}-\sign s(|s|-k-1) \frac{w'}{w}\swei{\upphi}{s}_{k+1}\\
&\quad=-\lp(\mathring{\slashed{\triangle}}^{[s]}+|s|+k(2|s|-k-1) -a^2\sin^2\theta TT-2aTZ\rp\}\swei{\upphi}{s}_k\\
&\quad\qquad -\frac{2Mr(r^2-a^2)}{(r^2+a^2)^2}\lp[1-3|s|+2s^2-3k(2|s|-k-1)\rp] \swei{\upphi}{s}_k- a^2w\lp[1-2|s|-2k(2|s|-k-1)\rp]\swei{\upphi}{s}_k\\
&\quad\qquad +\frac{2a\sign s}{r^2+a^2}\lp[(2|s|-2k-1)r+i|s|\cos\theta\rp]Z\swei{\upphi}{s}_k+\sum_{j=0}^{k-1}\lp(ac_{s,\,k,\,j}^{Z}Z+ ac_{s,\,k,\,j}^{\mr{id}}\rp)\swei{\upphi}{s}_{j}+\frac{\swei{\mathfrak{H}}{s}_k}{w}\,.\numberthis\label{eq:transformed-constraint}
\end{align*}
We note that we will sometimes use the notation:
\begin{align*}
\swei{U}{s}_k&:= w(|s|+k(2|s|-k-1))+ \frac{a^2\Delta^2}{(r^2+a^2)^4}\lp[1-2|s|-2k(2|s|-k-1)\rp]\\
&\qquad+\frac{2Mr(r^2-a^2)\Delta}{(r^2+a^2)^4}\lp[1-3|s|+2s^2-3k(2|s|-k-1)\rp]\,.
\end{align*}
For $k=|s|$, we drop the subscript on the operators and write $\swei{\Psi}{s}:=\swei{\uppsi}{s}_{|s|}$. Thus, from \eqref{eq:transformed-k}, we see that $\swei{\Psi}{s}$ solves the transformed equation
\begin{equation}
\begin{split}
&\frac12 \lp(L\underline{L}+\underline{L}L\rp)\swei{\Phi}{s} + w\lp(\mathring{\slashed{\triangle}}^{[s]}_T+ s^2\rp)\swei{\Phi}{s} +\frac{\Delta\lp[(1-2s^2)a^2\Delta + 2(1-s^2)Mr(r^2-a^2)\rp]}{(r^2+a^2)^4}\swei{\Phi}{s}\\&\quad=\swei{\mathfrak{H}}{s}+\sum_{k=0}^{|s|-1} aw s\lp[ c^{Z}_{s,\,|s|,\,k}(r) Z+c^\mr{id}_{s,\,|s|,\,k}(r)\rp]\swei{\upphi}{s}_{k}\,.
\end{split}\label{eq:Regge-Wheeler-eq}
\end{equation}
\end{itemize}
\end{definition}

The transformed system in Definition~\ref{def:transformed-system} can be thought of as a generalization of the Teukolsky equation. We say $\swei{\upalpha}{s}$ is a solution to the Teukolsky equation if $\Delta^{s}\swei{\upalpha}{s}\in\mathscr{S}_\infty^{[s]}(\mc{R})$ and
\begin{equation}
\rho^{2}\mathfrak{T}^{[s]}\upalpha^{[s]}=\frac{\Delta^{|s|/2(1-\sign s)}}{(r^2+a^2)^{|s|}}\swei{F}{s}. \label{eq:teukolsky-alpha}
\end{equation}
where, in Boyer--Lindquist coordinates,
\begin{equation} 
\begin{split}
\rho^2 \mathfrak{T}^{[s]}&=\rho^2\Box_g +2s(r-M)\p_r +2s\lp[\frac{a(r-M)}{\Delta}+i\frac{\cos\theta}{\sin^2\theta}\rp]\p_\phi 
\\
&\qquad +2s\lp(\frac{w'(r^2+a^2)}{2w} -ia\cos\theta\rp)\p_t +s-s^2\cot^2\theta 
\\
&=\Delta^{-s}\p_r(\Delta^{s+1}\p_r)-\frac{(r^2+a^2)^2}{\Delta}\lp(\p_t+\frac{a}{r^2+a^2}\p_\phi\rp)^2+s\frac{w'}{w^2}\lp(\p_t+\frac{a}{r^2+a^2}\p_\phi\rp)+\frac{4sar}{r^2+a^2}\p_\phi\\
&\qquad -\mathring{\slashed{\triangle}}^{[s]} +(2a\p_t\p_\phi+a^2\sin^2\theta \p_t^2-2ias\cos\theta\p_t)\,,
\end{split}\label{eq:teukolsky-operator}
\end{equation}
and  $\swei{F}{s}\in\mathscr{S}_\infty^{[s]}(\mc{R})$ is given.  With $\swei{\upalpha}{s}$ as the starting point, we have

\begin{lemma}[Teukolsky and the transformed system]\label{lemma:transformed-system} Fix $s\in\mathbb{Z}$. Let $\swei{\upalpha}{s}$ be a solution to the Teukolsky equation (\ref{eq:teukolsky-alpha}). Let $\swei{\upphi}{s}_{0}$ be the rescaling of the Teukolsky variable given by
\begin{equation}\label{eq:def-psi0}
\swei{\upphi}{s}_{0}:=(r^2+a^2)^{-|s|+1/2}\Delta^{|s|(1+\mr{sign}\,s)/2}\swei{\alpha}{s}\,,
\end{equation}
and $\swei{\mathfrak{H}}{s}_{0}$ be the rescaling of the Teukolsky inhomogeneity given by
\begin{equation}
\swei{\mathfrak{H}}{s}_{0}:=w^{|s|+1}(r^2+a^2)^{1/2}\swei{F}{s}\,,\label{eq:def-G0}
\end{equation}
Defining $\swei{\upphi}{s}_{k}$ and $\swei{\mathfrak H}{s}_{k}$ by the system of transport equations \eqref{eq:transformed-transport} and \eqref{eq:transformed-transport-inhom}, we see that for $k=0,\dots, |s|$, $\swei{\upphi}{s}_{k}$ solve the transformed system of Definition~\ref{def:transformed-system} with inhomogeneity $\swei{\mathfrak{H}}{s}_k$. 
\end{lemma}

\begin{proof} See our previous \cite{SRTdC2020}; in particular see \cite[Remark 3.2.1]{SRTdC2020} for explicit computations for $c_{s,k,j}^{\rm id}$ and $c_{s,k,j}^{Z}$ for some values of $s$.
\end{proof}

Note that our transformed system of Definition~\ref{def:transformed-system} generalizes that  of \cite{Dafermos2017} and, for $a=0$, matches those of \cite{Dafermos2016a,Pasqualotto2016}. It will sometimes be useful to consider  rescalings $\tilde{\upphi}_{k}=c_k(r)\upphi_{k}$. Then,  \eqref{eq:transformed-transport} and  \eqref{eq:transformed-transport-inhom} are replaced by
\begin{align}
\mc L \tilde\upphi_k&= \sign s \frac{c_k'}{c_k}\tilde\upphi_{k+1} + w\frac{c_{k+1}}{c_k}\tilde\upphi_{k+1}\,,\label{eq:transformed-transport-tilde}\\
\swei{\tilde{\mathfrak{H}}}{s}_{k}&=\frac{1}{{c_k}}\mc{L}\lp(\frac{c_{k-1}\swei{\tilde{\mathfrak H}}{s}_{k-1}}{w}\rp)\,,\qquad  k=1,...,|s|\,. \label{eq:transformed-transport-inhom-tilde}
\end{align}
In turn, equations \eqref{eq:transformed-k} are equivalent to
\begin{equation}
\swei{\tilde{\mathfrak{R}}}{s}_{k}\swei{\tilde\upphi}{s}_{k}=\swei{\tilde{\mathfrak{H}}}{s}_{k} -\sign s\lp(\frac{c_k'}{c_k}-(|s|-k)\frac{w'}{w}\rp)\frac{wc_{k+1}}{c_k}\tilde\upphi_{k+1}+a\sum_{i=0}^{k-1}\swei{\tilde{\mathfrak{J}}}{s}_{k,i}\,, \label{eq:transformed-k-tilde}
\end{equation}
where 
\begin{align}
\swei{\tilde{\mathfrak{R}}}{s}_{k}&:=\frac12 \lp(L\underline{L}+\underline{L}L\rp)+w\lp[\mathring{\slashed{\triangle}}^{[s]}_T+|s|+k(2|s|-k-1)\rp]-\sign s \frac{c_k'}{c_k}\underline{\mc L}\nonumber\\
&\qquad-\frac{4arw}{r^2+a^2}\sign{s}\lp(|s|-k\rp)Z
+ \frac{a^2\Delta^2}{(r^2+a^2)^4}\lp[1-2|s|-2k(2|s|-k-1)\rp] \label{eq:transformed-R-k-tilde}\\
&\qquad+\frac{2Mr(r^2-a^2)\Delta}{(r^2+a^2)^4}\lp[1-3|s|+2s^2-3k(2|s|-k-1)\rp]-\lp(\frac{c_k'}{c_k}\rp)'\,,\nonumber
\\
\swei{\tilde{\mathfrak{J}}}{s}_{k,i}&:=\frac{wc_i(r)}{c_k(r)}\lp(c_{s,\,k,\,i}^{Z}(r)Z+ c_{s,\,k,\,i}^{\mr{id}}(r)\rp)\swei{\tilde\upphi}{s}_{i}\,,\label{eq:transformed-J-k-tilde}
\end{align}
%and we note the equality
%\begin{align*}
%&\lp(-\frac{c_k'}{c_k}(L-\uL)-\sign{s}(|s|-k)\frac{w'}{w}\mc{L}-\frac{c_k''}{c_k}+(|s|-k)\frac{w'}{w}\frac{c_k'}{c_k}\rp)\tilde\upphi_k\\
%&=\sign s\lp(\frac{c_k'}{c_k}-(|s|-k)\frac{w'}{w}\rp)\frac{wc_{k+1}}{c_k}\tilde\upphi_{k+1}-\sign s \frac{c_k'}{c_k}\underline{\mc L}\tilde\upphi_k-\lp(\frac{c_k'}{c_k}\rp)'\tilde\upphi_k\,.
%\end{align*}
and \eqref{eq:transformed-constraint} becomes
\begin{align*}
&\underline{\mc L} \swei{\tilde\upphi}{s}_{k+1}+\sign s\lp[\frac{c_{k+1}'}{c_{k+1}}-(|s|-k-1) \frac{w'}{w}\rp]\swei{\tilde\upphi}{s}_{k+1}\\
&\quad=-\lp(\mathring{\slashed{\triangle}}^{[s]}+|s|+k(2|s|-k-1) -a^2\sin^2\theta TT-2aTZ\rp\}\swei{\tilde\upphi}{s}_k \numberthis\label{eq:transformed-constraint-tilde}\\
&\quad\qquad -\frac{2Mr(r^2-a^2)}{(r^2+a^2)^2}\lp[1-3|s|+2s^2-3k(2|s|-k-1)\rp] \swei{\tilde\upphi}{s}_k- a^2w\lp[1-2|s|-2k(2|s|-k-1)\rp]\swei{\tilde\upphi}{s}_k\\
&\quad\qquad +\frac{2a\sign s}{r^2+a^2}\lp[(2|s|-2k-1)r+i|s|\cos\theta\rp]Z\swei{\tilde\upphi}{s}_k+\sum_{j=0}^{k-1}\lp(ac_{s,\,k,\,j}^{Z}Z+ ac_{s,\,k,\,j}^{\mr{id}}\rp)\frac{c_j}{c_k}\swei{\tilde\upphi}{s}_{j}+\frac{\swei{\tilde{\mathfrak{H}}}{s}_k}{w}\,.
\end{align*}

Finally, for technical reasons it will be useful to consider a version of the system in Definition~\ref{def:transformed-system} where we assume that $\swei{\upphi}{s}_{k_0+1}$, for some $k_0=0, \dots, |s|-1$ is known already. 

\begin{definition}[Alternative transformed system]  \label{def:alt-transformed-system} Fix $s\in\mathbb{Z}$ and $k_0\in\{0,\dots |s|-1\}$. We say that, for $k\in\{0,\dots k_0\}$, the functions $\swei{\upphi}{s}_{k}$  are solutions to the alternative transformed system with inhomogeneities $\swei{\mathfrak H}{s}_k$ and $\swei{\mathfrak h}{s}_{k_0}$ and given $\swei{\upphi}{s}_{k_0+1}$ if $\Delta^{|s|-k}\swei{\upphi}{s}_{k}\in\mathscr S_\infty^{[s]}(\mc R)$ for $k=0,\dots,k_0+1$, $\Delta^{|s|-k-1}\swei{\mathfrak{H}}{s}_{k}\in S_\infty^{[s]}(\mc R)$ for $k=0,\dots k_0$, $\Delta^{|s|-k_0-1}\swei{\mathfrak{h}}{s}_{k_0}\in S_\infty^{[s]}(\mc R)$, and they satisfy the following PDEs.
\begin{itemize}
\item \textit{Transport equations}: if $s\neq 0$, \eqref{eq:transformed-transport} and \eqref{eq:transformed-transport-inhom} hold for $k=1,\dots k_0$; moreover, \eqref{eq:transformed-constraint} holds for each $k=0,\dots, k_0-1$ and it holds for $k=k_0$ if we replace $\swei{\mathfrak{H}}{s}_{k_0}$ by $\swei{\mathfrak{H}}{s}_{k_0}+\swei{\mathfrak{h}}{s}_{k_0}$.
\item \textit{Hyperbolic equations}: \eqref{eq:transformed-k} holds for each $k=0,\dots, k_0$.
\end{itemize}

\end{definition}

\subsection{Template energy norms}
\label{sec:template-energy-norms}

In this section, we introduce the main norms for the quantities in the transformed system of Definition~\ref{def:transformed-system} which will appear in the remainder of the paper. As we have mentioned in the aforementioned result, it useful to consider rescalings of our transformed variables. We write
\begin{align*}
\swei{\tilde\upphi}{s}_{k} := c_k(r)\swei{\upphi}{+s}_{k}\,, \qquad \swei{\dbtilde\upphi}{s}_{k} := \tilde c_k(r)\swei{\upphi}{+s}_{k}\,.
\end{align*}
for $c_k(r)$ and $\tilde{c}_k(r)$ given by
\begin{align*}
c_k(r)=
\begin{dcases}
(r^2+a^2)^{\frac{(|s|-k)}{2}}\,,&s>0\\
\lp(\frac{r^2+a^2}{\Delta}\rp)^{|s|-k}\,, &s<0
\end{dcases}\,,\qquad \tilde{c}_k(r)=
\begin{dcases}
(r^2+a^2)^{|s|-k}\,,&s>0\\
\lp(\frac{r^2+a^2}{\Delta}\rp)^{|s|-k}\,, &s<0
\end{dcases}\,,\numberthis\label{eq:rescalings}
\end{align*}
Note that there is a difference in these weights only for $s>0$, with the latter being consistent with peeling; see already Remark~\ref{rmk:peeling} for more comments. 
%We also record here the following simple computations:
%\begin{align*}
%-\sign s\frac{c_k'}{c_k} = \begin{dcases}
%\frac{r\Delta}{(r^2+a^2)^2}(|s|-k)(1-\delta)  \,, &s>0\\
%\frac{2M (r^2-^2)}{(r^2+a^2)^4}(|s|-k) \,, &s<0
%\end{dcases}\,,\qquad 
%-\sign s\frac{\tilde c_k'}{\tilde c_k} = \begin{dcases}
%\frac{2r\Delta}{(r^2+a^2)^2}(|s|-k)  \,, &s>0\\
%\frac{2M (r^2-^2)}{(r^2+a^2)^4}(|s|-k) \,, &s<0
%\end{dcases}
%\end{align*}

We will also find it convenient to introduce a function $\zeta$ given by
\begin{align*}
\zeta(r):=\lp(1-\frac{3M}{r}\rp)^2\lp(1-\mathbbm{1}_{[3M-\gamma_-,3M-\gamma_+]}\rp)\,,\numberthis\label{eq:def-zeta}
\end{align*}
with $\gamma_\pm(a_0,M)$ being trapping parameters identified in our previous \cite{SRTdC2020}, and to be given below in Definition~\ref{def:trapping-parameters}.

\subsubsection{First order norms}

Assume $\tau_1\leq \tau_2$. We begin with some definitions for all $0\leq k\leq |s|$:
\begin{align*}
\mb{E}_{\Sigma_\tau,p}\lp[\swei{\tilde\upphi}{s}_k\rp](\tau)&:=\int_{\Sigma_\tau} dr d\sigma \lp(r^p\big|L\swei{\tilde\upphi}{s}_k\big|^2+\frac{r^2+a^2}{\Delta}\big|\uL\swei{\tilde\upphi}{s}_k\big|^2r^{-2}+\big|\mathring{\slashed{\nabla}}^{[s]}\swei{\tilde\upphi}{s}_k\big|^2r^{-2}+s^2\big|\swei{\tilde\upphi}{s}_k\big|^2r^{-2}\rp)\,,\\
\mb{E}_{\mc H^+}\lp[\swei{\tilde\upphi}{s}_k\rp](\tau_1,\tau_2)&:=\int_{\mathbb{S}^2}\int_{\tau_1}^{\tau_2} d\tau d\sigma \,\big|L\swei{\tilde\upphi}{s}_k\big|^2\Big|_{r=r_+}\,,\\
\mb{E}_{\mc I^+,p}\lp[\swei{\tilde\upphi}{s}_k\rp](\tau_1,\tau_2)&:=\limsup_{v\to \infty}\int_{S_{(\tau_1,\tau_2)}(v)} d\tau d\sigma \lp( |\uL\swei{\tilde\upphi}{s}_k|^2 +r^{p-2}\big|\mathring{\slashed{\nabla}}^{[s]}\swei{\tilde\upphi}{s}_k\big|^2+s^2r^{p-2}\big|\swei{\tilde\upphi}{s}_k\big|^2\rp)\,,
\end{align*}
where $S_{(\tau_1,\tau_2)}(v)$ denote null hypersurfaces which, as $v\to \infty$, approach $\mc I^+_{(\tau_1,\tau_2)}$. When $k=|s|$, we define
\begin{align*}
&\mb{I}_{-\delta,p}\lp[\swei{\Phi}{s}\rp](\tau_1,\tau_2)\\
&\quad:=\int_{\tau_1}^{\tau_2}d\tau\int_{\Sigma_\tau} dr d\sigma \lp(|L\swei{\Phi}{s}|^2(pr^{p-1} + r^{-1-\delta})+\big|\uL\swei{\Phi}{s}\big|^2r^{-1-\delta}\rp)\\
&\quad\qquad + \int_{\tau_1}^{\tau_2}d\tau\int_{\Sigma_\tau} dr d\sigma r^{p-3}\lp((2-p+r^{-1})|\mathring{\slashed{\nabla}}^{[s]}\swei{\Phi}{s}|^2+\lp(r^{-p-\delta}+r^{-1}+s^2(2-p)\rp)|\swei{\Phi}{s}|^2\rp)\,,
\end{align*}
and we use the notation $\mb{I}_{p}\equiv\mb{I}_{-1,p}$, $\mb{I}_{-\delta}\equiv \mb{I}_{-\delta,0}$, $\mb{I}\equiv \mb{I}_{-1,0}$. We also define 
\begin{align*}
&\mb{I}_{-\delta,p}^{\rm deg}\lp[\swei{\Phi}{s}\rp](\tau_1,\tau_2)\\
&\quad:=\int_{\tau_1}^{\tau_2}d\tau\int_{\Sigma_\tau} dr d\sigma \lp(|L\swei{\Phi}{s}|^2\frac{(pr^{p} + r^{-\delta})}{r}+\frac{1}{r^{1+\delta}}\big|\uL\swei{\Phi}{s}\big|^2+\frac{(2-p+r^{-1})}{r^{3-p}}|\mathring{\slashed{\nabla}}^{[s]}\swei{\Phi}{s}|^2\rp)\zeta(r)\\
&\quad\qquad + \int_{\tau_1}^{\tau_2}d\tau\int_{\Sigma_\tau} dr d\sigma  \lp(r^{p-3}(s^2(2-p)+r^{-\delta})|\swei{\Phi}{s}|^2+r^{-1-\delta}|(L-\uL)\swei{\Phi}{s}|^2\rp)\,,
\end{align*}
and again write $\mb{I}_{p}^{\rm deg}\equiv\mb{I}_{-1,p}^{\rm deg}$, $\mb{I}^{\rm deg}_{-\delta}\equiv \mb{I}_{-\delta,0}^{\rm deg}$, $\mb{I}^{\rm deg}\equiv \mb{I}_{-1,0}$.

For $k<|s|$ and letting $s>0$ without loss of generality, we define
\begin{align*}
\mb{I}_{-\delta,p}\lp[\swei{\tilde\upphi}{-s}_{(k)}\rp](\tau_1,\tau_2)&:=\int_{\tau_1}^{\tau_2}d\tau\int_{\Sigma_\tau} dr d\sigma \lp(\frac{r^2+a^2}{\Delta}\big|\uL\swei{\tilde\upphi}{-s}_{(k)}\big|^2r^{-1-\delta}+|\mathring{\slashed{\nabla}}^{[s]}\swei{\tilde\upphi}{-s}_{(k)}|^2r^{p-2}+|\swei{\tilde\upphi}{-s}_{(k)}|^2r^{p-2}\rp)\,,\\
\mb{I}_{p}\lp[\swei{\tilde\upphi}{+s}_{(k)}\rp](\tau_1,\tau_2)&:=\int_{\tau_1}^{\tau_2}d\tau\int_{\Sigma_\tau} dr d\sigma \lp(\big|L\swei{\tilde\upphi}{+s}_{(k)}\big|^2r^{p-1}+(|\mathring{\slashed{\nabla}}^{[s]}\swei{\tilde\upphi}{+s}_{(k)}|^2+|\swei{\tilde\upphi}{+s}_{(k)}|^2)(2-p+r^{-1})r^{p-3}\rp)\,.
\end{align*}
We also have the following conventions: in the case of $-s$, we let $\mb{I}_{p}\equiv\mb{I}_{-1,p}$, $\mb{I}_{-\delta}\equiv \mb{I}_{-\delta,0}$ and  $\mb{I}\equiv \mb{I}_{-1,0}$; in the case of $+s$, we let $\mb{I}_{-\delta,p}\equiv\mb{I}_{p}$ and $\mb{I}\equiv \mb{I}_{-\delta,0}$.

If $|a|< M$, we will also use
\begin{align*}
\overline{\mb{E}}_{\mc H^+}\lp[\swei{\tilde\upphi}{s}_k\rp](\tau_1,\tau_2)&:=\int_{\mathbb{S}^2}\int_{\tau_1}^{\tau_2} d\tau d\sigma \lp(\big|L\swei{\tilde\upphi}{s}_k\big|^2+\big|\mathring{\slashed{\nabla}}^{[s]}\swei{\tilde\upphi}{s}_k\big|^2+\big|\swei{\tilde\upphi}{s}_k\big|^2\rp)\Big|_{r=r_+}\,,
\end{align*}
for $k=0,\dots,|s|$. We can add an overbar to any of the previous norms if we replace $\frac{r^2+a^2}{\Delta}|\uL\cdot |^2$ or $|\uL \cdot|^2$ by $|\frac{r^2+a^2}{\Delta}\uL\cdot |^2$ and the factor $s^2$ on the zeroth order term replaced by $1$.

Finally, for $A\subset[r_+,\infty]$, we will use the notation $\mathbb{E}[\mathbbm{1}_{r\in A}\cdot ]$ or $\mathbb{I}_{-\delta,p}[\mathbbm{1}_{r\in A}\cdot ]$ with or without additional superscripts to denote the same as above but replacing $\Sigma_\tau$ with $\Sigma_\tau\cup \{r\in A\}$. A similar convention holds if the set $A$ is given in terms of the $r^*$ radial coordinate.

\subsubsection{Higher order norms}

For $J\geq 1$ and $|a|<M$, we set
\begin{align*}
\overline{\mathbb{E}}^{J}_p[\swei{\tilde\upphi}{s}_k](\tau)&:=\sum_{X\in\{\mathrm{id}, T,X^*,r^{-1}\mathring{\slashed\nabla}^{[s]}\}} \overline{\mathbb{E}}^{J-1}_{p}[X\swei{\tilde\upphi}{s}_k](\tau)\,,
\end{align*}
with $\overline{\mathbb{E}}^{0}_{p}\equiv \overline{\mathbb{E}}_{p}$ introduced in the previous section. 
If, furthermore $J<|s|-k$, then we set
\begin{align*}
\overline{\mathbb{I}}^{\mathrm{deg},J}_{-\delta,p}[\swei{\tilde\upphi}{s}_k](\tau_1,\tau_2)=\overline{\mathbb{I}}^{J}_{-\delta,p}[\swei{\tilde\upphi}{s}_k](\tau_1,\tau_2)&:=\sum_{X\in\{\mathrm{id},T,X^*,r^{-1}\mathring{\slashed\nabla}^{[s]}\}} \overline{\mathbb{I}}^{J-1}_{-\delta,p}[X\swei{\tilde\upphi}{s}_k](\tau_1,\tau_2)\,,
\end{align*}
and if $J\geq |s|-k$, then we set
\begin{align*}
\overline{\mathbb{I}}^{\mathrm{deg},J}_{-\delta,p}[\swei{\tilde\upphi}{s}_k](\tau_1,\tau_2)&:=\overline{\mathbb{I}}_{-\delta,p}^{J-1}[X^*\swei{\tilde\upphi}{s}_k](\tau_1,\tau_2)+\sum_{X\in\{\mathrm{id},T,X^*,r^{-1}\mathring{\slashed\nabla}^{[s]}\}} \overline{\mathbb{I}}^{\mathrm{deg},J-1}_{-\delta,p}[X\swei{\tilde\upphi}{s}_k](\tau_1,\tau_2)\,;
\end{align*}
here, $\overline{\mathbb{I}}^{0}_{-\delta,p}\equiv \overline{\mathbb{I}}_{-\delta,p}$ and $\overline{\mathbb{I}}^{\mathrm{deg},0}_{-\delta,p}\equiv \overline{\mathbb{I}}^{\rm deg}_{-\delta,p}$ introduced in the previous section.  Finally, we set 
\begin{align*}
\overline{\mathbb{E}}^{J}_{\mc H^+}[\swei{\tilde\upphi}{s}_k](\tau_1,\tau_2)&:=\sum_{X\in\{\mathrm{id},L,\mathring{\slashed\nabla}^{[s]}\}} \overline{\mathbb{E}}^{J-1}_{\mc H^+}[X\swei{\tilde\upphi}{s}_k](\tau_1,\tau_2)\,,\\
\overline{\mathbb{E}}^{J}_{\mc I^+,p}[\swei{\tilde\upphi}{s}_k](\tau_1,\tau_2)&:=\sum_{X\in\{\mathrm{id},\uL,\mathring{\slashed\nabla}^{[s]}\}} \overline{\mathbb{E}}^{J-1}_{\mc I^+,p}[X\swei{\tilde\upphi}{s}_k](\tau_1,\tau_2)\,, \qquad k<|s| \text{~and~} J\leq |s|-k\,,\\
\overline{\mathbb{E}}^{J}_{\mc I^+,p}[\swei{\tilde\upphi}{s}_k](\tau_1,\tau_2)&:=\sum_{X\in\{\mathrm{id},\uL\}} \overline{\mathbb{E}}^{J-1}_{\mc I^+,p}[X\swei{\tilde\upphi}{s}_k](\tau_1,\tau_2)\,,
\end{align*}
with $\overline{\mathbb{E}}^{0}_{\mc H^+}\equiv \overline{\mathbb{E}}_{\mc H^+}$ and $\overline{\mathbb{E}}^{0}_{\mc I^+,p}\equiv \overline{\mathbb{E}}_{\mc I^+,p}$ introduced in the previous section.

As before, letting $A\subset[r_+,\infty]$, the notation $\mathbb{E}^J[\mathbbm{1}_{r\in A}\cdot ]$ or $\mathbb{I}_{-\delta,p}^J[\mathbbm{1}_{r\in A}\cdot ]$ with or without additional superscripts shall denote the same as above but replacing $\Sigma_\tau$ with $\Sigma_\tau\cup \{r\in A\}$; and similarly if the set $A$ is given in terms of the $r^*$ radial coordinate.

\section{A toolbox for physical space estimates}
\label{sec:toolbox-physical-space}

This section contains a sequence of basic but useful results. Thanks to some multiplier and commutator identities in Sections~\ref{sec:multiplier-identities} and \ref{sec:multiplier-identities}, respectively, we start by deriving some basic integrated and flux estimates for our transformed system in Sections~\ref{sec:first-order-basic-estimates} and \ref{sec:higher-order-estimates}. These can then be used to establish well-posedness (Section~\ref{sec:wellposedness}) and obtain a useful backwards extension result (Section~\ref{sec:extension}) for solutions to the transformed system.

\subsection{Multiplier identities}
\label{sec:multiplier-identities}

In this section, we derive several pointwise identities taking the form 
\begin{align*}
LF_L+\uL F_{\uL} + I +  O= 0\,,\qquad \int_{\mathbb{S}^2}O d\sigma=0\,.
\end{align*}
Diving by $\rho^2\Delta (r^2+a^2)^{-1}$ and integrating in $\mc{R}_{(0,\tau)}$ against the respective volume form, and appealing to \eqref{eq:commutation-mcL-volume-form}, we get the identity
\begin{align}
\begin{split}
&\int_{\Sigma_\tau}  \lp(v_1\Delta^{-1}F_L+v_2F_{\uL}\rp) drd\sigma
+\int_{\mc H^+_{(0,\tau)}} v_3F_{\uL} d\sigma d\tau+\int_{\mc I^+_{(0,\tau)}} v_4F_L d\sigma d\tau \\
&= \int_{\Sigma_0}  \lp(v_1\Delta^{-1}F_L+v_2F_{\uL}\rp) drd\sigma-\int_{\mc{R}_{(0,\tau)}} v_5 I \frac{r^2+a^2}{\Delta} dt dr d\sigma\,,
\end{split}\label{eq:template-physical-space-identity}
\end{align}
for some functions $v_i(r,\theta)$, $i=1,\dots,5$ which are bounded above and below by constants depending only on $M$ which can be computed explicitly from the definitions in Section~\ref{sec:kerr-prelims}.

In what follows, $\chi_R$ denotes  a cutoff function supported for $r^*\geq R^*$ and equal to 1  for $r^*\geq 2R^*$ for some large $R^*>0$. We will also drop the superscript on solutions of the transformed system of Definition~\ref{def:transformed-system}, to lighten the notation.

\subsubsection{Transport identities}

\begin{lemma}[Transport identities] \label{lemma:phys-space-transport-identity} Let $|s|\in\mathbb{Z}$ and $k<|s|$. Suppose $\tilde\upphi_k$ are homogeneous solutions to the transformed system of Definition~\ref{def:transformed-system}. Then, we have an identity of the form \eqref{eq:template-physical-space-identity} where the boundary terms are determined by
\begin{align*}
F_{L}= \frac12(1-\sign s)c|\tilde\upphi_k|^2\,,\qquad
F_{\uL}= \frac12(1+\sign s)c|\tilde\upphi_k|^2\,,
\end{align*}
and the bulk term is determined by 
\begin{align*}
I&=\sign s \lp(c'-\frac{2c_k'}{c_k}c\rp)|\tilde\upphi_k|^2-2w\frac{c_{k+1}}{c_k}c(r)\Re[\tilde\upphi_{k+1}\overline{\tilde\upphi_k}]\,.
\end{align*}
\end{lemma}

From Lemma~\ref{lemma:phys-space-transport-identity}, setting $c=-r^{p-2}(1+(M|s|+1)/r)\chi_R$ where $p\in\mathbb{R}_+$,  we deduce that for $s<0$ and sufficiently large $R^*$
\begin{align*}
&\int_{\Sigma_\tau\cap\{r^*\geq 2R^*\}}r^{p-2}\lp(1+\frac{M|s|+1}{r}\rp)|\tilde\upphi_k|^2 drd\sigma +\int_{0}^\tau \int_{\Sigma_{\tau'}\cap\{r^*\geq 2R^*\}} \lp(2-p+\frac{1}{r}\rp)r^{p-3}|\tilde\upphi_k|^2 drd\sigma d\tau'\\
&\qquad+\int_{\mc I^+_{(0,\tau)}}r^{p-2}|\tilde\upphi_k|^2\big|_{r=\infty} drd\sigma d\tau'  \\
&\quad\leq
B\mathbb{I}^{\rm deg}[\tilde\upphi_k\mathbbm{1}_{\{R^*\leq r^*\leq 2R^*\}}](0,\tau) 
+B\mathbb{E}[\tilde\upphi_k](0)\\
&\quad\qquad+B\int_{0}^\tau \int_{\Sigma_{\tau'}\cap\{r^*\geq 2R^*\}} \lp(2-p+\frac{1}{r}\rp)r^{p-5}|\tilde\upphi_{k+1}|^2 drd\sigma d\tau'\,.\numberthis\label{eq:phys-space-transport-identity-s<0}
\end{align*}
For $s>0$, we can choose  $c=(r-r_+)\chi_R(-r^*)/r$ to get, for sufficiently large $R^*$,
\begin{align*}
&\int_{\Sigma_\tau\cap\{r^*\leq -2R^*\}} (r-r_+)|\tilde\upphi_k|^2 dr^*d\sigma +\int_{0}^\tau \int_{\Sigma_{\tau'}\cap\{r^*\leq -2R^*\}} \Delta|\tilde\upphi_k|^2 dr^*d\sigma d\tau'  \\
&\quad\leq
B\mathbb{I}^{\rm deg}[\tilde\upphi_k\mathbbm{1}_{\{R^*\leq -r^*\leq 2R^*\}}](0,\tau) 
+B\int_{0}^\tau \int_{\Sigma_{\tau'}\cap\{r^*\leq -2R^*\}}\Delta(r-r_+)^2|\tilde\upphi_{k+1}|^2 dr^*d\sigma d\tau'
+B\mathbb{E}[\tilde\upphi_k](0)\,.\numberthis\label{eq:phys-space-transport-identity-s>0}
\end{align*}
Note that estimates \eqref{eq:phys-space-transport-identity-s<0} and \eqref{eq:phys-space-transport-identity-s>0} also hold replacing $\tilde\upphi_k$ with $\mathring{\slashed{\nabla}}$ or $Z\tilde\upphi_k$.

\subsubsection{Killing multiplier identities}

In this section, we obtain the identities that will allow us to bound the flux of the transformed variables through the future boundary of a spacetime slab by bulk error terms.  These follow by applying the Killing vector fields on Kerr as multipliers:
\begin{lemma}[Killing multiplier identity I] \label{lemma:phys-space-Killing-multiplier-identity}  Let $|s|\in\mathbb{Z}$, let $\chi_T$, $\chi_Z$ and $\tilde\chi$ be smooth functions of $r$, and let $\xi$ be a smooth function of $\tilde t^*$. Suppose that, for $0\leq k\leq |s|$, $\swei{\tilde\upphi}{s}_k$ are homogeneous solutions to the transformed system of Definition~\ref{def:transformed-system}; we drop the superscripts in what follows to ease the notation. We have an identity of the form \eqref{eq:template-physical-space-identity} where 
\begin{align*}
4F_L&=\chi_T|\uL\tilde \upphi_k|^2+a^2w\sin^2\theta|(\chi_T T+\upomega_+\chi_Z Z)\tilde\upphi_k|^2+\frac12a\upomega_+\chi_Z w\lp(2-\frac{a^2\sin^2\theta}{r_+^2+a^2}\rp)|Z\tilde\upphi_k|^2\\
&\qquad +\chi_Tw|\mathring{\slashed{\nabla}} \tilde\upphi_k|^2+\chi_T{U}_k |\tilde\upphi_k|^2-2ka\upomega_+\cos\theta w\chi_Z\sign s\Im[\tilde\upphi_k\overline{Z\tilde\upphi_k}]\\
&\qquad -2\lp(\frac{a\chi_T}{r^2+a^2}-\upomega_+\chi_Z\rp)\Re[Z\tilde\upphi_k\overline{\uL \tilde\upphi_k}]\\
&\qquad +2 w\chi_T\tilde\chi \xi \sum_{j=0}^{k-1}\frac{c_j}{c_k}\lp\{\Re[(ac_{s,k,j}^ZZ+ac_{s,k,j}^{\rm id})\tilde\upphi_j\overline{\tilde\upphi_k}]+ac_{s,k,j}^Z\sign s \Re[\tilde\upphi_j\overline{Z\tilde\upphi_k}]\rp\}\\
&\qquad -(1-\sign s)\frac{a}{r^2+a^2}\chi_T\tilde\chi \xi\lp\{2\sum_{j=0}^{k-2}\frac{c_jc_{k-1}}{c_k^2}\Re[Z\tilde\upphi_j\overline{Z\upphi_{k-1}}]+\frac{a_{s,k,k-1}^Zc_{k-1}^2}{c_k^2}|Z\tilde\upphi_{k-1}|^2\rp\}\,,
\end{align*}
and $F_{\uL}$, obtained from the latter by swapping $\uL$ by $L$ and vice versa,  and replacing $\sign s$ by $-\sign s$ in the last two lines, and where
\begin{align*}
I&=-\frac12\sign s \chi_T \frac{c_k'}{c_k}|\underline{\mc{L}}\tilde\upphi_k|^2+\frac12 \upomega_+\chi_Z'\Re\lp[(L-\uL)\tilde\upphi_k\overline{Z\tilde\upphi_k}\rp]+\frac18\chi_T'(|L\tilde\upphi_k|^2-|\uL\tilde\upphi_k|^2)\\
&\qquad-\frac12\chi_T\lp(\frac{c_k'}{c_k}\rp)'\lp(\frac{c_k'}{c_k}\sign s |\tilde\upphi_k|^2+\frac{wc_{k+1}}{c_k}\Re[\tilde\upphi_k\overline{\tilde\upphi_{k+1}}] +\Re[\underline{\mc L}\tilde\upphi_k\overline{\tilde\upphi_{k}}] \rp)-\frac{a\chi_T'}{4(r^2+a^2)}\Re[(L+\uL)\tilde\upphi_k\overline{Z\tilde\upphi_k}]\\
&\qquad+\frac12\sign s \lp(\frac{c_k'}{c_k}-(|s|-k)\frac{w'}{w}\rp)w\frac{c_{k+1}}{c_k}\Re[\tilde\upphi_{k+1}\overline{\lp(\sign s \chi_T\frac{c_k'}{c_k}\tilde\upphi_k-2\lp(\frac{a\chi_T}{r^2+a^2}-\upomega_+\chi_Z\rp)Z\tilde\upphi_k\rp)}]\\
&\qquad+\frac12\sign s \lp(\frac{c_k'}{c_k}-(|s|-k)\frac{w'}{w}\rp)\lp(w\frac{c_{k+1}}{c_k}\rp)^2|\tilde\upphi_{k+1}|^2\chi_T \\
&\qquad-\frac12\sign s \frac{c_k'}{c_k}\Re\lp[\overline{\underline{\mc{L}}\tilde\upphi_k}\lp(\sign s \chi_T\frac{c_k'}{c_k}\tilde\upphi_k+\chi_T\frac{wc_{k+1}}{c_k}\tilde\upphi_{k+1}-2\lp(\frac{a\chi_T}{r^2+a^2}-\upomega_+\chi_Z\rp)Z\tilde\upphi_k\rp)\rp]\\
&\qquad -\frac{2ar\chi_T}{r^2+a^2}w\sign s (|s|-k) \Re\lp[\lp(\underline{\mc L}\tilde\upphi_k+\frac{wc_{k+1}}{c_{k}}\tilde\upphi_{k+1}\rp) \overline{Z\tilde\upphi_k}\rp]\\
&\qquad -\frac{4ar}{r^2+a^2}w\sign s (|s|-k)\lp(\upomega_+\chi_Z -\frac{a\chi_T}{r^2+a^2}\rp)|Z\tilde\upphi_k|^2\\
&\qquad -a(|s|-k)\upomega_+ \chi_Z w\lp( \frac{c_k'}{c_k}\Im[\tilde\upphi_k\overline{Z\tilde\upphi_k}]+\sign s \frac{wc_{k+1}}{c_k}\Im[\tilde\upphi_{k+1}\overline{Z\tilde\upphi_k}]+\sign s\Im[\underline{\mc L}\tilde\upphi_k\overline{Z\tilde\upphi_k}]\rp)\\
&\qquad+w\sum_{j=0}^{k-1}\frac{c_j}{c_k}\Re\lp\{\lp[\lp((1-\xi)\tilde\chi+1-\tilde \chi\rp)\chi_T T+\upomega_+\chi_Z Z\rp]\overline{\tilde\upphi_k}(ac_{s,k,j}^{Z}Z+ac_{s,k,j}^{\rm id})\tilde\upphi_j\rp\}\\
&\qquad- \chi_T\tilde\chi \xi w\sum_{j=0}^{k-1}ac_{s,k,j}^{\rm id}\frac{c_j}{c_k}\Re\lp[\lp(w\frac{c_{j+1}}{c_j}\tilde\upphi_{j+1}-\frac{a}{r^2+a^2}Z \tilde\upphi_j+\sign s \tilde\upphi_j'+\sign s\frac{c_j'}{c_j}\tilde\upphi_j\rp)\overline{\tilde\upphi_k}\rp]\\
&\qquad + \xi \sign s \sum_{j=0}^{k-1} \lp\{\lp[\lp(\frac{c_j}{c_k}ac_{s,k,j}^{\rm id}\chi_T\tilde\chi\rp)'-ac_{s,k,j}^{\rm id}\frac{c_j'}{c_k}\chi_T\tilde\chi\rp]\Re[Z\tilde\upphi_j\overline{\tilde\upphi_{k}}] + ac_{s,k,j}^{\rm id}\frac{c_j}{c_k}\chi_T\tilde\chi\Re[Z\tilde\upphi_j\overline{\tilde\upphi_{k}'}] \rp\}\\
&\qquad +\sum_{j=0}^{k-2}\lp\{c_j c_{k-1}\mc L\lp(\chi_T\tilde\chi\xi a \frac{c_{s,k,j}^Z}{c_k^2}\rp)\Re[Z\tilde{\upphi_j}\overline{Z\tilde\upphi_{k-1}}] -  a c_{s,k,j}^Z\chi_T\tilde\chi\xi\frac{c_{j+1}}{c_k}\Re[Z\tilde\upphi_{j+1}\overline{\upphi_k}] \rp\}\\
&\qquad +\sum_{j=0}^{k-2}\frac{aw}{r^2+a^2}\chi_T\tilde\chi\xi ac_{s,k,j}^Z\frac{c_{k-1}c_{j+1}}{c_jc_k}\Re[Z\tilde\upphi_{j+1}\overline{Z\tilde\upphi_{k-1}}]\\
&\qquad +\lp[c_{k-1}^2\mc L\lp(\frac{a}{r^2+a^2}\chi_T\tilde\chi\xi\frac{ac_{s,k,k-1}^Z}{2c_k^2}\rp)+\frac{a}{r^2+a^2}\chi_T\tilde\chi\xi\frac{ac_{s,k,k-2}^Zc_{k-1}^2}{c_k c_{k-2}}\rp]|Z\tilde\upphi_{k-1}|^2\\
\end{align*}
\end{lemma}

\begin{proof}
The proof follows from multiplying \eqref{eq:transformed-k-tilde} by $\overline{(\chi_T T+\upomega_+\chi_Z Z)\tilde\upphi_k}$, taking the real part and  integrating by parts. The contributions  from the left hand side of \eqref{eq:transformed-k-tilde}  are easiest when $k=|s|$, and the reader may see the steps in \cite[Appendix B]{Dafermos2017}. In the case $k<|s|$, we use the transport identity \eqref{eq:transformed-transport-tilde} whenever possible to deal with the additional terms.

For the terms due to coupling, we can move the $T$ derivative onto the lower-level transformed variables:
\begin{align*}
&aw\Re\lp[\overline{T\tilde\upphi_k}(c_{s,k,j}^{\rm id}+c_{s,k,j}^{Z}Z)\upphi_j\rp] \\
&\quad= T\Re\lp[aw\overline{\tilde\upphi_k}(c_{s,k,i}^{\rm id}+c_{s,k,i}^{Z}Z)\upphi_j\rp]
-Z\Re\lp[aw\overline{\tilde\upphi_k}c_{s,k,j}^{Z}T\upphi_j\rp]
+aw\Re\lp[T\upphi_j(c_{s,k,j}^{Z}Z-c_{s,k,j}^{\rm id})\overline{\tilde\upphi_k}\rp] \\
&\quad=T\lp\{aw\Re\lp[\overline{\tilde\upphi_k}(c_{s,k,j}^{\rm id}+c_{s,k,j}^{Z}Z)\upphi_j\rp]\rp\}
-Z\lp\{aw\Re\lp[\overline{\tilde\upphi_k}c_{s,k,j}^{Z}T\upphi_j\rp]\rp\}\\
&\qquad+awc_j\Re\lp[\lp(w\frac{c_{j+1}}{c_j}\tilde\upphi_{j+1}-\frac{a}{r^2+a^2}Z \tilde\upphi_j+\sign s \tilde\upphi_j'+\sign s\frac{c_j'}{c_j}\tilde\upphi_j\rp)(c_{s,k,j}^{Z}Z-c_{s,k,j}^{\rm id})\overline{\tilde\upphi_k}\rp]\,.
\end{align*}
When $k=|s|$, we will find it convenient to deal with the contributions from $Z\overline{\tilde\upphi_k}$ more carefully, see also our previous \cite[Lemma 6.3.2]{SRTdC2020}. We have
\begin{align*}
&aw\frac{c_j}{c_k}\tilde{\chi}\Re\lp[\lp( \tilde\upphi_j'+\frac{c_j'}{c_j}\tilde\upphi_j\rp)c_{s,k,j}^{Z}\overline{Z\tilde\upphi_k}\rp] \\
&\quad=Z\lp(\frac{awc_{s,k,j}^Z}{c_k}\tilde\chi \upphi_j'\overline{\tilde\upphi_k}\rp)- \frac{1}{2}(L-\uL)\lp[aw\frac{c_j}{c_k}c_{s,k,j}^Z\tilde \chi \Re[Z\tilde\upphi_j\overline{\tilde\upphi_k}]\rp]+\frac{awc_j}{c_k}c_{s,k,j}^Z \tilde \chi\Re[Z\tilde\upphi_j\overline{\tilde\upphi_k'}]\\
&\qquad\quad+\lp(\frac{aw}{c_k}c_{s,k,j}^Z\tilde\chi\rp)' c_j\Re[Z\tilde\upphi_j\overline{\tilde\upphi_k}]
\end{align*}
For the remaining terms, let us distinguish between two cases: the coupling with $j<k-1$ and the coupling with $j=k-1$. In the latter case,
\begin{align*}
&aw\frac{c_j}{c_k}c_{s,k,j}^Z\tilde \chi\Re\lp[\lp(\frac{wc_k}{c_j}\tilde\upphi_k-\frac{a}{r^2+a^2}\tilde\upphi_j\rp)\frac{Z\upphi_k}{c_k}\rp]\\
&=Z\lp(\frac12aw^2\frac{c_j}{c_k}c_{s,k,j}^Z\tilde\chi|\tilde\upphi_k|^2\rp)-\mc L\lp(\frac{a^2}{2(r^2+a^2)}\frac{c_j^2}{c_k^2}c_{s,k,j}^Z\tilde\chi |Z\tilde \upphi_j|^2\rp)-\sign sc_j^2\lp(\frac{a^2}{2(r^2+a^2)}\frac{c_{s,k,j}^Z}{c_k^2}\tilde\chi\rp)' |Z\tilde \upphi_j|^2
\end{align*}
whereas in the former case,
\begin{align*}
&aw\frac{c_j}{c_k}c_{s,k,j}^Z\tilde\chi\Re\lp[\lp(w\frac{c_{j+1}}{c_j}\tilde\phi_{j+1}-\frac{a}{r^2+a^2}Z\tilde\upphi_j\rp)\overline{Z\tilde\upphi_k}\rp] \\
&\quad=Z\lp(a w^2 \frac{c_{j+1}}{c_k}\tilde\chi\Re[\tilde\upphi_{j+1}\overline{\tilde\upphi_k}]\rp)-aw^2\frac{c_{j+1}}{c_k}c_{s,k,j}^Z\tilde\chi\Re[Z\tilde\upphi_{j+1}\overline{\tilde\upphi_k}]-\frac{a^2}{r^2+a^2}\frac{c_{s,k,j}^Z}{c_kc_{k-1}}\Re[Z\mc L\upphi_{k-1}\overline{Z\upphi_j}]\\
&\quad= Z\lp(a w^2 \frac{c_{j+1}}{c_k}\tilde\chi\Re[\tilde\upphi_{j+1}\overline{\tilde\upphi_k}]\rp) -
\mc L\lp(\frac{a^2\tilde\chi}{r^2+a^2}\frac{c_{s,k,j}^Z}{c_kc_{k-1}}\Re[Z\upphi_{k-1}\overline{Z\upphi_j}]\rp)-aw^2\frac{c_{j+1}}{c_k}c_{s,k,j}^Z\tilde\chi\Re[Z\tilde\upphi_{j+1}\overline{\tilde\upphi_k}]\\
&\quad\qquad+\sign sc_{j}c_{k-1}\lp(\frac{a^2\tilde\chi}{r^2+a^2}\frac{c_{s,k,j}^Z}{c_kc_{k-1}}\rp)'\Re[Z\tilde \upphi_{k-1}\overline{Z\tilde\upphi_j}]+\frac{a^2w}{r^2+a^2}\frac{c_{j+1}}{c_k}c_{s,k,j}^Z\tilde\chi\Re[Z\tilde\upphi_{k-1}\overline{Z\tilde\upphi_{j+1}}]\,.
\end{align*}
To conclude, note that we can also introduce a cutoff in (hyperboloidal) time $\xi$ in the above computations, noting that $1=\xi+(1-\xi)$. 
\end{proof}

The identity in Lemma~\ref{lemma:phys-space-Killing-multiplier-identity} is, in large part, the motivation for the norms considers for the fluxes of $\tilde\upphi_k$ across the hypersurfaces bounding $\mc{R}_{(0,\tau)}$ and even some of the norms considered on the time slab  $\mc{R}_{(0,\tau)}$ in some cases. As an example, we consider some applications of the result where  $1-\tilde\chi=\zeta$ if $k=|s|$ and $\tilde\chi=0$ otherwise, and $\xi$ is a bump function supported in $\tilde t^*\in (0,\tau)$ and equal to 1 for $\tilde t^*\in (0,\tau-1)$. Choosing $\chi_T\equiv 1$, and $0\leq \chi_Z \leq 1$ to be a smooth compactly supported function equal to 1 at $r=r_+$ such that $T+\upomega_+\chi_Z\Phi$ is is timelike in the support of $\tilde\upphi_k$ for all $k=0,\dots, |s|$, we have
\begin{align*}
&\mathbb{E}[\tilde\upphi_k](\tau)+\mathbb{E}_{\mc I^+}[\tilde\upphi_k](0,\tau)+\mathbb{E}_{\mc H^+}[\tilde\upphi_k](0,\tau)+\int_0^\tau \int_{\Sigma_{\tau'}}\frac{(-\sign s )}{c_k}\frac{dc_k}{dr}|\underline{\mc{L}}\tilde\upphi_k|^2 drd\sigma d\tau' \\
&\qquad +\int_0^\tau \int_{\Sigma_{\tau'}}a\upomega_+\chi_Z\sum_{j=0}^{k-1}w\frac{c_j}{c_k}c_{s,k,j}^Z(r)\Re[Z\tilde\upphi_k\overline{Z}\tilde\upphi_j]drd\sigma d\tau'\\
%&\quad\leq B\sum_{j=0}^k \mathbb{E}[\tilde\upphi_j](0) +B\sum_{j=0}^{k-1} \mathbb{E}[\tilde\upphi_j](\tau)+B\varepsilon \sum_{j=0}^{k-1}\mathbb{I}[\tilde\upphi_j](0,\tau)+B\varepsilon\mathbb{I}^{\rm deg}[\tilde\upphi_k](0,\tau)\\
%&\quad\qquad+B\varepsilon^{-1}\int_0^\tau \int_{\Sigma_{\tau'}}\lp[\frac{(|s|-k)}{r^3}\lp(|\tilde\upphi_{k+1}|^2+|\tilde\upphi_k|^2+|Z\tilde\upphi_k|^2\rp)+a^2\mathbbm{1}_{\{k=s\}}r^{-3}\sum_{j=0}^{|s|}|Z\tilde\upphi_j|^2+\upomega_+\lp|\frac{d\chi_Z}{dr}\rp||Z\tilde\upphi_k|^2\rp]drd\sigma d\tau' \\
&\quad\leq B\sum_{j=0}^k \mathbb{E}[\tilde\upphi_j](0)+B(\varepsilon+|a|\mathbbm{1}_{\{k=|s|\}}) \sum_{j=0}^{k-1}\mathbb{I}^{\rm deg}[\tilde\upphi_k](0,\tau)+B(|s|-k)\varepsilon\mathbb{I}[\tilde\upphi_k](0,\tau)\\
&\quad\qquad+B(|s|-k)\varepsilon^{-1}\int_0^\tau \int_{\Sigma_{\tau'}}r^{-3}\lp(|\tilde\upphi_{k+1}|^2+|\tilde\upphi_k|^2+a^2|Z\tilde\upphi_k|^2\rp)drd\sigma d\tau'\\
&\quad\qquad +Ba^2\varepsilon^{-1}\mathbbm{1}_{\{k=|s|\}}\int_0^\tau \int_{\Sigma_{\tau'}}\lp[r^{-3}\sum_{j=0}^{|s|-1}|Z\tilde\upphi_j|^2+\lp|\frac{d\chi_Z}{dr}\rp||Z\tilde\upphi_k|^2\rp]drd\sigma d\tau'\,, \numberthis \label{eq:phys-space-Killing-multiplier-consequence-1}
\end{align*}
for some small $\varepsilon>0$. To obtain this estimate, we have used the fact that any bulk terms involving $|T\Psi|^2$ come with smallness and: either are supported for $r$ outside the trapping region, in which case they are controlled by $\mathbb I^{\rm deg}$; or are localized to $\mc{R}_{(0,1)}\cup\mc{R}_{(\tau-1,\tau)}$, and thus are controlled by energy fluxes. Similarly, if we simply change from $c_k$ to $\tilde c_k$ for $s>0$, we obtain
\begin{align*}
&\mathbb{E}[\dbtilde\upphi_k](\tau)+\mathbb{E}_{\mc I^+}[\dbtilde\upphi_k](0,\tau)+\mathbb{E}_{\mc H^+}[\dbtilde\upphi_k](0,\tau)+\int_0^\tau \int_{\Sigma_{\tau'}}\frac{2r}{r^2+a^2}(|s|-k)|\underline{\mc{L}}\dbtilde\upphi_k|^2 drd\sigma d\tau' \\
&\quad\leq B\sum_{j=0}^k \mathbb{E}[\dbtilde\upphi_j](0)+B(\varepsilon+|a|\mathbbm{1}_{\{k=|s|\}}) \sum_{j=0}^{k-1}\mathbb{I}^{\rm deg}[\dbtilde\upphi_k](0,\tau)+B(|s|-k)\varepsilon\mathbb{I}[\dbtilde\upphi_k](0,\tau)\\
&\quad\qquad+B(|s|-k)\varepsilon^{-1}\int_0^\tau \int_{\Sigma_{\tau'}}r^{-2}\lp(|\dbtilde\upphi_{k+1}|^2+|\dbtilde\upphi_k|^2+a^2|Z\dbtilde\upphi_k|^2\rp)drd\sigma d\tau'\\
&\quad\qquad +Ba^2\varepsilon^{-1}\mathbbm{1}_{\{k=|s|\}}\int_0^\tau \int_{\Sigma_{\tau'}}\lp[r^{-2}\sum_{j=0}^{|s|-1}|Z\dbtilde\upphi_j|^2+\lp|\frac{d\chi_Z}{dr}\rp||Z\dbtilde\upphi_k|^2\rp]drd\sigma d\tau'\,. \numberthis \label{eq:phys-space-Killing-multiplier-consequence-1-peel}
\end{align*}

We can also choose $\tilde \chi\equiv 0$, and $\chi_T=\chi_R$; then, we obtain
\begin{align*}
&\mathbb{E}[\tilde\upphi_k\mathbbm{1}_{\{r^*\geq 2R^*\}}](\tau)+\mathbb{E}_{\mc I^+}[\tilde\upphi_k](0,\tau)+\int_0^\tau \int_{\Sigma_{\tau'}\cap\{r^*\geq 2R^*\}}\frac{(-\sign s )}{c_k}\frac{dc_k}{dr}|\underline{\mc{L}}\tilde\upphi_k|^2 drd\sigma d\tau'\\
 &\quad\leq  B\sum_{j=0}^k \lp(\mathbb{E}[\tilde\upphi_j](0)+\mathbb{E}[\tilde\upphi_j\mathbbm{1}_{\{R^*\leq r^*\leq 2R^*\}}](\tau)\rp)+\varepsilon\mathbb{I}_0[\tilde\upphi_k\mathbbm{1}_{\{r^*\geq R^*\}}](0,\tau)\\
 &\quad\qquad+B\sum_{j=0}^{k-1}\lp(\mathbb{E}[\tilde\upphi_k\mathbbm{1}_{\{r^*\geq 2R^*\}}](\tau)+\frac{\varepsilon^{-1}}{R^*}\mathbb{I}_0[\tilde\upphi_j\mathbbm{1}_{\{r^*\geq R^*\}}](0,\tau)\rp)\\
 &\quad\qquad +B(|s|-k)\varepsilon^{-1}\int_0^\tau \int_{\Sigma_\tau'\cap[R^*,\infty)} r^{-3}|\tilde\upphi_{k+1}|dr d\sigma d\tau'\\
 &\quad\leq  B\sum_{j=0}^k \lp(\mathbb{E}[\tilde\upphi_j](0)+\mathbb{E}[\tilde\upphi_j\mathbbm{1}_{\{R^*\leq r^*\leq 2R^*\}}](\tau)+\varepsilon\mathbb{I}_0[\tilde\upphi_j\mathbbm{1}_{\{r^*\geq R^*\}}](0,\tau)\rp)\\
 &\quad\qquad+B(|s|-k)\varepsilon^{-1}\int_0^\tau \int_{\Sigma_\tau'\cap[R^*,\infty)} r^{-3}|\tilde\upphi_{k+1}|dr d\sigma d\tau'\,, \numberthis \label{eq:phys-space-Killing-multiplier-consequence-3}
\end{align*}
for sufficiently large $R^*=\varepsilon^{-2}$, where we are using the transport estimate \eqref{eq:phys-space-transport-identity-s<0} to obtain the gain in $R^*$ in the third line in the case $s\leq 0$. As before, the same estimate holds for the weighted solutions $\dbtilde\upphi_k$ if $r^{-3}$ is replaced by $r^{-2}$ in the bulk integral and if $\mathbb{I}_0$ is replaced by $\mathbb{I}_1$.

A similar procedure for $r^*\to -\infty$, choosing $\chi=\chi_T=\chi_R(-r^*)$, 
\begin{align*}
&\mathbb{E}[\tilde\upphi_k\mathbbm{1}_{\{r^*\leq -2R^*\}}](\tau)+\mathbb{E}_{\mc H^+}[\tilde\upphi_k](0,\tau)+\int_0^\tau \int_{\Sigma_{\tau'}\cap\{r^*\leq -2R^*\}}\frac{(-\sign s )}{c_k}\frac{dc_k}{dr}|\underline{\mc{L}}\tilde\upphi_k|^2 drd\sigma d\tau'\\
 &\quad\leq  B\sum_{j=0}^k \lp(\mathbb{E}[\tilde\upphi_j](0)+\mathbb{E}[\tilde\upphi_j\mathbbm{1}_{\{R^*\leq -r^*\leq 2R^*\}}](\tau)+\varepsilon\mathbb{I}_0[\tilde\upphi_j\mathbbm{1}_{\{r^*\leq -R^*\}}](0,\tau)\rp)\\ 
 &\quad\qquad+B(|s|-k)\varepsilon^{-1}\int_0^\tau \int_{\Sigma_\tau'\cap(-\infty,-R^*]} |\tilde\upphi_{k+1}|dr d\sigma d\tau'\,,\numberthis \label{eq:phys-space-Killing-multiplier-consequence-4}
\end{align*}
for sufficiently large $R^*=\varepsilon^{-2}$, and using the transport estimate \eqref{eq:phys-space-transport-identity-s>0} in the case $s>0$. In the above, $\tilde\upphi_k$ may be replaced by $\dbtilde\upphi$, as those $r$-weights are not important.

As we can see from all the above examples, estimates for the data of $\tilde\upphi_{k}$ and for certain spacetime integrals lead to control over the fluxes across $\Sigma_\tau$, $\mc{H}^+_{(0,\tau)}$, $\mc{I}^+_{(0,\tau)}$ and, if $s\neq 0$, integrated control over a null derivative on $\mc{R}_{(0,\tau)}$.

\subsubsection{Virial multiplier identities}

In this section we derive the identities, following from using $L-\uL$ as a commutator, which will allows us to control the bulk terms for the transformed variables on a spacetime slab $\mc{R}_{(0,\tau)}$, in terms of the fluxes through its boundary.

\begin{lemma}[Virial multiplier identity] \label{lemma:phys-space-y-multiplier-identity} Let $|s|\in\mathbb{Z}$, and $y$ be a continuous, piecewise $C^1$ function. Suppose that, for $0\leq k\leq |s|$, $\swei{\tilde\upphi}{s}_k$ are homogeneous solutions to the transformed system of Definition~\ref{def:transformed-system}; we drop the superscripts in what follows to ease the notation. We have an identity of the form \eqref{eq:template-physical-space-identity} where  the bulk term is given by
\begin{align*}
I&=\frac12 y'|\mc L\tilde\upphi_k|^2+\frac12\lp( y'-2\frac{c_k'}{c_k}y\rp)|\underline{\mc L}\tilde\upphi_k|^2-(wy)'|\mathring{\slashed{\nabla}} \tilde\upphi_k|^2-\lp[y \lp(U_k-\lp(\frac{c_k'}{c_k}\rp)'\rp)\rp]'|\tilde\upphi_k|^2-\frac{2a}{r^2+a^2}(wy)'|Z\tilde\upphi_k|^2\\
&\qquad+a\lp((wy)'-\frac{2ryw}{r^2+a^2}\rp)\Re[Z\tilde\upphi_k\overline{(L+\uL)\tilde\upphi_k}]+(wy)'a^2\sin^2\theta|T\tilde\upphi_k|^2-as\cos\theta(wy)'\Im[\tilde\upphi_k\overline{T\tilde\upphi_k}]\\
&\qquad + y\frac{c_k'}{c_k}\Re\lp[\underline{\mc L} \tilde\upphi_k\overline{\lp(\sign s \frac{c_k'}{c_k}\tilde\upphi_k+\frac{wc_{k+1}}{c_k}\tilde\upphi_{k+1}\rp)}\rp] -y\lp(\frac{c_k'}{c_k}-(|s|-k)\frac{w'}{w}\rp)\lp(\frac{wc_{k+1}}{c_k}\rp)^2|\tilde\upphi_{k+1}|^2\\
&\qquad -y\lp(\frac{c_k'}{c_k}-(|s|-k)\frac{w'}{w}\rp)\frac{wc_{k+1}}{c_k}\Re\lp[\lp(\sign s\frac{c_k'}{c_k}\tilde\upphi_k-\underline{\mc L}\tilde\upphi_k\rp)\overline{\tilde\upphi_{k+1}}\rp]\\
&\qquad +\frac{4arw}{r^2+a^2}(|s|-k)y\frac{wc_{k+1}}{c_k}\Re[\overline{Z\tilde\upphi_k}\tilde\upphi_{k+1}]-\frac{4arw}{r^2+a^2}y(|s|-k)\Re[\overline{Z\tilde\upphi_k}\underline{\mc L}\tilde\upphi_{k}]\\
&\qquad+\sign s\sum_{j=0}^{k-1}wy\frac{c_j}{c_k}\sign s\Re\lp[(ac_{s,k,j}^{\rm id}+ac_{s,k,j}^{Z}Z)\tilde\upphi_j\overline{\underline{\mc{L}}\tilde\upphi_k}\rp]\\
&\qquad-\sum_{j=0}^{k-1}\lp[\lp(yw\frac{c_j}{c_k}ac_{s,k,j}^{\rm id}\rp)'\Re[\tilde\upphi_j\overline{\tilde\upphi_k}]+\lp(yw\frac{c_j}{c_k}ac_{s,k,j}^{Z}\rp)'\Re[Z\tilde\upphi_j\overline{\tilde\upphi_k}]\rp]\\
&\qquad+\sign s\sum_{j=0}^{k-1}wy\frac{c_j}{c_k}\sign s\Re\lp[\lp(\sign s \frac{c_j'}{c_j}\tilde\upphi_j+w\frac{c_{j+1}}{c_j}\tilde\upphi_{j+1}\rp)(ac_{s,k,j}^{\rm id}-ac_{s,k,j}^{Z}Z)\overline{\tilde\upphi_k}\rp]\,,
\end{align*}
and the boundary terms are determined through:
\begin{align*}
2F_L&=-\frac12y|\uL \tilde\upphi_k|^2+yw|\mathring{\slashed{\nabla}} \tilde\upphi_k|^2+y\lp[U_k-\lp(\frac{c_k'}{c_k}\rp)'\rp]|\tilde\upphi_k|^2+\frac{2a^2}{r^2+a^2}wy|Z\tilde\upphi_k|^2-2awy\Re[Z\tilde\upphi_k\overline{\uL\upphi_k}]\\
&\qquad-\frac{a}{r^2+a^2}\lp(\frac{a^3wy\sin^2\theta}{r^2+a^2}|Z\tilde\upphi_k|^2-2a\Re[Z\tilde\upphi_k\overline{\uL\tilde\upphi_k}]\rp)+awys\cos\theta\Im\lp[\tilde\upphi_k\overline{\lp(2\uL-\frac{a}{r^2+a^2}Z\rp)\tilde\upphi_k}\rp]\\
&\qquad-\frac14a^2wy\sin^2\theta\lp(|\uL\tilde\upphi_k|^2+3|L\tilde\upphi_k|^2+2\Re[L\tilde\upphi_k\overline{\uL\tilde\upphi_k}]\rp)\\
&\qquad +\frac{(1-\sign s)}{2} \sum_{j=0}^{k-1}yw\frac{c_j}{c_k}\Re[\lp(ac_{s,k,j}^{\rm id}+ac_{s,k,j}^{Z}Z\rp)\tilde\upphi_j\overline{\tilde\upphi_k}]
\end{align*}
with $F_{\uL}$ obtained by replacing $L$ by $\uL$, inverting the sign in every term and replacing $\sign s$ by $-\sign s$. 

Alternatively, we may use the identity 
\begin{align*}
&y\frac{c_k'}{c_k}\Re\lp[\underline{\mc L} \tilde\upphi_k\overline{\lp(\sign s \frac{c_k'}{c_k}\tilde\upphi_k+\frac{wc_{k+1}}{c_k}\tilde\upphi_{k+1}\rp)}\rp]\\
%&= y\frac{c_k'}{c_k}\Re\lp[\underline{\mc L}\tilde\upphi_k\overline{{\mc L}\tilde\upphi_k}\rp]\\
&\quad  = \mc L \lp(y\frac{c_k'}{c_k}\Re\lp[\underline{\mc L}\tilde\upphi_k\overline{\tilde\upphi_k}\rp]\rp) +\sign s \lp(y\frac{c_k'}{c_k}\rp)'\sign s \Re\lp[\underline{\mc L}\tilde\upphi_k\overline{\tilde\upphi_k}\rp]- y\frac{c_k'}{c_k}\Re\lp[\mc L \underline{\mc L}\tilde\upphi_k\overline{{\mc L}\tilde\upphi_k}\rp]\,,
\end{align*}
 replace the last term on the right hand side through  \eqref{eq:transformed-k-tilde} and integrate by parts once more to split up the derivatives.
\end{lemma}
\begin{proof}
The proof follows from multiplying \eqref{eq:transformed-k-tilde} by $y\overline{(L-\uL)\tilde\upphi_k}$, taking the real part and  integrating by parts. The contributions from the left hand side of \eqref{eq:transformed-k-tilde}  are easy to derive when $k=|s|$, see the steps in \cite[Appendix B]{Dafermos2017}. In the case $k<|s|$, we write  $(L-\uL)=-\sign s (\mc L-\underline{\mc L})$, and use the identity \eqref{eq:transformed-transport-tilde} in all the new terms except for the one arising from $\lp(\frac{c_k'}{c_k}\rp)'$ in \eqref{eq:transformed-k-tilde}. For the coupling terms, we integrate by parts to shift the $\mc L$ derivative onto the $j$th transformed variable. 
\end{proof}

Let  $\chi_R(r^*)$ be a cutoff equal to 1 for $r^*\geq 2R^*$ and supported in $\{r^*\geq R^*\}$. Consider Lemma~\ref{lemma:phys-space-y-multiplier-identity} with $y=\chi_R(r^*)(1-|r^*|^{-\delta})$ for some $\delta\in(0,1]$. Then, if $R^*$ is sufficiently large,
\begin{align*}
&b(\delta)\mathbb{I}_{-\delta}[\tilde\upphi_k\mathbbm{1}_{\{r^*\geq 2R^*\}}](0,\tau) \\
&\quad\leq\sum_{j=0}^k\lp(\mathbb{E}[\tilde\upphi_j\mathbbm{1}_{\{r^*\geq R^*\}}](\tau)+\mathbb{E}[\tilde\upphi_j](0)+\mathbb{I}_0[\tilde\upphi_j\mathbbm{1}_{\{R^*\leq r^*\leq 2R^*\}}](0,\tau)\rp)+\mathbb{E}_{\mc I^+}[\tilde\upphi_k](0,\tau)\\
&\quad\qquad+\frac{1}{R^*}\sum_{j=0}^{k-1}\mathbb{I}_{-\delta}[\tilde\upphi_j](0,\tau)+(|s|-k)\int_{\mc R_{(0,\tau)}}r^{-3}|\tilde\upphi_{k+1}|^2\mathbbm{1}_{\{r^*\geq R^*\}}drd\sigma d\tau'\\
&\quad\leq\sum_{j=0}^{k}\lp(\mathbb{E}[\tilde\upphi_j\mathbbm{1}_{\{r^*\geq R^*\}}](\tau)+\mathbb{E}[\tilde\upphi_j](0)+\mathbb{E}_{\mc I^+}[\tilde\upphi_j](0,\tau)+\mathbb{I}_0[\tilde\upphi_j\mathbbm{1}_{\{R^*\leq r^*\leq 2R^*\}}](0,\tau)\rp)\\
&\quad\qquad+(|s|-k)\int_{\mc R_{(0,\tau)}}r^{-3}|\tilde\upphi_{k+1}|^2\mathbbm{1}_{\{r^*\geq R^*\}}drd\sigma d\tau'
\,. \numberthis\label{eq:phys-space-y-multiplier-consequence-infty}
\end{align*}
Here, for $s<0$, we have made use of the estimate \eqref{eq:phys-space-transport-identity-s<0} when needed to obtain the gain in $R^*$ in the first term on the second line. A similar estimate holds for the weighted solutions $\dbtilde\upphi_k$ if $s>0$, but where the term in $\tilde \upphi_{k+1}$ is absent from the right hand side; this is because $\lp(\frac{\dbtilde c_k'}{\tilde c_k}-(|s|-k)\frac{w'}{w}\rp)<0$ so that one in fact controls the corresponding bulk term in $\tilde \upphi_{k+1}$ rather than it emerging as an error.

By a similar argument, choosing $y=\chi_R(-r^*)(1-|r^*|^{-1})$ we get, for  $R^*$ is sufficiently large,
\begin{align*}
&\mathbb{I}_{0}[\tilde\upphi_k\mathbbm{1}_{\{r^*\leq -2R^*\}}](0,\tau)  \\
&\quad\leq B\sum_{j=0}^{k}\lp(\mathbb{E}[\tilde\upphi_j\mathbbm{1}_{\{r^*\leq -R^*\}}](\tau)+\mathbb{E}[\tilde\upphi_j](0)+\mathbb{E}_{\mc H^+}[\tilde\upphi_j](0,\tau)+\mathbb{I}_{0}[\tilde\upphi_j\mathbbm{1}_{\{-R^*\leq r^*\leq -2R^*\}}](0,\tau)\rp)\\
&\quad\qquad+B(|s|-k)\int_{\mc R_{(0,\tau)}}|\tilde\upphi_{k+1}|^2\mathbbm{1}_{\{r^*\leq -R^*\}}drd\sigma d\tau'
\,, \numberthis\label{eq:phys-space-y-multiplier-consequence-hor}
\end{align*}
making use of the \eqref{eq:phys-space-transport-identity-s>0} for $s>0$ when necessary. The exception to the  procedure we have just outlined to arrive at \eqref{eq:phys-space-y-multiplier-consequence-infty} and \eqref{eq:phys-space-y-multiplier-consequence-hor} is the case $s<0$ and $k<|s|$, where Lemma~\ref{lemma:phys-space-y-multiplier-identity} does not yield always control over the $\frac{r}{(r-r_+)}|\uL\tilde\upphi_k|^2r^{-1-\delta}$ term in $\mathbb{I}_{-\delta}$ and $\mathbb{I}_{0}$.  For such $k$ and $s$, we appeal to \eqref{eq:phys-space-Killing-multiplier-consequence-3} to recover control over that term and obtain \eqref{eq:phys-space-y-multiplier-consequence-infty} and \eqref{eq:phys-space-y-multiplier-consequence-hor}.

\subsubsection{Redshift and \texorpdfstring{$r^p$-weighted}{Dafermos--Rodnianski} multiplier identities}

In this section, we study identities which induce stronger weighted-norms, as $r\to r_+ $ and $r\to \infty$, for the transformed variables.

\begin{lemma}[Redshift multiplier identities] \label{lemma:phys-space-redshift-multiplier-identity}Let $|s|\in\mathbb{Z}$, and $\xi$ be a smooth function. Suppose that, for $0\leq k\leq |s|$, $\swei{\tilde\upphi}{s}_k$ are homogeneous solutions to the transformed system of Definition~\ref{def:transformed-system}; we drop the superscripts in what follows to ease the notation. Then, an identity of the form \eqref{eq:template-physical-space-identity} holds where the boundary terms are determined by 
\begin{align*}
F_{L}&=\frac12\frac{\xi}{r-r_+} \lp|\uL \tilde\upphi_k\rp|^2+\frac12a\frac{\xi w}{r-r_+}\Re[Z\tilde\upphi_k\overline{\uL\tilde\upphi_k}]+\frac12a^2\sin^2\theta\frac{\xi w}{r-r_+}\Re[T\tilde\upphi_k\overline{\uL\tilde\upphi_k}]\\
&\qquad +\mathbbm{1}_{\{s<0\}} \lp[\frac{\xi}{2(r-r_+)}\lp(\frac{c_k'}{c_k}-(|s|-k)\frac{w'}{w}\rp)\rp]'|\tilde\upphi_k|^2\\
&\qquad +\mathbbm{1}_{\{s<0\}}\frac{\xi w}{2(r-r_+)}\lp(\frac{c_k'}{c_k}-(|s|-k)\frac{w'}{w}\rp)a^2\sin^2\theta\Re[T\tilde\upphi_k\overline{\tilde\upphi_k}]\,,\\
F_{\uL}&=\frac{\xi w}{r-r_+}\lp(\frac12|\mathring{\slashed{\nabla}} \tilde\upphi_k|^2+\frac{a^2}{r^2+a^2}|Z\tilde\upphi_k|^2+\frac{U_k}{2w}|\tilde\upphi_k|^2-\frac12 a^2\sin^2\theta|T\tilde\upphi_k|^2-\frac12a\Re[Z\tilde\upphi_k\overline{L\tilde\upphi_k}]\rp)\\
&\qquad+\frac12a^2\sin^2\theta\frac{\xi w}{r-r_+}\Re[T\tilde\upphi_k\overline{\uL \tilde\upphi_k}]+\mathbbm{1}_{\{s>0\}}\sum_{j=0}^{k-1}\frac{c_j}{c_k}\frac{\xi w}{r-r_+}\Re[(a c_{s,k,j}^{\rm id}+ac_{s,k,j}^Z Z)\tilde\upphi_j\overline{\tilde\upphi_k}]\\
&\qquad -\mathbbm{1}_{\{s<0\}}\frac{\xi}{2(r-r_+)}\lp\{2\lp(\frac{c_k''}{c_k}-(|s|-k)\frac{w'}{w}\frac{c_k'}{c_k}\rp)|\tilde\upphi_k|^2+\frac{\xi}{r-r_+}\lp(\frac{c_k'}{c_k}-(|s|-k)\frac{w'}{w}\rp)\Re[L\tilde\upphi_k\overline{\tilde\upphi_k}]\rp\}\\
&\qquad -\mathbbm{1}_{\{s<0\}}\frac{\xi}{2(r-r_+)}\frac{c_k'}{c_k}\lp(\frac{c_k'}{c_k}-(|s|-k)\frac{w'}{w}\rp)|\tilde\upphi_k|^2\\
&\qquad +\mathbbm{1}_{\{s<0\}}\frac{\xi w}{2(r-r_+)}\lp(\frac{c_k'}{c_k}-(|s|-k)\frac{w'}{w}\rp)a^2\sin^2\theta\Re[T\tilde\upphi_k\overline{\tilde\upphi_k}]\,,
\end{align*}
and the bulk term is determined via
\begin{align*}
I&=\lp(\frac{r-r_-}{2(r^2+a^2)}+\frac{c_k'}{c_k}\mathbbm{1}_{\{s<0\}}\rp)\frac{\xi}{r-r_+}\lp|\uL\tilde\upphi_k\rp|^2+\frac{2ar(r-r_-)}{(r^2+a^2)^3}\xi\Re[Z\tilde\upphi_k\overline{\uL\tilde\upphi_k}]+\lp(\frac{U_k\xi}{2(r-r_+)}\rp)'|\tilde\upphi_k|^2\\
&\qquad-\frac{\xi'}{2(r-r_+)}|\uL\tilde\upphi_k|^2-\frac{\xi'(r-r_-)}{(r^2+a^2)^2}\lp(|\mathring{\slashed{\nabla}} \tilde\upphi_k|^2+a\Re[Z\tilde\upphi_k\overline{T\tilde\upphi_k}]+\frac12a^2\sin^2\theta|T\tilde\upphi_k|^2\rp)\\
&\qquad-2as\cos\theta\frac{\xi(r-r_-)}{(r^2+a^2)^2}\Im[T\tilde\upphi_k\overline{\uL\tilde\upphi_k}]\\
&\qquad +\mathbbm{1}_{\{s>0\}}\sum_{j=0}^{k-1}\lp\{\lp(a c_{s,k,j}^{\rm id}\frac{c_j}{c_k}\frac{\xi(r-r_-)}{(r^2+a^2)^2}\rp)'\Re[\tilde\upphi_j\overline{\tilde\upphi_k}]+\lp(a c_{s,k,j}^{\rm Z}\frac{c_j}{c_k}\frac{\xi(r-r_-)}{(r^2+a^2)^2}\rp)'\Re[Z\tilde\upphi_j\overline{\tilde\upphi_k}]\rp\}\\
&\qquad +\sum_{j=0}^{k-1}\frac{\xi(r-r_-)}{(r^2+a^2)^2}\Re\lp[(a c_{s,k,j}^{\rm id}+ac_{s,k,j}^Z Z)\lp(\frac{c_{j}}{c_k}\tilde\upphi_{j}\uL \overline{\tilde\upphi_k}\mathbbm{1}_{\{s<0\}}-\frac{c_{j+1}}{c_k}w\tilde\upphi_{j+1}\overline{\tilde\upphi_k}\mathbbm{1}_{\{s>0\}}\rp)\rp]\\
&\qquad -\mathbbm{1}_{\{s>0\}}\frac{\xi}{r-r_+}\lp\{\lp(\frac{c_k'}{c_k}\rp)'\Re[\uL\tilde\upphi_k\overline{\tilde\upphi_k}]+\lp(\frac{c_k'}{c_k}\rp)^2\Re[\uL\tilde\upphi_k\overline{\tilde\upphi_k}]\rp\}\\
&\qquad +\mathbbm{1}_{\{s>0\}}\frac{\xi}{r-r_+}\frac{wc_{k+1}}{c_k}\lp\{\lp(\frac{c_k'}{c_k}-(|s|-k)\frac{w'}{w}\rp)\Re[\uL\tilde\upphi_k\overline{\tilde\upphi_{k+1}}]-\frac{c_k'}{c_k}\Re[L\tilde\upphi_k\overline{\tilde\upphi_{k+1}}]\rp\}\\
&\qquad +\mathbbm{1}_{\{s<0\}}\lp\{\lp[\frac{\xi}{2(r-r_+)}\lp(\frac{c_k''}{c_k}-(|s|-k)\frac{w'}{w}\frac{c_k'}{c_k}\rp)\rp]' -\lp[\frac{\xi}{2(r-r_+)}\lp(\frac{c_k'}{c_k}-(|s|-k)\frac{w'}{w}\rp)\rp]''\rp\}|\tilde\upphi_k|^2\\
&\qquad +\mathbbm{1}_{\{s<0\}}\lp\{\lp[\frac{\xi}{2(r-r_+)}\frac{c_k'}{c_k}\lp(\frac{c_k'}{c_k}-(|s|-k)\frac{w'}{w}\rp)\rp]'+\frac{\xi}{r-r_+}\lp[\lp(\frac{c_k'}{c_k}\rp)'-U_k\rp]\lp(\frac{c_k'}{c_k}-(|s|-k)\frac{w'}{w}\rp)\rp\}|\tilde\upphi_k|^2\\
&\qquad+\mathbbm{1}_{\{s<0\}}\frac{\xi}{r-r_+}\lp(\frac{c_k'}{c_k}-(|s|-k)\frac{w'}{w}\rp)^2\frac{wc_{k+1}}{c_k}\Re[\tilde\upphi_{k+1}\overline{\tilde\upphi_k}]-\mathbbm{1}_{\{s<0\}}\frac{\xi w}{r-r_+}\lp(\frac{c_k'}{c_k}-(|s|-k)\frac{w'}{w}\rp)|\mathring{\slashed\nabla}\tilde\upphi_k|^2\\
&\qquad-\mathbbm{1}_{\{s<0\}}\frac{\xi w}{r-r_+}\lp(\frac{c_k'}{c_k}-(|s|-k)\frac{w'}{w}\rp)\lp[a^2\sin^2\theta|T\tilde\upphi_k|^2 +2a\Re[T\tilde\upphi_k\overline{Z\tilde\upphi_k}]+2as\cos\theta\Im[T\tilde\upphi_k\overline{\upphi_k}]\rp]
 \,.
\end{align*}
Alternatively, we can remove all terms with $\mathbbm{1}_{\{s<0\}}$ in $F_L$, $F_{\uL}$ and in the last 4 lines of $I$ by adding to $I$ 
\begin{align*}
\mathbbm{1}_{\{s<0\}}\frac{\xi}{r-r_+}\lp\{\lp(\frac{c_k'}{c_k}-(|s|-k)\frac{w'}{w}\rp)\frac{wc_{k+1}}{c_k}\Re[\tilde\upphi_{k+1}\overline{\uL\tilde\upphi_k}]-\lp(\frac{c_k'}{c_k}\rp)'\Re[\tilde\upphi_k\overline{\uL\tilde\upphi_k}]\rp\}\,.
\end{align*}
\end{lemma}
\begin{proof}
The proof follows from multiplying \eqref{eq:transformed-k-tilde} by $(r-r_+)^{-1}\xi\overline{\uL \tilde\upphi_k}$, taking the real part and  integrating by parts. In the case $k=|s|$, the contributions from the left hand side of \eqref{eq:transformed-k-tilde}  are easily obtained, see the steps in \cite[Appendix B]{Dafermos2017}. 

In the case $k<|s|$ and $s>0$, we use \eqref{eq:transformed-transport-tilde} to simplify
\begin{align*}
&-\frac{c_k'}{c_k}\xi\Re[L\tilde\upphi_{k}\overline{\uL\tilde\upphi_k}]=-
\lp(\frac{c_k'}{c_k}\rp)^2\xi\Re[L\tilde\upphi_{k}\overline{\tilde\upphi_{k}}]-\frac{wc_{k+1}}{c_k}\frac{c_k'}{c_k}\xi\Re[L\tilde\upphi_{k}\overline{\tilde\upphi_{k+1}}]\,.
\end{align*}
For the coupling terms to $j<k$, if $s>0$ we integrate by parts to move the $\uL$ derivative onto the $j$th transformed variable and then make use of \eqref{eq:transformed-transport}. 

In the case $k<|s|$ but $s<0$, we want to move the $\uL$ derivative in the terms where there is coupling to the $k+1$ equation. We deduce:
\begin{align*}
&\frac{\xi}{r-r_+}\Re\lp\{\lp[\lp(-\frac{c_k''}{c_k}+(|s|-k)\frac{w'}{w}\frac{c_k'}{c_k}\rp)\tilde\upphi_k-\frac{c_k'}{c_k}(L-\uL)\tilde\upphi_k+(|s|-k)\frac{w'}{w}L\tilde\upphi_k\rp]\overline{\uL\tilde\upphi_k}\rp\}\\
&\quad= -\frac{\xi}{2(r-r_+)}\lp(\frac{c_k''}{c_k}-(|s|-k)\frac{w'}{w}\frac{c_k'}{c_k}\rp)\uL|\tilde\upphi_k|^2+\frac{\xi}{r-r_+}\frac{c_k'}{c_k}|\uL\tilde\upphi_k|^2\\
&\qquad\quad-\frac{\xi}{r-r_+}\lp(\frac{c_k'}{c_k}-(|s|-k)\frac{w'}{w}\rp)\Re[L\tilde\upphi_k\overline{\uL\tilde\upphi_k}]\\
&\quad =-\uL\lp\{\frac{\xi}{2(r-r_+)}\lp(\frac{c_k''}{c_k}-(|s|-k)\frac{w'}{w}\frac{c_k'}{c_k}\rp)|\tilde\upphi_k|^2+\frac{\xi}{r-r_+}\lp(\frac{c_k'}{c_k}-(|s|-k)\frac{w'}{w}\rp)\Re[L\tilde\upphi_k\overline{\tilde\upphi_k}]\rp\}\\
&\quad\qquad +L\lp\{\lp[\frac{\xi}{2(r-r_+)}\lp(\frac{c_k'}{c_k}-(|s|-k)\frac{w'}{w}\rp)\rp]'|\tilde\upphi_k|^2\rp\} +\frac{\xi}{r-r_+}\frac{c_k'}{c_k}|\uL\tilde\upphi_k|^2\\
&\quad\qquad +\lp[\frac{\xi}{2(r-r_+)}\lp(\frac{c_k''}{c_k}-(|s|-k)\frac{w'}{w}\frac{c_k'}{c_k}\rp)\rp]'|\tilde\upphi_k|^2 -\lp[\frac{\xi}{2(r-r_+)}\lp(\frac{c_k'}{c_k}-(|s|-k)\frac{w'}{w}\rp)\rp]''|\tilde\upphi_k|^2\\
&\quad\qquad +\frac{\xi}{r-r_+}\lp(\frac{c_k'}{c_k}-(|s|-k)\frac{w'}{w}\rp)\Re[\uL L\tilde\upphi_k\overline{\tilde\upphi_k}]\,,
\end{align*}
for which we can now invoke the PDE \eqref{eq:transformed-k-tilde}.
\end{proof}

We will consider Lemma~\ref{lemma:phys-space-redshift-multiplier-identity} with $\xi(r^*)=\chi_R(-r^*)$, and assuming $|a|<M$. For the case $s>0$, using additionally the transport estimate \eqref{eq:phys-space-transport-identity-s>0}, we can easily deduce that 
\begin{align*}
&\overline{\mathbb{E}}[\tilde\upphi_k\mathbbm{1}_{(-\infty,-2R^*]}](\tau)+\overline{\mathbb{E}}_{\mc H^+}[\tilde\upphi_k](0,\tau)+\overline{\mathbb{I}}[\tilde\upphi_k\mathbbm{1}_{(-\infty,-2R^*]}](0,\tau)\\
&\quad \leq B(a)\sum_{j=0}^k\lp(\mathbb{E}_0[\tilde\upphi_j\mathbbm{1}_{(-\infty,-R^*]}](\tau)+\mathbb{E}_{\mc I^+,0}[\tilde\upphi_j](0,\tau)+\frac{1}{R^*}\mathbb{I}_0[\tilde\upphi_j\mathbbm{1}_{(-\infty,-R^*}](0,\tau)\rp)\\ &\quad\qquad+B(a)\sum_{j=0}^k \overline{\mathbb{E}}[\tilde\upphi_j](0)+B(a)(|s|-k)\int_0^\tau\int_{\Sigma_{\tau'}\cap(-\infty,-R^*]}|\tilde\upphi_{k+1}|^2drd\sigma d\tau\,. \numberthis\label{eq:phys-space-redshift-multiplier-consequence-s>0}
\end{align*}
In the case $s<0$, we can use the alternative given at the end of Lemma~\ref{lemma:phys-space-redshift-multiplier-identity} to obtain
\begin{align*}
&\overline{\mathbb{E}}[\tilde\upphi_k\mathbbm{1}_{(-\infty,-2R^*]}](\tau)+\overline{\mathbb{E}}_{\mc H^+}[\tilde\upphi_k](0,\tau)+\overline{\mathbb{I}}[\tilde\upphi_k\mathbbm{1}_{(-\infty,-2R^*]}](0,\tau)\\
&\quad\leq B\overline{\mathbb{E}}[\tilde\upphi_k](0)+B\mathbb{E}[\tilde\upphi_k\mathbbm{1}_{(-\infty,-R^*]}](\tau)+B\mathbb{E}_{\mc H^+}[\tilde\upphi_k](0,\tau)+B\mathbb{I}[\tilde\upphi_k\mathbbm{1}_{(-\infty,-R^*]}](0,\tau)\\
&\quad\qquad+ \frac{B}{R^*}\sum_{j=0}^{k-1}\mathbb{I}[\tilde\upphi_j\mathbbm{1}_{(-\infty,-R^*]}](0,\tau)+B(|s|-k)\int_0^\tau\int_{\Sigma_{\tau'}\cap(-\infty,-R^*]}|\tilde\upphi_{k+1}|^2dr^*\,. \numberthis\label{eq:phys-space-redshift-multiplier-consequence-s<0}
\end{align*}

\begin{lemma}[$r^p$ multiplier identity] \label{lemma:phys-space-rp-multiplier-identity} Let $|s|\in\mathbb{Z}$, and $\xi$ be a smooth function. Suppose that, for $0\leq k\leq |s|$, $\swei{\tilde\upphi}{s}_k$ are homogeneous solutions to the transformed system of Definition~\ref{def:transformed-system}; we drop the superscripts in what follows to ease the notation. Then, an identity of the form \eqref{eq:template-physical-space-identity} holds where the boundary terms are determined by
\begin{align*}
F_{\uL}&=\frac12 r^p \xi \lp|L \tilde\upphi_k\rp|^2+\frac12awr^p\xi \Re[Z\tilde\upphi_k\overline{L\tilde\upphi_k}]+\frac12a^2\sin^2\theta wr^p\xi\Re[T\tilde\upphi_k\overline{L\tilde\upphi_k}]\\
&\qquad -\mathbbm{1}_{\{s>0\}} \lp[\frac12r^p\xi\lp(\frac{c_k'}{c_k}-(|s|-k)\frac{w'}{w}\rp)\rp]'|\tilde\upphi_k|^2\\
&\qquad +\mathbbm{1}_{\{s>0\}} wr^p\xi\lp(\frac{c_k'}{c_k}-(|s|-k)\frac{w'}{w}\rp)a^2\sin^2\theta \Re[T\tilde\upphi_k\overline{\tilde\upphi_k}]\\
&\qquad +\mathbbm{1}_{\{s>0\}}\frac12 wr^p\xi\lp(\frac{c_k'}{c_k}-(|s|-k)\frac{w'}{w}\rp)a^2\sin^2\theta\Re[T\tilde\upphi_k\overline{\tilde\upphi_k}]\,,\\
F_{L}&=\frac12 wr^p\xi |\mathring{\slashed{\nabla}} \tilde\upphi_k|^2+\frac{a^2}{r^2+a^2}wr^p\xi |Z\tilde\upphi_k|^2+\frac{1}{2} \lp(U_k-\lp(\frac{c_k'}{c_k}\rp)'\mathbbm{1}_{\{s>0\}}\rp) r^p\xi|\tilde\upphi_k|^2-\frac12awr^p\xi\Re[Z\tilde\upphi_k\overline{\uL\tilde\upphi_k}]\\
&\qquad-\frac12a^2\sin^2\theta wr^p\xi\lp(\frac14|L\tilde\upphi_k|^2-\frac14|\uL\tilde\upphi_k|^2-\frac{a^2}{(r^2+a^2)^2}|Z\tilde\upphi_k|^2+\frac{a}{r^2+a^2}\Re[\uL\tilde\upphi_k\overline{Z\tilde\upphi_k}]\rp)\\
&\qquad+\mathbbm{1}_{\{s<0\}}\sum_{j=0}^{k-1}\frac{c_j}{c_k}wr^p\xi\Re[(a c_{s,k,j}^{\rm id}+ac_{s,k,j}^Z Z)\tilde\upphi_j\overline{\tilde\upphi_k}]\\
&\qquad -\mathbbm{1}_{\{s>0\}}r^p\xi\lp\{\frac12\lp[\lp(\frac{c_k'}{c_k}\rp)'+\lp(\frac{c_k'}{c_k}-(|s|-k)\frac{w'}{w}\rp)\frac{c_k'}{c_k}\rp]|\tilde\upphi_k|^2-r^p\xi\lp(\frac{c_k'}{c_k}-(|s|-k)\frac{w'}{w}\rp)\Re[\uL\tilde\upphi_k\overline{\tilde\upphi_k}]\rp\}\\
&\qquad +\mathbbm{1}_{\{s>0\}}\frac12 wr^p\xi\lp(\frac{c_k'}{c_k}-(|s|-k)\frac{w'}{w}\rp)a^2\sin^2\theta\Re[T\tilde\upphi_k\overline{\tilde\upphi_k}]\,,
\end{align*}
and the bulk term is determined via
\begin{align*}
I&=\lp((r^p\xi)'-\frac{c_k'}{c_k}r^p\xi\mathbbm{1}_{\{s>0\}}\rp)\lp|L\tilde\upphi_k\rp|^2-\frac{2ar}{r^2+a^2}wr^p\xi\Re[Z\tilde\upphi_k L\overline{\tilde\upphi_k}]-\frac12\lp(U_k r^p\xi\rp)'|\tilde\upphi_k|^2\\
&\qquad-(wr^p\xi)'\lp(|\mathring{\slashed{\nabla}} \tilde\upphi_k|^2-\frac12a\Re[Z\tilde\upphi_k\overline{(L+\uL)\tilde\upphi_k}]+\frac{a}{r^2+a^2}|Z\tilde\upphi_k|^2+\frac12a^2\sin^2\theta|T\tilde\upphi_k|^2\rp)\\
&\qquad-as\cos\theta wr^p\xi\Im\lp[\lp(\uL\tilde\upphi_k-\frac{2a}{r^2+a^2}\rp)\overline{L\tilde\upphi_k}\rp]\\
&\qquad -\mathbbm{1}_{\{s<0,k=|s|\}}\sum_{j=0}^{k-1}\lp\{\lp(a c_{s,k,j}^{\rm id}\frac{c_j}{c_k}wr^p\xi\rp)'\Re[\tilde\upphi_j\overline{\tilde\upphi_k}]+\lp(a c_{s,k,j}^{\rm Z}\frac{c_j}{c_k}wr^p\xi\rp)'\Re[Z\tilde\upphi_j\overline{\tilde\upphi_k}]\rp\}\\
&\qquad +\sum_{j=0}^{k-1}wr^p\xi\Re\lp[(a c_{s,k,j}^{\rm id}+ac_{s,k,j}^Z Z)\lp(\frac{c_{j}}{c_k}\tilde\upphi_{j}L\overline{\tilde\upphi_k}\mathbbm{1}_{\{s>0\}}-\frac{c_{j+1}}{c_k}w\tilde\upphi_{j+1}\overline{\tilde\upphi_k}\mathbbm{1}_{\{s<0\}}\rp)\rp]\\
&\qquad +\mathbbm{1}_{\{s<0\}}r^p\xi\lp\{\frac{c_k'}{c_k}\frac{wc_{k+1}}{c_k}\Re[\tilde\upphi_{k+1}\overline{\uL\tilde\upphi_k}]-\lp(\frac{c_k'}{c_k}-(|s|-k)\frac{w'}{w}\rp)\frac{wc_{k+1}}{c_k}\Re[L\tilde\upphi_k\overline{\tilde\upphi_{k+1}}]\rp\}\\
&\qquad -\mathbbm{1}_{\{s<0\}}r^p\xi\lp\{\lp(\frac{c_k'}{c_k}\rp)'\Re[L\tilde\upphi_k\overline{\tilde\upphi_k}]+\lp(\frac{c_k'}{c_k}\rp)^2\Re[\uL\tilde\upphi_k\overline{\tilde\upphi_k}]\rp\}\\
&\qquad +\mathbbm{1}_{\{s>0\}}\lp\{\lp[\frac12 r^p\xi\lp(\frac{c_k'}{c_k}\rp)'\rp]'|\tilde\upphi_k|^2 +\lp[\frac12 r^p\xi\lp(\frac{c_k'}{c_k}-(|s|-k)\frac{w'}{w}\rp)\rp]''|\tilde\upphi_k|^2\rp\}\\
&\qquad +\mathbbm{1}_{\{s>0\}}r^p\xi\lp[\lp(\frac{c_k'}{c_k}\rp)'-U_k\rp]\lp(\frac{c_k'}{c_k}-(|s|-k)\frac{w'}{w}\rp)|\tilde\upphi_k|^2\\
&\qquad+\mathbbm{1}_{\{s>0\}}r^p\xi\lp(\frac{c_k'}{c_k}-(|s|-k)\frac{w'}{w}\rp)^2\frac{wc_{k+1}}{c_k}\Re[\tilde\upphi_{k+1}\overline{\tilde\upphi_k}]-\mathbbm{1}_{\{s>0\}}wr^p\xi \lp(\frac{c_k'}{c_k}-(|s|-k)\frac{w'}{w}\rp)|\mathring{\slashed\nabla}\tilde\upphi_k|^2\\
&\qquad-\mathbbm{1}_{\{s>0\}}w r^p\xi\lp(\frac{c_k'}{c_k}-(|s|-k)\frac{w'}{w}\rp)\lp[a^2\sin^2\theta|T\tilde\upphi_k|^2 +2a\Re[T\tilde\upphi_k\overline{Z\tilde\upphi_k}]+2as\cos\theta\Im[T\tilde\upphi_k\overline{\tilde\upphi_k}]\rp] \,.
\end{align*}
Alternatively, we can remove all terms with $\mathbbm{1}_{\{s>0\}}$ in $F_L$, $F_{\uL}$ and in the last 4 lines of $I$ by adding 
\begin{align*}
\mathbbm{1}_{\{s>0\}}\lp\{r^p\xi\lp(\frac{c_k'}{c_k}-(|s|-k)\frac{w'}{w}\rp)\frac{wc_{k+1}}{c_k}\Re[\tilde\upphi_{k+1}\overline{L\tilde\upphi_k}]\rp\}\,.
\end{align*}
to $I$  and replacing, in $F_L$ and in $I$, $ U_k$ by $U_k-\lp(\frac{c_k'}{c_k}\rp)'$.
\end{lemma}

\begin{proof} The proof is entirely analogous to that of Lemma~\ref{lemma:phys-space-redshift-multiplier-identity}, and we again direct the reader to \cite[Appendix B]{Dafermos2017} for some computations in the case $k=|s|$. The identity
\begin{align*}
\Re[T\tilde\upphi_k\overline{(L-T)\tilde\upphi_k}]=\frac14|L\tilde\upphi_k|^2-\frac14|\uL\tilde\upphi_k|^2-\frac{a^2}{(r^2+a^2)^2}|Z\tilde\upphi_k|^2+\frac{a}{r^2+a^2}\Re[\uL\tilde\upphi_k\overline{Z\tilde\upphi_k}]
\end{align*}
was used to simplify the expressions shown here, and make the estimates that follow more easily obtained from the identities.
\end{proof}

We will consider Lemma~\ref{lemma:phys-space-rp-multiplier-identity} with $\xi=(1+4M/r)\tilde\chi_R(r)$. For the case $s<0$, using additionally the transport estimate \eqref{eq:phys-space-transport-identity-s<0}, we can easily deduce that for $p\in[0,2]$
\begin{align*}
&\mathbb{E}_p[\tilde\upphi_k\mathbbm{1}_{[2R^*,\infty)}](\tau)+\mathbb{E}_{\mc I^+,p}[\tilde\upphi_k](0,\tau)+\mathbb{I}_p[\tilde\upphi_k\mathbbm{1}_{[2R^*,\infty)}](0,\tau)\\
%&\quad \leq B\sum_{j=0}^k\lp(\mathbb{E}_0[\tilde\upphi_j\mathbbm{1}_{[R^*,\infty)}](\tau)+\mathbb{E}_{\mc I^+,0}[\tilde\upphi_j](0,\tau)+\mathbb{I}_0[\tilde\upphi_j\mathbbm{1}_{[R^*,\infty)}](0,\tau)\rp)\\
%&\quad \qquad +B\sum_{j=0}^{k-1}\lp(\mathbb{E}_p[\tilde\upphi_j\mathbbm{1}_{[2R^*,\infty)}](\tau)+\mathbb{E}_{\mc I^+,p}[\tilde\upphi_j](0,\tau)+ \mathbb{I}_p[\tilde\upphi_j\mathbbm{1}_{[2R^*,\infty)}](0,\tau)\rp)+B\sum_{j=0}^k \mathbb{E}_p[\tilde\upphi_j](0)+B(|s|-k)\int_0^\tau\int_{\Sigma_{\tau'}\cap[R^*,\infty]}r^{p-5}|\dbtilde\upphi_{k+1}|^2drd\sigma d\tau\\
&\quad \leq B\sum_{j=0}^k\lp(\mathbb{E}_0[\tilde\upphi_j\mathbbm{1}_{[R^*,\infty)}](\tau)+\mathbb{E}_{\mc I^+,0}[\tilde\upphi_j](0,\tau)+\mathbb{I}_0[\tilde\upphi_j\mathbbm{1}_{[R^*,\infty)}](0,\tau)\rp)+\mathbbm{1}_{\{p=2\}}\mathbb{I}_1[\tilde\upphi_k\mathbbm{1}_{[R^*,\infty)}](0,\tau)\\ &\quad\qquad+B\sum_{j=0}^k \mathbb{E}_p[\tilde\upphi_j](0)+B(|s|-k)\int_0^\tau\int_{\Sigma_{\tau'}\cap[R^*,\infty]}r^{p-5}|\tilde\upphi_{k+1}|^2drd\sigma d\tau\,. \numberthis\label{eq:phys-space-rp-multiplier-consequence-s<0}
\end{align*}
Note that the case $p=2$ is special because we lose $r$-weights on the zeroth order term and angular derivatives; thus its proof requires appealing to the statement for some other $p\in[0,2)$. In the case $s>0$, we must distinguish between the cases with weights $\tilde c_k$ and with weights $c_k$. In the former setting, we note that
\begin{align*}
\mathbbm{1}_{\{s>0\}}\lp(\frac{\tilde c_k'}{\tilde c_k}-(|s|-k)\frac{w'}{w}\rp) = -\frac{2M(r^2-a^2)(|s|-k)}{(r^2+a^2)^2}\leq 0
\end{align*}
and, similarly, the coefficient on the $|\dbtilde\upphi_k|^2$ term in $I$ is positive. Therefore, we have for $p\in[0,2]$
\begin{align*}
&\mathbb{E}_p[\dbtilde\upphi_k\mathbbm{1}_{[2R^*,\infty)}](\tau)+\mathbb{E}_{\mc I^+,p}[\dbtilde\upphi_k](0,\tau)+\mathbb{I}_p[\dbtilde\upphi_k\mathbbm{1}_{[2R^*,\infty)}](0,\tau)\\
&\quad\leq \frac{B}{R^*}\mathbb{E}_p[\dbtilde\upphi_k\mathbbm{1}_{[2R^*,\infty)}](\tau)+B\mathbb{E}_p[\dbtilde\upphi_k](0)+B\mathbb{E}_0[\dbtilde\upphi_k\mathbbm{1}_{[R^*,\infty)}](\tau)+B\mathbb{E}_{\mc I^+,0}[\dbtilde\upphi_k](0,\tau)+B\mathbb{I}_0[\dbtilde\upphi_k\mathbbm{1}_{[R^*,\infty)}](0,\tau)\\
&\quad\qquad+B\mathbbm{1}_{\{p=2\}}\mathbb{I}_1[\dbtilde\upphi_k\mathbbm{1}_{[R^*,\infty)}](0,\tau)+ \frac{B}{R^*}\sum_{j=0}^{k-1}\mathbb{I}_p[\dbtilde\upphi_j](0,\tau)+B(|s|-k)\int_0^\tau\int_{\Sigma_{\tau'}\cap[R^*,\infty]}r^{p-4}|\dbtilde\upphi_{k+1}|^2dr^*\\
&\quad\leq B\lp(\mathbb{E}_p[\dbtilde\upphi_k](0)+\mathbb{E}_0[\dbtilde\upphi_k\mathbbm{1}_{[R^*,\infty)}](\tau)+\mathbb{E}_{\mc I^+,0}[\dbtilde\upphi_j](0,\tau)+\mathbb{I}_0[\dbtilde\upphi_k\mathbbm{1}_{[R^*,\infty)}](0,\tau)+\mathbbm{1}_{\{p=2\}}\mathbb{I}_1[\dbtilde\upphi_k\mathbbm{1}_{[R^*,\infty)}](0,\tau)\rp)\\
&\quad\qquad+\frac{B}{R^*}\sum_{j=0}^{k-1} \lp(\mathbb{E}_p[\dbtilde\upphi_j](0)+\mathbb{E}_0[\dbtilde\upphi_j\mathbbm{1}_{[R^*,\infty)}](\tau)+\mathbb{E}_{\mc I^+,0}[\tilde\upphi_j](0,\tau)+\mathbb{I}_0[\dbtilde\upphi_j\mathbbm{1}_{[R^*,\infty)}](0,\tau)\rp)\\
&\quad\qquad+B(|s|-k)\int_0^\tau\int_{\Sigma_{\tau'}\cap[R^*,\infty)}r^{p-4}|\dbtilde\upphi_{k+1}|^2drd\sigma d\tau\,. \numberthis\label{eq:phys-space-rp-multiplier-consequence-s>0-peeling}
\end{align*}
If we choose $c_k$ weights, i.e.\ $\tilde\upphi_k$ instead of $\dbtilde\upphi_k$, then we work with the alternative version of Lemma~\ref{lemma:phys-space-rp-multiplier-identity}, and obtain
\begin{align*}
&\mathbb{E}_p[\tilde\upphi_k\mathbbm{1}_{[2R^*,\infty)}](\tau)+\mathbb{E}_{\mc I^+,p}[\tilde\upphi_k](0,\tau)+\mathbb{I}_p[\tilde\upphi_k\mathbbm{1}_{[2R^*,\infty)}](0,\tau)\\
&\quad\leq B\mathbb{E}_p[\tilde\upphi_k](0)+B\mathbb{E}[\tilde\upphi_k\mathbbm{1}_{[R^*,2R^*]}](\tau)+B\mathbb{E}_{\mc I^+,0}[\tilde\upphi_k](0,\tau)+B\mathbb{I}[\tilde\upphi_k\mathbbm{1}_{[R^*,2R^*]}](0,\tau)\\
&\quad\qquad+ B\int_0^\tau\int_{\Sigma_{\tau'}\cap[R^*,\infty]}\lp(\sum_{j=0}^{k-1}r^{p-5}|\tilde\upphi_j|^2+(|s|-k)r^{p-3-}|\tilde\upphi_{k+1}|^2\rp)dr^*\,, %\numberthis\label{eq:phys-space-rp-multiplier-consequence-s>0-no-peeling}
\end{align*}
so that, if $p\in[0,2)$ we have a gain in the bulk term involving lower level variables: 
\begin{align*}
&\sum_{j=0}^k\lp(\mathbb{E}_p[\tilde\upphi_j\mathbbm{1}_{[2R^*,\infty)}](\tau)+\mathbb{E}_{\mc I^+,p}[\tilde\upphi_j](0,\tau)+\mathbb{I}_p[\tilde\upphi_j\mathbbm{1}_{[2R^*,\infty)}](0,\tau)\rp)\\
&\quad\leq B\sum_{j=0}^k\lp(\mathbb{E}_p[\tilde\upphi_j](0)+\mathbb{E}[\tilde\upphi_j\mathbbm{1}_{[R^*,2R^*]}](\tau)+B\mathbb{E}_{\mc I^+,0}[\tilde\upphi_j](0,\tau)\rp)+\mathbb{I}[\tilde\upphi_j\mathbbm{1}_{[R^*,2R^*]}](0,\tau)\\
&\quad\qquad+ B(|s|-k)\int_0^\tau\int_{\Sigma_{\tau'}\cap[R^*,\infty]}r^{p-3}|\tilde\upphi_{k+1}|^2dr^*\,. \numberthis\label{eq:phys-space-rp-multiplier-consequence-s>0-no-peeling}
\end{align*}

\subsection{First order estimates}
 \label{sec:first-order-basic-estimates}
 
As an application of the results of the previous Section~\ref{sec:multiplier-identities}, we deduce a series of basic estimates involving first order derivatives of the transformed variables of the sytem in Definition~\ref{def:transformed-system}.

\begin{proposition}[Finite-in-time energy estimates]\label{prop:finite-in-time-first-order} Fix $s=\{0,\pm 1,\pm 2\}$, $a_0\in[0,M)$, and let $p\in[0,2)$. Let $0< \tau<\infty$, $k\in\{0,\dots,|s|\}$ and $\chi=\chi(r^*)$. Then, if $\swei{\tilde\upphi}{s}_k$, for $k=0,\dots,|s|$  are homogeneous solutions to the transformed system of Definition~\ref{def:transformed-system}, we have the following estimates.
\begin{itemize}
\item Finite-in-time boundedness of fluxes through on hyperboloidal foliation: 
\begin{align}
\mathbb{E}_p[\swei{\tilde\upphi}{s}_{k}](\tau) &\leq B(\tau) \sum_{j=0}^{|s|} \mathbb{E}_p[\swei{\tilde\upphi}{s}_{j}](0)\,,
\end{align}
and, if $|a|\leq a_0$, the same holds replacing $\mathbb{E}$ with $\overline{\mathbb{E}}$ or $\overline{\mathbb{E}}_p$ and $B(\tau)$ with $B(a_0,\tau)$. If $s>0$, the same holds for $\swei{\dbtilde\upphi}{s}_k$ if we replace the $p$ range by $p\in[1,2]$.
\item Finite-in-time boundedness of fluxes through $\mc{I}^+$:
\begin{align*}
\mathbb{E}_{\mc I^+,p}[\swei{\tilde\upphi}{s}_{k}](0,\tau)\leq B(\tau) \sum_{j=0}^{|s|} \mathbb{E}_p[\swei{\tilde\upphi}{s}_{j}](0)\,.
\end{align*}
If $s>0$, the same holds for $\swei{\dbtilde\upphi}{s}_k$ as long as $p\in[1,2]$.
\item Finite in time boundedness of fluxes through $\mc H^+$:
\begin{align*}
\mathbb{E}_{\mc H^+}[\swei{\tilde\upphi}{s}_{k}](0,\tau)\leq B(\tau) \sum_{j=0}^{|s|} \mathbb{E}[\swei{\tilde\upphi}{s}_j](0)\,,
\end{align*}
and, if $|a|\leq a_0$, the same holds replacing $\mathbb{E}$ and $\mathbb{E}_{\mc H^+}$  with $\overline{\mathbb{E}}$ and $\overline{\mathbb{E}}_{\mc H^+}$, respectively, and $B(\tau)$ with $B(a_0,\tau)$. If $s>0$, the same holds for $\swei{\dbtilde\upphi}{s}_k$ after replacing $\mathbb{E}$ by $\mathbb{E}_1$ on the right hand side.
\item Finite-in-time energy boundedness on hyperboloidal slab: 
\begin{align}
\mathbb{I}_p^{\rm deg}[\swei{\tilde\upphi}{s}_{k}](0,\tau)\leq \mathbb{I}_p[\swei{\tilde\upphi}{s}_k](0,\tau) &\leq B(\tau) \sum_{j=0}^{|s|} \mathbb{E}_p[\swei{\tilde\upphi}{s}_j](0)\,,
\end{align}
and, if $|a|\leq a_0$ the same holds replacing  $\mathbb{E}_p$ with $\overline{\mathbb{E}}_p$, $\mathbb{I}_p$ with $\overline{\mathbb{I}}_p$ and $B(\tau)$ with $B(a_0,\tau)$. If $s>0$, the same holds for $\swei{\dbtilde\upphi}{s}_k$ as long as $p\in[1,2]$.
\end{itemize} 
The above estimates also hold for $p=2$ if $s\leq 0$.
\end{proposition}
\begin{proof}
Let us begin the case where we consider the usual $c_k$, starting by establishing the results for the fluxes. The result with $p=0$ and no overbar follows from \eqref{eq:phys-space-Killing-multiplier-consequence-1}, for instance with $\chi\equiv 1$, using Gr\"{o}nwall's lemma. For the variations, one combines \eqref{eq:phys-space-Killing-multiplier-consequence-1} with \eqref{eq:phys-space-redshift-multiplier-consequence-s>0} and \eqref{eq:phys-space-redshift-multiplier-consequence-s<0}  for the result with an overbar in the case $|a|<M$, or with \eqref{eq:phys-space-rp-multiplier-consequence-s<0} and \eqref{eq:phys-space-rp-multiplier-consequence-s>0-no-peeling}  for the result with $p\neq 0$. Finally, for the result concerning the slab, we simply note that 
\begin{align*}
\mathbb{I}_p[\swei{\tilde\upphi}{s}_k](0,\tau)\leq B \int_0^\tau \mathbb{E}_p[\swei{\tilde\upphi}{s}_k](\tau')d\tau' \,,
\end{align*}
and similarly for the energy norms with an overbar. Thus, applying the bounds on the fluxes concludes the proof. 

For the statements involving the weighted $\dbtilde\upphi_k$, we appeal to the variations of the above estimates in this case, which show that the usual Killing and virial currents lead to error terms with can be controlled by adding the $r^p$ currents with $p\geq 1$.
\end{proof}

\begin{proposition}[Large-$|r^*|$ bulk and flux estimates]\label{prop:bulk-flux-large-r} Fix $s\in\mathbb{Z}$, $a_0\in[0,M)$.  For some sufficiently large $R^*$, and for any $\tau>0$, any $\delta\in(0,1]$, the following estimates hold:
\begin{align*} 
&b(\delta)\mathbb{I}_{-\delta,p}[\swei{\tilde\upphi}{s}_k\mathbbm{1}_{\{|r^*|\geq R^*\}}](0,\tau)+\mathbb{E}_{p}[\swei{\tilde\upphi}{s}_{k}\mathbbm{1}_{\{|r^*|\geq R^*\}}](\tau)+\mathbb{E}_{\mc H^+}[\swei{\tilde\upphi}{s}_{k}](0,\tau)+\mathbb{E}_{\mc I^+,p}[\swei{\tilde\upphi}{s}_{k}](0,\tau)\\
&\quad\leq B\sum_{j=0}^{k}\lp(\mathbb{E}_{p}[\swei{\tilde\upphi}{s}_j](0)+\mathbb{E}[\swei{\tilde\upphi}{s}_j\mathbbm{1}_{\{R^*\leq |r^*|\leq 2R^*\}}](\tau)+\mathbb{I}[\swei{\tilde\upphi}{s}_j\mathbbm{1}_{\{R^*\leq |r^*|\leq 2R^*\}}](0,\tau)\rp)\\
&\quad\qquad+B(|s|-k)\int_{0}^\tau\int_{\Sigma_\tau'}(r^{p-4}\mathbbm{1}_{\{s<0\}}+r^{p-3}\mathbbm{1}_{\{s>0\}})|\tilde\upphi_{k+1}|^2\mathbbm{1}_{\{|r^*|\geq R^*\}} dr d\sigma d\tau'\,. \numberthis \label{eq:bulk-flux-estimate-large-r}
\end{align*}
(If $s>0$, then the above estimates also hold for $\swei{\dbtilde\upphi}{s}_k$ for $p\in[1,2]$.) In particular, we have for $p\in[0,2)$
\begin{align*}
&\sum_{k=0}^{|s|}\lp(b(\delta)\mathbb{I}_{-\delta,p}[\swei{\tilde\upphi}{s}_k\mathbbm{1}_{\{|r^*|\geq R^*\}}](0,\tau)+\mathbb{E}_p[\swei{\tilde\upphi}{s}_{k}\mathbbm{1}_{\{|r^*|\geq R^*\}}](\tau)+\mathbb{E}_{\mc H^+}[\swei{\tilde\upphi}{s}_{k}](0,\tau)+\mathbb{E}_{\mc I^+,p}[\swei{\tilde\upphi}{s}_{k}](0,\tau)\rp)\\
&\quad\leq B\sum_{k=0}^{|s|}\lp(\mathbb{E}_{p}[\swei{\tilde\upphi}{s}_k](0)+\mathbb{E}_0[\swei{\tilde\upphi}{s}_k\mathbbm{1}_{\{R^*\leq |r^*|\leq 2R^*\}}](\tau)+\mathbb{I}_0[\swei{\tilde\upphi}{s}_k\mathbbm{1}_{\{R^*\leq |r^*|\leq 2R^*\}}](0,\tau)\rp)\,, \numberthis \label{eq:bulk-flux-estimate-large-r-summed}
\end{align*}
and the same holds for $p=2$ if $s\leq 0$.  Furthermore, if $|a|\leq a_0$, then we can replace $\mathbb{I}_p$, $\mathbb{E}_p$ and $\mathbb{E}_{\mc H^+}$ on the left hand side in the above estimates by their versions with overbar if we moreover add an overbar to the $\mathbb{E}$ quantities on the right hand side; then the constant $B$ depends on $a_0$.
\end{proposition}
\begin{proof} We start with the case $p=0$ with no overbars on the energy norms: we combine the bulk estimates \eqref{eq:phys-space-y-multiplier-consequence-infty} and \eqref{eq:phys-space-y-multiplier-consequence-hor} with the boundary term estimates \eqref{eq:phys-space-Killing-multiplier-consequence-3} and  \eqref{eq:phys-space-Killing-multiplier-consequence-4} to control the first term on the left hand side. The cases of $p\neq 0$ and energies with an overbar then follow by invoking the estimates \eqref{eq:phys-space-redshift-multiplier-consequence-s>0}, \eqref{eq:phys-space-redshift-multiplier-consequence-s<0}, \eqref{eq:phys-space-rp-multiplier-consequence-s<0} and \eqref{eq:phys-space-rp-multiplier-consequence-s>0-no-peeling} in Section~\ref{sec:multiplier-identities}. Finally, for the statement with $\dbtilde\upphi_k$, we rely on the alternative versions of the aforementioned estimates obtained in the last section for this rescaled variable.
\end{proof}

To conclude the section, we state some conditional estimates which will be useful later on.

\begin{proposition}[Two conditional estimates] \label{prop:estimates-assuming-ILED} Fix $s\in\mathbb{Z}$, $M>0$, and $\tau>0$.
\begin{itemize}
\item For $k\in\{0,\dots,|s|\}$, let $\swei{\tilde\upphi}{s}_k$ be solutions to the inhomogeneous transformed system of Definition~\ref{def:transformed-system}, with inhomogeneity $\swei{\tilde{\mathfrak H}}{s}_k$, satisfying the integrated local energy decay estimate 
\begin{align*}
\mathbb{I}^{\rm deg}[\tilde\upphi_k](-\infty,\infty)\leq B\sum_{j=0}^{|s|}\mathbb{E}[\tilde\upphi_j](0)\,,
\end{align*}
and let $\chi=\chi(r^*)$ be a smooth function with the property that $T+\upomega_+\chi Z$ is timelike in the support of $\swei{\tilde\upphi}{s}_k$ for all $k=0,\dots, |s|$. Then for any $\tau_0<\tau$ we have the energy estimate
\begin{align}
\begin{split}
\mathbb{E}[\swei{\tilde\upphi}{s}_k](\tau) &\leq B\sum_{j=0}^{|s|}\mathbb{E}[\swei{\tilde\upphi}{s}_j](\tau_0) +B\upomega_+\int_{\tau_0}^\tau \int_{\Sigma_{\tau'}}\lp|\frac{d\chi}{dr}\rp||Z\swei{\Phi}{s}|^2drd\sigma d\tau' \\
&\qquad +B\sum_{j=0}^{|s|-1}\upomega_+ \lp|\int_{\tau_0}^\tau\int_{\Sigma_{\tau'}}\chi wc_jc_{s,j}^Z(r)\Re[Z\swei{\Phi}{s}\overline{Z\swei{\tilde\upphi}{s}_j}]drd\sigma d\tau'\rp| \\
&\qquad +B\sum_{j=0}^{|s|}\lp|\int_{\tau_0}^\tau\int_{\Sigma_{\tau'}}\Re[\swei{\tilde{\mathfrak{H}}}{s}_j\overline{(T+\upomega_+\chi Z)\swei{\tilde\upphi}{s}_j}]drd\sigma d\tau'\rp|\,.
\end{split}\label{eq:1st-order-estimate-for-energy-bddness}
\end{align}

\item For $k\in\{0,\dots,|s|\}$, let $\swei{\tilde\upphi}{s}_k$ be solutions to the homogeneous transformed system of Definition~\ref{def:transformed-system}, satisfying the integrated local energy decay estimate
\begin{align*}
\sum_{k=0}^{|s|}\mathbb{I}_{-\delta,p}^{\rm deg}[\swei{\dbtilde\upphi}{s}_{k}](0,\tau)\leq B(\delta)\sum_{k=0}^{|s|}\lp(\overline{\mathbb{E}}_{p}[\swei{\dbtilde\upphi}{s}_{k}](\tau)+\overline{\mathbb{E}}_{p}[\swei{\dbtilde\upphi}{s}_{k}](0)\rp)
\end{align*}
for any $p\in(1,2)$ and $\delta\in(0,1]$. Then, we have the 
\begin{align*}
\sum_{k=0}^{|s|}\overline{\mathbb{E}}_{p}[\swei{\tilde\upphi}{s}_k](\tau)
&\leq B\sum_{k=0}^{|s|}\int_0^\tau\int_{\Sigma_{\tau'}} r^{-3}(s^2+r^{-1})\lp(|\swei{\tilde\upphi}{s}_k|^2+|Z\swei{\tilde\upphi}{s}_k|^2\rp)+B\sum_{k=0}^{|s|}\overline{\mathbb{E}}_p[\swei{\tilde\upphi}{s}_k](0)\,.\numberthis\label{eq:estimate-for-openness-first-order}
\end{align*}
\end{itemize}
\end{proposition}

\begin{proof}  The first statement follows easily by combining the assumptions with \eqref{eq:phys-space-Killing-multiplier-consequence-1}. Now let us turn to the second statement. From estimate \eqref{eq:phys-space-Killing-multiplier-consequence-1}
\begin{align*}
\sum_{k=0}^{|s|}{\mathbb{E}}_0[\tilde\upphi_k](\tau)
&\leq B(a_0)\varepsilon^{-1}\sum_{k=0}^{|s|}\int_0^\tau\int_{\Sigma_{\tau'}} r^{-3}(s^2+r^{-1})\lp(|\tilde\upphi_k|^2+|Z\tilde\upphi_k|^2\rp)\\
&\qquad+B(a_0)\varepsilon\sum_{k=0}^{|s|}{\mathbb{I}}^{\mathrm{deg}}_0[\tilde\upphi_k](0,\tau)+B(a_0)\sum_{k=0}^{|s|}{\mathbb{E}}_0[\tilde\upphi_k](0)\\
&\leq B(a_0)\sum_{k=0}^{|s|}\int_0^\tau\int_{\Sigma_{\tau'}} r^{-3}(s^2+r^{-1})\lp(|\tilde\upphi_k|^2+|Z\tilde\upphi_k|^2\rp)+B(a_0)\sum_{k=0}^{|s|}{\mathbb{E}}_0[\tilde\upphi_k](0)\,,\numberthis\label{eq:estimate-for-openness-intermediate-1}
\end{align*}
after using the assumptions. It remains to improve the weights at $r=r_+$ and $r=\infty$ on the left hand side. 

We begin with the large $r$ weights. Fix some $\eta\in(0,1)$. By Proposition~\ref{prop:bulk-flux-large-r}, we have
\begin{align*}
\sum_{k=0}^{|s|}\overline{\mathbb{E}}_{1+2\eta}[\tilde\upphi_k](\tau)
&\leq B\sum_{k=0}^{|s|}\lp(\mathbb{E}_{0}[\tilde\upphi_k](\tau) +{\mathbb{I}}^{\mathrm{deg}}_0[\tilde\upphi_k](0,\tau)\rp)+B\sum_{k=0}^{|s|}\overline{\mathbb{E}}_{1+2\eta}[\tilde\upphi_k](0)\,,
\end{align*}
as the errors on the right hand side of \eqref{eq:bulk-flux-estimate-large-r} are only located in a bounded $|r^*|$ region where the weights are unimportant. By the assumptions, we then have
\begin{align*}
\sum_{k=0}^{|s|}\overline{\mathbb{E}}_{1+2\eta}[\tilde\upphi_k](\tau)
&\leq B\sum_{k=0}^{|s|}\overline{\mathbb{E}}_{1+\eta}[\tilde\upphi_k](\tau) +B\sum_{k=0}^{|s|}\overline{\mathbb{E}}_{1+2\eta}[\tilde\upphi_k](0)\,. \numberthis\label{eq:estimate-for-openness-intermediate-2}
\end{align*}
Note that we can interpolate
\begin{align*}
\overline{\mathbb{E}}_{1+\eta}[\tilde\upphi_k](\tau) \leq B\overline{\mathbb{E}}_{0}[\tilde\upphi_k](\tau)+R^{-\eta}\overline{\mathbb{E}}_{1+2\eta}[\tilde\upphi_k](\tau)\,.
\end{align*}
By choosing $R$ sufficiently large depending on $\eta$, we can absorb the second term by the left hand side of \eqref{eq:estimate-for-openness-intermediate-2}. Since our choice of $\eta$ was arbitrary, we conclude 
\begin{align*}
\sum_{k=0}^{|s|}\overline{\mathbb{E}}_{p}[\tilde\upphi_k](\tau)
&\leq B\sum_{k=0}^{|s|}\overline{\mathbb{E}}_{0}[\tilde\upphi_k](\tau) +B\sum_{k=0}^{|s|}\overline{\mathbb{E}}_{p}[\tilde\upphi_k](0)\,, \numberthis\label{eq:estimate-for-openness-intermediate-3}
\end{align*}
for any $p\in(1,2)$.

We now turn to the weights as $r\to r_+$. By Proposition~\ref{prop:bulk-flux-large-r}, we again have
\begin{align*}
\sum_{k=0}^{|s|}\lp(\overline{\mathbb{E}}_{0}[\tilde\upphi_k](\tau)+\overline{\mathbb{I}}_{1}[\tilde\upphi_k](0,\tau)\rp)
&\leq B\sum_{k=0}^{|s|}\lp(\mathbb{E}_{0}[\tilde\upphi_k](\tau) +{\mathbb{I}}_1[\tilde\upphi_k](0,\tau)\rp)+B\sum_{k=0}^{|s|}\overline{\mathbb{E}}_{1}[\tilde\upphi_k](0)\\
&\leq B\sum_{k=0}^{|s|}{\mathbb{I}}_1[\tilde\upphi_k](0,\tau)+B\sum_{k=0}^{|s|}\overline{\mathbb{E}}_{1+2\eta}[\tilde\upphi_k](0)\,,
\end{align*}
using the appropriate version of \eqref{eq:phys-space-Killing-multiplier-consequence-1} in the last inequality. Hence, 
\begin{align*}
\sum_{k=0}^{|s|}\overline{\mathbb{E}}_0[\tilde\upphi_k](\tau)+\sum_{k=0}^{|s|}\int_0^\tau \overline{\mathbb{E}}_0[\tilde\upphi_k](\tau')d\tau' 
&\leq B\sum_{k=0}^{|s|}\int_0^\tau {\mathbb{E}}_0[\tilde\upphi_k](\tau')d\tau' + B\sum_{k=0}^{|s|}\overline{\mathbb{E}}_1[\tilde\upphi_k](0)\\
&\leq B\tau \sup_{0\leq \tau'\leq \tau} \sum_{k=0}^{|s|}\mathbb{E}_0[\tilde\upphi_k](\tau')+ B\sum_{k=0}^{|s|}\overline{\mathbb{E}}_1[\tilde\upphi_k](0)\,,
 \end{align*}
from which we deduce (see \cite[last paragraph of proof of Lemma 11.2.1]{Dafermos2016b}) the estimate.
 \begin{align*}
\sum_{k=0}^{|s|}\overline{\mathbb{E}}_0[\tilde\upphi_k](\tau)\leq  B\sup_{0\leq \tau'\leq \tau}\sum_{k=0}^{|s|}{\mathbb{E}}_0[\tilde\upphi_k](\tau')+ B\overline{\mathbb{E}}_1[\tilde\upphi_k](0)\,. \numberthis\label{eq:estimate-for-openness-intermediate-4}
 \end{align*}

By combining \eqref{eq:estimate-for-openness-intermediate-1}, \eqref{eq:estimate-for-openness-intermediate-3} and \eqref{eq:estimate-for-openness-intermediate-4}, we conclude the proof.
\end{proof}

\subsection{Commutator identities}

In this section, we derive some commutation identities which are useful in obtaining higher order estimates for the transformed system of Definition~\ref{def:transformed-system}. We begin with a commutation relation which will be useful in recovering derivatives of the lower level quantities as well as eliminate the higher order bulk norms:

\begin{lemma}[Radial derivative commutation] \label{lemma:radial-commutation} Fix $M>0$, $a\in[0,M)$ and $s\in\mathbb{Z}$. Suppose that, for $k\in\{0,\dots,|s|\}$, $\swei{\tilde\upphi}{s}$ are solutions to the homogeneous transformed system of Definition~\ref{def:transformed-system}; we drop superscripts in what follows to easy the notation. We have 
\begin{align*}
\tilde{\mathfrak{R}}_k\lp[\tilde\upphi_k'\rp]= \lp[\tilde{\mathfrak{R}}_k, \frac12(L-\uL) \rp]\tilde\upphi_k  + \lp(\tilde{\mathfrak{R}}_k\tilde\upphi_k\rp)'\,,
\end{align*}
where the first term is given by
\begin{align*}
\lp[\tilde{\mathfrak{R}}_k, \frac12(L-\uL) \rp]
&=\frac{2arw}{r^2+a^2}Z(L+\uL)-w'\mathring{\slashed{\triangle}}^{[s]}_T  +\frac12\sign s \lp(\frac{c_k'}{c_k}\rp)'\underline{\mc L }\\
&\qquad +\sign s \lp[(|s|-k)\lp(\frac{4arw}{r^2+a^2}\rp)'-\frac{2arw}{r^2+a^2}\frac{c_k'}{c_k}\rp]Z-\lp(U_k-\lp(\frac{c_k'}{c_k}\rp)'\rp)'\,.
\end{align*}
\end{lemma}

Let us note already that, by the repeating the procedure followed to obtain \eqref{eq:phys-space-Killing-multiplier-consequence-1}, we may use the identities in Lemma~\ref{lemma:phys-space-Killing-multiplier-identity} to deduce
\begin{align*}
\mathbb{E}[\tilde\upphi_k'](\tau)
& \leq B\sum_{j=0}^{k}\mathbb{I}^{\rm deg, 1}[X\tilde\upphi_j](0,\tau)+B(|s|-k)\mathbb{I}^{\rm deg}[\tilde\upphi_{k+1}](0,\tau)\\
&\qquad +B\sum_{j=0}^k\lp|\int_0^\tau\int_{\Sigma_{\tau'}} \Re\lp([\tilde{\mathfrak{R}}_j,\p_{r^*}]\tilde\upphi_j\overline{(T+\upomega_+\chi Z)\tilde\upphi_j'}\rp) drd\sigma d\tau\rp| \numberthis \label{eq:phys-space-Killing-multiplier-consequence-1-radial-commuted}\\
&\leq B(|s|-k)\mathbb{I}^{\rm deg}[\tilde\upphi_{k+1}](0,\tau)+ B\varepsilon^{-1}\sum_{j=0}^{k}\mathbb{I}^{\rm deg, 1}[\tilde\upphi_j](0,\tau)+B\varepsilon\sum_{j=0}^{k}\mathbb{I}^1[\tilde\upphi_{j}](0,\tau)+ \sum_{j=0}^k\mathbb{E}[\tilde\upphi_j'](0)\,. 
\end{align*}

In order to establish higher order estimates close to $r=r_+$, it will be useful to have the following commutation lemma:

\begin{lemma}[Redshift commutation] \label{lemma:redshift-commutation} Fix $M>0$, $a\in[0,M)$ and $s\in\mathbb{Z}$. Suppose that, for $k\in\{0,\dots,|s|\}$, $\swei{\tilde\upphi}{s}$ are solutions to the homogeneous transformed system of Definition~\ref{def:transformed-system}; we drop superscripts in what follows to easy the notation. We have 
\begin{align*}
\tilde{\mathfrak{R}}_k\lp[\frac{r^2+a^2}{\Delta}\underline{L}\tilde\upphi_k\rp]= w\lp[\frac{(r^2+a^2)^2}{\Delta}\tilde{\mathfrak{R}}_k, \frac{r^2+a^2}{\Delta}\underline{L} \rp]\tilde\upphi_k +(r^2+a^2)\uL\lp(\frac{1}{w}\tilde{\mathfrak{R}}_k[\tilde\upphi_k]\rp)
\end{align*}
where the first term on the right hand side is given by
\begin{align*}
&\lp[\frac{(r^2+a^2)^2}{\Delta}\tilde{\mathfrak{R}}_k, \frac{r^2+a^2}{\Delta}\underline{L} \rp]\\
&\quad= -\frac{2M (r^2-a^2)}{r^2+a^2}\lp(\frac{r^2+a^2}{\Delta}\uL\rp)^2+2r\lp(\frac{r^2+a^2}{\Delta}\uL\rp)L-\frac{4ar}{r^2+a^2}Z\lp(\frac{r^2+a^2}{\Delta}\uL\rp)\\
&\quad\qquad -\lp[\frac{4Mr(r^2-3a^2)}{(r^2+a^2)^2}-\frac{1-\sign s}{2}\frac{d}{dr}\lp(\frac{(r^2+a^2)c_k'}{c_k}\rp)+\frac{1+\sign s}{2}\frac{2M(r^2-a^2)}{(r^2+a^2)c_k}\frac{dc_k}{dr}\rp]\lp(\frac{r^2+a^2}{\Delta}\uL\rp)\\
&\quad\qquad -\frac{1+\sign s}{2}\frac{d}{dr}\lp(\frac{c_k'}{c_k w}\rp) L+\frac{d}{dr}\lp(\frac{U_k}{w}-\frac{d}{dr}\lp(\frac{c_k'}{c_k}\rp)(r^2+a^2)\rp)\\
&\quad\qquad-\frac{2a}{(r^2+a^2)^2}\lp[5r^2-a^2-2\sign s (|s|-k)(r^2-a^2) +(1+\sign s)\frac{c_k'}{c_k w}\rp]Z\,.
\end{align*}
\end{lemma}

We can now essentially repeat the arguments in Section~\ref{sec:multiplier-identities} with the transformed system of PDEs verified by ${\textstyle\frac{r^2+a^2}{\Delta}\uL}\tilde\upphi_k$, noting that the transport equations now only hold up to first order errors and, moreover, that there are additional terms in the bulk term $I$ due to the commutator in Lemma~\ref{lemma:redshift-commutation}.

Let us start with the Killing $K$ energy current from Lemma~\ref{lemma:phys-space-Killing-multiplier-identity}. The commutator in Lemma~\ref{lemma:redshift-commutation} creates a lot of extra bulk errors. For the zeroth and first order terms, we apply Cauchy--Schwarz. Then, we use that 
\begin{align*}
&\Re\lp[\lp({\textstyle\frac{r^2+a^2}{\Delta}\uL}\rp)L\tilde\upphi_k K\lp({\textstyle\frac{r^2+a^2}{\Delta}\uL}\rp)\overline{\tilde\upphi_k}\rp] \\
&\quad = 2\lp|{\textstyle\frac{r^2+a^2}{\Delta}\uL}K\tilde\upphi_k\rp|^2-\frac{\Delta}{r^2+a^2}\Re\lp[\lp({\textstyle\frac{r^2+a^2}{\Delta}\uL}\rp)^2\tilde\upphi_k {\textstyle\frac{r^2+a^2}{\Delta}\uL}K\overline{\tilde\upphi_k}\rp]\\
&\quad\qquad-\frac{r^2-r_+^2}{r^2+a^2}\upomega_+\Re[{\textstyle\frac{r^2+a^2}{\Delta}\uL}Z\tilde\upphi_k {\textstyle\frac{r^2+a^2}{\Delta}\uL}K\overline{\tilde\upphi_k}]\,,
\end{align*}
which can be controlled, after spacetime integration in $\{r^*\leq -R^*\}$, by $\overline{\mathbb{I}}[K\tilde\upphi_k]+(R^*)^{-1}\overline{\mathbb{I}}[{\textstyle\frac{r^2+a^2}{\Delta}\uL}\tilde\upphi_k]$; and we apply  Cauchy--Schwarz making sure to introduce a small parameter multiplying any bulk term in $\lp({\textstyle\frac{r^2+a^2}{\Delta}\uL}\rp)^2\tilde\upphi_k$. Hence our analogue of \eqref{eq:phys-space-Killing-multiplier-consequence-4} will have  $\overline{\mathbb{I}}$ instead of $\mathbb{I}$:
\begin{align*}
&\mathbb{E}[{\textstyle\frac{r^2+a^2}{\Delta}\uL}\tilde\upphi_k\mathbbm{1}_{\{r^*\leq -2R^*\}}](\tau)+\mathbb{E}_{\mc H^+}[{\textstyle\frac{r^2+a^2}{\Delta}\uL}\tilde\upphi_k](0,\tau)\\
&\qquad +\int_0^\tau \int_{\Sigma_{\tau'}\cap\{r^*\leq -2R^*\}}\frac{(-\sign s )}{c_k}\frac{dc_k}{dr}|\underline{\mc{L}}\lp({\textstyle\frac{r^2+a^2}{\Delta}\uL}\rp)\tilde\upphi_k|^2 drd\sigma d\tau'\\
 &\quad\leq  B\sum_{j=0}^k \lp(\mathbb{E}[{\textstyle\frac{r^2+a^2}{\Delta}\uL}\tilde\upphi_j](0)+\overline{\mathbb{E}}[\tilde\upphi_j](0)+\mathbb{E}[{\textstyle\frac{r^2+a^2}{\Delta}\uL}\tilde\upphi_j\mathbbm{1}_{\{-2R^*\leq r^*\leq -R^*\}}](\tau)\rp)\\
 &\quad\qquad +B\varepsilon\sum_{j=0}^k \overline{\mathbb{I}}[{\textstyle\frac{r^2+a^2}{\Delta}\uL}\tilde\upphi_j\mathbbm{1}_{\{r^*\leq -R^*\}}](0,\tau)+B(|s|-k)\varepsilon^{-1}\int_{\mc R_{(0,\tau)}}\lp|{\textstyle\frac{r^2+a^2}{\Delta}\uL}\tilde\upphi_{k+1}\rp|^2\mathbbm{1}_{\{r^*\leq -R^*\}}]drd\sigma d\tau'\\
 &\qquad\quad+B\sum_{j=0}^k\overline{\mathbb{E}}[\tilde\upphi_j\mathbbm{1}_{\{r^*\leq -R^*\}}](\tau)+ B\varepsilon^{-1}\sum_{j=0}^k\sum_{X\in\{T,Z,\mathrm{id}\}} \overline{\mathbb{I}}[X\tilde\upphi_j\mathbbm{1}_{\{r^*\leq -R^*\}}](0,\tau)\,,\numberthis\label{eq:phys-space-Killing-multiplier-consequence-4-Y-version}
\end{align*}
for sufficiently large $R^*=\varepsilon^{-2}$. Next, we revisit the virial identity in Lemma~\ref{lemma:phys-space-y-multiplier-identity}. Noting in the applications of Cauchy--Schwarz that  $L-\uL=2K-\frac{r^2-r_+^2}{r^2+a^2}\upomega_+ Z-2\frac{\Delta}{r^2+a^2}\lp(\frac{r^2+a^2}{\Delta}\uL\rp)$,  we can similarly re-derive an analogue of \eqref{eq:phys-space-y-multiplier-consequence-hor} for sufficiently large $R^*$:
\begin{align*}
&\mathbb{I}[{\textstyle\frac{r^2+a^2}{\Delta}\uL}\tilde\upphi_k\mathbbm{1}_{\{r^*\leq -2R^*\}}](0,\tau) \\
&\quad\leq B\sum_{j=0}^{k}\lp(\mathbb{E}[{\textstyle\frac{r^2+a^2}{\Delta}\uL}\tilde\upphi_j\mathbbm{1}_{\{r^*\leq -R^*\}}](\tau)+\mathbb{E}_{\mc H^+}[{\textstyle\frac{r^2+a^2}{\Delta}\uL}\tilde\upphi_j](0,\tau)+\mathbb{I}[{\textstyle\frac{r^2+a^2}{\Delta}\uL}\tilde\upphi_j\mathbbm{1}_{\{-2R^*\leq r^*\leq -R^*\}}](\tau)\rp)\\
&\quad\qquad+
B(|s|-k)\int_{\mc R_{(0,\tau)}}\lp|{\textstyle\frac{r^2+a^2}{\Delta}\uL}\tilde\upphi_{k+1}\rp|^2\mathbbm{1}_{\{r^*\leq -R^*\}}drd\sigma d\tau'+B\sum_{j=0}^k\sum_{X\in\{T,Z,\mathrm{id}\}}\overline{\mathbb{I}}[X\tilde\upphi_k\mathbbm{1}_{\{r^*\leq -R^*\}}](0,\tau)\\
&\quad\qquad +B\sum_{j=0}^k\overline{\mathbb{E}}[{\textstyle\frac{r^2+a^2}{\Delta}\uL}\tilde\upphi_k](0)+B\sum_{j=0}^k\mathbb{E}[\tilde\upphi_k](0) +B\sum_{j=0}^{k}\overline{\mathbb{E}}[\tilde\upphi_k\mathbbm{1}_{\{r^*\leq -R^*\}}](\tau)\,.\numberthis\label{eq:phys-space-y-multiplier-consequence-hor-Y-version}
\end{align*}

Finally, we turn to the redshift estimates. Our starting point is Lemma~\ref{lemma:phys-space-redshift-multiplier-identity} and we note that, among the additional terms arising from the commutator, the term with $\lp(\frac{r^2+a^2}{\Delta}\uL\rp)^2$ has a good sign. We apply Cauchy--Schwarz to all others, obtaining an analogue of \eqref{eq:phys-space-redshift-multiplier-consequence-s<0} and \eqref{eq:phys-space-redshift-multiplier-consequence-s>0}:
\begin{align*}
&\overline{\mathbb{E}}[{\textstyle\frac{r^2+a^2}{\Delta}\uL}\tilde\upphi_k\mathbbm{1}_{\{r^*\leq -2R^*\}}](\tau)+\overline{\mathbb{E}}_{\mc H^+}[{\textstyle\frac{r^2+a^2}{\Delta}\uL}\tilde\upphi_k](0,\tau)+\overline{\mathbb{I}}_0[{\textstyle\frac{r^2+a^2}{\Delta}\uL}\tilde\upphi_k\mathbbm{1}_{\{r^*\leq -2R^*\}}](0,\tau) \\
&\quad\leq B(a_0)\sum_{j=0}^k\lp({\mathbb{E}}[{\textstyle\frac{r^2+a^2}{\Delta}\uL}\tilde\upphi_k\mathbbm{1}_{\{r^*\leq -R^*\}}](\tau)+{\mathbb{E}}_{\mc H^+}[{\textstyle\frac{r^2+a^2}{\Delta}\uL}\tilde\upphi_k](0,\tau)+\mathbb{I}[{\textstyle\frac{r^2+a^2}{\Delta}\uL}\tilde\upphi_j\mathbbm{1}_{\{r^*\leq -R^*\}}](0,\tau)\rp)\\
&\quad\qquad +B(a_0)\overline{\mathbb{E}}[{\textstyle\frac{r^2+a^2}{\Delta}\uL}\tilde\upphi_k](0)+B(|s|-k)\int_{\mc R_{(0,\tau)}}\lp|{\textstyle\frac{r^2+a^2}{\Delta}\uL}\tilde\upphi_{k+1}\rp|^2\mathbbm{1}_{\{r^*\leq -R^*\}}drd\sigma d\tau'\\
&\quad\qquad +B\sum_{X\in\{T,Z,\mathrm{id}\}}\overline{\mathbb{I}}[X\tilde\upphi_k\mathbbm{1}_{\{r^*\leq -R^*\}}](0,\tau)+B\overline{\mathbb{E}}[\tilde\upphi_k](0)\,.\numberthis\label{eq:phys-space-redshift-multiplier-consequence-Y-version}
\end{align*}

\subsection{Higher order estimates}
\label{sec:higher-order-estimates}

In this section, we derive higher order estimates for solutions to the transformed system of Definition~\ref{def:transformed-system}, conditional on having first order estimates for bulk and flux terms.

\begin{proposition}[Higher order bulk and flux estimates] \label{prop:higher-order-bulk-flux} Fix $s\in\mathbb{Z}$, $a_0\in[0,M)$, and $J\geq 0$. For $k\in\{0,\dots,|s|\}$, let $\swei{\tilde\upphi}{s}_k$ be solutions to the homogeneous transformed system of Definition~\ref{def:transformed-system}. Suppose $|a|\leq a_0$. Then for all $p\in[0,2)$ excluding $p=0$ if $J\neq 0$, and all $R^*>0$ sufficiently large, we have the flux estimates 
\begin{align*}
&\sum_{k=0}^{|s|}\overline{\mathbb{E}}^{J+|s|-k}_p[\tilde\upphi_k](\tau) +\mathbb{E}_{\mc I^+,p}^{J}[\Phi](0,\tau) + \overline{\mathbb{E}}^{J}_{\mc H^+}[\Phi](0,\tau)\\
&\qquad+\sum_{k=0}^{|s|-1}\lp(\mathbbm{1}_{\{s<0\}}\mathbb{E}_{\mc I^+,p}^{|s|-k}[\tilde\upphi_{k}](0,\tau)+ \mathbbm{1}_{\{s>0\}}\overline{\mathbb{E}}^{|s|-k}_{\mc H^+}[\tilde\upphi_k](0,\tau)\rp) \\
&\quad\leq B(a_0,J)\sum_{k=0}^{|s|}\overline{\mathbb{E}}^{J+|s|-k}_p[\tilde\upphi_k](0)+B(a_0,J)\sum_{J_1+J_2=0}^J{\mathbb{E}}[T^{J_1}Z^{J_2}\Phi\mathbbm{1}_{[-R^*,R^*]}](\tau)\\
&\quad\qquad +B(a_0,J)\sum_{k=0}^{|s|-1}\sum_{J_1+J_2=0}^{J+|s|-k-1}\sum_{X\in\{\mathrm{id},X^*\}}\lp(\mathbb{E}[T^{Z_1}Z^{J_2}X\tilde\upphi_k\mathbbm{1}_{[-R^*,R^*]}](\tau)+\mathbb{I}[T^{Z_1}Z^{J_2}X\tilde\upphi_k\mathbbm{1}_{[-R^*,R^*]}](\tau)\rp)
\\
&\quad\qquad+B(a_0,J)\mathbb{I}^{\mathrm{deg},J+|s|-k}[\Phi\mathbbm{1}_{[-R^*,R^*]}](0,\tau)\,, \numberthis\label{eq:higher-order-energy-bddness}
\end{align*}
and for any $\delta\in(0,1)$, we have the bulk estimates
\begin{align*}
&b(\delta)\sum_{k=0}^{|s|}\overline{\mathbb{I}}^{{\rm deg},J+|s|-k}_{-\delta,p}[\tilde\upphi_k](0,\tau)\\
&\quad\leq B(a_0,J)\sum_{k=0}^{|s|}\overline{\mathbb{E}}^{J+|s|-k}_p[\tilde\upphi_k](0)+B(a_0,J)\sum_{J_1+J_2=0}^{J}{\mathbb{I}}^{{\rm deg}}[T^{J_1}Z^{J_2}\Phi\mathbbm{1}_{[-R^*,R^*]}](0,\tau)\\
&\quad\qquad+B(a_0,J)\sum_{k=0}^{|s|-1}\sum_{J_1+J_2=0}^{J+|s|-k-1}\sum_{X\in\{\mathrm{id},X^*\}}{\mathbb{I}}^{\mathrm{deg}}[T^{J_1}Z^{J_2}X\tilde\upphi_k\mathbbm{1}_{[-R^*,R^*]}](0,\tau)\\
&\quad\qquad+B(a_0,J)\sum_{k=0}^{|s|}\mathbb{E}^{J+|s|-k}[\tilde\upphi_k\mathbbm{1}_{[-R^*,R^*]}](\tau)\,, \numberthis\label{eq:higher-order-ILED}
\end{align*}
and
\begin{align*}
\sum_{k=0}^{|s|}\overline{\mathbb{I}}^{{\rm deg},J+|s|-k}_{-\delta,p}[\tilde\upphi_k](0,\tau) &\leq B(a_0,J)\sum_{k=0}^{|s|}\lp(\overline{\mathbb{I}}^{J+|s|-k}_{-\delta,p}[\tilde\upphi_k](0,\tau)+\mathbb{E}^{J+|s|-k}_p[\tilde\upphi_k](\tau)\rp) \\
&\qquad + B(a_0,J)\sum_{X\in\{T,Z\}}\sum_{k=0}^{|s|}\lp(\overline{\mathbb{I}}^{J+|s|-k}_{0}[X\tilde\upphi_k](0,\tau)+\mathbb{E}^{J+|s|-k}_0[X\tilde\upphi_k](\tau)\rp)\\
&\qquad + B(a_0,J)\sum_{k=0}^{|s|}\Big(\mathbb{E}^{J+|s|-k}_p[\tilde\upphi_k](0)+\sum_{X\in\{T,Z\}}\mathbb{E}^{J+|s|-k}_0[X\tilde\upphi_k](0)\Big)\,.\numberthis\label{eq:ILED-for-decay}
\end{align*}
The above estimates also hold for $p=2$ if $s\leq 0$.
In the case of $s>0$, we can drop the overlines from both sides of the inequalities as long as $\overline{\mathbb{E}}_{\mc H^+}^J[\Phi]$ is dropped from the left hand side. The above estimates also hold for $p\in[1,2)$ if we replace $\tilde\upphi_k$ by $\dbtilde\upphi_k$. 
\end{proposition}

\begin{proof} The procedure to upgrade the estimates with $J-1$ to that with $J$ is the same for all $J\geq 1$, so we present only the case $J=1$. This follows several steps: we separate the spacetime region into the large $r$ region (step 1), region of $r$ close to $r_+$ (step 2), the boundaries $\mc H^+$ and $\mc I^+$ (step 3) and the bounded $|r^*|$ region (step 4).

\medskip
\noindent\textit{Step 1: the case of $r\to \infty$.} Using the identity $\uL=-L+2\lp(T+\frac{a}{r^2+a^2}Z\rp)$, we get
\begin{align*}
&\mathbb{E}_p[r^{-1}\mathring{\slashed\nabla}\tilde\upphi_k\mathbbm{1}_{[2R^*,\infty)}](\tau)\\
&\quad\leq 
B \int_{\Sigma_\tau\cap[2R^*,\infty)} 
\lp(r^{-4}|\mathring{\slashed\nabla}^2\tilde\upphi_k|^2+r^{p-2}|L\mathring{\slashed\nabla}\tilde\upphi_k|^2+r^{-4}|T\mathring{\slashed\nabla}\tilde\upphi_k|^2+a^2r^{-6}|Z\mathring{\slashed\nabla}\tilde\upphi_k|^2 +s^2r^{-4}|\mathring{\slashed\nabla}\tilde\upphi_k|^2\rp)drd\sigma \\
&\quad \leq B\sum_{X\in\{\mathrm{id},T\}}\mathbb{E}_p[X\tilde\upphi_k\mathbbm{1}_{[2R^*,\infty)}](\tau)+ B\int_{\Sigma_\tau\cap[2R^*,\infty)} \lp(r^{p-2}|L\mathring{\slashed\nabla}\tilde\upphi_k|^2+r^{-4}|\mathring{\slashed\nabla}^2\tilde\upphi_k|^2\rp)drd\sigma\,.
\end{align*}
Similarly, from $\p_{r^*}=-L+T+\frac{a}{r^2+a^2}Z$, we obtain
\begin{align*}
&\mathbb{E}_p[X^*\tilde\upphi_k\mathbbm{1}_{[2R^*,\infty)}](\tau)\\
&\quad\leq 
B \int_{\Sigma_\tau\cap[2R^*,\infty)} 
\lp(r^p|L^2\tilde\upphi_k|^2+r^{-2}|L\mathring{\slashed\nabla}\tilde\upphi_k|^2+r^{-2}|L\uL\tilde\upphi_k|^2 +s^2r^{-2}|L\tilde\upphi_k|^2\rp)drd\sigma \\
&\quad \qquad +B\mathbb{E}_p[T\tilde\upphi_k\mathbbm{1}_{[2R^*,\infty)}](\tau)+B(R^*)^{-1}\mathbb{E}_p[r^{-1}\mathring{\slashed\nabla}\tilde\upphi_k\mathbbm{1}_{[2R^*,\infty)}](\tau)\\
&\quad \leq B\sum_{X\in\{\mathrm{id},T\}}\mathbb{E}_p[X\tilde\upphi_k\mathbbm{1}_{[2R^*,\infty)}](\tau)\\
&\quad\qquad + B\int_{\Sigma_\tau\cap[2R^*,\infty)} \lp(r^{p}|L\uL\tilde\upphi_k|^2+r^p|L^2\tilde\upphi_k|^2+r^{p-2}|L\mathring{\slashed\nabla}\tilde\upphi_k|^2+r^{-4}|\mathring{\slashed\nabla}^2\tilde\upphi_k|^2\rp)drd\sigma\\
&\quad \leq B\mathbb{E}_p[T\tilde\upphi_k\mathbbm{1}_{[2R^*,\infty)}](\tau)+B\sum_{j=0}^{k}\mathbb{E}_p[\tilde\upphi_j\mathbbm{1}_{[2R^*,\infty)}](\tau) +B(|s|-k)\int_{\Sigma_\tau\cap[2R^*,\infty)} r^{p-4}|\tilde\upphi_{k+1}|^2 dr d\sigma\\
&\quad\qquad + B\int_{\Sigma_\tau\cap[2R^*,\infty)} \lp(r^p|L^2\tilde\upphi_k|^2+r^{p-2}|L\mathring{\slashed\nabla}\tilde\upphi_k|^2+r^{p-4}|\mathring{\slashed\nabla}^2\tilde\upphi_k|^2\rp)drd\sigma\,,
\end{align*}
using the PDE \eqref{eq:transformed-k-tilde} in the last inequality. Combining the previous estimates, we deduce
\begin{align*}
&\sum_{X\in\{\mathrm{id}, T,X^*,r^{-1}\mathring{\slashed\nabla}\}}\mathbb{E}_p[X\tilde\upphi_k\mathbbm{1}_{[2R^*,\infty)}](\tau)  \\
&\quad\leq B\mathbb{E}_p[T\tilde\upphi_k\mathbbm{1}_{[2R^*,\infty)}](\tau)+B\sum_{j=0}^{k}\mathbb{E}_p[\tilde\upphi_j\mathbbm{1}_{[2R^*,\infty)}](\tau) +B(|s|-k)\int_{\Sigma_\tau\cap[2R^*,\infty)} r^{p-4}|\tilde\upphi_{k+1}|^2 dr d\sigma\\
&\quad\qquad + B\int_{\Sigma_\tau\cap[2R^*,\infty)} \lp(r^p|L^2\tilde\upphi_k|^2+r^{p-2}|L\mathring{\slashed\nabla}\tilde\upphi_k|^2+r^{p-4}|\mathring{\slashed\nabla}^2\tilde\upphi_k|^2\rp)drd\sigma\,.\numberthis \label{eq:2nd-order-large-r-flux-intermediate}
\end{align*}
To control the last line in \eqref{eq:2nd-order-large-r-flux-intermediate}, we employ an elliptic estimate for \eqref{eq:transformed-k-tilde}:  we rewrite
\begin{align*}
\frac{1}{2}(L\uL+\uL L)+w\mathring{\slashed\triangle} &= \lp[w\mathring{\slashed\triangle}-L^2\rp]+2L\lp(T+\frac{a}{r^2+a^2} Z\rp) +\frac{4arw}{r^2+a^2} Z  \,,
\end{align*}
and use the ellipticity of the operator in square brackets to obtain
\begin{align*}
&\int_{\Sigma_\tau\cap[2R^*,\infty)}\lp(r^p|L^2\tilde\upphi_k|^2+r^{p-2}|L\mathring{\slashed\nabla}\tilde\upphi_k|^2+r^{p-4}|\mathring{\slashed\nabla}^2\tilde\upphi_k|^2\rp)drd\sigma\\
&\quad\leq  B\mathbb{E}_p[T\tilde\upphi_k\mathbbm{1}_{[R^*,\infty)}](\tau)+B\sum_{j=0}^{k}\mathbb{E}_p[\tilde\upphi_j\mathbbm{1}_{[R^*,\infty)}](\tau) +B(|s|-k)\int_{\Sigma_\tau\cap[R^*,\infty)} r^{p-4}|\tilde\upphi_{k+1}|^2 dr d\sigma\,.
\end{align*}
We conclude 
\begin{align*}
&\sum_{X\in\{\mathrm{id}, T,X^*,r^{-1}\mathring{\slashed\nabla}\}}\mathbb{E}_p[X\tilde\upphi_k\mathbbm{1}_{[2R^*,\infty)}](\tau)  \\
&\quad\leq B\mathbb{E}_p[T\tilde\upphi_k\mathbbm{1}_{[2R^*,\infty)}](\tau)+B\sum_{j=0}^{k}\mathbb{E}_p[\tilde\upphi_j\mathbbm{1}_{[2R^*,\infty)}](\tau) +B(|s|-k)\int_{\Sigma_\tau\cap[2R^*,\infty)} r^{p-4}|\tilde\upphi_{k+1}|^2 dr d\sigma\,.\numberthis \label{eq:2nd-order-large-r-flux}
\end{align*}

In an entirely similar fashion, one can show that
\begin{align*}
&\sum_{X\in\{\mathrm{id}, T,X^*,r^{-1}\mathring{\slashed\nabla}\}}\mathbb{I}_{-\delta,p}[X\tilde\upphi_k\mathbbm{1}_{[2R^*,\infty)}](0,\tau)  \\
&\quad\leq B\mathbb{I}_{-\delta,p}[T\tilde\upphi_k\mathbbm{1}_{[2R^*,\infty)}](0,\tau)+B\sum_{j=0}^{k}\mathbb{I}_{-\delta,p}[\tilde\upphi_j\mathbbm{1}_{[2R^*,\infty)}](0,\tau) \numberthis \label{eq:2nd-order-large-r-bulk-intermediate} \\
&\quad\qquad+B(|s|-k)\int_0^\tau\int_{\Sigma_{\tau'}\cap[2R^*,\infty)} (pr^{p-2}+r^{-2-\delta})r^{-3}|\tilde\upphi_{k+1}|^2 dr d\sigma d\tau'\\
&\quad\qquad + B\int_0^\tau\int_{\Sigma_{\tau'}\cap[2R^*,\infty)} \lp((pr^{p-1}+r^{-1-\delta})|L^2\tilde\upphi_k|^2+r^{p-3}|L\mathring{\slashed\nabla}\tilde\upphi_k|^2+r^{p-5}(2-p+r^{-1})|\mathring{\slashed\nabla}^2\tilde\upphi_k|^2\rp)drd\sigma\,.
\end{align*}
Then, by the same elliptic estimate as before, we obtain
\begin{align*}
&\int_{\mc R_{(0,\tau)}\cap[2R^*,\infty)}(pr^{p-1}+r^{-1-\delta})\lp(|L^2\tilde\upphi_k|^2+r^{-2}|L\mathring{\slashed\nabla}\tilde\upphi_k|^2+r^{-4}|\mathring{\slashed\nabla}^2\tilde\upphi_k|^2\rp)drd\sigma d\tau'\\
&\quad\leq  B\mathbb{I}_{-\delta,p}[T\tilde\upphi_k\mathbbm{1}_{[R^*,\infty)}](0,\tau)+B\sum_{j=0}^{k}\mathbb{I}_{-\delta,p}[\tilde\upphi_j\mathbbm{1}_{[R^*,\infty)}](0,\tau)\\
&\quad\qquad+B(|s|-k)\int_{\mc R_{(0,\tau)}\cap[R^*,\infty)} (pr^{p-1}+r^{-1-\delta})r^{-4}|\tilde\upphi_{k+1}|^2 dr d\sigma\,.
\end{align*}
Thus, for $p\neq 0$,
\begin{align*}
&\sum_{X\in\{\mathrm{id}, T,X^*,r^{-1}\mathring{\slashed\nabla}\}}\mathbb{I}_{-\delta,p}[X\tilde\upphi_k\mathbbm{1}_{[2R^*,\infty)}](0,\tau)  \\
&\quad\leq B\mathbb{I}_{-\delta,p}[T\tilde\upphi_k\mathbbm{1}_{[2R^*,\infty)}](0,\tau)+B\sum_{j=0}^{k}\mathbb{I}_{-\delta,p}[\tilde\upphi_j\mathbbm{1}_{[2R^*,\infty)}](0,\tau) \numberthis \label{eq:2nd-order-large-r-bulk-p-not-zero} \\
&\quad\qquad+B(|s|-k)\int_0^\tau\int_{\Sigma_{\tau'}\cap[2R^*,\infty)} r^{p-5}|\tilde\upphi_{k+1}|^2 dr d\sigma d\tau'\,.
\end{align*}
In the case $p=0$, our elliptic estimates are not enough to control the last two terms in \eqref{eq:2nd-order-large-r-bulk-intermediate}. However, in the $k<|s|$ we can control these terms in a different manner: if $s<0$ we can use the transport equation \eqref{eq:transformed-k-tilde}, and if $s>0$ we can use the fact that $\uL=-L+2\lp(T+\frac{a}{r^2+a^2}Z\rp)$ as well as the PDE  \eqref{eq:transformed-constraint} to deduce that 
\begin{align*}
&\int_{\mc R_{(0,\tau)}\cap[2R^*,\infty)}\lp(r^{-3}|L\mathring{\slashed\nabla}\tilde\upphi_k|^2+r^{-5}|\mathring{\slashed\nabla}^2\tilde\upphi_k|^2\rp)drd\sigma d\tau' \\
&\quad\leq B\mathbb{I}[T\tilde\upphi_{k}\mathbbm{1}_{[2R^*,\infty)}](0,\tau)+B\sum_{j=0}^{k}\mathbb{I}[\tilde\upphi_{j}\mathbbm{1}_{[2R^*,\infty)}](0,\tau)+ B\mathbb{I}[\tilde\upphi_{k+1}\mathbbm{1}_{[2R^*,\infty)}](0,\tau)\,.
\end{align*}
Hence, we obtain for $k<|s|$, 
\begin{align*}
&\sum_{X\in\{\mathrm{id}, T,X^*,r^{-1}\mathring{\slashed\nabla}\}}\mathbb{I}_{-\delta,p}[X\tilde\upphi_k\mathbbm{1}_{[2R^*,\infty)}](0,\tau)  \\
&\quad\leq B\mathbb{I}_{-\delta,p}[T\tilde\upphi_k\mathbbm{1}_{[2R^*,\infty)}](0,\tau)+B\sum_{j=0}^{k}\mathbb{I}_{-\delta,p}[\tilde\upphi_j\mathbbm{1}_{[2R^*,\infty)}](0,\tau)+B\mathbb{I}[\tilde\upphi_{k+1}\mathbbm{1}_{[2R^*,\infty)}](0,\tau)\,. \numberthis \label{eq:2nd-order-large-r-bulk-lower-level} 
\end{align*}

To conclude the present step, we would like to appeal to Proposition~\ref{prop:bulk-flux-large-r} and the fact that $T$ is Killing to control the first term on the right hand side in the estimates \eqref{eq:2nd-order-large-r-flux}, \eqref{eq:2nd-order-large-r-bulk-p-not-zero}, and \eqref{eq:2nd-order-large-r-bulk-lower-level}. However, the $r$-weights on $T\tilde\upphi_{k+1}$ that would come out of the $T$-commuted Proposition~\ref{prop:bulk-flux-large-r} are not ideal; to improve them, we revisit the proof, in particular the proof of estimate \eqref{eq:phys-space-rp-multiplier-consequence-s>0-no-peeling}. Note that the error term involving the $k+1$ variable is a bulk term
\begin{align*}
&c(r)T\tilde\upphi_{k+1}\overline{LT\tilde\upphi_k}\,,
\end{align*}
where $c(r)=r^p\xi\frac{wc_{k+1}}{c_k}\lp(\frac{c_k'}{c_k}-(|s|-k)\frac{w'}{w}\rp)$. If $s<0$, we can use the transport equation \eqref{eq:transformed-transport-tilde};  if $s>0$, we may apply Cauchy--Schwarz to close an estimate for $p\in[0,1]$. However, if $s>0$ then we can instead rely on the identities
\begin{align*}
&c(r)T\tilde\upphi_{k+1}\overline{LT\tilde\upphi_k}\\
&\quad = \frac12 c(r)T\tilde\upphi_{k+1}\overline{L\lp(\frac{c_k'}{c_k}\tilde\upphi_k+\frac{wc_{k+1}}{c_k}\tilde\upphi_{k+1}+\frac{2a}{r^2+a^2}Z\tilde\upphi_k\rp)}+L\lp(c(r)T\tilde\upphi_{k+1}\overline{L\tilde\upphi_k}\rp)-c'(r)T\tilde\upphi_{k+1}\overline{L\tilde\upphi_k}\\
&\qquad \quad-T\lp(c(r) L\tilde\upphi_{k+1}\overline{L\tilde\upphi_k}\rp)+c(r)L\tilde\upphi_{k+1}\overline{LT\tilde\upphi_k}\,,
\end{align*}
so, we obtain for $p\in[0,2)$
\begin{align*}
&\mathbb{E}_p[T\tilde\upphi_k\mathbbm{1}_{[2R^*,\infty)}](\tau)+b(\delta)\mathbb{I}_{-\delta,p}[T\tilde\upphi_k\mathbbm{1}_{[2R^*,\infty)}](0,\tau)\\
&\quad\leq B\sum_{j=0}^{k}\sum_{X\in\{\mathrm{id},T\}}\lp(\mathbb{E}_p[X\tilde\upphi_j](0)+\mathbb{E}[X\tilde\upphi_j\mathbbm{1}_{[R^*,2R^*]}](\tau)+\mathbb{I}[X\tilde\upphi_j\mathbbm{1}_{[R^*,2R^*]}](0,\tau)\rp) \\
&\quad\qquad  +B\int_0^\tau\int_{\Sigma_{\tau'}\cap[R^*,\infty)} \lp(r^{p-4}(|\tilde\upphi_{k+1}|^2+|T\tilde\upphi_{k+1}|^2)\mathbbm{1}_{\{s\leq 0\}}+r^{p-3}(|\tilde\upphi_{k+1}|^2+|L\tilde\upphi_{k+1}|^2)\mathbbm{1}_{\{s> 0\}}\rp) dr d\sigma\,,
\end{align*}
where we may include $p=2$ if $s>0$. For $\dbtilde\upphi_{k}$, a similar statement to the case $s<0$ holds but where the $p-4$ exponents above are replaced by $p-3$; thus, the last term is controlled by $\mathbb{I}_{-\delta,p}[\dbtilde\upphi_{k+1}](0,\tau)$ for $p\in(1,2)$ if we choose $\delta$ appropriately depending on $p$.

Combining all the previous estimates, we see that for any $\delta\in(0,1]$, we have 
\begin{align*}
&\sum_{X\in\{\mathrm{id}, T,X^*,r^{-1}\mathring{\slashed\nabla}\}}\lp(\mathbb{E}_p[X\tilde\upphi_k\mathbbm{1}_{[2R^*,\infty)}](\tau)+b(\delta)\mathbb{I}_{-\delta,p}[X\tilde\upphi_k\mathbbm{1}_{[2R^*,\infty)}](0,\tau)\rp)\\
&\quad\leq B\sum_{j=0}^{k}\Big[\sum_{X\in\{\mathrm{id},T,X^*,r^{-1}\mathring{\slashed\nabla}\}}\mathbb{E}_p[X\tilde\upphi_j](0)+\sum_{X\in\{\mathrm{id},T\}}\lp(\mathbb{E}[X\tilde\upphi_j\mathbbm{1}_{[R^*,2R^*]}](\tau)+\mathbb{I}[X\tilde\upphi_j\mathbbm{1}_{[R^*,2R^*]}](0,\tau)\rp)\Big]\\
&\quad\qquad  +B(|s|-k)\lp(\mathbb{E}[\tilde\upphi_{k+1}\mathbbm{1}_{[R^*,\infty)}](\tau)+\mathbb{I}[\tilde\upphi_{k+1}\mathbbm{1}_{[R^*,\infty)}](0,\tau)\rp)\,,\numberthis \label{eq:2nd-order-large-r-bulk-flux-conclusion-1}
\end{align*}
with the following constraints on $p$: if $k=|s|$ then $p\in(0,2]$; otherwise if $k<|s|$ and $s<0$ then $p\in[0,2]$, and if $k<|s|$ with  $s>0$ then $p\in[0,2)$. Note that we can apply Proposition~\ref{prop:bulk-flux-large-r} to the last line in \eqref{eq:2nd-order-large-r-bulk-flux-conclusion-1}.

\medskip
\noindent\textit{Step 2: the case of $r\to r_+$.} We turn now to the region close to $\mc H^+$ and we assume that $|a|\leq a_0<M$. We may again use the identities $L=-\uL+2K-\frac{\upomega_+(r^2-r_+^2)}{r^2+a^2}Z$ and $X^*=\frac{r^2+a^2}{\Delta}\uL+T$ to obtain 
\begin{align*}
&\sum_{X\in\{\mathrm{id},T,X^*,r^{-1}\mathring{\slashed\nabla}^{[s]}\}}\lp(\overline{\mathbb{E}}[X\tilde\upphi_k\mathbbm{1}_{(-\infty,-2R^*]}](\tau)+\overline{\mathbb{E}}_{\mc H^+}[X\tilde\upphi_k](0,\tau)+\overline{\mathbb{I}}[X\tilde\upphi_k\mathbbm{1}_{(-\infty,-2R^*]}](0,\tau)\rp)\\
&\quad \leq B(a_0)\sum_{X\in\{\mathrm{id},T,Z,{\frac{r^2+a^2}{\Delta}}\uL\}}\lp(\overline{\mathbb{E}}[X\tilde\upphi_k\mathbbm{1}_{(-\infty,-2R^*]}](\tau)+\overline{\mathbb{E}}_{\mc H^+}[X\tilde\upphi_k](0,\tau)+\overline{\mathbb{I}}[X\tilde\upphi_k\mathbbm{1}_{(-\infty,-2R^*]}](0,\tau)\rp) \\
&\quad\qquad +B(a_0)\int_{\Sigma_{\tau}\cap(-\infty,-2R^*]}|\mathring{\slashed\nabla}^2\tilde\upphi_k|^2 drd\sigma +B(a_0)\int_{\mc H^+_{(0,\tau)}}|\mathring{\slashed\nabla}^2\tilde\upphi_k|^2 \sigma d\tau'\\
&\quad\qquad +B(a_0)\int_0^\tau \int_{\Sigma_{\tau'}\cap(-\infty,-2R^*]}|\mathring{\slashed\nabla}^2\tilde\upphi_k|^2 drd\sigma d\tau'\,. \numberthis\label{eq:2nd-order-hor-intermediate-0} 
\end{align*}

To control the angular derivatives in \eqref{eq:2nd-order-hor-intermediate-0}, we use the standard elliptic estimate on $\mathbb{S}^2$ by interpreting the PDE  \eqref{eq:transformed-k-tilde} as an identity for $\mathring{\slashed \triangle}\tilde\upphi_k$. We rewrite
\begin{align*}
\frac{1}{2w}(L\uL+\uL L)&=L\lp(\frac{r^2+a^2}{\Delta}\uL\rp)+(r^2+a^2)\lp[\frac{r^2+a^2}{\Delta}\uL, L\rp] +\frac{1}{2w}[L,\uL]\\
&=L\lp(\frac{r^2+a^2}{\Delta}\uL\rp)-\frac{2M(r^2-a^2)}{r^2+a^2}\lp(\frac{r^2+a^2}{\Delta}\uL\rp)+\frac{2ar}{r^2+a^2}Z\,,
\end{align*}
obtaining the estimate
\begin{align*}
\int_{\Sigma_\tau \cap (-\infty,-2R^*]} |\mathring{\slashed\nabla}^2\tilde\upphi_k|^2 drd\sigma  
&\leq B(a_0)\sum_{X\in\{\mathrm{id},K,\frac{r}{r-r_+}\uL\}}\overline{\mathbb{E}}[X\tilde\upphi_k\mathbbm{1}_{\{r^*\leq -2R^*\}}](\tau)\\
&\qquad +B(a_0)(|s|-k)\int_{\Sigma_\tau\cap (-\infty,-2R^*]}|\tilde\upphi_{k+1}|^2 drd\sigma \,.
\end{align*}
Thus, \eqref{eq:2nd-order-hor-intermediate-0} becomes
\begin{align*}
&\sum_{X\in\{\mathrm{id},T,X^*,r^{-1}\mathring{\slashed\nabla}^{[s]}\}}\lp(\overline{\mathbb{E}}[X\tilde\upphi_k\mathbbm{1}_{(-\infty,-2R^*]}](\tau)+\overline{\mathbb{E}}_{\mc H^+}[X\tilde\upphi_k](0,\tau)+\overline{\mathbb{I}}[X\tilde\upphi_k\mathbbm{1}_{(-\infty,-2R^*]}](0,\tau)\rp)\\
&\quad \leq B(a_0)\sum_{X\in\{\mathrm{id},T,Z,{\frac{r^2+a^2}{\Delta}}\uL\}}\lp(\overline{\mathbb{E}}[X\tilde\upphi_k\mathbbm{1}_{(-\infty,-2R^*]}](\tau)+\overline{\mathbb{E}}_{\mc H^+}[X\tilde\upphi_k](0,\tau)+\overline{\mathbb{I}}[X\tilde\upphi_k\mathbbm{1}_{(-\infty,-2R^*]}](0,\tau)\rp) \\
&\quad\qquad+B(a_0)(|s|-k)\lp(\int_{\Sigma_\tau\cap (-\infty,-2R^*]}|\tilde\upphi_{k+1}|^2 drd\sigma+\int_0^\tau\int_{\Sigma_{\tau'}\cap (-\infty,-R^*]}|\tilde\upphi_{k+1}|^2 drd\sigma d\tau'\rp)\,. \numberthis\label{eq:2nd-order-hor-intermediate-1}
\end{align*}

We turn to the terms with $X=\frac{r^2+a^2}{\Delta}\uL$ on the right hand side of \eqref{eq:2nd-order-hor-intermediate-1}. Putting together estimates \eqref{eq:phys-space-Killing-multiplier-consequence-4-Y-version}--\eqref{eq:phys-space-redshift-multiplier-consequence-Y-version} as in the proof of Proposition~\ref{prop:bulk-flux-large-r}, we obtain its analogue for ${\textstyle\frac{r^2+a^2}{\Delta}\uL}{\tilde\upphi}_k$ in the region $r\to r_+$:
\begin{align*} 
&\overline{\mathbb{I}}[{\textstyle\frac{r^2+a^2}{\Delta}\uL}{\tilde\upphi}_k\mathbbm{1}_{\{r^*\leq -2R^*\}}](0,\tau) +\overline{\mathbb{E}}[{\textstyle\frac{r^2+a^2}{\Delta}\uL}{\tilde\upphi}_{k}\mathbbm{1}_{\{r^*\leq -2R^*\}}](\tau)+\overline{\mathbb{E}}_{\mc H^+}[{\textstyle\frac{r^2+a^2}{\Delta}\uL}{\tilde\upphi}_{k}](0,\tau)\\
&\quad\leq B(a_0)\sum_{j=0}^k\lp(\overline{\mathbb{E}}[\tilde\upphi_j\mathbbm{1}_{\{r^*\leq -R^*\}}](\tau)+\sum_{X\in\{\mathrm{id},T,Z\}}\overline{\mathbb{I}}[X\tilde\upphi_j\mathbbm{1}_{\{r^*\leq -R^*\}}](\tau)\rp)\\
&\quad\qquad + B(a_0)\sum_{j=0}^k\lp({\mathbb{I}}[{\textstyle\frac{r^2+a^2}{\Delta}\uL}{\tilde\upphi}_k\mathbbm{1}_{\{-2R^*\leq r^*\leq -R^*\}}](0,\tau) +{\mathbb{E}}[{\textstyle\frac{r^2+a^2}{\Delta}\uL}{\tilde\upphi}_{k}\mathbbm{1}_{\{-2R^*\leq r^*\leq -R^*\}}](\tau)\rp)\\
&\quad\qquad +B(a_0)\sum_{X\in\{\mathrm{id},\frac{r^2+a^2}{\Delta}\uL\}}\sum_{j=0}^{k}\overline{\mathbb{E}}[X{\tilde\upphi}_j](0) + B(a_0)(|s|-k)\int_{\mc R_{(0,\tau)}\cap\{r^*\leq -R^*\}} \lp|{\textstyle\frac{r^2+a^2}{\Delta}\uL}\tilde\upphi_{k+1}\rp|^2 dr d\sigma d\tau'\,.
\end{align*}

Finally, we note $\mathrm{span}\lp\{T,\frac{r^2+a^2}{\Delta}\uL\rp\}=\mathrm{span}\{T,X^*\}$ and, moreover, that for any terms on the right hand side we can replace $Z$ by $\mathring{\slashed\nabla}$ by Lemma~\ref{lemma:spherical-angular-ode}. Hence, after we appeal to Proposition~\ref{prop:bulk-flux-large-r} and the fact that $T$  and $Z$ are Killing, \eqref{eq:2nd-order-hor-intermediate-1} becomes
\begin{align*}
&\sum_{X\in\{\mathrm{id},T,X^*,r^{-1}\mathring{\slashed\nabla}^{[s]}\}}\lp(\overline{\mathbb{E}}[X\tilde\upphi_k\mathbbm{1}_{(-\infty,-2R^*]}](\tau)+\overline{\mathbb{E}}_{\mc H^+}[X\tilde\upphi_k](0,\tau)+\overline{\mathbb{I}}[X\tilde\upphi_k\mathbbm{1}_{(-\infty,-2R^*]}](0,\tau)\rp)\\
%&\quad \leq B(a_0)\sum_{X\in\{\mathrm{id},T,Z\}}\lp(\overline{\mathbb{E}}[X\tilde\upphi_k\mathbbm{1}_{(-\infty,-R^*]}](\tau)+\overline{\mathbb{E}}_{\mc H^+}[X\tilde\upphi_k](0,\tau)+\overline{\mathbb{I}}[X\tilde\upphi_k\mathbbm{1}_{(-\infty,-R^*]}](0,\tau)\rp) \\
%&\quad\qquad+B(a_0)(|s|-k)\lp(\int_{\Sigma_\tau\cap (-\infty,-R^*]}|\tilde\upphi_{k+1}|^2 drd\sigma+\int_0^\tau\int_{\Sigma_{\tau'}\cap (-\infty,-R^*]}|\tilde\upphi_{k+1}|^2 drd\sigma d\tau'\rp)\\
%&\quad\qquad +B(a_0)\sum_{j=0}^k\sum_{X\in\{\mathrm{id},T,X^*,r^{-1}\mathring{\slashed\nabla}^{[s]}\}}\overline{\mathbb{E}}[X\tilde\upphi_j](0)\\
%&\quad \qquad + B(a_0)\sum_{j=0}^k\lp({\mathbb{I}}[{\textstyle\frac{r^2+a^2}{\Delta}\uL}{\tilde\upphi}_k\mathbbm{1}_{\{-2R^*\leq r^*\leq -R^*\}}](0,\tau) +{\mathbb{E}}[{\textstyle\frac{r^2+a^2}{\Delta}\uL}{\tilde\upphi}_{k}\mathbbm{1}_{\{-2R^*\leq r^*\leq -R^*\}}](\tau)\rp)\\
&\quad \leq  B(a_0)(|s|-k)\sum_{X\in\{\mathrm{id},T,Z\}}\lp(\int_{\Sigma_\tau\cap (-\infty,-R^*]}|X\tilde\upphi_{k+1}|^2 drd\sigma+\int_0^\tau\int_{\Sigma_{\tau'}\cap (-\infty,-R^*]}|X\tilde\upphi_{k+1}|^2 drd\sigma d\tau'\rp)\\
&\quad\qquad+ B(a_0)\sum_{j=0}^k\sum_{X\in\{\mathrm{id}, T, X^*,r^{-1}\mathring{\slashed\nabla}^{[s]}\}}\lp({\mathbb{I}}[X{\tilde\upphi}_k\mathbbm{1}_{[-2R^*, -R^*]}](0,\tau) +{\mathbb{E}}[X{\tilde\upphi}_{k}\mathbbm{1}_{[-2R^*, -R^*]}](\tau)\rp)\\
&\quad\qquad+B(a_0)\sum_{X\in\{\mathrm{id},T,X^*,r^{-1}\mathring{\slashed\nabla}^{[s]}\}}\sum_{j=0}^k\overline{\mathbb{E}}[X\tilde\upphi_j](0)\\
&\quad \leq   B(a_0)\sum_{j=0}^k\sum_{X\in\{\mathrm{id}, T, X^*,r^{-1}\mathring{\slashed\nabla}^{[s]}\}}\lp(\overline{\mathbb{E}}[X\tilde\upphi_j](0)+{\mathbb{I}}[X{\tilde\upphi}_k\mathbbm{1}_{[-2R^*, -R^*]}](0,\tau) +{\mathbb{E}}[X{\tilde\upphi}_{k}\mathbbm{1}_{[-2R^*, -R^*]}](\tau)\rp)\\
&\qquad\quad + B(a_0)(|s|-k)\lp({\mathbb{I}}[{\tilde\upphi}_{k+1}\mathbbm{1}_{(-\infty,-R^*]}](0,\tau) +{\mathbb{E}}[{\tilde\upphi}_{k+1}\mathbbm{1}_{(-\infty,-R^*]}](\tau)\rp)
\,. \numberthis\label{eq:2nd-order-hor-conclusion}
\end{align*}
Note that we can invoke Proposition~\ref{prop:bulk-flux-large-r} to control the last line of \eqref{eq:2nd-order-hor-conclusion}. We also note that if $s>0$ and $k<|s|$, 
\begin{align*}
\sum_{X\in\{\mathrm{id}, T, X^*,r^{-1}\mathring{\slashed\nabla}^{[s]}\}}\overline{\mathbb{E}}[X\tilde\upphi_k](0)&\leq  \sum_{X\in\{\mathrm{id}, T, X^*,r^{-1}\mathring{\slashed\nabla}^{[s]}\}}{\mathbb{E}}[X\tilde\upphi_k](0)+ \int_{\Sigma_\tau}\lp|\frac{\uL}{w}(X\tilde\upphi_j)\rp|^2 drd\sigma \\
&\leq B{\mathbb{E}}[X\tilde\upphi_{k+1}](0)+B\sum_{X\in\{\mathrm{id}, T, X^*,r^{-1}\mathring{\slashed\nabla}^{[s]}\}}{\mathbb{E}}[X\tilde\upphi_k](0)\,,
\end{align*}
and hence we do not need the norms with an overbar with optimal weights near $r_+$. More generally, if $s>0$ and $k<|s|$, we need not appeal to the commutator in Lemma~\ref{lemma:redshift-commutation} to control the right hand side of \eqref{eq:2nd-order-hor-intermediate-0}: we may instead rely on the identity \eqref{eq:transformed-constraint} and the transport equation \eqref{eq:transformed-transport-tilde} to show that we do not need any overbars in \eqref{eq:higher-order-energy-bddness} and \eqref{eq:higher-order-ILED} as long as we drop the term $\mathbb{E}_{\mc H^+}[\Phi]$ from the left hand side.

\medskip
\noindent \textit{Step 3: fluxes through $\mc H^+$ and $\mc I^+$.} Let us first consider the case $s<0$ and $k<|s|$. Since $\uL\tilde\upphi_k=-L\tilde\upphi_k+2\lp(T+\frac{a^2}{r^2+a^2}Z\rp)\tilde\upphi_k$, the transport equation \eqref{eq:transformed-transport-tilde} implies that we can replace $\uL$ by $T$ derivatives along $\mc I^+$:
\begin{align*}
&\sum_{X\in\{\mathrm{id},\mathring{\slashed\nabla},\uL\}}\mathbb{E}_{\mc I^+,p}[X\tilde\upphi_k](0,\tau)\\
&\quad\leq B\mathbb{E}_{\mc I^+,p}[\tilde\upphi_k](0,\tau)+B\lim_{r\to \infty}\int_{\mc I^+_{(0,\tau)}} \lp(r^{p-2}|\mathring{\slashed\nabla}\uL\tilde\upphi_k|^2+r^{p-}|\mathring{\slashed\triangle}\tilde\upphi_{k}|^2+|\uL\uL\tilde\upphi_k|^2\rp)d\sigma d\tau'\\
&\quad\leq B\sum_{X\in\{\mathrm{id},T\}}\mathbb{E}_{\mc I^+,p}[X\tilde\upphi_k](0,\tau)+B\lim_{r\to \infty}\int_{\mc I^+_{(0,\tau)}} r^{p-2}|\mathring{\slashed\triangle}\tilde\upphi_{k}|^2 d\sigma d\tau'\,.
\end{align*}
Now, from \eqref{eq:transformed-constraint}, we deduce
\begin{align*}
\int_{\mc I^+_{(0,\tau)}}|\mathring{\slashed\triangle}\tilde\upphi_{k}|^2 d\sigma d\tau \leq B\int_{\mc I^+_{(0,\tau)}}|\uL\tilde\upphi_{k+1}|^2 d\sigma d\tau'+B\sum_{X\in\{\mathrm{id},T\}}\mathbb{E}_{\mc I^+,2}[X\tilde\upphi_k](0,\tau)+\sum_{j=0}^k\mathbb{E}_{\mc I^+,2}[\tilde\upphi_j](0,\tau)\,,
\end{align*}
thus obtaining 
\begin{align*}
\sum_{j=0}^k\sum_{X\in\{\mathring{\slashed\nabla},\uL\}}\mathbb{E}_{\mc I^+,p}[X\tilde\upphi_j](0,\tau)&
\leq B\sum_{j=0}^k\sum_{X\in\{\mathrm{id},T\}}\mathbb{E}_{\mc I^+,p}[X\tilde\upphi_k](0,\tau)+B\mathbb{E}_{\mc I^+,0}[\tilde\upphi_{k+1}](0,\tau)\,.
\end{align*}

An analogous reasoning at $\mc H^+$ with $s>0$ and $k<|s|$ leads to 
\begin{align*}
\sum_{j=0}^k\sum_{X\in\{\mathring{\slashed\nabla},L\}}\overline{\mathbb{E}}_{\mc H^+}[X\tilde\upphi_j](0,\tau)\leq B\sum_{j=0}^k\sum_{X\in\{\mathrm{id},T,Z\}}\overline{\mathbb{E}}_{\mc H^+}[X\tilde\upphi_j](0,\tau)+B\overline{\mathbb{E}}_{\mc H^+}[\tilde\upphi_{k+1}](0,\tau)\,.
\end{align*}

\medskip
\noindent \textit{Step 4: bounded $|r^*|$}. It remains to understand how to obtain higher order estimates away from $r^*=\pm \infty$. Using the wave equation \eqref{eq:transformed-k-tilde},
\begin{align*}
-\p_{r^*}^2+w\mathring{\slashed{\triangle}}&=  \lp(T+\frac{a}{r^2+a^2}Z\rp)^2 +w(a^2\sin^2\theta TT+2aTZ-2ias\cos\theta T)-\frac{4arw}{r^2+a^2}\sign{s}\lp(|s|-k\rp)Z\\
&\qquad -\sign s \frac{c_k'}{c_k}\underline{\mc L}\tilde\upphi_k
+ \frac{a^2\Delta^2}{(r^2+a^2)^4}\lp[1-2|s|-2k(2|s|-k-1)\rp]+w[|s|+k(2|s|-k-1)]\\
&\qquad +\frac{2Mr(r^2-a^2)\Delta}{(r^2+a^2)^4}\lp[1-3|s|+2s^2-3k(2|s|-k-1)\rp]-\lp(\frac{c_k'}{c_k}\rp)'\tilde\upphi_k\\
&\qquad +\sign s\lp(\frac{c_k'}{c_k}-(|s|-k)\frac{w'}{w}\rp)\frac{wc_{k+1}}{c_k}\tilde\upphi_{k+1}\,,
\end{align*}
and the fact that $T$ and $Z$ span a timelike direction for bounded $|r^*|$, we can run elliptic estimates: an $H^2$ estimate  over the hypersurfaces $\Sigma_\tau$,
\begin{align*}
&\int_{\Sigma_\tau\cap \{|r^*|\leq R^*\}} \lp(|\tilde\upphi_k''|^2+|\mathring{\slashed\nabla}\tilde\upphi_k'|^2+|\mathring{\slashed\nabla}^2\tilde\upphi_k|^2\rp)drd\sigma d\tau' \\
&\quad\leq B(R^*)\int_{\Sigma_\tau\cap \{|r^*|\leq 2R^*\}} \Big[\Big|\Big(T+\frac{a}{r^2+a^2}Z\Big)^2\tilde\upphi_k\Big|^2+(|s|-k)|\tilde\upphi_{k+1}|^2\Big]drd\sigma  \\
&\qquad + B(R^*)\sum_{j=0}^k\mathbb{E}[\tilde\upphi_j\mathbbm{1}_{[-2R^*,2R^*]}](\tau) \numberthis\label{eq:H2-estimate}
\,, 
\end{align*}
and an $H^1$ estimate over $\mc R_{(0,\tau)}$,
\begin{align*}
&\int_{\mc R_{(0,\tau)}\cap \{|r^*|\leq R^*\}} \lp(|\tilde\upphi_k'|^2+|\mathring{\slashed\nabla}\tilde\upphi_k|^2\rp)drd\sigma d\tau' \\
&\quad\leq B(R^*)\int_{\mc R_{(0,\tau)}\cap \{|r^*|\leq 2R^*\}} \Big[\Big|\Big(T+\frac{a}{r^2+a^2}Z\Big)\tilde\upphi_k\Big|^2+\sum_{j=0}^{k}|\tilde\upphi_j|^2+(|s|-k)|\tilde\upphi_{k+1}|^2\Big]drd\sigma d\tau' \numberthis\label{eq:H1-estimate}\\
&\quad\qquad +B(R^*)\mathbb{E}[\tilde\upphi_{k}\mathbbm{1}_{[-2R^*,2R^*]}](\tau)+B(R^*)\sum_{j=0}^k\mathbb{E}[\tilde\upphi_{j}\mathbbm{1}_{[-2R^*,2R^*]}](0) \,. 
\end{align*}

Let us first review the applications of the $H^2$ estimate \eqref{eq:H2-estimate} towards the present step. For $k<|s|$, we note that by the transport equation \eqref{eq:transformed-transport-tilde}, we have 
\begin{align*}
\mathbb{E}[\big({\textstyle T+\frac{a}{r^2+a^2}Z}\big)\tilde\upphi_k\mathbbm{1}_{[-R^*,R^*]}](\tau) &\leq B(R^*)\mathbb{E}[X^*\tilde\upphi_{k}\mathbbm{1}_{[-2R^*,2R^*]}](\tau)+B(R^*)\mathbb{E}[\tilde\upphi_{k}\mathbbm{1}_{[-2R^*,2R^*]}](\tau)\\
&\qquad +B(R^*)\int_{\Sigma_\tau\cap [-2R^*,2R^*]} |\tilde\upphi_{k+1}|^2 dr d\sigma\,;
\end{align*}
thus, combining with \eqref{eq:H2-estimate}, we obtain 
\begin{align*}
\sum_{X\in\{\mathrm{id},T,X^*,r^{-1}\mathring{\slashed\nabla}\}\}}\mathbb{E}[X\tilde\upphi_k\mathbbm{1}_{[-R^*,R^*]}](\tau) &\leq B(R^*)\mathbb{E}[X^*\tilde\upphi_{k}\mathbbm{1}_{[-2R^*,2R^*]}](\tau)+B(R^*)\sum_{j=0}^k\mathbb{E}[\tilde\upphi_{j}\mathbbm{1}_{[-2R^*,2R^*]}](\tau)\\
&\qquad +B(R^*)\int_{\Sigma_\tau\cap [-2R^*,2R^*]} |\tilde\upphi_{k+1}|^2 dr d\sigma\,. \numberthis\label{eq:2nd-order-bdd-r-bulk-low-level}
\end{align*}
We apply \eqref{eq:2nd-order-bdd-r-bulk-low-level} whenever $k=|s|-1$. An even more direct estimate following from \eqref{eq:H2-estimate}, which we apply in other settings, is 
\begin{align*}
&\sum_{X\in\{\mathrm{id},T,X^*,r^{-1}\mathring{\slashed\nabla}\}}\mathbb{E}[X{\tilde\upphi}_k\mathbbm{1}_{[-R^*,R^*]}](\tau)\\
&\quad \leq B(R^*)\sum_{X\in\{\mathrm{id},T, Z\}}\sum_{j=0}^{k}\mathbb{E}[X{\tilde\upphi}_j\mathbbm{1}_{[-2R^*,2R^*]}](\tau) +B(R^*)(|s|-k)\int_{\Sigma_\tau\cap [-2R^*,2R^*]} |\tilde\upphi_{k+1}|^2 dr d\sigma \numberthis\label{eq:2nd-order-bdd-r-bulk-intermediate-higher}\,.
\end{align*}
Furthermore, by adapting the (proof of the) $H^2$ estimate \eqref{eq:H2-estimate} and integrating in time, we can show
\begin{align*}
&\sum_{X\in\{\mathrm{id},T,X^*,r^{-1}\mathring{\slashed\nabla}\}}\mathbb{I}^{\rm deg}[X\tilde\upphi_k\mathbbm{1}_{[-R^*,R^*]}](0,\tau)\\
&\quad \leq B(R^*)\sum_{X\in\{\mathrm{id},T,X^*,r^{-1}\mathring{\slashed\nabla}\}}\int_{\Sigma_\tau\cap [-2R^*,2R^*]} \zeta\lp(|TX\tilde\upphi_k|^2+|\mathring{\slashed\nabla} X\tilde\upphi_k|^2\rp) dr d\sigma\\
&\quad\qquad+B(R^*)\mathbb{I}[\tilde\upphi_k\mathbbm{1}_{[-2R^*,2R^*]}](0,\tau) +B(R^*)\mathbb{I}[X^*\tilde\upphi_k\mathbbm{1}_{[-2R^*,2R^*]}](0,\tau) \\
&\quad\leq B(R^*)\sum_{X\in\{\mathrm{id},T, X^*,r^{-1}\mathring{\slashed\nabla}\}}\int_{\Sigma_\tau\cap [-2R^*,2R^*]} \zeta\Big|X\Big(T+\frac{a}{r^2+a^2}Z\Big)\tilde\upphi_k\Big|^2 dr d\sigma\\
&\quad\qquad+B(R^*)(|s|-k)\int_{\Sigma_\tau\cap [-2R^*,2R^*]} |\tilde\upphi_{k+1}|^2 dr d\sigma\\
&\quad\qquad+B(R^*)\sum_{j=0}^k\mathbb{I}[\tilde\upphi_j\mathbbm{1}_{[-2R^*,2R^*]}](0,\tau) +B(R^*)\mathbb{I}[X^*\tilde\upphi_k\mathbbm{1}_{[-2R^*,2R^*]}](0,\tau)\,.
\end{align*} 
If $k<|s|$, by the transport equation \eqref{eq:transformed-transport-tilde},
\begin{align*}
&\sum_{X\in\{\mathrm{id},T,X^*,r^{-1}\mathring{\slashed\nabla}\}}\mathbb{I}^{\rm deg}[X\tilde\upphi_k\mathbbm{1}_{[-R^*,R^*]}](0,\tau)\\
&\quad\leq B(R^*)\mathbb{I}^{\rm deg}[\tilde\upphi_{k+1}\mathbbm{1}_{[-2R^*,2R^*]}](0,\tau)\\
&\qquad\quad +B(R^*)\sum_{j=0}^k\mathbb{I}[\tilde\upphi_j\mathbbm{1}_{[-2R^*,2R^*]}](0,\tau) +B(R^*)\mathbb{I}[X^*\tilde\upphi_k\mathbbm{1}_{[-2R^*,2R^*]}](0,\tau)\,; \numberthis\label{eq:2nd-order-bdd-r-bulk-intermediate-low-level}
\end{align*} 
we invoke this estimate whenever $k=|s|-1$. In other settings, we rely on the more direct estimate 
\begin{align*}
&\sum_{X\in\{\mathrm{id},T,X^*,r^{-1}\mathring{\slashed\nabla}\}}\mathbb{I}^{\rm deg}[X\tilde\upphi_k\mathbbm{1}_{[-R^*,R^*]}](0,\tau)\\
&\quad \leq B(R^*)\sum_{X\in\{\mathrm{id},T, Z\}}\sum_{j=0}^{k}\mathbb{I}^{\rm deg}[X{\tilde\upphi}_j\mathbbm{1}_{[-2R^*,2R^*]}](0,\tau) +B(R^*)(|s|-k)\int_{\Sigma_\tau\cap [-2R^*,2R^*]} |\tilde\upphi_{k+1}|^2 dr d\sigma\\
&\quad\qquad+B(R^*)\sum_{j=0}^k\mathbb{I}[\tilde\upphi_j\mathbbm{1}_{[-2R^*,2R^*]}](0,\tau) +B(R^*)\mathbb{I}[X^*\tilde\upphi_k\mathbbm{1}_{[-2R^*,2R^*]}](0,\tau) \,.\numberthis\label{eq:2nd-order-bdd-r-bulk-intermediate-high-level}
\end{align*} 

To conclude, we turn to applying the $H^1$ estimate \eqref{eq:H1-estimate} to control the last line in \eqref{eq:2nd-order-bdd-r-bulk-intermediate-low-level} and \eqref{eq:2nd-order-bdd-r-bulk-intermediate-high-level}. For the first term, we have 
\begin{align*}
&\mathbb{I}[\tilde\upphi_k\mathbbm{1}_{[-R^*,R^*]}](0,\tau)\\
&\quad \leq \mathbb{I}^{\rm deg}[\tilde\upphi_k\mathbbm{1}_{[-R^*,R^*]}](0,\tau) + \int_{\mc R_{(0,\tau)}\cap \{|r^*|\leq R^*\}} \lp(|T\tilde\upphi_k|^2+|\mathring{\slashed\nabla}\tilde\upphi_k|^2\rp)drd\sigma d\tau' \\
&\quad \leq B(R^*)\sum_{X\in\{\mathrm{id},T, Z\}}\sum_{j=0}^{k}\mathbb{I}^{\rm deg}[X{\tilde\upphi}_j\mathbbm{1}_{[-2R^*,2R^*]}](0,\tau)+B(|s|-k)\int_0^\tau\int_{\Sigma_{\tau'}\cap [-2R^*,2R^*]}|\tilde\upphi_{k+1}|^2dr d\sigma d \tau'\\
&\quad\qquad +B(R^*)\sum_{X\in\{\mathrm{id},T\}}\mathbb{E}[X{\tilde\upphi}_k\mathbbm{1}_{[-2R^*,2R^*]}](\tau)+B\sum_{X\in\{\mathrm{id},T\}}\mathbb{E}[X{\tilde\upphi}_k](0)\,,
\end{align*}
whenever $\mathbb{I}$ and $\mathbb{I}^{\rm deg}$ do not agree. For the second term in the last line of \eqref{eq:2nd-order-bdd-r-bulk-intermediate-low-level} and \eqref{eq:2nd-order-bdd-r-bulk-intermediate-high-level}, we adapt (the proof of the) $H^1$ estimate \eqref{eq:H1-estimate} to deal with an $\p_{r^*}$ commuted version of the PDE~\eqref{eq:transformed-k-tilde},  see Lemma~\ref{lemma:radial-commutation}:
\begin{align*}
&\mathbb{I}[X^*\tilde\upphi_k\mathbbm{1}_{[-R^*,R^*]}](0,\tau)\\
&\quad\leq B(R^*)\int_{\mc R_{(0,\tau)}\cap \{|r^*|\leq 2R^*\}} \Big[\Big|\Big(T+\frac{a}{r^2+a^2}Z\Big)\tilde\upphi_k'\Big|^2+\sum_{j=0}^{k}(|\tilde\upphi_j|^2+|\tilde\upphi_j'|^2)+|\tilde\upphi_{k+1}'|^2+|\tilde\upphi_{k+1}|^2\Big]drd\sigma d\tau' \\
&\quad\qquad +B(R^*)\mathbb{I}[\tilde\upphi_{k}\mathbbm{1}_{[-2R^*,2R^*]}](0,\tau)\\
&\qquad\quad +B(R^*)\mathbb{E}[\tilde\upphi_{k}'\mathbbm{1}_{[-2R^*,2R^*]}](\tau)+B(R^*)\sum_{j=0}^k\mathbb{E}[\tilde\upphi_{j}'\mathbbm{1}_{[-2R^*,2R^*]}](0)\\
&\quad\leq B(R^*)\int_{\mc R_{(0,\tau)}\cap \{|r^*|\leq 2R^*\}} \Big[\Big|\Big(T+\frac{a}{r^2+a^2}Z\Big)\tilde\upphi_k'\Big|^2 drd\sigma d\tau' \\
&\quad\qquad +B(R^*)\sum_{j=0}^k\mathbb{I}[\tilde\upphi_{j}\mathbbm{1}_{[-2R^*,2R^*]}](0,\tau) +B(R^*)\mathbb{I}^{\rm deg}[\tilde\upphi_{k+1}\mathbbm{1}_{[-2R^*,2R^*]}](0,\tau)\\
&\qquad\quad +B(R^*)\mathbb{E}[\tilde\upphi_{k}'\mathbbm{1}_{[-2R^*,2R^*]}](\tau)+B(R^*)\sum_{j=0}^k\mathbb{E}[\tilde\upphi_{j}'\mathbbm{1}_{[-2R^*,2R^*]}](0)\\
&\quad \leq B(R^*)\sum_{X\in\{\mathrm{id},T,Z\}}\sum_{j=0}^k\mathbb{I}^{\rm deg}[X\tilde\upphi_{j}\mathbbm{1}_{[-2R^*,2R^*]}](0,\tau) +B(|s|-k)\mathbb{I}^{\rm deg}[\tilde\upphi_{k+1}\mathbbm{1}_{[-2R^*,2R^*]}](0,\tau) \\
&\quad\qquad +B(R^*)\mathbb{E}[\tilde\upphi_{k}'\mathbbm{1}_{[-2R^*,2R^*]}](\tau)+B(R^*)\sum_{j=0}^k\mathbb{E}[\tilde\upphi_{j}'\mathbbm{1}_{[-2R^*,2R^*]}](0)\,;
\end{align*}
where we find it convenient to consider the middle inequality if $k=|s|-1$. This concludes step 4 and, thus, the entire proof.
\end{proof}

\begin{remark}[Peeling for fixed $\tau$] \label{rmk:peeling-fixed-tau} It is straightforward to use these various weighted estimates, the equation~\eqref{eq:transformed-k-tilde}, Sobolev inequalities on $\mathbb{S}^2$, and the fundamental theorem of calculus,  to show that if $\upphi_k^{[s]}$ arises from smooth compactly supported data, then there exists a constant $C(\tau,\upphi_k^{[s]})$ depending only on $\tau$ and $\upphi_k^{[s]}$ so that we have
\[\left|\dbtilde\upphi_k^{[s]}\right| \leq C\left(\tau,\upphi_k^{[s]}\right),\qquad \left|L\dbtilde\upphi_k^{[s]}\right| \leq C\left(\tau,\upphi_k^{[s]}\right)r^{-2}.\]
The same statement holds if we replace $\dbtilde\upphi_k^{[s]}$ with any sequence of derivatives from $\left\{X^*,T,r^{-1}\mathring{\slashed\nabla}^{[s]}\right\}$ applied to $\dbtilde\upphi_k^{[s]}$.
\end{remark}

\subsection{Well-posedness, regularity and smooth dependence}
\label{sec:wellposedness}

Standard theory yields that the transformed system of Definition~\ref{def:transformed-system}, with the choice of $r$-weights leading to the rescaled variables $\swei{\tilde\upphi}{s}_k$ or $\swei{\dbtilde\upphi}{s}_k$, is well-posed. We refer the reader to \cite{Dafermos2010,Dafermos2017} for further details.

\begin{proposition}[Well-posedness] \label{prop:well-posedness}
Let $j\geq |s|$, and take $(\tilde\upphi^{[s],\circ}_0, \tilde\upphi^{[s],\circ\circ}_0)\in {}^{[s]}H^j_{\rm loc}(\Sigma_0)\times {}^{[s]}H^{j-1}_{\rm loc}(\Sigma_0)$ to be  $s$-spin-weighted functions.

Then, for each $k=0,\dots, |s|$ there exist unique $s$-spin-weighted $\swei{\tilde\upphi}{s}_k$ on $\mathcal{R}_0$ satisfying the homogeneous wave-type PDE \eqref{eq:transformed-k-tilde} and the transport relations \eqref{eq:transformed-transport-tilde} with 
$\swei{\tilde\upphi}{s}_k \in {}^{[s]}H^{j-k}_{\rm loc} (\Sigma_\tau)$, $n_{\Sigma_{\tau}} \swei{\tilde\upphi}{s}_0\in {}^{[s]}H^{j-k-1}_{\rm loc}(\Sigma_\tau)$ such that $\swei{\tilde\upphi}{s}_0\big|_{\Sigma_0}=\tilde\upphi^{[s],\circ}_0$, and 
$(n_{\Sigma_0}\swei{\tilde\upphi}{s}_0)\big|_{\Sigma_0}=\tilde\upphi^{[s],\circ\circ}_0$. 
In particular, if 
 $(\tilde\upphi^{[s],\circ}_0, \tilde\upphi^{[s],\circ\circ}_0)\in \mathscr{S}^{[s]}_{\infty}(\Sigma_0)$ 
then $\swei{\tilde\upphi}{s}_k\in \mathscr{S}^{[s]}_{\infty}(\mathcal{R}_0)$ for all $k=0,\dots,|s|$. Finally, the map $(\tilde\upphi^{[s],\circ}_0, \tilde\upphi^{[s],\circ\circ}_0)\times a\mapsto \{\swei{\upphi}{s}_{k,a}\}_{k=0}^{|s|}$, where $\swei{\upphi}{s}_{k,a}$ denotes the above $\swei{\upphi}{s}_{k}$ considering the equations with respect to Kerr parameters $(a,M)$, is $C^0\times C^\infty$.

The same statement holds with $\Sigma_0$, $\Sigma_\tau$, and 
$\tilde{\mathcal{R}}_0$ in place of $\Sigma_0$, $\Sigma_\tau$, and $\mathcal{R}_0$,
respectively.
\end{proposition}

\subsection{A backwards extension procedure}
\label{sec:extension}

In this section, we present a backwards extension argument which will be convenient in simplifying the setting of the proof of our main result, see already Lemma~\ref{lemma:redshift-commutation}, and in our proof of energy boundedness in Theorem~\ref{thm:energy-bddness-futurint}. The extensions argument is as follows:

\begin{proposition}[An extension procedure] \label{prop:scattering-construction} Let $s\in\{0,\pm 1,\pm 2\}$. For $k=0,\dots,|s|$, let $\swei{\tilde\upphi}{s}_{k}$ be  solutions to the homogeneous transformed system on $\tilde{\mc{R}}_{(0,1)}$ arising from smooth data on $\Sigma_0$, with the following additional properties:
\begin{itemize}
\item $\swei{\Phi}{s}|_{\Sigma_0}$ is compactly supported and, if $|a|=M$, is compactly supported away from $\mc H^+\cap \Sigma_0$;
\item if $s>0$, $\swei{\tilde\upphi}{s}_k|_{\Sigma_0}$ is compactly supported  for $k=0,\dots,|s|$;
\item if $s<0$ and $|a|=M$, $\swei{\tilde\upphi}{s}_k|_{\Sigma_0}$ is supported away from $\mc H^+\cap \Sigma_0$ for $k=0,\dots,|s|$.
\end{itemize}
Then, for any large integer $q$, we can find extensions, $\swei{\hat{\upphi}}{s}_{k} \in C^q$, solving the transformed system in the spacetime slab $\tilde{\mc{R}}_{(-1,1)}$ such that 
\begin{enumerate}[label=(\roman*)]
\item $\swei{\hat{\upphi}}{s}_{k}=\swei{\tilde{\upphi}}{s}_{k}$ on $\tilde{\mc{R}}_{(0,1)}$; \label{it:scattering-construction-equalities-1}
\item If $s > 0$, then $\swei{\hat{\upphi}}{s}_{k}$ vanishes for $r$ sufficiently large in $\tilde{\mc{R}}_{(-1,1)}$;
\item $\swei{\hat{\upphi}}{s}_{k}\equiv 0$ on $\mc I^+_{(-1,-\frac12)}$ and on $\mc H^+_{(-1,-\frac12)}$; \label{it:scattering-construction-equalities-2}
\item we have the following estimates: if $|a|=M$, then \label{it:scattering-construction-energy}
\begin{align*}
\sum_{k=0}^{|s|}\sum_{0\leq j_1+j_2\leq |s|-k}\mathbb{E}[Z^{j_1}T^{j_2}\swei{\hat{\upphi}}{s}_{k}](-1)\leq B\sum_{k=0}^{|s|}\sum_{0\leq j_1+j_2\leq |s|-k}\mathbb{E}[Z^{j_1}T^{j_2}\swei{\tilde{\upphi}}{s}_{k}](0)\,,
\end{align*}
but if $|a|\leq a_0<M$, then
\begin{align*}
\sum_{k=0}^{|s|}\sum_{0\leq j_1+j_2\leq |s|-k}\overline{\mathbb{E}}[Z^{j_1}T^{j_2}\swei{\hat{\upphi}}{s}_{(k)}](-1)\leq B(a_0)\sum_{k=0}^{|s|}\sum_{0\leq j_1+j_2\leq |s|-k}\overline{\mathbb{E}}[Z^{j_1}T^{j_2}\swei{\tilde{\upphi}}{s}_{k}](0)\,.
\end{align*}
\end{enumerate}
\end{proposition}

To prove Proposition~\ref{prop:scattering-construction}, we want to consider a mixed spacelike-characteristic initial value problem for the transformed system, where data is set\footnote{To say one sets data on $\mc I^+$ is an abuse of notation which we will use throughout what follows. More rigorously, we set data on $\Sigma_0\cup\mc{H}_{(-1,0)}^+\cup \mc{C}$ where $\mc{C}$ is a suitable null cone and later take the limit $r(\Sigma_0\cap \mc{C})\to \infty$.} on $\Sigma_0\cup\mc{I}_{(-1,0)}^+\cup\mc{H}_{(-1,0)}^+$, we require that $\left(\swei{\hat{\upphi}}{s}_{k},n_{\Sigma_0}\swei{\hat{\upphi}}{s}_{k}\right)|_{\Sigma_0} = \left(\swei{\tilde\upphi}{s}_k,n_{\Sigma_0}\swei{\tilde\upphi}{s}_k\right)|_{\Sigma_0}$, we solve to the past of $\Sigma_0$, and then glue the result to the previously defined solution $\swei{\tilde\upphi}{s}_k$. In order for the resulting solution $\swei{\hat{\upphi}}{s}_{k}$ to be $C^q$ across $\Sigma_0$ we must have that the data we set for $\swei{\hat{\upphi}}{s}_{k}$ along $\mc{H}_{(-1,0)}$ and $\mc{I}_{(-1,0)}^+$ is compatible with the traces of $\swei{\tilde\upphi}{s}_k$ along $\mc{H}_{(0,1)}$ and $\mc{I}_{(0,1)}^+$. More concretely, if the traces along $\mc{I}_{(-1,1)}^+$ and $\mc{H}_{(-1,1)}^+$ are $C^{\tilde{q}}$ for a sufficiently large $\tilde{q} > q$, then the solution will also be $C^q$. The following lemma will be useful for the construction of data satisfying these compatibility conditions.

\begin{lemma}[$C^k$ extension]\label{lemma:Ck-extension} Let $k\geq 0$ and consider a function $f\in C^k([0,1)\times\mathbb{S}^2)$. Then, there is a map $\mathrm{Ext}_k[f]\colon (-1,1)\times\mathbb{S}^2\to \mathbb{R}$ such that 
\begin{itemize}
\item $\mathrm{Ext}_k[f]=f$ for $x\in[0,1)$, 
\item $\mathrm{Ext}_k[f]\in C^k((-1,1)\times\mathbb{S}^2)$, and
\item $\lVert \mathrm{Ext}_k[f]\rVert_{H^k([-1,0]\times\mathbb{S}^2)}\lesssim_k \lVert \mathrm{Ext}_k[f]\rVert_{H^k([0,1]\times\mathbb{S}^2)}$.
\end{itemize}
\end{lemma}
%\begin{proof}
%For $x>0$, we define
%\begin{align*}
%\mathrm{Ext}_0[f](-x)&=f(x)\,,
%\mathrm{Ext}_1[f](-x)&=-3f(x)+4f(x/2)\,,
%\mathrm{Ext}_1[f](-x)&=af(x)+bf(x/2)+cf(x/4)\,,
%\end{align*}
%with $a+b+c=a+b/2+c/4=a+b/4+c/16=1$, and so on. See Proposition 2.27 in \url{https://web.math.princeton.edu/~seri/homepage/courses/Analysis2011.pdf}.
%\end{proof}

We are now ready to establish points \ref{it:scattering-construction-equalities-1} and \ref{it:scattering-construction-equalities-2} of Lemma~\ref{prop:scattering-construction}:

\begin{proof}[Proof of Lemma~\ref{prop:scattering-construction}(i,ii)] By well-posedness for the transformed system, the solutions $\tilde\upphi_{k}$ to the future of $\Sigma_0$ are smooth, in particular on the hypersurface\footnote{We are again employing a slight abuse of notation: $\mc{I}_{(0,1)}^+$ should be replaced by suitable future directed cones $\{\mc {C}_r\}_{r > 0}$ which approaches $\mc{I}_{(0,1)}^+$ as $r\to \infty$.} $\Sigma_0\cup\mc{I}_{(0,1)}^+\cup\mc{H}_{(0,1)}^+$. Thus, a natural first attempt at an ansatz for the characteristic data in our mixed spacelike-characteristic problem is to consider an extension by 0 via  Lemma~\ref{lemma:Ck-extension}:
\begin{align}
\hat{\upphi}_{0}\Big|_{\mc I^+}(u,\theta,\phi):=\chi_0(u)\mathrm{Ext}_{N}[\tilde\upphi_{0}](u,\theta,\phi) \,, \qquad \hat{\upphi}_{0}\Big|_{\mc H^+}(v,\theta,\phi):=\chi_0(v)\mathrm{Ext}_{N}[\tilde\upphi_{0}](v,\theta,\phi) \,, \label{eq:scattering-ansatz-naive}
\end{align}
where $u\in[-1,1]$ is an affine parameter along $\mc{I}^+_{(-1,1)}$, $v\in[-1,1]$ is an affine parameter along $\mc{I}^+_{(-1,1)}$, $N$ is sufficiently large and $0\leq \chi_0(z)\leq 1$ is a smooth function which is 1 for $|z|\ll 1$ and 0 for $z\leq -1/4$. In fact, we can and will drop the tilde on the right hand side of the $\mc H^+$ ansatz if $s>0$, since $\tilde\upphi_k|_{\mc H^+}$ and $\upphi_k|_{\mc H^+}$ are the same up to a constant. 

The above ansatzes yield the desired statement when $s=0$, but they fail in the case of $s\neq 0$, due to the existence of conservations laws along $\mc{I^+}$, if $s<0$, and  $\mc H^+$, if $s>0$. Indeed, taking the limit as $r\to\infty$ in \ref{eq:transformed-constraint}, we get for $s<0$
\begin{align}
|s|\leq 2\implies \uL \hat\upphi_{1}\Big|_{\mc I^+}&=- (\mathring{\slashed{\triangle}}^{[s]}+|s|)\hat\upphi_{0}\Big|_{\mc I^+}+\lp(\frac14a^2\sin^2\theta \uL\uL+aZ\uL+ia|s|\cos\theta \uL\rp)\hat\upphi_{0}\Big|_{\mc I^+}\,,\label{eq:constraint-null-infty-1}\\
|s|=2\implies \uL \hat\upphi_{2}\Big|_{\mc I^+}&=- (\mathring{\slashed{\triangle}}^{[-2]}+4)\hat\upphi_{1}\Big|_{\mc I^+}+\lp(\frac14a^2\sin^2\theta \uL\uL+aZ\uL+2ia\cos\theta \uL\rp)\hat\upphi_{1}\Big|_{\mc I^+}\nonumber\\
&\qquad-2\lp(3M+5aZ\rp)\hat\upphi_{0}\Big|_{\mc I^+}\,,\label{eq:constraint-null-infty-2}
\end{align}
and taking the limit $r\to r_+$ we get for $s>0$ 
\begin{align*}
|s|\leq 2\implies  L\hat\upphi_1\Big|_{\mc H^+}
&=-\lp(\mathring{\slashed{\triangle}}^{[+s]}_{\mc H}+|s|-\frac{2a}{M}(s-1)Z+\kappa r_+(s-1)(2s-1)\rp)\hat\upphi_0\Big|_{\mc H^+}\\
&\qquad+\lp(\frac{a\rho^2(r_+)}{2Mr_+}Z+\frac14a^2\sin^2\theta L-ias\cos\theta\rp)L\hat\upphi_0\Big|_{\mc H^+}+(s-1)\kappa \hat\upphi_1\Big|_{\mc H^+}\,,\numberthis\label{eq:constraint-horizon-1}\\
|s|=2\implies L\hat\upphi_2\Big|_{\mc H^+}
&=-\lp(\mathring{\slashed{\triangle}}^{[+2]}_{\mc H}+4-3\kappa r_+\rp)\hat\upphi_1\Big|_{\mc H^+} +\kappa r_+\lp(3r_+-10aZ\rp)\hat\upphi_0\Big|_{\mc H^+} \\
&\qquad  +\lp(\frac{a\rho^2(r_+)}{2Mr_+}Z+\frac14s^2\sin^2\theta L-2ia\cos\theta\rp)L\hat\upphi_1\Big|_{\mc H^+}\,,\numberthis\label{eq:constraint-horizon-2}
\end{align*}
where we have introduced the notation $\kappa=(r_+-r_-)/(2Mr_+)$ and 
\begin{align}
\mathring{\slashed{\triangle}}^{[+s]}_{\mc H}:=\mathring{\slashed{\triangle}}^{[+s]}+\frac{a^2}{Mr_+}ZZ-\frac{a^4}{(2Mr_+)^2}\sin^2\theta ZZ -\frac{a^2}{Mr_+}is\cos\theta Z-\frac{a}{M}Z\,. \label{eq:def-triangle-H}
\end{align}
In order to accommodate these extra constraints in the allowed data across the characteristic hypersurfaces, we must change the ansatz \eqref{eq:scattering-ansatz-naive}: for $s<0$ we change the ansatz at $\mc{I^+}$ to be 
\begin{align}
\hat{\upphi}_{0}\Big|_{\mc I^+}(u,\theta,\phi):=\chi_0(u)\mathrm{Ext}_{N}[\upphi_0](u,\theta,\phi) + \sum_{j=1}^{|s|}\chi_j(u)\hat{\upphi}_{0,j}(\theta,\phi)\,, \label{eq:scattering-ansatz-s<0}
\end{align}
and for $s>0$ we change the ansatz at $\mc{H^+}$ to be 
\begin{align}
\hat{\upphi}_{0}\Big|_{\mc H^+}(v,\theta,\phi):=\chi_0(v)\mathrm{Ext}_{N}[\upphi_{0}](u,\theta,\phi) + \sum_{j=1}^{|s|}\chi_j(v)\hat{\upphi}_{0,j}(\theta,\phi)\,. \label{eq:scattering-ansatz-s>0}
\end{align}
Here, $\hat{\upphi}_{0,j}$ denotes smooth $s$-spin-weighted functions on $\mathbb{S}^2$ which are, for the moment, free. Furthermore, $\chi_{1}^\pm(z)$ and $\chi_{2}^\pm(z)$ are smooth functions supported in $(-1/2,-1/4)$ such that
\begin{gather*}
\int_{-1/2}^{-1/4}\chi_1^-(z)dz=b_{1,1}\,,\quad \int_{-1/2}^{-1/4}\int_{-1/2}^{z}\chi_1^-(\tilde z)d\tilde z  dz=b_{1,2}\,,\\
\int_{-1/2}^{-1/4}\chi_2^-(z)dz=0\,,\quad\int_{-1/2}^{-1/4}\int_{-1/2}^{z}\chi_2^-(\tilde z)d\tilde z =b_{2,2}\,,
\\
\int_{-1/2}^{-1/4}\chi_1^+(v)dv=\tilde b_{1,1}\,, \qquad\int_{-1/2}^{-1/4}e^{\kappa v}\chi_1^+(v)dv=:b_{1,1}\,,\quad \int_{-1/2}^{-1/4}e^{-\kappa v}\int_{-1/2}^{v}e^{\gamma \tilde v}\chi_1^+(\tilde v)d\tilde vdv =b_{1,2}\,,\\
\int_{-1/2}^{-1/4}\chi_2^+(v)dv=0\,, \qquad \int_{-1/2}^{-1/4}e^{\kappa v}\chi_2^+(v)dv=0\,,\quad\int_{-1/2}^{-1/4}e^{-\kappa v}\int_{-1/2}^{v}e^{\gamma \tilde v}\chi_2^+(\tilde v)d\tilde vdv =b_{2,2}\,,
\end{gather*}
for some fixed $b_{1,1},b_{2,2}\in\mathbb{R}\backslash\{0\}$ and $\tilde b_{1,1},b_{1,2}\in\mathbb{R}$. For instance, $\chi_{2}^-(z)$ can be taken to be an odd function given by the sum of two bump functions, and $\chi_2^+$ can be taken to be, for $0<\epsilon\ll 1$,
\begin{align*}
\chi_2^+(z)=\chi_{2,1}^+(z)+A_{1/3}\chi_{2,\epsilon}^+(z+1/3)+A_{5/12}\chi_{2,\epsilon}^+(z+5/12)\,,
\end{align*}
for $\chi_{2,1}^+$ a bump function supported for $z\in(-1/4,-1/3)$, $\chi_{2,\epsilon}^+$ an approximation to the identity as $\epsilon\to 0$ and $A_{1/3}$ and $A_{5/12}$ uniquely determined by the above conditions on $\chi_2^+$.

\medskip
\noindent \textit{The case of $s<0$.} We can now integrate \eqref{eq:constraint-null-infty-1} and \eqref{eq:constraint-null-infty-2} along $\mc{I}^+$. In light of the ansatz \eqref{eq:scattering-ansatz-s<0}, integrating along $u\in[0,-1/4]$ yields that $\hat\upphi_{k}\Big|_{\mc I^+}(u=-1/4,\theta,\phi)$ are uniquely determined by the initial data at $u=0$. Next, we integrate along $u\in[-1/4,-1/2]$ to determine the second term of \eqref{eq:scattering-ansatz-s<0} in such a way that $\hat\upphi_{k}\Big|_{\mc I^+}(u=-1/2,\theta,\phi)=0$.

Indeed, if $\hat\upphi_{k}\Big|_{\mc I^+}(u=-1/2,\theta,\phi)=0$, we have 
\begin{align*}
\hat\upphi_{1}\Big|_{\mc I^+}(u,\theta,\phi)&= -\sum_{j=1}^{|s|}\lp(\int_{-1/2}^{u}\chi_j^-( u')d u'\rp) \lp(\mathring{\slashed{\triangle}}^{[s]}+|s|\rp)\hat{\upphi}_{0,j}(\theta,\phi)\\
&\qquad + \sum_{j=1}^{|s|}\lp[\frac14a^2\sin^2\theta\p_u\chi_j^-(u)+a\chi_j^-(u)\lp(Z+i|s|\cos\theta \rp)\rp]\hat{\upphi}_{0,j}(\theta,\phi)\,,
\end{align*}
and for $s=-2$,
\begin{align*}
&\hat\upphi_{2}\Big|_{\mc I^+}(u,\theta,\phi)\\
&\quad=\sum_{j=1}^{2}\lp(\int_{-1/2}^{u}\int_{-1/2}^{ u'}\chi_j^-( u'')d u''du'\rp) \lp(\mathring{\slashed{\triangle}}^{[-2]}+4\rp)\lp(\mathring{\slashed{\triangle}}^{[-2]}+2\rp)\hat{\upphi}_{0,j}(\theta,\phi)\\
&\quad\qquad -\sum_{j=1}^{2}\int_{-1/2}^u a\chi_j^-( u)d u'\lp[4i\cos\theta+(Z+2i\cos\theta)\lp(2\mathring{\slashed{\triangle}}^{[-2]}+6-a(Z+2i\cos\theta)\rp)\rp]\hat{\upphi}_{0,j}(\theta,\phi)\\
&\quad\qquad -\frac14\sum_{j=1}^2 \lp[a^2\chi_j^-(u)\lp(\mathring{\slashed{\triangle}}^{[-2]}+4-a(Z+2i\cos\theta)\rp)\hat{\upphi}_{(0,j)}(\theta,\phi)+\frac{1}{4}a^4\sin^4\theta \p_u\chi_j^-(u)\rp]\hat{\upphi}_{0,j}(\theta,\phi)\\
&\quad\qquad -\sum_{j=1}^2\int_{-1/2}^u\chi_j^-( u')d u' \lp(3M+5aZ\rp)\hat\upphi_{0,1}(\theta,\phi)\,.
\end{align*}
With the assumptions made on the cutoff functions $\chi_j^-$, we deduce that 
\begin{align*}
\hat\upphi_{1}\Big|_{\mc I^+}(u=-1/4,\theta,\phi) = -b_{1,1} \lp(\mathring{\slashed{\triangle}}^{[s]}+|s|\rp)\hat{\upphi}_{0,1}(\theta,\phi)\,,
\end{align*}
from which $\hat{\upphi}_{0,1}$ can be uniquely determined, and for $s=-2$ additionally we have
\begin{align*}
&\hat\upphi_{2}\Big|_{\mc I^+}(u=-1/4,\theta,\phi)\\
&\quad= b_{2,2}\lp(\mathring{\slashed{\triangle}}^{[-2]}+4\rp)\lp(\mathring{\slashed{\triangle}}^{[-2]}+2\rp)\hat{\upphi}_{0,2}(\theta,\phi)+b_{1,2}\lp(\mathring{\slashed{\triangle}}^{[-2]}+4\rp)\lp(\mathring{\slashed{\triangle}}^{[-2]}+2\rp)\hat{\upphi}_{0,1}(\theta,\phi)\\
&\quad\qquad- b_{1,1}\lp[4ia\cos\theta+a(Z+2i\cos\theta)\lp(2\mathring{\slashed{\triangle}}^{[-2]}+6-a(Z+2i\cos\theta)\rp)+2(3M+5aZ)\rp]\hat{\upphi}_{0,1}(\theta,\phi)\,,
\end{align*}
from which $\hat{\upphi}_{0,2}$ can subsequently be uniquely determined.

\medskip
\noindent \textit{The case of $s>0$.} When $s<0$, we proceed similarly: we integrate \eqref{eq:constraint-horizon-1} and \eqref{eq:constraint-horizon-2} along $\mc{H}^+$. By \eqref{eq:scattering-ansatz-s>0}, integrating along $v\in[0,-1/4]$ yields that $\hat\upphi_{(k)}\Big|_{\mc H^+}(v=-1/4,\theta,\phi)$ are uniquely determined by the initial data at $v=0$; integrating along $u\in[-1/4,-1/2]$ allows us to determine the second term of \eqref{eq:scattering-ansatz-s>0} in such a way that $\hat\upphi_{k}\Big|_{\mc H^+}(v=-1/2,\theta,\phi)=0$. If $\hat\upphi_{k}\Big|_{\mc H^+}(v=-1/2,\theta,\phi)=0$, we have
\begin{align*}
&e^{\kappa v}\hat\upphi_1\Big|_{\mc H^+}(v,\theta,\phi)-e^{\kappa v}\lp(\frac{a\rho^2(r_+)}{2Mr_+}Z+\frac14a^2\sin^2\theta L-ias\cos\theta\rp)\hat\upphi_0\Big|_{\mc H^+}(v,\theta,\phi)\\
&\quad=-\sum_{j=1}^{|s|}\int_{-1/2}^ve^{\kappa  v'}\chi_j^+(v')d v'\lp(\mathring{\slashed{\triangle}}^{[+s]}_{\mc H}+s-\frac{2a}{M}(s-1)Z+\kappa r_+(s-1)(2s-1)\rp)\hat\upphi_{0,j}(\theta,\phi)\,,
\end{align*}
and, if $s=+2$,
\begin{align*}
&\hat\upphi_2\Big|_{\mc H^+}(v,\theta,\phi)-\lp(\frac{a\rho^2(r_+)}{2M^2}+\frac14a^2\sin^2\theta L -2ia\cos\theta\rp)\hat \upphi_1\Big|_{\mc H^+}(v,\theta,\phi)\\
&\qquad +\lp(\mathring{\slashed{\triangle}}^{[+2]}_{\mc H}+4-3\kappa r_+\rp)\lp(\frac14a^2\sin^2\theta\hat\upphi_0\Big|_{\mc H^+}(v,\theta,\phi)\rp)\\
&\quad=\sum_{j=1}^2\lp(\int_{-1/2}^v e^{-\kappa v'}\int_{-1/2}^{ v'} e^{\kappa v''}\chi_j^+(v'') dv'' d v'\rp) \lp(\mathring{\slashed{\triangle}}^{[+2]}_{\mc H}+4-3\kappa r_+\rp)\lp(\mathring{\slashed{\triangle}}^{[+2]}_{\mc H}+2-\frac{2a}{M}Z+3\kappa r_+\rp) \hat \upphi_{0,j}(\theta,\phi) \\
&\quad\qquad -\sum_{j=1}^{2}\int_{-1/2}^v\chi_j^+( v')d v'\lp[\lp(\mathring{\slashed{\triangle}}^{[+2]}_{\mc H}+4-3\kappa r_+\rp)\lp(\frac{a\rho(r_+)}{2Mr_+}Z-2ia\cos\theta\rp)-\kappa r_+(3r_+-10aZ)\rp]\hat \upphi_{0,j}(\theta,\phi) \,.
\end{align*}
Using the assumptions on $\chi_j^+$, we deduce that
\begin{align*}
&\hat\upphi_1\Big|_{\mc H^+}(v=-1/4,\theta,\phi)=-b_{1,1} \lp(\mathring{\slashed{\triangle}}^{[+s]}_{\mc H}+|s|-\frac{2a}{M}(s-1)Z+\kappa r_+(s-1)(2s-1)\rp)\hat\upphi_{0,1}(\theta,\phi)\,,
\end{align*}
from which $\hat\upphi_{0,1}(\theta,\phi)$ can be uniquely determined, and 
\begin{align*}
&\hat\upphi_2\Big|_{\mc H^+}(v=-1/4,\theta,\phi)-\lp(\frac{a\rho^2(r_+)}{2M^2}+\frac14a^2\kappa\sin^2\theta -2ia\cos\theta\rp)\upphi_1\Big|_{\mc H^+}(v=-1/4,\theta,\phi)\\
&\quad=\sum_{j=1}^2b_{j,2} \lp(\mathring{\slashed{\triangle}}^{[+2]}_{\mc H}+4-3\kappa r_+\rp)\lp(\mathring{\slashed{\triangle}}^{[+2]}_{\mc H}+2-\frac{2a}{M}Z+3\kappa r_+\rp) \hat \upphi_{0,j}(\theta,\phi) \\
&\quad\qquad -\tilde b_{1,1}\lp[\lp(\mathring{\slashed{\triangle}}^{[+2]}_{\mc H}+4-3\kappa r_+\rp)\lp(\frac{a\rho(r_+)}{2Mr_+}Z-2ia\cos\theta\rp)-\kappa r_+(3r_+-10aZ)\rp]\hat \upphi_{0,1}(\theta,\phi)  \,,
\end{align*}
from which one can then uniquely determine $\hat\upphi_{0,2}(\theta,\phi)$. 

Note that, in order to argue that $\hat\upphi_{0,j}$ can be obtained \textit{uniquely} from the previous identities, we are appealing to the invertibility of the angular operators 
\begin{align*}
\mathring{\slashed{\triangle}}^{[+2]}_{\mc H}+4-3\kappa r_+\,,\quad \mathring{\slashed{\triangle}}^{[+2]}_{\mc H}+2-\frac{2a}{M}Z+3\kappa r_+\,,\quad \mathring{\slashed{\triangle}}^{[+1]}_{\mc H}+1-\frac{2a}{M}Z\,,
\end{align*}
where $\mathring{\slashed{\triangle}}^{[s]}_{\mc H}$ is defined in  \eqref{eq:def-triangle-H}. Indeed, if $a=0$ and in the space of axisymmetric functions the above operators reduce to, respectively,
\begin{align*}
\mathring{\slashed{\triangle}}^{[+2]}+1+3(1-\kappa r_+)\,,\quad \mathring{\slashed{\triangle}}^{[+2]}_{\mc H}+2+3\kappa r_+\,,\quad \mathring{\slashed{\triangle}}^{[+1]}+1\,,
\end{align*}
whose smallest eigenvalues are lower bounded by that of $\mathring{\slashed{\triangle}}^{[s]}+1$, and so the operators are invertible. If $a\neq 0$ and we do not restrict to axisymmetry, then the terms $(a/M)Z$ lead to a nontrivial imaginary part of the eigenvalues of the operators which ensure their invertibility. 
\end{proof}

It remains to obtain energy estimates for the past extended solutions:

\begin{proof}[Proof of Lemma~\ref{prop:scattering-construction}\ref{it:scattering-construction-energy}] 
Since $\hat\upphi_k$ solve the transformed system on $\mc{R}_{(-1,0)}$ it follows from the finite in time estimates of Proposition~\ref{prop:finite-in-time-first-order} that
\begin{align*}
\sum_{k=0}^{|s|}\mathbb{E}[\hat{\upphi}_{(k)}](-1)\leq B\sum_{k=0}^{|s|}\lp[\mathbb{E}[\tilde{\upphi}_{(k)}](0)+\mathbb{E}_{\mc I^+}[\hat{\upphi}_{(k)}](-1,0)+\mathbb{E}_{\mc H^+}[\hat{\upphi}_{(k)}](-1,0)\rp]\,.
\end{align*}
Thus, it remains to show that, for the extension we constructed, the energy fluxes through $\mc I^+_{(-1,0)}$ and $\mc H^+_{(-1,0)}$ can be controlled from data on $\Sigma_0$.

\medskip
\noindent \textit{The case of $\mc I^+_{(-1,0)}$}. 
In the case $s\geq 0$, our ansatz for the extension ensures that 
\begin{align*}
\int_{\mc I_{(-1,0)}^+}|\uL\hat\upphi_{k}|^2d\sigma d\tau &=  \int_{\mc I_{(-1,0)}^+}|\uL\uL^{k}\hat\upphi_{0}|^2d\sigma d\tau=0\,,
\end{align*}

In the case $s<0$, the construction at $\mc I^+_{(-1,0)}$ leads to the estimate
\begin{align*}
\int_{\mc I_{(-1,0)}^+}|\uL\hat\upphi_{k}|^2d\sigma d\tau &= \int_{\mc I_{(-\frac12,-\frac14)}^+}|\uL\hat\upphi_{k}|^2d\sigma d\tau + \int_{\mc I_{(-\frac14,0)}^+}|\uL\hat\upphi_{k}|^2d\sigma d\tau\\
&\leq \int_{\mc I_{(-\frac12,-\frac14)}}|\uL\hat\upphi_{k}|^2d\sigma d\tau + B\int_{\mathbb{S}^2}\lp|\hat\upphi_k\big|_{\mc I^+}(u=-1/4,\theta,\phi)\rp|^2 d\sigma \\
&\leq B\int_{\mc I_{(-\frac12,-\frac14)}^+}|\uL\hat\upphi_{k}|^2d\sigma d\tau \,.
\end{align*}
Suppose $k=1$; from \eqref{eq:constraint-null-infty-1}  we compute 
\begin{align*}
\uL \hat{\upphi}_1\Big|_{\mc I^+}(u)&=\chi_0(u)\mathrm{Ext}_N[\uL{\upphi}_1](u)+\chi_0'(u)\lp(\frac12a^2\sin^2\theta\uL+a(Z+i|s|\cos\theta)\rp)\mathrm{Ext}_N[{\upphi}_0](u)\\
&\qquad +\frac14a^2\sin^2\theta\chi_0''(u)(\mathring{\slashed{\triangle}}^{[s]}+|s|)^{-1}\mathrm{Ext}_N[(\mathring{\slashed{\triangle}}^{[s]}+|s|){\upphi}_0](u)
\\
&\qquad+\chi_0(u)\lp[\frac14a^2\sin^2\theta \uL\uL+aZ\uL+ia|s|\cos\theta \uL,\mathrm{Ext}_N\rp]{\upphi}_0(u)\\\
&=\chi_0(u)\mathrm{Ext}_N[\uL{\upphi}_1](u)+\chi_0'(u)\lp(\frac12a^2\sin^2\theta\uL+a(Z+i|s|\cos\theta)\rp)\mathrm{Ext}_N[{\upphi}_0](u)\\
&\qquad +\frac14a^2\sin^2\theta\chi_0''(u)(\mathring{\slashed{\triangle}}^{[s]}+|s|)^{-1}\mathrm{Ext}_N\lp[\lp(\frac14a^2\sin^2\theta\uL+a(Z+i|s|\cos\theta)\rp)\uL\upphi_0-\uL{\upphi}_1\rp](u)
\\
&\qquad+\chi_0(u)\lp[\frac14a^2\sin^2\theta \uL\uL+aZ\uL+ia|s|\cos\theta \uL,\mathrm{Ext}_N\rp]{\upphi}_0(u)\,,
\end{align*}
for $u\in(-1/4,0)$. Thus, we have the estimate
\begin{align*}
\int_{\mc I_{(-\frac12,-\frac14)}^+}|\uL\hat\upphi_{1}|^2d\sigma d\tau \leq B \int_{\mc I_{(\frac14,\frac12)}^+}\lp(|\uL\upphi_{1}|^2+|\uL\uL\upphi_0|^2+|Z\uL\upphi_0|^2+|\uL\upphi_0|^2\rp)d\sigma d\tau\,;
\end{align*} 
a similar computation for $k=2$ yields  
\begin{align*}
&\int_{\mc I_{(-\frac12,-\frac14)}^+}|\uL\hat\upphi_{2}|^2d\sigma d\tau \\
&\quad\leq B \int_{\mc I_{(\frac14,\frac12)}^+}\lp(|\uL\upphi_{2}|^2+|\uL\uL\upphi_1|^2+|\uL\upphi_1|^2+|Z\uL\upphi_1|^2\rp)d\sigma d\tau\\
&\quad\qquad + B \int_{\mc I_{(\frac14,\frac12)}^+}\lp(|\uL\uL\uL\upphi_0|^2+|Z\uL\uL\upphi_0|^2+|ZZ\uL\upphi_0|^2+|\uL\uL\upphi_0|^2+|Z\uL\upphi_0|^2+|\uL\upphi_0|^2\rp)d\sigma d\tau\,.
\end{align*}
Furthermore, since $\uL=-L+2\lp(T+\frac{a}{r^2+a^2}Z\rp)$ and we have the transport equation \eqref{eq:transformed-transport}, we can replace $\uL$ by $T$ above. Moreover, since $T$ and $Z$ are Killing, we can commute the previous estimates with these vector fields. Thus, from the finite in time estimates of Proposition~\ref{prop:finite-in-time-first-order} and the fact that $T$ and $Z$ are Killing vectors, we have
\begin{align*}
\sum_{k=0}^{|s|}\sum_{0\leq j_1+j_2\leq |s|-k}\mathbb{E}_{\mc I^+}[Z^{j_1}T^{j_2}\hat{\upphi}_{(k)}](-1,1)&\leq B\sum_{k=0}^{|s|}\sum_{0\leq j_1+j_2\leq |s|-k}\mathbb{E}_{\mc I^+}[Z^{j_1}T^{j_2}\tilde{\upphi}_{(k)}](0,1)\\
&\leq B\sum_{k=0}^{|s|}\sum_{j_1+j_2\leq |s|-k}\mathbb{E}[Z^{j_1}T^{j_2}\tilde{\upphi}_{(k)}](0)\,. 
\end{align*}

\medskip
\noindent \textit{The case of $\mc H^+_{(-1,0)}$}. 
If $s>0$, analogous formulas as those above may be derived from the construction of the ansatz for data on $\mc H^+_{(-1,0)}$, keeping in mind the triangle inequality
\begin{align*}
\lp|L\upphi_{1}\rp|^2\leq \lp|L(e^{\kappa v}\upphi_{1})\rp|^2+\kappa^2|\hat\upphi_{1}|\,.
\end{align*}
Repeating the procedure of the previous step, we have
\begin{align*}
\sum_{k=0}^{|s|}\mathbb{E}_{\mc H^+}[\hat{\upphi}_{k}](-1,1)&\leq B\sum_{k=0}^{|s|}\sum_{j_1+j_2\leq |s|-k}\mathbb{E}_{\mc H^+}[Z^{j_1}L^{j_2}\tilde{\upphi}_{k}](0,1)+B(M^2-a^2)\sum_{k=1}^{|s|-1}\overline{\mathbb{E}}_{\mc H^+}[\tilde{\upphi}_{k}](0,1) \\
&\leq B\sum_{k=0}^{|s|}\sum_{j_1+j_2\leq |s|-k}\mathbb{E}_{\mc H^+}[Z^{j_1}T^{j_2}\tilde{\upphi}_{k}](0,1)+B(M^2-a^2)\sum_{k=1}^{|s|-1}\overline{\mathbb{E}}_{\mc H^+}[\tilde{\upphi}_{k}](0,1) \,, 
\end{align*} 
where in the last line we used that, similarly to the previous case, the identity $L=-\uL+2\lp(T+\frac{a}{r^2+a^2}Z\rp)$ and the transport equation \eqref{eq:transformed-transport} allows us to replace $L$ by $K$.  Finite in time estimates of Proposition~\ref{prop:finite-in-time-first-order} and the fact that $T$ and $Z$ are Killing vectors then implies
\begin{align*}
&\sum_{k=0}^{|s|}\sum_{j_1+j_2\leq |s|-k}\mathbb{E}_{\mc H^+}[Z^{j_1}T^{j_2}\hat{\upphi}_{k}](-1,1)\\
&\quad\leq B\sum_{k=0}^{|s|}\sum_{j_1+j_2\leq |s|-k}\mathbb{E}_{\mc H^+}[Z^{j_1}L^{j_2}\tilde{\upphi}_{k}](0,1)+B(M^2-a^2)\sum_{k=1}^{|s|-1}\overline{\mathbb{E}}_{\mc H^+}[\tilde{\upphi}_{k}](0,1) \\
&\quad\leq B\sum_{k=0}^{|s|}\sum_{0\leq j_1+j_2\leq |s|-k}\mathbb{E}[Z^{j_1}T^{j_2}\tilde{\upphi}_{k}](0) +B(M^2-a^2)\sum_{k=1}^{|s|-1}\overline{\mathbb{E}}[\tilde{\upphi}_{k}](0)\,. 
\end{align*}

In the case $s\leq 0$, we have
\begin{align*}
\int_{\mc H_{(-1,0)}^+}|L\hat\upphi_{k}|^2d\sigma d\tau &=  \int_{\mc H_{(-1,0)}^+}\lp|L \lp(2Mr_+L-(|s|-k)\frac{r_+-M}{Mr_+}\rp)^k \hat\upphi_{0}\rp|^2d\sigma d\tau\,.
\end{align*}
If $|a|=M$, then our ansatz dictates that $\hat\upphi_{0}$ and its $L$ derivatives vanish. However, if $|a|< M$, the ansatz \eqref{eq:scattering-ansatz-naive} leads to 
\begin{align*}
\int_{\mc H_{(-1,0)}^+}|L\hat\upphi_{k}|^2d\sigma d\tau &\leq B \sum_{j=0}^k\int_{\mc H_{(0,1)}^+}\lp(|\tilde\upphi_j|^2+|L\tilde\upphi_{j}|^2\rp)d\sigma d\tau\,.
\end{align*}
Thus, we have
\begin{align*}
\sum_{k=0}^{|s|}\mathbb{E}_{\mc H^+}[\hat{\upphi}_{(k)}](-1,1)&\leq B\sum_{k=0}^{|s|}\mathbb{E}_{\mc H^+}[\tilde{\upphi}_{(k)}](0,1)\leq B \sum_{k=0}^{|s|}{\mathbb{E}}[\tilde{\upphi}_{(k)}](0)
\end{align*}
with $\mathbb{E}$ replaced by $\overline{\mathbb{E}}$ if $|a|<M$.
\end{proof}

\section{Precise statement of the main theorem}
\label{sec:precise-statement}

Having defined the energy norms we will consider in Section \ref{sec:template-energy-norms} and motivated them through the basic physical space estimates of Section \ref{sec:toolbox-physical-space}, we are finally ready to give a precise statement of our main theorem:

\begin{theorem}[Main result] \label{thm:main}
Fix $s\in\{0,\pm 1,\pm 2\}$, $M>0$, $a_0\in[0,M)$. For all $\delta\in(0,1]$ and $p\in[0,2)$, there exists a constant $C=C(a_0,M,|s|,\delta, p)$ such that all sufficiently regular solutions the \textit{homogeneous} system of PDEs from Definition~\ref{def:transformed-system} on the exterior of any subextremal Kerr black hole spacetime with parameters $(a,M)$ where $|a|\leq a_0$ satisfy the following estimates.
\begin{itemize}[noitemsep]
\item Energy flux boundedness:
\begin{align*}
&\overline{\mathbb{E}}_{p}[\swei{\Phi}{s}](\tau)+\overline{\mathbb{E}}_{\mc H^+}[\swei{\Phi}{s}](0,\tau)+\mathbb{E}_{\mc I^+,p}[\swei{\Phi}{s}](0,\tau)\\
&\qquad + \sum_{k=0}^{|s|-1}\lp(\overline{\mathbb{E}}^{|s|-k}_{p}[\swei{\tilde\upphi}{s}_k](\tau)+\mathbbm{1}_{\{s>0\}}\overline{\mathbb{E}}_{\mc H^+}^{|s|-k}[\swei{\tilde\upphi}{s}_k](0,\tau)+\mathbbm{1}_{\{s<0\}}{\mathbb{E}}_{\mc I^+,p}^{|s|-k}[\swei{\tilde\upphi}{s}_k](0,\tau)\rp)\\
&\quad \leq C\overline{\mathbb{E}}_{p}[\swei{\Phi}{s}](0)+ C\sum_{k=0}^{|s|-1}\overline{\mathbb{E}}^{|s|-k}_{p}[\swei{\tilde\upphi}{s}_k](0)\,.
\end{align*}
\item Integrated local energy decay: 
\begin{align*}
&\overline{\mathbb{I}}^{\mathrm{deg},J}_{-\delta,p}[\swei{\Phi}{s}](0,\tau)+\sum_{k=0}^{|s|-1}\lp(\overline{\mathbb{I}}^{\mathrm{deg},|s|-k}_{-\delta,p}[\swei{\tilde\upphi}{s}_k](0,\tau)+\overline{\mathbb{I}}^{|s|-k-1}_{-\delta,p}[\swei{\tilde\upphi}{s}_k](0,\tau)\rp)\\
&\quad\leq C\overline{\mathbb{E}}_{p}[\swei{\Phi}{s}](0)+ C\sum_{k=0}^{|s|-1}\overline{\mathbb{E}}^{|s|-k}_{p}[\swei{\tilde\upphi}{s}_k](0)\,.
\end{align*}

\item Higher-order energies estimates: for instance, as long as $p\neq 0$, we have
\begin{align*}
&\overline{\mathbb{E}}^{J}_{p}[\swei{\Phi}{s}](\tau)+\overline{\mathbb{I}}^{\mathrm{deg}}_{-\delta,p}[\swei{\Phi}{s}_k](0,\tau)\\
&\qquad+\sum_{k=0}^{|s|-1}\lp(\overline{\mathbb{E}}^{J+|s|-k}_{p}[\swei{\tilde\upphi}{s}_k](\tau)+\overline{\mathbb{I}}^{\mathrm{deg},{J+|s|-k}}_{-\delta,p}[\swei{\tilde\upphi}{s}_k](0,\tau)+\overline{\mathbb{I}}^{{|s|-k-1}}_{-\delta,p}[\swei{\tilde\upphi}{s}_k](0,\tau)\rp)\\
&\quad\leq C\overline{\mathbb{E}}^{J}_{p}[\swei{\Phi}{s}](0)+C\sum_{k=0}^{|s|-1}\overline{\mathbb{E}}^{J+|s|-k}_{p}[\swei{\tilde\upphi}{s}_k](0)\,,
\end{align*}
where $C$ depends on $J\in \mathbb{N}_{\geq 1}$ as well.
\end{itemize}
The estimates above also hold for $p=2$ if $s\leq 0$.
\end{theorem}

From the result above, we obtain the boundedness and decay in time of solutions to \eqref{eq:teukolsky-alpha}:

\begin{corollary}[Decay] \label{cor:main} Fix $s\in\{0,\pm 1,\pm 2\}$, $M>0$, $a_0\in[0,M)$ as in Theorem~\ref{thm:main}. For all $\eta\in(0,1)$, there exists a constant $C=C(a_0,M,|s|, \eta)$ such that all sufficiently regular solutions the \textit{homogeneous} system of PDEs from Definition~\ref{def:transformed-system} on the exterior of any subextremal Kerr black hole spacetime with parameters $(a,M)$ where $|a|\leq a_0$ remain uniformly bounded in time  and, in fact, decay in time:
\begin{align*}
s\leq 0 &\implies \sup_{\Sigma_\tau} \lp(|\swei{\Phi}{s}|^2+ |\swei{\tilde\upphi}{s}_k|^2\rp) \leq C \sqrt{\overline{\mathbb{D}}_{2,4}} \tau^{-1/2}\,,\\
s> 0 &\implies\sup_{\Sigma_\tau} \lp(|\swei{\Phi}{s}|^2+ |\swei{\tilde\upphi}{s}_k|^2\rp) \leq C \sqrt{\overline{\mathbb{D}}_{2-\eta,4}} \tau^{-1/2-\eta/2}\,,
\end{align*}
for $\tau>1$, where $\overline{\mathbb{D}}_{2,4}$ is a higher order norm at the level $J=4$ in the notation of Theorem~\ref{thm:main}.
\end{corollary}

\begin{remark}[On the $s>0$ weights] \label{rmk:peeling} For $s>0$, we note that in  Theorem~\ref{thm:main} and Corollary~\ref{thm:main} the norms that we consider on the initial data, and that we propagate in time, have weaker $r$-weights than those consistent with so-called peeling properties for $\swei{\upalpha}{s}$. We direct the reader to \cite{Christodoulou2012,Kehrberger2022} for further comments on peeling, and why it is considered an unphysical assumption. Our choice of norms matches, up to factors of $r^\eta$, the norms in \cite{Dafermos2017}. 

Finally, note that by assuming that stronger weighted norms of the initial data are controlled, one naturally recovers stronger results in the evolution: see, for instance, \cite[Remark 4.4]{Pasqualotto2016} for $s=\pm 1$ and $a=0$, and the more recent \cite{Ma2021a} for sharp decay under stronger conditions.
\end{remark}

By the well-posedness statement in Proposition~\ref{prop:well-posedness}, starting from suitably regular data for the homogeneous transformed system on $\Sigma_0$, we can evolve it to produce a solution in a finite hyperboloidal slab, say, $\tilde{\mc{R}}_{(0,1)}$. It turns out that, without loss of generality, we can assume that such a solution has some useful support properties:

\begin{lemma}[Reduction to compact support in hyperboloidal slab] \label{lemma:reduction-argument} It is enough to prove Theorem~\ref{thm:main} and Corollary~\ref{thm:main} in the case where the solutions, $\swei{\tilde\upphi}{s}_k$, of the homogeneous transformed system of Definition \ref{def:transformed-system} satisfy the following conditions:
\begin{enumerate}[label=(\roman*)]
\item $\swei{\tilde\upphi}{s}_k\big|_{\tilde{\mc{R}}_{(0,1)}}$ is smooth for all $k=0,\dots,|s|$;
\item $\swei{\tilde\upphi}{s}_k\big|_{\tilde{\mc{R}}_{(0,1)}}$ is compactly supported for all $k=0,\dots,|s|$;
%\item $\swei{\tilde\upphi}{s}_k\big|_{\tilde{\mc{R}}_{(0,1)}}$ is supported away from $\mc H^+_{(0,1)}$ for $k=|s|$ and, if $s< 0$, $0\leq k<|s|$,  as long as we assume $|a|=M$;
\item $\swei{\tilde\upphi}{s}_k\big|_{\tilde{\mc{R}}_{(0,1)}}$ is supported away from $\mc H^+_{(0,1)}$ for all $k=0,\dots,|s|$.
\end{enumerate}
\end{lemma}

\begin{proof}[Proof of Lemma~\ref{lemma:reduction-argument}] The reduction proceeds in several steps, combining density arguments with the extension procedure of Proposition~\ref{prop:scattering-construction}. We drop the superscripts below for readability. 
\begin{itemize}
\item We start with some basic density arguments. In the case $s\geq 0$, one can show that smooth compactly supported, in  $r\in[r_+,\infty)$, data $\tilde\upphi_0|_{\Sigma_0}$ are dense with respect to  the $\mathbb{E}_0$ norm which we assume are finite for  Theorem~\ref{thm:main}. Thus, without loss of generality, we can assume $\tilde\upphi_k|_{\Sigma_0}$ are smooth and compactly supported for $s<0$. (A similar result would hold for $s>0$ near $r=r_+$ if $|a|=M$: one could approximate the data by smooth functions  supported away from $r=r_+$.) In the case $s<0$, this is not true for all levels of the transformed system, but one can nevertheless show that it holds for the top level: without loss of generality, we can assume $\Phi|_{\Sigma_0}$ is smooth and compactly supported; for $k=0,\dots, |s|-1$, we can assume $\tilde\upphi_k|_{\Sigma_0}$ are smooth but not necessarily compactly supported.
\item We now invoke  Proposition~\ref{prop:scattering-construction}. After shifting and rescaling the $\tau$ coordinate, we have reduced the proof of Theorem \ref{thm:main} to initial data such that $\tilde\upphi_k$ vanish at $\mc{I}^+_{(0,1)}$ and at $\mc{H}^+_{(0,1)}$, for all $k=0,\dots, |s|$, but are not necessarily compactly supported.
\item Another density argument yields that the initial data of the previous step, $\tilde\upphi_k|_{\Sigma_0}$ for $k=0,\dots, |s|$,  can be assumed, without loss of generality, to be data which is smooth and compactly supported in $r\in(r_+,\infty)$ \textit{regardless of the sign of $s$}.
\item Finally, an extension procedure as in Proposition~\ref{prop:scattering-construction}, now with identically zero data set on $\mc{I}^+_{(-1,0)}$ and $\mc H^+_{(-1,0)}$, then shows that is is enough to consider data such that  $\tilde\upphi_k$ will be compactly supported away from $r=r_+$ along each hyperboloidal slice in a finite-time slab.
\end{itemize}
This yields all of the conditions in the statement. We emphasize that, if $|a|<M$, then the compact support near $r=r_+$ requires assuming that non-degenerate-at-$\mc H^+$ norms are finite, i.e.\ $\overline{\mathbb{E}}^{|s|-k}[\tilde\upphi_k](0)<\infty$. This condition would not be necessary in the case $|a|=M$.
\end{proof}

The proofs of Theorem~\ref{thm:main} and Corollary~\ref{cor:main}, which rely on results from the previous an ensuing sections, will be concluded in Section~\ref{sec:proof-main-thm}.

\section{The Teukolsky radial ODE and the transformed system of ODEs}
\label{sec:odes-big}

In this section, we present the main ordinary differential equations featuring in the present paper. We define our space of admissible frequency parameters in Section~\ref{sec:admissible-freqs}. In Section~\ref{sec:radial-ODEs}, we introduce the (radial) Teukolsky and transformed ODEs and discuss their boundary conditions and general properties. Finally, in Section~\ref{sec:physical-to-frequency} we discuss the connections between this section and  Section \ref{sec:pdes-big}.

\subsection{Admissible frequencies}
\label{sec:admissible-freqs}

For $\omega\in\mathbb{R}$ and $m\in\frac12\mathbb{Z}$, it will be convenient to define:
\begin{equation}
\xi := -i\frac{2M r_+}{r_+-r_-}(\omega-m\upomega_+) \,,\quad \beta:=2iM^2(\omega-m\upomega_+)\,, \quad \upomega_+:=\frac{a}{2Mr_+}\,. \label{eq:xi-upomega+}
\end{equation}

For the remainder of this paper, we will be interested in the following parameters:
\begin{definition}[Admissible frequencies] \label{def:admissible-freqs} Fix $s\in\frac12\mathbb{Z}$. 
\begin{enumerate}
\item We say the frequency parameter $m$ is admissible with respect to $s$ when, if $s$ is an integer, $m$ is also an integer and when, if $s$ is a half-integer, so is $m$.
\item We say the frequency pair $(m,l)$ is admissible with respect to $s$ when $m$ is admissible with respect to $s$, $l$ is an integer or half-integer if $s$ is an integer or half-integer, respectively, and $l\geq \max\{|m|,|s|\}$. 
\item We say the frequency triple $(\omega,m,l)$ is admissible with respect to $s$ when the pair $(m,l)$ is admissible with respect to $s$ and $\omega\in\mathbb{R}$.
\item We say the frequency triple $(\omega,m,\Lambda)$ is admissible with respect to $s$ when $m$ is admissible with respect to $s$ and $\Lambda\in\mathbb{R}$ satisfies
\begin{enumerate}[label = (\alph*), ref=\theenumi{}(\alph*)]
\item if $s=0$, $\Lambda \geq 2am\omega$ and  $\Lambda \geq m^2$; \label{it:admissible-freqs-triple-wave}
\item $\Lambda \geq \max\{|m|,|s|\}(\max\{|m|,|s|\}+1)-s^2-2|s||a\omega|$; \label{it:admissible-freqs-triple-Lambda-lower-bound}
\item $|\Lambda| \leq \Lambda +s^2+a^2\omega^2+4|s||a\omega|$; \label{it:admissible-freqs-triple-Lambda-upper-bound}
\item assuming $q:=a\omega/m\in [-B,-b]\cup[b,B]$ for some $0<b<B<\infty$, \label{it:admissible-freqs-triple-Lambda-q-bound}
\begin{itemize} 
\item if $q\in(1,B]$, $
\Lambda -2am\omega = 4\lp(l-|m|-|s|-\frac12\rp)|a\omega|\sqrt{\frac{q-1}{q}}+O(|a\omega|^{1/2})$ as $|a\omega|\to\infty$;
\item if $q=1$, $\Lambda -2am\omega \geq 3|a\omega|^{2/3}|s|^{4/3}+o(|a\omega|^{2/3})$ as $|a\omega|\to\infty$;
\item if $q\in[-B,-b]\cup[b,1)$, $\Lambda -2am\omega \geq m^2(q-1)^2+o(|a\omega|^{2})$ as $|a\omega|\to\infty$;
\end{itemize}
\item $\mathfrak{C}_s\geq 0$, for $|s|\leq 2$, for $\mathfrak{C}_s$ as given in \eqref{eq:TS-radial-constants}. \label{it:admissible-freqs-triple-Cs-positivity}
\end{enumerate}
\item We say the frequency triples $(\omega,m,l)$ and $(\omega,m,\Lambda)$ are admissible with respect to $s$ and $a$ when they are admissible with respect to $s$ and, moreover, $\omega$ satisfies $\omega\in\mathbb{R}\backslash\{0\}$ and, if $|a|=M$ $\omega\neq m\upomega_+$. 
\end{enumerate}
\end{definition}

\subsection{The radial ODEs and their boundary conditions}
\label{sec:radial-ODEs}

As in Section~\ref{sec:PDEs}, we will introduce a system of radial ODEs:
\begin{definition}[Transformed system of ODEs] \label{def:transformed-system-ODE}
Fix $s\in\mathbb{Z}$, and $(\omega,m,\Lambda)\in\mc F_{\rm admiss}$. Write $\mathcal{L}=L$ if $s<0$, $\mathcal{L}=\underline{L}$ if $s>0$; here and throughout the section we are considering, by abuse of notation, $L$ and $\uL$ to be given by their frequency space analogues,
\begin{align*}
L=\frac{d}{dr^*}-i\omega+\frac{iam}{r^2+a^2}\,,\quad \uL=-\frac{d}{dr^*}-i\omega+\frac{iam}{r^2+a^2}\,.
\end{align*}
We say that, for $k\in\{0,\dots |s|\}$, the functions $\smlambdak{\uppsi}{s}{k}$  are solutions to the transformed system with inhomogeneity $\smlambdak{\mathfrak H}{s}{k}$ if they satisfy the following ODEs.
\begin{itemize}
\item \textit{Transport ODEs}: if $s\neq 0$, 
\begin{align}
\smlambdak{\uppsi}{s}{k}&=\frac{1}{w}\mc{L}\smlambdak{\uppsi}{s}{k-1}\,,\qquad  k=1,...,|s|\,, \label{eq:transformed-transport-separated}\\
\smlambdak{\mathfrak{G}}{s}{k}&=\mc{L}\lp(\frac{\smlambdak{\mathfrak G}{s}{k-1}}{w}\rp)\,,\qquad  k=1,...,|s|\,, \label{eq:transformed-transport-inhom-separated}
\end{align}
for $k=1,\dots |s|$.

\item \textit{Radial ODEs}:  for $k=0,\dots,|s|$,
\begin{equation}
\begin{gathered}
\lp(\smlambdak{\uppsi}{s}{k}\rp)''-(|s|-k)\lp(\frac{w'}{w}\rp)\lp(\smlambdak{\uppsi}{s}{k}\rp)' + \lp(\omega^2-\smlambdak{\mc V}{s}{k}\rp)\smlambdak{\uppsi}{s}{k}\\
= \smlambdak{\mathfrak G}{s}{k}+aw\sum_{j=0}^{k-1}\lp(im c_{s,\,k,\,j}^{Z}+ c_{s,\,k,\,j}^{\mr{id}}\rp)\smlambdak{\uppsi}{s}{j}\,,
\end{gathered} \label{eq:transformed-k-separated}
\end{equation}
where   $c_{s,\,k,\,j}^{Z}$ and $c_{s,\,k,\,j}^{\mr{id}}$ are the functions introduced in Definition~\ref{def:transformed-system} (recall $ac_{s,\,1,\,0}^{\mr{id}}$ in \eqref{eq:transformed-k-separated} should be replaced by $c_{s,\,1,\,0}^{\mr{id}}$ when $|s|\neq 1$) and
the potential $\smlambdak{\mc V}{s}{k}$ is characterized by
\begin{align*}
\Re\smlambdak{\mc V}{s}{k}&=\frac{\Delta\Lambda+4Mram\omega-a^2m^2}{(r^2+a^2)^2}+w\lp[|s|+k(2|s|-k-1)\rp]\\
&\qquad+\frac{a^2\Delta w}{(r^2+a^2)^2} [1-2|s|-2k(2|s|-k-1)]\\
&\qquad +\frac{2Mr(r^2-a^2)w}{(r^2+a^2)^2}[1-3|s|+2s^2-3k(2|s|-k-1)]\,,\numberthis \label{eq:radial-ODE-potential-k-tilde}\\
\Im\smlambdak{\mc V}{s}{k}&=\sign s (|s|-k)\lp[\frac{w'}{w}\lp(\omega-\frac{am}{r^2+a^2}\rp)-\frac{4amrw}{(r^2+a^2)}\rp]\,. 
\end{align*}
In light of \eqref{eq:transformed-transport-separated} and \eqref{eq:transformed-k-separated}, we can recast the latter for $k<|s|$ as a transport ODE:
\begin{align*}
&\lp[\sign s\frac{d}{dr^*}-i\lp(\omega-\frac{am}{r^2+a^2}\rp)\rp]\smlambdak{\uppsi}{s}{k+1}-\sign s(|s|-k-1) \frac{w'}{w}\smlambdak{\uppsi}{s}{k+1}\\
&\quad =(2am\omega-\Lambda)\smlambdak{\uppsi}{s}{k}+\frac{2ar}{r^2+a^2}\sign s(2|s|-2k-1)im\smlambdak{\uppsi}{s}{k}-\lp(|s|+k(2|s|-k-1)\rp)\smlambdak{\uppsi}{s}{k}\\
&\quad\qquad -\lp(a^2w(1-2|s|-2k(2|s|-k-1))+\frac{2M r(r^2-a^2)}{(r^2~+a^2)^2}(1-3|s|-3k(2|s|-k-1))\rp)\smlambdak{\uppsi}{s}{k}\\
&\quad\qquad +\sum_{j=0}^{k}(ac_{s,k,j}^Zim+ac_{s,k,j}^{\rm id})\smlambdak{\uppsi}{s}{j}+\frac{1}{w}\smlambdak{\mathfrak G}{s}{k}\,.\numberthis\label{eq:transformed-constraint-separated}
\end{align*}

In particular, $\smlambda{\Psi}{s}:=\smlambdak{\uppsi}{s}{|s|}$ solves the transformed radial ODE
\begin{align}
\lp(\smlambda{\Psi}{s}\rp)''+\lp(\omega^2-\smlambda{\mc{V}}{s}\rp)\smlambda{\Psi}{s}=\smlambda{\mathfrak G}{s}+aw\sum_{k=0}^{|s|-1}\lp(im c_{s,\,|s|,\,k}^{Z}+ c_{s,\,|s|,\,k}^{\mr{id}}\rp)\smlambdak{\uppsi}{s}{k}\,, \label{eq:radial-ODE-Psi}
\end{align}
where $\smlambda{\mc{V}}{s}=\mc{V}^{[s]}_{m\Lambda,\,0}+{\mc{V}}^{[s]}_{m\Lambda,\,1}$ is real and given by
\begin{equation}
\begin{split}
\mc{V}^{[s]}_{m\Lambda,\,0}&:=\frac{4Mram\omega-a^2m^2+\Delta\Lambda}{(r^2+a^2)^2} \\
\mc{V}^{[s]}_{m\Lambda,\,1}&:=s^2\frac{\Delta}{(r^2+a^2)^2}+\frac{\Delta}{(r^2+a^2)^4}\lp[(1-2s^2)a^2\Delta + 2(1-s^2)Mr(r^2-a^2)\rp] \,.
\end{split} \label{eq:radial-ODE-Psi-potentials}
\end{equation}
\end{itemize}
\end{definition}

The transformed system in Definition~\ref{def:transformed-system-ODE} can be thought of as a generalization of the Teukolsky ODE when $s\in\mathbb{Z}$. Let $s\in\frac12\mathbb{Z}$, $M>0$, $|a|\leq M$ and $(\omega,m,\Lambda)$ be an admissible frequency triple. We say $\smlambda{\alpha}{s}$ is a solution to the Teukolsky radial ODE if
\begin{align*}
&\lp[\Delta^{-s}\frac{d}{dr}\lp(\Delta^{s+1}\frac{d}{dr}\rp)  +\frac{[\omega(r^2+a^2)-am]^2-2is(r-M)[\omega(r^2+a^2)-am]}{\Delta}\rp] \smlambda{\upalpha}{s}(r)\\
&\qquad\quad+\lp(4is\omega r -\Lambda-s+2am\omega\rp)\smlambda{\upalpha}{s}(r)=\frac{\Delta^{|s|/2(1-\sign s)}}{(r^2+a^2)^{|s|}}\smlambda{F}{s}(r)\,,\numberthis \label{eq:radial-ODE-alpha}
\end{align*}
where $\smlambda{F}{s}(r)$ is given. Defining 
\begin{equation}
\smlambda{u}{s}:=\Delta^{s/2}(r^2+a^2)^{1/2}\smlambda{\upalpha}{s}\,, \label{eq:def-u}
\end{equation}
 the Teukolsky radial ODE \eqref{eq:radial-ODE-alpha} can also be rewritten as
\begin{gather} \label{eq:radial-ODE-u}
\lp(\smlambda{u}{s}\rp)''  + \lp(\omega^2- V^{[s]}\rp) \smlambda{u}{s} =\smlambda{H}{s}\,,\qquad  \sml{H}{s}:=\frac{\Delta^{|s|/2+1}}{(r^2+a^2)^{|s|+3/2}} \smlambda{F}{s}\,,
\end{gather}
letting $'$ denote $r^*$ derivatives as usual. Here, $\swei{V}{s}=\swei{V}{s}_0+V_1+i\swei{V}{s}_2$, where
\begin{equation}
\begin{split}
\swei{V}{s}_0&:=\frac{4Mra m\omega -a^2m^2+\Delta\Lambda_{ml}^{[s]}}{(r^2+a^2)^2}\,, \\
\swei{V}{s}_1&:=s^2\frac{(r-M)^2}{(r^2+a^2)^2}+\frac{\Delta}{(r^2+a^2)^4}\lp(a^2\Delta+2Mr(r^2-a^2)\rp) \geq 0 \,,  \\
\swei{V}{s}_2&:=-\frac{2s\omega \lp(r(r^2+a^2)+M(a^2-3r^2)\rp)+2sam(r-M)}{(r^2+a^2)^2}\,.
\end{split} \label{eq:radial-ODE-u-potentials} 
\end{equation}
With $\smlambda{\upalpha}{s}$ or $\smlambda{u}{s}$ as the starting point, we have

\begin{lemma}[Teukolsky ODE and the transformed system of ODEs]\label{lemma:transformed-system-ODE} Fix $s\in\mathbb{Z}$ and $(\omega,m,\Lambda)\in\mc F_{\rm admiss}$. Let $\smlambda{\upalpha}{s}$ be a solution to the Teukolsky ODE (\ref{eq:radial-ODE-alpha}). Let $\swei{\uppsi}{s}_{0}$ be the rescaling of the Teukolsky radial variable given by
\begin{equation}\label{eq:def-psi0-separated}
\smlambdak{\uppsi}{s}{0}:=(r^2+a^2)^{-|s|+1/2}\Delta^{|s|(1+\mr{sign}\,s)/2}\smlambda{\alpha}{s}\,,
\end{equation}
and $\smlambdak{\mathfrak{H}}{s}{0}$ be the rescaling of the Teukolsky radial inhomogeneity given by
\begin{equation}
\smlambdak{\mathfrak{G}}{s}{0}:=w^{|s|+1}(r^2+a^2)^{1/2}\smlambda{F}{s}\,,\label{eq:def-G0-separated}
\end{equation}
Defining $\smlambda{\uppsi}{s}{k}$ and $\smlambda{\mathfrak G}{s}{k}$ by the system of transport equations \eqref{eq:transformed-transport-separated} and \eqref{eq:transformed-transport-inhom-separated}, we see that for $k=0,\dots, |s|$, $\smlambda{\uppsi}{s}{k}$ solve the transformed ODE system of Definition~\ref{def:transformed-system-ODE} with inhomogeneity $\smlambda{\mathfrak{G}}{s}{k}$. 
\end{lemma}

\begin{proof} Identical to the proof of Lemma~\ref{lemma:transformed-system}.
\end{proof}

As in Section~\ref{sec:PDEs}, for technical reasons it will be useful to consider a version of the system in Definition~\ref{def:transformed-system-ODE} where we assume that $\smlambda{\upphi}{s}{k_0+1}$, for some $k_0=0, \dots, |s|-1$, is known already. 

\begin{definition}[Alternative transformed system of ODEs]  \label{def:alt-transformed-system-ODE} Fix $s\in\mathbb{Z}$, $k_0\in\{0,\dots |s|-1\}$ and $(\omega,m,\Lambda)\in\mc F_{\rm admiss}$. We say that, for $k\in\{0,\dots k_0\}$, the functions $\smlambda{\uppsi}{s}{k}$  are solutions to the alternative transformed system of ODEs with inhomogeneities $\smlambdak{\mathfrak G}{s}{k}$ and $\smlambda{\mathfrak g}{s}{k_0}$ and given $\smlambda{\uppsi}{s}{k_0+1}$ if they satisfy the following PDEs.
\begin{itemize}
\item \textit{Transport ODEs}: if $s\neq 0$, \eqref{eq:transformed-transport-separated} and \eqref{eq:transformed-transport-inhom-separated} hold for $k=1,\dots k_0$; moreover, \eqref{eq:transformed-constraint} holds for each $k=0,\dots, k_0-1$ and it holds for $k=k_0$ if we replace $\smlambda{\mathfrak{G}}{s}{k_0}$ by $\smlambda{\mathfrak{G}}{s}{k_0}+\smlambda{\mathfrak{G}}{s}{k_0}$.
\item \textit{Radial ODEs}: \eqref{eq:transformed-k} holds for each $k=0,\dots, k_0$.
\end{itemize}
\end{definition}

To discuss boundary conditions for the radial ODEs introduced, it is worth introducing some model solutions.
\begin{definition} Fix $M>0$, $|a|\leq M$, $s\in\mathbb{Z}$ and an admissible frequency triple $(\omega,m,\Lambda)$.
\begin{enumerate}
\item Define $\uppsi^{[s],\, a,\,\omega}_{(0),\, m\Lambda,\,\mc{H}^+}$ as the unique solution to the {\normalfont homogeneous} radial ODE~(\ref{eq:transformed-k-separated}) with $k=0$ and with the following boundary condition
\begin{enumerate}
\item if $|a|<M$,
  \begin{enumerate}
\item $\uppsi^{[s],\, a,\,\omega}_{(0),\,m\Lambda,\,\mc{H}^+ }(r)(r-r_+)^{-\xi - \frac12|s|(1-\sign s)}$ are smooth at $r=r_+\,,$ and
\item $\lp|\lp(w^{-\frac{|s|}{2}}\Delta^{\frac{s}{2}}(r-r_+)^{-\xi}\uppsi^{[s],\, a,\,\omega}_{(0),\, ml,\,\mc{H}^+}\rp)\Big|_{r=r_+}\rp|^2=1\,;$
  \end{enumerate}
\item if $|a|=M$,
  \begin{enumerate}
\item $(r-M)^{2iM\omega-|s|(1- \sign s)}\uppsi^{[s],\, a,\,\omega}_{(0),\, ml,\,\mc{H}^+}(r)e^{- \beta(r-M)^{-1}}$ are  smooth at $r=M\,,$ and 
\item $\lp|\lp(w^{-\frac{|s|}{2}}(r-M)^{2iM\omega}\Delta^{\frac{s}{2}}e^{-\beta(r-M)^{-1}}\uppsi^{[s],\, a,\,\omega}_{(0),\, ml,\,\mc{H}^+}\rp)\Big|_{r=M}\rp|^2=1\,.$
  \end{enumerate}
\end{enumerate}
\item Define $\uppsi^{[s],\, a,\,\omega}_{(0),\,ml,\,\mc{I}^+}$  as the unique classical solution to the {\normalfont homogeneous} radial ODE~(\ref{eq:transformed-k-separated}) with $k=0$ and with the following boundary condition
\begin{enumerate}
\item $\uppsi^{[s],\, a,\, \omega}_{(0),\, ml,\,\mc{I}^+}\sim e^{ i\omega r}r^{ 2Mi\omega-|s|(1+ \sign s)}$ asymptotically  as $r\to \infty\,,$ and
\item $\lp|\lp(e^{- i\omega r}r^{-2iM\omega}w^{-\frac{|s|}{2}}\Delta^{\frac{s}{2}}\uppsi^{[s],\, a,\,\omega}_{(0),\, ml,\,\mc{I}^\pm}\rp)\big|_{r=\infty}\rp|^2=1\,.$ 
\end{enumerate}
\end{enumerate}
\label{def:psihor-psiout}

Finally, we define $u^{[s],\, (a\omega)}_{ml,\,\mc{I}^\pm}$ and $u^{[s],\, (a\omega)}_{ml,\,\mc{H}^\pm}$ by the following rescaling of the previously introduced functions:
\begin{align*}
u^{[s],\, (a\omega)}_{m\Lambda,\,\mc{I}^\pm}=w^{-|s|/2}\uppsi^{[s],\, (a\omega)}_{(0),\, m\Lambda,\,\mc{I}^\pm}\,, \qquad u^{[s],\, (a\omega)}_{m\Lambda,\,\mc{H}^\pm}=w^{-|s|/2}\uppsi^{[s],\, (a\omega)}_{(0),\,m\Lambda,\,\mc{H}^\pm}\,.
\end{align*}
\end{definition}

In the present paper, an important choice of solutions for the radial ODEs \eqref{eq:radial-ODE-alpha} and \eqref{eq:transformed-k-separated} are those which are so-called ingoing at $\mc H^+$ and outgoing at $\mc I^+$. We call these solutions outgoing, for short:

\begin{definition}[Outgoing ODE solutions] \label{def:outgoing-bdry-freq-space} Fix $s\in\mathbb{Z}$, $M>0$, $|a|\leq M$ and an admissible frequency triple $(\omega,m,\Lambda)$. Suppose that, for every $k\in\{0,\dots,|s|\}$, $\smlambdak{\mathfrak G}{s}{k}=O(w)$ as $|r^*|\to \infty$.  Then, we say that $\smlambdak{\uppsi}{s}{k}$ are outgoing solutions to the radial ODEs \eqref{eq:transformed-k-separated} with inhomogeneities $\smlambdak{\mathfrak G}{s}{k}$ if we can find coefficients $\swei{A}{s}_{k,\mc{I}^+},\swei{A}{s}_{k,\mc{H}^+}\in\mathbb{C}$, depending only on $(\omega,m,\Lambda)$, such that 
\begin{align} 
\smlambdak{\uppsi}{s}{k}&= \swei{A}{s}_{k,\mc{I}^+}\uppsi^{[s],\, (a\omega)}_{(k),\,m\Lambda,\,\mc{I}^+} +O(r^{-1})\,, \label{eq:def-outgoing-bdry-freq-space-infty}\\
\smlambdak{\uppsi}{s}{k}&= \swei{A}{s}_{k,\mc{H}^+}\uppsi^{[s],\, (a\omega)}_{(k),\,m\Lambda,\,\mc{H}^+}+O(r-r_+)
\,. \label{eq:def-outgoing-bdry-freq-space-hor}
\end{align}
\end{definition}
Let us note that the second term in \eqref{eq:def-outgoing-bdry-freq-space-infty} and \eqref{eq:def-outgoing-bdry-freq-space-hor} is not necessarily lower order if $k<|s|$, as can be seen readily from Definition~\ref{def:psihor-psiout}. If the decay assumptions on the inhomogeneities are strong enough, we can obtain more complete information regarding the relation between boundary terms of the elements in the transformed system of ODEs:

\begin{lemma} \label{lemma:bdry-term-relations} Fix $M>0$, $|a|\leq M$, $s\in\mathbb{Z}$ and an admissible frequency triple $(\omega,m,\Lambda)$.  We have the following asymptotics for solutions, $\smlambdak{\uppsi}{s}{k}$ to the transformed system of radial ODEs in Definition~\ref{def:transformed-system-ODE} with inhomogeneities $\smlambdak{\mathfrak G}{s}{k}$, for $k=0,\dots,|s|$.
\begin{enumerate}[label=(\roman*)]
\item \label{it:bdry-term-relations-infty} Suppose that, for every $k\in\{0,\dots,|s|\}$, $\smlambdak{\mathfrak G}{s}{k}=O(w)$ as $r\to \infty$ and $\smlambdak{\uppsi}{s}{k}$ have outgoing boundary conditions. Then we have the boundary term relations
\begin{align}\label{eq:a-to-A-infinity-plus}
\lp|\swei{A}{+|s|}_{k,\,\mc{I}^+}\rp|^2&=(2\omega)^{2k}\lp|\swei{A}{+ |s|}_{0,\,\mc{I}^+}\rp|^2\,,
\end{align}
and, if furthermore $\smlambdak{\mathfrak G}{-|s|}{k}=O(w^{1+|s|-k})$ as $r\to \infty$, 
\begin{align}\label{eq:a-to-A-infinity-minus}
\lp|\swei{A}{- |s|}_{k,\,\mc{I}^+}\rp|^2=\frac{\mathfrak{D}_{s,k}^{\mc{I}}}{(2\omega)^{2k}}\lp|\swei{A}{- |s|}_{0,\,\mc{I}^+}\rp|^2\,.
\end{align}

\item \label{it:bdry-term-relations-hor} Suppose that, for every $k\in\{0,\dots,|s|\}$, $\smlambdak{\mathfrak G}{s}{k}=O(\Delta)$ as $r\to r_+$ and $\smlambdak{\uppsi}{s}{k}$ have outgoing boundary conditions. Then we have the following boundary term relations
\begin{equation}\label{eq:a-to-A-horizon-extremal-minus}
\lp|\swei{A}{-|s|}_{k,\,\mc{H}^+}\rp|^2=\lp\{\begin{array}{lr}
[4M^2(\omega-m\upomega_+)]^{2k},&|a|=M\\
\displaystyle \prod_{j=0}^{k-1}\lp[\lp(\frac{4Mr_+}{r_+-r_-}\rp)^2(\omega-m\upomega_+)^2+(s-j)^2\rp], &|a|<M
\end{array}\rp\}
\lp|\swei{A}{-|s|}_{0,\,\mc{H}^+}\rp|^2\,, 
\end{equation}
and, if furthermore $\smlambdak{\mathfrak G}{+|s|}{k}=O(w^{1+|s|-k})$ as $r\to r_+$,
\begin{align}\label{eq:a-to-A-horizon-sub-plus}
\lp|\swei{A}{+|s|}_{k,\,\mc{H}^+}\rp|^2&=\lp\{
\begin{array}{lr}
\displaystyle \frac{\mathfrak{D}_{s,k}^{\mc{H}}(2M^2)^{2k}}{[4M^2(\omega-m\upomega_+)]^{2k}},&|a|=M\\
\displaystyle \frac{\mathfrak{D}_{s,k}^{\mc{H}}(2Mr_+)^{2k}}{\prod_{j=1}^{k}\lp\{\lp[4Mr_+(\omega-m\upomega_+)\rp]^2+(s-j)^2(r_+-r_-)^2\rp\}},&|a|<M
\end{array}\rp\}
\lp|\swei{A}{+|s|}_{0,\,\mc{H}^+}\rp|^2\,.
\end{align}
\end{enumerate}

In the formulas above, while $\mathfrak{D}_{s,0}^\mc{I}=\mathfrak{D}_{s,0}^\mc{H}=1$, the remaining coefficients are generally non-trivial: for $|s|\in\{0,1,2\}$, these are
\begin{align} 
\begin{split}
\mathfrak{D}_{s,1}^{\mc{I}}&=(\Lambda-2am\omega+|s|)^2\,,\\
\mathfrak{D}_{s,2}^{\mc{I}}&=\lp[(\Lambda-2am\omega+2)(\Lambda-2am\omega+3|s|-2)+4(2|s|-1)am\omega\rp]^2+16(|s|-1)^2(2|s|-1)^2M^2\omega^2\,;
\end{split}\label{eq:Ds-infinity}\\
\begin{split}
\mathfrak{D}_{s,1}^{\mc{H}}&=\lp(\Lambda-2am\omega+|s|+(|s|-1)(2|s|-1)\frac{r_+-M}{M}\rp)^2+\frac{a^2m^2}{M^2}(2|s|-1)^2\,,\\
\mathfrak{D}_{s,2}^{\mc{H}}&=\lp[(\Lambda-2am\omega+|s|)(\Lambda-2am\omega+3|s|-2)+\frac{4(r_+-M)}{M}(\Lambda-2am\omega+|s|)(|s|-2)(|s|-1)\rp.\\
&\qquad\qquad \lp. +4am\omega\frac{r_+-M}{M}(2|s|-1)-m^2(2|s|-1)\lp(\frac{a^2}{M^2}(2|s|-3)+\frac{2r_-(r_+-M)}{M^2}\rp)\rp.\\
&\qquad\qquad \lp. +\frac{(r_+-M)^2}{M^2}(|s|-2)(|s|-1)(2|s|-1)\lp(2|s|-1-\frac{2M}{r_+}\rp)\rp]^2\\
&\qquad+\lp[\frac{2\omega}{M}(2|s|-1)\lp(2a^2(|s|-1)-M(r_+-M)-(r_+-M)^2(2|s|-3)\rp)\rp.\\
&\qquad\qquad\lp.+\frac{am(r_+-M)}{Mr_+}(2|s|-3)(2|s|-1)\lp(2|s|-3+\frac{r_+-M}{M}(2|s|-1)\rp)\rp.\\
&\qquad\qquad\lp.+\frac{4am}{M}(|s|-1)(\Lambda-2am\omega+|s|)\rp]^2\,.
\end{split}\label{eq:Ds-horizon}
\end{align}
For convenience, when $k=|s|$, we write $\mathfrak{D}_{s}^{\mc H}=\mathfrak{D}_{s,k}^{\mc H}$, $\mathfrak{D}_{s}^{\mc I}=\mathfrak{D}_{s,k}^{\mc I}$, $\swei{A}{s}_{\mc{I}^\pm}=\swei{A}{s}_{|s|,\,\mc{I}^\pm}$ and $\swei{A}{s}_{\mc{H}^\pm}=\swei{A}{s}_{|s|,\,\mc{H}^\pm}$.
\end{lemma}

Finally, one of the interesting features of the Teukolsky radial ODEs~\eqref{eq:radial-ODE-alpha} is that they verifies so-called Teukolsky--Starobinsky identities, named after the seminal papers \cite{Teukolsky1974,Starobinsky1974} for $|s|=1,2$, see also the generalization \cite{Kalnins1989}. 

\begin{proposition}[Radial Teukolsky--Starobinsky identities] \label{prop:TS-radial-constant-identities} Fix $s\in\{0,\frac12,1,\frac32,2\}$, $|a|\leq M$, and an admissible frequency triple $(\omega,m,\Lambda)$. Define
\begin{align}
\hat{\mc{D}}^{\pm}_n &= \frac{d}{dr}\pm i\lp(\frac{\omega(r^2+a^2)}{\Delta} -\frac{am}{\Delta}\rp)+\frac{2n(r-M)}{\Delta}\,. \label{eq:def-D-pm}
\end{align}

Dropping most subscripts, for $\swei{\upalpha}{\pm s}$ solutions of the homogenenous radial ODE~\eqref{eq:radial-ODE-alpha} of spin $\pm s$, set
\begin{align*}
P^{[+s]}:=\Delta^s \swei{\upalpha}{+s}\,,\qquad P^{[-s]}:=\swei{\upalpha}{-s}\,;
\end{align*}
then $P^{[\pm s]}$ is an eigenfunction of the operator
\begin{align*}
\Delta^s\lp(\hat{\mc{D}}^{\mp}_0\rp)^{2s}\lp[\Delta^s\lp(\hat{\mc{D}}^{\pm}_0\rp)^{2s}\rp] \equiv \prod_{j=0}^{2s-1}\lp(\Delta^{1/2}\hat{\mc{D}}^{\mp}_{j/2}\rp)\prod_{k=0}^{2s-1}\lp(\Delta^{1/2}\hat{\mc{D}}^{\pm}_{k/2}\rp)\,, 
\end{align*}
with indices $j,k$ increasing from right to left in the product, and the latter being replaced by the identity if $s=0$, for the same eigenvalue. This eigenvalue, $\mathfrak C_s=\mathfrak C_s(\omega,m,\Lambda)\in\mathbb{R}$ is called the {\normalfont radial Teukolsky--Starobinsky constant}; if $\mathfrak C_s=0$, we say $(\omega,m,\Lambda)$ is an {\normalfont algebraically special frequency triple}. The Teukolsky--Starobinksy constants $\mathfrak C_s$ can be computed explicity; for instance
\begin{align}\label{eq:TS-radial-constants}
\begin{split} 
%\mathfrak C_{1/2}&= \Lambda-2am\omega+1/2\,,\\
\mathfrak C_{1}&= (\Lambda-2am\omega+1)^2+4am\omega-4a^2\omega^2\,,\\
%\mathfrak C_{3/2}&=  \lp(\Lambda-2am\omega+\frac32\rp)^3+\lp(\Lambda-2am\omega+\frac32\rp)^2-16\lp(\Lambda-2am\omega+\frac32\rp)(a^2\omega^2-am\omega)+16a^2\omega^2\,,\\
\mathfrak C_2&= \lp[(\Lambda-2am\omega+2)(\Lambda-2am\omega+4)\rp]^2+40a\omega(\Lambda-2am\omega+2)^2(m-a\omega)\\
&\qquad+48a\omega(\Lambda-2am\omega+2)(m+a\omega)+144a^2\omega^2(m-a\omega)^2+144M^2\omega^2\,,
\end{split}
\end{align}
The Teukolsky--Starobinksy constants figure in the so-called Teukolsky--Starobinsky identities:
\begin{itemize}
\item Suppose $\smlambdak{\uppsi}{+s}{0}$ is a solution to the radial ODE~\eqref{eq:transformed-k-separated} with $k=0$,  spin $+s$, and  inhomogeneity $\smlambdak{\mathfrak G}{+s}{0}$. Define 
\begin{align*}
\smlambdak{\uppsi}{-s}{0}&:=(r^2+a^2)^{1/2-|s|}\Delta^s(\mc D_0^+)^{2|s|}\lp((r^2+a^2)^{|s|-1/2}\smlambdak{\uppsi}{+s}{0}\rp)\,,\\
\smlambdak{\mathfrak G}{-s}{0}&:=w(r^2+a^2)^{1/2-|s|}\Delta^s (\mc D_0^+)^{2|s|}\lp((r^2+a^2)^{|s|-1/2}\frac{\smlambdak{\mathfrak G}{+s}{0}}{w}\rp)\,.
\end{align*}
It follows that $\smlambdak{\uppsi}{-s}{0}$ is a solution to the radial ODE~\eqref{eq:transformed-k-separated}  with $k=0$,  spin $-s$, and  inhomogeneity $\smlambdak{\mathfrak G}{-s}{0}$. Now assume additionally that $\smlambdak{\mathfrak G}{+s}{0}=O(\Delta)$ as $r^*\to -\infty$, $\smlambdak{\mathfrak G}{+s}{0}=O(r^{-2})$ as $r^*\to \infty$ and  $\smlambdak{\uppsi}{+s}{0}$  is outgoing in the sense of Definition~\ref{def:outgoing-bdry-freq-space}. Then,
\begin{align*}
\swei{A}{-s}_{0,\,\mc{I}^{+}}=\mathfrak{C}_s^{(1)}\swei{A}{+s}_{0,\,\mc{I}^{+}}\,,
\qquad
\swei{A}{-s}_{0,\,\mc{H}^{-}}=\mathfrak{C}_s^{(6)}\swei{A}{+s}_{0,\,\mc{H}^{-}}\,.
\end{align*}
If furthermore $\smlambdak{\mathfrak G}{+s}{0}=O(w^{1+2(|s|-k)})$ as $r\to r_+$ or as $r\to \infty$ then, respectively,
\begin{align*}
\swei{A}{-s}_{0,\,\mc{H}^{+}}&=\mathfrak{C}_s^{(4)}\swei{A}{+s}_{0,\,\mc{H}^{+}}\,,\qquad 
\swei{A}{-s}_{0,\,\mc{I}^{-}}=\mathfrak{C}_s^{(7)}\swei{A}{+s}_{0,\,\mc{I}^{-}}\,.
\end{align*}

\item Suppose $\smlambdak{\uppsi}{-s}{0}$ is a solution to the radial ODE~\eqref{eq:transformed-k-separated} with $k=0$,  spin $-s$, and  inhomogeneity $\smlambdak{\mathfrak G}{-s}{0}$. Define 
\begin{align*}
\smlambdak{\uppsi}{+s}{0}&:=(r^2+a^2)^{1/2-|s|}\Delta^s(\mc D_0^-)^{2|s|}\lp((r^2+a^2)^{|s|-1/2}\smlambdak{\uppsi}{-s}{0}\rp)\,,\\
\smlambdak{\mathfrak G}{+s}{0}&:=w(r^2+a^2)^{1/2-|s|}\Delta^s(\mc D_0^-)^{2|s|}\lp((r^2+a^2)^{|s|-1/2}\frac{\smlambdak{\mathfrak G}{-s}{0}}{w}\rp)\,.
\end{align*}
It follows that $\smlambdak{\uppsi}{+s}{0}$ is a solution to the radial ODE~\eqref{eq:transformed-k-separated}  with $k=0$,  spin $+s$, and  inhomogeneity $\smlambdak{\mathfrak G}{-s}{0}$. Now assume additionally that $\smlambdak{\mathfrak G}{-s}{0}=O(\Delta)$ as $r^*\to -\infty$, $\smlambdak{\mathfrak G}{-s}{0}=O(r^{-2})$ as $r^*\to \infty$, and  $\smlambdak{\uppsi}{-s}{0}$  is outgoing in the sense of Definition~\ref{def:outgoing-bdry-freq-space}. Then,
\begin{align*}
\swei{A}{+s}_{0,\,\mc{H}^{+}}&=\mathfrak{C}_s^{(2)}\swei{A}{-s}_{0,\,\mc{H}^{+}}\,,\qquad 
\swei{A}{+s}_{0,\,\mc{I}^{-}}=\mathfrak{C}_s^{(5)}\swei{A}{-s}_{0,\,\mc{I}^{-}}\,.
\end{align*}
If furthermore $\smlambdak{\mathfrak G}{+s}{0}=O(w^{1+2(|s|-k)})$ as $r\to \infty$ or as $r\to r_+$ then, respectively,
\begin{align*}
\swei{A}{+s}_{0,\,\mc{I}^{+}}=\mathfrak{C}_s^{(3)}\swei{A}{-s}_{0,\,\mc{I}^{+}}\,,
\qquad
\swei{A}{+s}_{0,\,\mc{H}^{-}}=\mathfrak{C}_s^{(8)}\swei{A}{-s}_{0,\,\mc{H}^{-}}\,.
\end{align*}
\end{itemize}
In the above,  $\mathfrak{C}_s$ is the Teukolsky--Starobinsky constant, and $\mathfrak{C}_{s}^{(i)}$ are explicit constants are given by:
\begin{gather*}
\overline{\mathfrak{C}_s^{(6)}}=\mathfrak{C}_s^{(2)}=(2Mr_+)^{-s|}\begin{dcases}
(-4M^2i(\omega-m\upomega_+))^{2|s|} &\text{~if~} |a|=M\\
\frac{(r_+-r_-)^{2|s|}\mathfrak{C}_s^{(9)}\mathfrak{C}_s^{(10)}}{4Mr_+i(\omega-m\upomega_+)[4Mr_+i(\omega-m\upomega_+)+|s|(r_+-r_-)]}&\text{~if~} |a|<M
\end{dcases}\,, \\
\mathfrak{C}_s^{(1)}=(2i\omega)^{2|s|}=\overline{\mathfrak{C}_s^{(5)}}\,, \quad
\mathfrak{C}_s^{(3)}=\frac{\mathfrak{C}_s}{\mathfrak{C}_s^{(1)}}\,, \quad
\mathfrak{C}_s^{(4)}=\frac{\mathfrak{C}_s}{\mathfrak{C}_s^{(2)}}\,, \quad
\mathfrak{C}_s^{(7)}= \frac{\mathfrak{C}_s}{\mathfrak{C}_s^{(5)}}\,,\quad
 \mathfrak{C}_s^{(8)}=\frac{\mathfrak{C}_s}{\mathfrak{C}_s^{(6)}}\,,\\
\mathfrak{C}_s^{(9)}=\prod_{j=1}^{|s|}\lp\{\lp[4Mr_+(\omega-m\upomega_+)\rp]^2+(s-j)^2(r_+-r_-)^2\rp\}\,, \\ \mathfrak{C}_s^{(10)}=\prod_{j=0}^{|s|-1}\lp\{\lp[4Mr_+(\omega-m\upomega_+)\rp]^2+(s-j)^2(r_+-r_-)^2\rp\}\,.
\end{gather*}
\end{proposition}

\subsection{From physical space to frequency space}
\label{sec:physical-to-frequency}

In this section we justify that, under suitable regularity conditions, the study of the PDEs in Section~\ref{sec:PDEs} can be reduced to studying the ODEs in Section~\ref{sec:radial-ODEs} in a uniform-in-frequency manner. We begin with a notion of sufficient integrability:

\begin{definition}[Sufficiently integrable solutions]\label{def:suf-integrability} Fix $s\in\mathbb{Z}$ and some $k_0\in\{0,\dots, |s|\}$. We say that solutions $\swei{\tilde \upphi}{s}_k\in\mathscr{S}_\infty^{[s]}(\mc R)$ to the transformed system of PDEs in Definitions~\ref{def:transformed-system} or \ref{def:alt-transformed-system} with inhomogeneities $\swei{\tilde{\mathfrak{H}}}{s}_k$ and, for $k_0<|s|$, $\swei{\tilde{\mathfrak{h}}}{s}_{k_0}$ are sufficiently integrable if the following conditions hold for every $J\geq 1$. Firstly, we have
\begin{align*}
    &\sum_{J_1+J_2+J_3=0}^J\int_{\mathbb{S}^2}\int_{-\infty}^\infty \big|(r^{-1}\mathring{\slashed\nabla}^{[s]})^{J_1}T^{J_2}\p_{r^*}^{J_3}\swei{\mathfrak H}{s}_k\big|^2+\mathbbm{1}_{\{k=k_0<|s|\}}\big|(r^{-1}\mathring{\slashed\nabla}^{[s]})^{J_1}T^{J_2}\p_{r^*}^{J_3}\swei{\mathfrak h}{s}_k\big|^2dt d\sigma \\
    &\quad = O(w^2r^{(|s|-k)(1+\sign s)}\Delta^{(|s|-k)(1-\sign s)})\,, \numberthis\label{eq:integrability-Hk}
\end{align*} 
as $r^*\to\pm \infty$ for every $k=0,\dots, k_0$. Secondly, when $|a|=M$,
\begin{align*}
 \sup_{r^*\in[-R^*,R^*]}\sum_{J_1+J_2+J_3=0}^J\int_{\mathbb{S}^2}\int_{-\infty}^\infty r^{-2-2J_1}\big|(\mathring{\slashed\nabla}^{[s]})^{J_1}T^{J_2}\p_{r^*}^{J_3}\swei{ \tilde\upphi}{s}_k\big|^2dt d\sigma &<\infty\,, \quad \forall R^*\in(0,\infty)\,,
\end{align*}
and when $|a|<M$, 
\begin{align*}
\sup_{r\in[r_+,R]}\sum_{J_1+J_2+J_3=0}^J\int_{\mathbb{S}^2}\int_{-\infty}^\infty r^{-2-2J_1}\big|(\mathring{\slashed\nabla}^{[s]})^{J_1}T^{J_2}(X^*)^{J_3}\swei{ \tilde\upphi}{s}_k\big|^2dt d\sigma &<\infty\,, \quad\forall R\in(r_+,\infty)\,,
\end{align*}
for every $k=0,\dots, k_0$. Thirdly, if $k_0<|s|$, when $|a|=M$,
\begin{align*}
\sum_{J_1+J_2+J_3=0}^J\int_{\mathbb{S}^2}\int_{-\infty}^\infty r^{-2-2J_1}\big|(\mathring{\slashed\nabla}^{[s]})^{J_1}T^{J_2}\p_{r^*}^{J_3}\swei{ \upphi}{s}_{k_0+1}\big|^2dt d\sigma =O(w)\,, \quad \forall R^*\in(0,\infty)\,,
\end{align*}
and when $|a|<M$, 
\begin{align*}
\sum_{J_1+J_2+J_3=0}^J\int_{\mathbb{S}^2}\int_{-\infty}^\infty r^{-2-2J_1}\big|(\mathring{\slashed\nabla}^{[s]})^{J_1}T^{J_2}(X^*)^{J_3}\swei{ \upphi}{s}_{k_0+1}\big|^2dt d\sigma =O(w)\,, \quad\forall R\in(r_+,\infty)\,,
\end{align*}
as $r*\to \pm\infty$.
\end{definition}

\begin{lemma}[Dyadic energy decay under sufficient integrability] \label{lemma:pigeonhole-energy-decay} Fix $M>0$, $s\in\mathbb Z$, $k_0\in\{0,\dots,|s|\}$, $R^*>0$ and $J\geq 0$. For $k=0,\dots,\min\{|s|,k_0+1\}$, $\swei{\upphi}{s}_k$ be sufficiently integrable solutions to the transformed system of Definitions~\ref{def:transformed-system} or \ref{def:alt-transformed-system}. Then, there is a constant $C$ and a dyadic sequence $\tau_n$ with either $\tau_n\to\infty$ or either $\tau_n\to-\infty$ as $n\to \infty$, which depend only on $a$, $M$, $s$, $k_0$, $R^*$, $J$ and the functions $\swei{\upphi}{s}_k$, $k=0,\dots |s|$, and such that 
\begin{align*}
    \sum_{k=0}^{k_0}\overline{\mathbb{E}}^J[\upphi_k\mathbbm{1}_{\{ r^*\leq R^*\}}](\tau_n)\leq \frac{C}{|\tau_n|}\,, \qquad \text{if~} |a|<M\,;\\
        \sum_{k=0}^{k_0}\mathbb{E}^J[\upphi_k\mathbbm{1}_{\{|r^*|\leq R^*\}}](\tau_n)\leq \frac{C}{|\tau_n|}\,, \qquad \text{if~} |a|=M\,;
\end{align*}
\end{lemma}
\begin{proof}
By Definition~\ref{def:suf-integrability},  if $|a|<M$,
\begin{align*}
    \int_{-\infty}^{-\infty} \sum_{j=0}^{k_0}\overline{\mathbb{E}}^J[\tilde\upphi_j\mathbbm{1}_{\{r_+\leq r\leq R\}}](\tau) d\tau <\infty\,,
\end{align*}
so the result follows by the pigeonhole principle. We proceed similarly for $|a|=M$.
\end{proof}

The notion of sufficient integrability put forth in Definition~\ref{def:suf-integrability} is (more than) enough to ensure that solutions to the transformed system of PDEs in Definitions~\ref{def:transformed-system} and \ref{def:alt-transformed-system} admit a Fourier transform:

\begin{lemma}[Plancherel's theorem] \label{lemma:Plancherel}
Fix $s\in\mathbb{Z}$ and $k_0\in\{0,\dots,|s|\}$. Suppose that, for  $k\in\{0,\dots,|s|\}$ or, if $k_0<|s|$, $k\in\{0,\dots,k_0+1\}$, $\swei{\upphi}{s}_{k}$ are sufficiently integrable solutions to the transformed system of PDEs in Definitions~\ref{def:transformed-system} or \ref{def:alt-transformed-system} with corresponding inhomogeneities $\swei{\mathfrak{H}}{s}_k$ and, if $k_0<|s|$, $\swei{\mathfrak{H}}{s}_{k_0}$. Then, we may define, for $\omega\in\mathbb{R}$ and $r\in(r_+,\infty)$,
\begin{align*}
\hat\uppsi_{(k)}^{[s],\,a,\,\omega}(r,\theta,\phi)&:=\frac{1}{\sqrt{2\pi}}\int_{-\infty}^\infty e^{i\omega t} \swei{\upphi}{s}_{k}(t,r,\theta,\phi) dt\,,\\
\hat{\mathfrak{G}}_{(k)}^{[s],\,a,\,\omega}(r,\theta,\phi)&:=\frac{1}{\sqrt{2\pi}}\int_{-\infty}^\infty e^{i\omega t} \swei{\mathfrak{H}}{s}_{k}(t,r,\theta,\phi) dt\,,\\
\hat{\mathfrak{g}}_{(k_0)}^{[s],\,a,\,\omega}(r,\theta,\phi)&:=\frac{1}{\sqrt{2\pi}}\int_{-\infty}^\infty e^{i\omega t} \swei{\mathfrak{h}}{s}_{k_0}(t,r,\theta,\phi) dt\,,
\end{align*}
By construction, the restriction of these functions to constant $r\in(r_+,\infty)$ is a smooth $s$-spin-weighted functions on $\mathbb{S}^2$; thus, we can define 
\begin{align}
\begin{split}
\uppsi_{(k_0),\,ml}^{[s],\,a,\,\omega}(r)&:=\int_{\mathbb{S}^2}\hat\uppsi_{(k)}^{[s],\,a,\,\omega}(r,\theta,\phi)e^{im\phi}{S}^{[s],\, a\omega}_{ml}(\theta,\phi) d\sigma\\
&=\frac{1}{\sqrt{2\pi}}\int_{\mathbb{S}^2}\int_{-\infty}^\infty e^{i\omega t} \swei{\upphi}{s}_{k}(t,r,\theta,\phi) e^{-im\phi}{S}^{[s],\, a\omega}_{ml}(\theta,\phi)  dt d\sigma\,,\\
\mathfrak{G}_{(k),\,ml}^{[s],\,a,\,\omega}(r)&:=\frac{1}{\sqrt{2\pi}}\int_{\mathbb{S}^2}\int_{-\infty}^\infty e^{i\omega t} \swei{\mathfrak{H}}{s}_{k}(t,r,\theta,\phi) e^{-im\phi}{S}^{[s],\, a\omega}_{ml}(\theta,\phi) dt\,,\\
\mathfrak{g}_{(k_0),\,ml}^{[s],\,a,\,\omega}(r)&:=\frac{1}{\sqrt{2\pi}}\int_{\mathbb{S}^2}\int_{-\infty}^\infty e^{i\omega t} \swei{\mathfrak{h}}{s}_{k_0}(t,r,\theta,\phi) e^{-im\phi}{S}^{[s],\, a\omega}_{ml}(\theta,\phi) dt\,,
\end{split}\label{eq:def-sml-uppsi-G}
\end{align}
for $m$ and $l$ as in Lemma~\ref{lemma:spheroidal-angular-ode}. It also follows that 
\begin{align}
\begin{split}
\swei{\upphi}{s}_{k}(t,r,\theta,\phi) &=\frac{1}{\sqrt{2\pi}}\int_{\mathbb{S}^2}\int_{-\infty}^\infty e^{-i\omega t} \uppsi_{(k),\,ml}^{[s],\,a,\,\omega}(r)e^{im\phi}{S}^{[s],\, a\omega}_{ml}(\theta,\phi)  dt d\sigma\,,\\
\swei{\mathfrak{H}}{s}_{k}(t,r,\theta,\phi)&=\frac{1}{\sqrt{2\pi}}\int_{\mathbb{S}^2}\int_{-\infty}^\infty e^{-i\omega t} \mathfrak{G}_{(k),\,ml}^{[s],\,a,\,\omega}(r) e^{-im\phi}{S}^{[s],\, a\omega}_{ml}(\theta,\phi) dt\,,\\
\swei{\mathfrak{h}}{s}_{k_0}(t,r,\theta,\phi)&=\frac{1}{\sqrt{2\pi}}\int_{\mathbb{S}^2}\int_{-\infty}^\infty e^{-i\omega t} \mathfrak{g}_{(k_0),\,ml}^{[s],\,a,\,\omega}(r) e^{-im\phi}{S}^{[s],\, a\omega}_{ml}(\theta,\phi) dt\,,
\end{split} \label{eq:def-upphi-H-from-sml}
\end{align}
and we note the following trivial applications of Plancherel's theorem and orthogonality of the spheroidal harmonics:
\begin{align*}
\int_{\mathbb{S}^2}\int_{-\infty}^\infty\Re[\swei{\upphi}{s}_k\overline{\swei{\upphi}{s}_j}] dt d\sigma 
&= \sum_{ml}\int_{-\infty}^\infty\uppsi^{[s],\,a\omega}_{(k),\, ml}\overline{\uppsi^{[s],\,a\omega}_{(j),\, ml} }d\omega\,, \\
\int_{\mathbb{S}^2}\int_{-\infty}^\infty|T\swei{\upphi}{s}_k|^2 dt d\sigma 
&= \sum_{ml}\int_{-\infty}^\infty\omega^2|\uppsi^{[s],\,a\omega}_{(k),\, ml}|^2dr^* d\omega\,, \\
\int_{\mathbb{S}^2}\int_{-\infty}^\infty|Z\swei{\upphi}{s}_k|^2dt d\sigma 
&= \sum_{ml}\int_{-\infty}^\infty m^2|\uppsi^{[s],\,a\omega}_{(k),\, ml}|^2 d\omega\,, \\
\int_{\mathbb{S}^2}\int_{-\infty}^\infty|\mathring{\slashed\nabla}^{[s]}\swei{\upphi}{s}_k|^2 dt d\sigma 
&= \sum_{ml}\int_{-\infty}^\infty\Lambda_{ml}^{[s],\,a\omega}|\uppsi^{[s],\,a\omega}_{(k),\, ml}|^2 d\omega\\
&\qquad-\int_{\mathbb{S}^2}\int_{-\infty}^\infty\lp(a^2\sin^2\theta|T\swei{\upphi}{s}_k|^2+2a\cos\theta \Im[T\swei{\upphi}{s}_k\overline{\swei{\upphi}{s}_k}]\rp) dt d\sigma\,.
\end{align*}
Similar identities hold replacing $\swei{\upphi}{s}_k$ by $\swei{\mathfrak H}{s}_k$ or $\swei{\mathfrak h}{s}_k$ and $\smlk{\upphi}{s}{k}$ by $\smlk{\mathfrak G}{s}{k}$ and $\smlk{\mathfrak g}{s}{k}$.
\end{lemma}

From Lemma~\ref{lemma:Plancherel} it is clear that, formally, $\smlk{\uppsi}{s}{k}$ solve the radial ODEs of Definitions~\ref{def:transformed-system-ODE} or \ref{def:alt-transformed-system-ODE}  with $\Lambda=\bm\Lambda^{[s],\,a\omega}_{ml}$ one of the eigenvalues from Lemma~\ref{lemma:spheroidal-angular-ode} and with inhomogeneities $\smlk{\mathfrak G}{s}{k}$ and $\smlk{\mathfrak g}{s}{k}$. To make this precise, it is useful to introduce one more definition:

\begin{definition}[Outgoing PDE solutions] \label{def:outgoing-bdry-phys-space} Fix $s\in\mathbb{Z}$, and some $k_0\in \{0,\dots |s|\}$. 
We say that $\swei{\upphi}{s}_k$ are outgoing solutions of the inhomogeneous transformed system of Definitions~\ref{def:transformed-system} or \ref{def:alt-transformed-system} with inhomogeneities $\swei{\mathfrak{H}}{s}_k$ and, if $k_0<|s|$, $\swei{\mathfrak{h}}{s}_{k_0}$ if there exists an $\epsilon>0$ such that for all $\tau\leq -\epsilon^{-1}$,
\begin{align*}
\swei{\upphi}{s}_k\Big|_{\Sigma_\tau\cap \lp(\{r\leq r_++\epsilon\}\cup \{r\geq\epsilon^{-1}\}\rp)}=0\,,\qquad 
\swei{\mathfrak{H}}{s}_k\Big|_{\Sigma_\tau\cap \lp(\{r\leq r_++\epsilon\}\cup \{r\geq\epsilon^{-1}\}\rp)}=0\,,\qquad 
\swei{\mathfrak{h}}{s}_{k_0}\Big|_{\Sigma_\tau\cap \lp(\{r\leq r_++\epsilon\}\cup \{r\geq\epsilon^{-1}\}\rp)}=0\,,
\end{align*}
in the first case for $k=0,\dots \min\{|s|,k_0+1\}$ and in the second case for $k=0,\dots, k_0$.
\end{definition}

Under the outgoing assumption, sufficient integrability follows if certain bulk and energy norms in $\mc R$ are finite. Conversely, outgoing and sufficient integrability imply finiteness of bulk norms in $\mc R$: 

\begin{lemma}[Sufficient integrability and outgoing property] \label{lemma:suf-integrability} Fix $s\in\mathbb{Z}$, $k_0\in\{0,\dots,|s|\}$ and $M>0$. Suppose that, for each $k=0,\dots,\min\{k_0+1,|s|\}$, $\swei{\upphi}{s}_k$ denotes an outgoing solution, in the sense of Definition~\ref{def:outgoing-bdry-phys-space} to the inhomogeneous transformed system of Definitions~\ref{def:transformed-system} or Definitions~\ref{def:alt-transformed-system}, with inhomogeneities $\swei{\mathfrak H}{s}_k$ and, if $k_0<|s|$, $\swei{\mathfrak h}{s}_{k_0}$.

If for each $R^*>0$, $J\geq 0$, and $k\in\{0,\dots |s|\}$ we have
\begin{align*}
\overline{\mathbb{I}}^J[\swei{\tilde\upphi}{s}_k\mathbbm{1}_{[-R^*,R^*]}](-\infty,\infty)+\sup_{\tau\leq 0}\overline{\mathbb{E}}^J[\swei{\tilde\upphi}{s}_k\mathbbm{1}_{[-R^*,R^*]}](\tau) <\infty\,,\qquad \text{~if~}|a|<M\,,\\
\overline{\mathbb{I}}^J[\swei{\tilde\upphi}{s}_k\mathbbm{1}_{[-R^*,R^*]}](-\infty,\infty)<\infty\,,\qquad \text{~if~}|a|=M\,,
\end{align*}
and conditions~\ref{eq:integrability-Hk} for the inhomogeneities; then, $\swei{\tilde\upphi}{s}_k$ are sufficiently integrable for each $k=0,\dots,|s|$. 

Conversely, if $\swei{\tilde\upphi}{s}_k$ are sufficiently integrable, then we have that for all $J\geq 0$ and $k\in\{0,\dots k_0\}$
\begin{align*}
\overline{\mathbb{I}}_1^{J}[{\tilde\upphi}_k](-\infty,\infty)<\infty\,,\qquad \text{~if~}|a|<M\,,\\
\overline{\mathbb{I}}_1[{\tilde\upphi}_k](-\infty,\infty)<\infty\,,\qquad \text{~if~}|a|=M\,.
\end{align*}
\end{lemma}

\begin{proof} Let us first assume only that $\upphi_k$ are outgoing. 
By the fundamental theorem of calculus, 
\begin{align*}
&\sup_{r\in[r_+,R]}\sum_{J_1+J_2+J_3=0}^{J}\int_{\mathbb{S}^2}\int_{-\infty}^\infty r^{-2-2J_1}\big|(\mathring{\slashed\nabla})^{J_1}T^{j_2}(X^*)^{J_3}\upphi_k\big|^2dt d\sigma \\
&\quad \leq B\overline{E}_{\mc H^+}^J[\tilde\upphi_k](-\infty,\infty) +B\mathbb{I}^{J+1}[\tilde\upphi_k\mathbbm{1}_{[r_+,R]}](-\infty,\infty)\,.
\end{align*}
Take $p\in[0,1]$ excluding $p=0$ if $J\neq 0$. By Propositions~\ref{prop:bulk-flux-large-r} and \ref{prop:higher-order-bulk-flux}, we have for each $|a|<M$, 
\begin{align*}
&\sum_{k=0}^{|s|}\lp(\overline{\mathbb{I}}^{J+|s|-k}_p[{\tilde\upphi}_k](\tau_1,\tau_2)+\overline{\mathbb{E}}_{\mc H^+}^{J+|s|-k}[{\tilde\upphi}_k](\tau_1,\tau_2)\rp) \\
&\quad\leq B\sum_{k=0}^{|s|}\lp(\overline{\mathbb{I}}^{J+|s|-k}[{\tilde\upphi}_k\mathbbm{1}_{[-R^*,R^*]}](\tau_1,\tau_2) +B\overline{\mathbb{E}}_p^{J+|s|-k}[{\tilde\upphi}_k](\tau_1)\rp) \\
&\qquad\quad+ B\sum_{k=0}^{|s|}\sum_{J_1+J_2+J_3=0}^{J+|s|-k}\int_{\mc R_{(\tau_1,\tau_2)}}wr^{-2J_1}\lp|(\mathring{\slashed\nabla})^{J_1}T^{j_2}(X^*)^{J_3}\frac{\tilde{\mathfrak{H}}_k}{w}\rp|^2dr^*d\sigma d\tau\,, \numberthis\label{eq:suf-int-integrated}
\end{align*}
for sufficiently large $R^*$; here we have repeated the proof of the propositions with simple applications of Cauchy--Schwarz to handle the inhomogeneities. If $\tau_1$ is chosen to be sufficiently negative, then by the support conditions in Definition~\ref{def:outgoing-bdry-phys-space}, we can introduce a factor of $\mathbbm{1}_{[-R^*,R^*]}$ inside the energy $\mathbb{E}$.
\begin{align*}
&\sum_{k=0}^{|s|}\lp(\overline{\mathbb{I}}_p^{J+|s|-k}[{\tilde\upphi}_k](\tau_1,\tau_2)+\overline{\mathbb{E}}_{\mc H^+}^{J+|s|-k}[{\tilde\upphi}_k](\tau_1,\tau_2)\rp) \\
&\quad\leq B\sum_{k=0}^{|s|}\lp(\overline{\mathbb{I}}^{J+|s|-k}[{\tilde\upphi}_k\mathbbm{1}_{[-R^*,R^*]}](\tau_1,\tau_2) +B\overline{\mathbb{E}}^{J+|s|-k}[{\tilde\upphi}_k\mathbbm{1}_{[-R^*,R^*]}](\tau_1)\rp) \\
&\qquad\quad+ B\sum_{k=0}^{|s|}\sum_{J_1+J_2+J_3=0}^{J+|s|-k}\int_{\mc R_{(\tau_1,\tau_2)}}wr^{-2J_1}\lp|(\mathring{\slashed\nabla})^{J_1}T^{j_2}(X^*)^{J_3}\frac{\tilde{\mathfrak{H}}_k}{w}\rp|^2dr^*d\sigma d\tau\,,
\end{align*}
where the last line is bounded by the assumption that \eqref{eq:integrability-Hk} holds.

If $\upphi_k$ are further assumed to be sufficiently integrable for $k=0,\dots, |s|$,  Lemma~\ref{lemma:pigeonhole-energy-decay} produces a dyadic sequence $\{\tau_n\}_{n=1}^\infty$ with $\tau_n\to -\infty$ as $n\to \infty$ such that, choosing $\tau_1=\tau_n$ and taking the limit $n\to \infty$, we have
\begin{align*}
&\sum_{k=0}^{|s|}\lp(\overline{\mathbb{I}}_p^{J+|s|-k}[{\tilde\upphi}_k](-\infty,\tau_2)+\overline{\mathbb{E}}_{\mc H^+}^{J+|s|-k}[{\tilde\upphi}_k](-\infty,\tau_2)\rp) \\
&\quad\leq B\sum_{k=0}^{|s|}\overline{\mathbb{I}}^{J+|s|-k}[{\tilde\upphi}_k\mathbbm{1}_{[-R^*,R^*]}](-\infty,\tau_2)  \\
&\qquad\quad+ B\sum_{k=0}^{|s|}\sum_{J_1+J_2+J_3=0}^{J+|s|-k}\int_{\mc R_{(-\infty,\tau_2)}}wr^{-2J_1}\lp|(\mathring{\slashed\nabla})^{J_1}T^{j_2}(X^*)^{J_3}\frac{\tilde{\mathfrak{H}}_k}{w}\rp|^2dr^*d\sigma d\tau\,,
\end{align*}
Now we can also take $\tau_2\to \infty$, since the right hand side will be finite by Definition~\ref{def:suf-integrability}. In the case where we assume sufficient integrability for $k=0,\dots, k_0+1$ with $k_0<|s|$ follows similarly.
\end{proof}

We are finally ready to relate the ODEs of Definitions~\ref{def:transformed-system-ODE} and \ref{def:alt-transformed-system-ODE} with the PDEs of Definitions~\ref{def:transformed-system} and \ref{def:alt-transformed-system}:

\begin{lemma}[Reduction to classical solutions to ODEs] \label{lemma:reduction-classical-odes} Fix $M>0$ and $s\in\mathbb Z$. Suppose that, for $k=0,\dots,|s|$, $\swei{\upphi}{s}_k$ denote sufficiently integrable and outgoing solutions to the transformed system of Definition~\ref{def:transformed-system} with inhomogeneities $\swei{\mathfrak{H}}{s}_k$. Let $\smlk{\uppsi}{s}{k}$, $\smlk{\mathfrak G}{s}{k}$  be as in Lemma~\ref{lemma:Plancherel}. Then, for almost every $\omega\in\mathbb{R}$, $\smlk{\uppsi}{s}{k}$ are smooth solutions, with outgoing boundary conditions as per Definition~\ref{def:outgoing-bdry-freq-space}, to the transformed system of ODEs of Definition~\ref{def:transformed-system-ODE}, with $\Lambda=\bm\Lambda^{[s],\,a\omega}_{ml}$ being one of the eigenvalues from Lemma~\ref{lemma:spheroidal-angular-ode}, and with inhomogeneities $\smlk{\mathfrak G}{s}{k}$.

A similar statement holds for the alternative transformed system of Definition~\ref{def:transformed-system} when $k_0<|s|$ is fixed: if, for $k=0,\dots k_0+1$, $\swei{\upphi}{s}_k$ denote sufficiently integrable and outgoing solutions to the transformed system of Definition~\ref{def:alt-transformed-system} with inhomogeneities $\swei{\mathfrak{H}}{s}_k$ and  $\swei{\mathfrak{h}}{s}_{k_0}$ and $\smlk{\uppsi}{s}{k}$, $\smlk{\mathfrak G}{s}{k}$, $\smlk{\mathfrak g}{s}{k_0}$  be as in Lemma~\ref{lemma:Plancherel}, then for almost every $\omega\in\mathbb{R}$, $\smlk{\uppsi}{s}{k}$ are smooth solutions, with outgoing boundary conditions as per Definition~\ref{def:outgoing-bdry-freq-space}, to the transformed system of ODEs of Definition~\ref{def:alt-transformed-system-ODE} and with inhomogeneities $\smlk{\mathfrak G}{s}{k}$ and $\smlk{\mathfrak g}{s}{k_0}$.
\end{lemma}

\begin{proof} The proof is identical to the case of $s=0$, established in \cite[Sections 5.3 and 5.4]{Dafermos2016b}. It is also similar to prove the statement with $k_0=|s|$ or $k_0<|s|$, and so we will focus on the former case.

\medskip
\noindent\textit{Step 1: smoothness}. By convention, see Lemma~\ref{lemma:spheroidal-angular-ode}, $\Lambda_{ml}^{[s],\,a\omega}$ is a measurable function of $\omega$, so it is enough to prove the result for a fixed admissible pair $(m,l)$. 

Let us first deal with the case $|a|<M$. For $n\in\mathbb{Z}_+$, define
\begin{align*}
 U_n^{(J)}:=\{\omega\in\mathbb{R}\colon \sml{\uppsi}{s}{k}\in C^J(r_+,r_++n),\,\, k=0,\dots,|s|\}\,,\qquad U:=\cap_{J,n=1}^\infty U_n^{(J)};
\end{align*}
The goal of this step is to show that $U$, containing the set of $\omega\in\mathbb{R}$ such that $\sml{\uppsi}{s}{k}$ is smooth, has full measure. Since
\begin{align*}
|U^c|=|\cup_{n,J=1}^\infty (U_n^{(J)})^c|\leq \sum_{n,J=1}^\infty|U_n^{(J)}|\,,
\end{align*}
it is enough to show that $|U_n^{(J)}|=0$ for every $n,J\in\mathbb{Z}_+$. This is guaranteed if, for instance the function  $\sup_{r\in[r_+,r_++n]}\sum_{k=0}^{|s|}|(X^*)^j \sml{\uppsi}{s}{k}|^2$ is in $L^1_\omega$, or if 
\begin{align*}
\sum_{0\leq j\leq J}\int_{-\infty}^\infty \sum_{ml} \sum_{k=0}^{|s|}\sup_{r\in[r_+,r_++n]}\big|(X^*)^j \smlk{\uppsi}{s}{k}\big|^2d\omega <\infty\,, \numberthis\label{eq:smoothness-intermediate}
\end{align*}
for each $n,J\in\mathbb{Z}_+$. 

We have reduced the proof to establishing \eqref{eq:smoothness-intermediate}. First, we note that, after applying the Plancherel, Lemma~\ref{lemma:Plancherel}, to the sufficient integrability condition in Definition~\ref{def:suf-integrability}, we almost obtain the formula on the left hand side of \eqref{eq:smoothness-intermediate}: the only difference will be that the supremum in $r$ is taken after, not before, integration in $\omega$. To get around this difficulty, we apply the fundamental theorem of calculus:
\begin{align*}
\sup_{r\in[r_+,r_++n]}\big|(X^*)^j \smlk{\uppsi}{s}{k}\big|^2 \leq \big|(X^*)^j\smlk{\uppsi}{s}{k}\big|^2\Big|_{r=r_+}+\int_{r_+}^{r_++n}\big|(X^*)^{j+1} \smlk{\uppsi}{s}{k}\big|^2\,.
\end{align*}
Thus, applying Plancherel in the forms in Lemma~\ref{lemma:Plancherel},
\begin{align*}
&\sum_{0\leq j\leq J}\int_{-\infty}^\infty \sum_{ml} \sum_{k=0}^{|s|}\sup_{r\in[r_+,r_++n]}\big|(X^*)^j \smlk{\uppsi}{s}{k}\big|^2d\omega\\
&\quad\leq
\sum_{k=0}^{|s|}\sum_{0\leq j\leq J}\lp(\int_{-\infty}^\infty\int_{\mathbb{S}^2} \big|(X^*)^j \swei{\upphi}{s}_{k}\big |^2\Big|_{r=r_+} d\sigma dt +\int_{-\infty}^\infty\int_{\mathbb{S}^2}\int_{r_+}^{r_++n} \big|(X^*)^{j+1} \swei{\upphi}{s}_{k}\big |^2 dr d\sigma dt \rp)\\
&\quad \leq n\sum_{k=0}^{|s|}\sum_{0\leq j\leq J+1} \sup_{r\in[r_+,r_++n]}\big|(X^*)^j \swei{\upphi}{s}_{k}\big |^2<\infty\,,
\end{align*}
by the sufficient integrability assumption, see Definition~\ref{def:suf-integrability}.

Now, for the case $|a|=M$, we follow a similar strategy, replacing $X^*$ with $\p_{r^*}$ and the intervals $r\in[r_+,r_++n]$ with intervals of the form $r^*\in[-n,n]$.

\medskip
\noindent\textit{Step 2: boundary conditions}.  As a consequence of Lemma~\ref{lemma:suf-integrability}, in particular we have that
\begin{align*}
&\int_{-\infty}^{\infty}\sum_{ml}\int_{R}^\infty\lp|(\p_{r^*}-i\omega)\lp(w^{-\frac{|s|-k}{2}}\Delta^{\frac{|s|-k}{2}\sign s}(r^2+a^2)^{-\frac{|s|-k}{4}(1+\sign s)} \smlk{\uppsi}{s}{k}\rp)\rp|^2 dr d\omega  \\
&=
\int_{\mathbb{S}^2}\int_{-\infty}^{\infty}\int_{R}^\infty\lp|\lp(\p_{r^*}+\p_t\rp)\swei{\tilde\upphi}{s}_k\rp|^2 dr d\tau d\sigma \leq B\overline{\mathbb{I}}_{1}[\swei{\dbtilde\upphi}{s}_k](-\infty,\infty)<\infty\,,
\end{align*}
so, by the pigeonhole principle, we can find a dyadic sequence $\{r_n\}_{n=1}^\infty$ with $r_n\to\infty$ as $n\to \infty$ such that
\begin{align*}
\lim_{n\to \infty}\int_{-\infty}^{\infty}\lp|(\p_{r^*}-i\omega)\lp(c_k \smlk{\uppsi}{s}{k}\rp)\rp|^2\Big|_{r=r_n} d\omega = 0 \,.
\end{align*}
We deduce that, for almost every $\omega\in\mathbb{R}$, 
\begin{align*}
\lim_{n\to \infty}\lp|\lp(c_k\smlk{\uppsi}{s}{k}\rp)'-i\omega \lp(c_k \smlk{\uppsi}{s}{k}\rp)\rp|^2(r_n) = 0 \,.
\end{align*} 

In the case $|a|=M$, the same argument, \textit{mutatis mutandis}, may be applied close to $r=r_+$ to obtain the boundary condition there: there is a sequence $\{r_n\}_{n=1}^\infty$, with the property that $r_n\to r_+$ as $n\to \infty$ such that, for almost every $\omega\in\mathbb{R}$, 
\begin{align*}
\lim_{n\to \infty}\lp|\lp(c_k \smlk{\uppsi}{s}{k}\rp)'+i(\omega-m\upomega_+) \lp(c_k \smlk{\uppsi}{s}{k}\rp)\rp|^2(r_n) = 0 \,.
\end{align*} 
On the other hand, an asymptotic analysis at these irregular singularities tells us that there are two linearly independent behaviors which the solution can take, see our previous \cite[Section 4.3.2]{SRTdC2020}, but the above condition is only verified if \eqref{eq:def-outgoing-bdry-freq-space-infty} in the case $r\to \infty$, or \eqref{eq:def-outgoing-bdry-freq-space-hor} in the case $r\to r_+$ for $|a|=M$, are satisfied. In the case of $|a|<M$, since we have that
\begin{align*}
&\int_{-\infty}^{\infty}\sum_{ml}\int^{r_++\epsilon}_{r_+}\frac{1}{(r-r_+)^2}\lp|(\p_{r^*}+i(\omega-m\upomega_+))\lp(w^{-\frac{|s|-k}{2}}\Delta^{\frac{|s|-k}{2}\sign s} \smlk{\uppsi}{s}{k}\rp)\rp|^2 dr d\omega  \\
&\quad=
\int_{\mathbb{S}^2}\int_{-\infty}^{\infty}\int^{r_++\epsilon}_{r_+}\lp|\frac{1}{r-r_+}\lp(\p_{r^*}-T-\frac{a}{2Mr_+}Z\rp)\swei{\tilde\upphi}{s}_k\rp|^2 dr d\tau d\sigma \\
&\quad\leq B\overline{\mathbb{I}}[\swei{\tilde\upphi}{s}_k](-\infty,\infty)<\infty\,,
\end{align*}
by Lemma~\ref{lemma:suf-integrability}, we deduce that as $r\to r_+$
\begin{align*}
\int_{-\infty}^{\infty}\sum_{ml}\lp|(\p_{r^*}+i(\omega-m\upomega_+))\lp(w^{-\frac{|s|-k}{2}}\Delta^{\frac{|s|-k}{2}\sign s} \smlk{\uppsi}{s}{k}\rp)\rp|^2 dr d\omega =O(r-r_+)\,.
\end{align*}
Once more, the asymptotic analysis at regular singularity $r=r_+$ tells us that there are two linearly independent behaviors which the solution can take, see our previous \cite[Section 4.3.2]{SRTdC2020}, but the above condition is only verified if \eqref{eq:def-outgoing-bdry-freq-space-hor} is satisfied. 
\end{proof}

Finally, we introduce some notation which will be useful to keep in mind in the following sections: for a spin-weighted function $\swei{f}{s}$ which is sufficiently integrable, defining $\sml{f}{s}$ as in Lemma~\ref{lemma:Plancherel}, we define the microlocal operators
\begin{align*}
\mc{P}_{\rm trap}[\swei{f}{s}]&:=\frac{1}{\sqrt{2\pi}}\int_{-\infty}^\infty \sum_{ml}\lp|1-(1-\zeta)\frac{r_{\rm trap}(\omega,m,l)}{r}\rp|e^{-i\omega t} \uppsi_{(k),\,ml}^{[s],\,a,\,\omega}(r)e^{im\phi}{S}^{[s],\, a\omega}_{ml}(\theta,\phi)  dt d\sigma\,,\\
\widetilde{\mc{P}}_{\rm trap}[\swei{f}{s}]&:=\frac{1}{\sqrt{2\pi}}\int_{-\infty}^\infty \sum_{ml}\mathbbm{1}_{\{r_{\rm trap}(\omega,m,l)\neq 0\}}e^{-i\omega t} \uppsi_{(k),\,ml}^{[s],\,a,\,\omega}(r)e^{im\phi}{S}^{[s],\, a\omega}_{ml}(\theta,\phi)  dt d\sigma\,,\\
T^{1/2}[\swei{f}{s}]&:=\frac{1}{\sqrt{2\pi}}\int_{-\infty}^\infty \sum_{ml}|\omega|^{1/2}e^{-i\omega t} \uppsi_{(k),\,ml}^{[s],\,a,\,\omega}(r)e^{im\phi}{S}^{[s],\, a\omega}_{ml}(\theta,\phi)  dt d\sigma\,,\\
Z^{1/2}[\swei{f}{s}]&:=\frac{1}{\sqrt{2\pi}}\int_{-\infty}^\infty \sum_{ml}|m|^{1/2}e^{-i\omega t} \uppsi_{(k),\,ml}^{[s],\,a,\,\omega}(r)e^{im\phi}{S}^{[s],\, a\omega}_{ml}(\theta,\phi)  dt d\sigma\,.
\end{align*}

%---------------
%Frequency space analysis
%---------------

\section{Frequency space estimates}
\label{sec:frequency-estimates-review}

In this section, we briefly recall for the reader the results, for the transformed system of radial ODEs in Definition~\ref{def:transformed-system-ODE}, which have been previously obtained in \cite{Shlapentokh-Rothman2015,TeixeiradaCosta2019, SRTdC2020}. As always, we take $s\in\mathbb{Z}$, $|a|\leq M$ and let $(\omega,m,\Lambda)$ be an admissible frequency triple. For some $0<\omega_{\rm low}<\omega_{\rm high}<\infty$ and $\varepsilon_{\rm width}>0$, it will also be convenient to define
\begin{align*}
\mc{F}_{\rm low}&:=\lp\{(\omega,m,\Lambda) \text{~admissible~}\colon \Lambda\leq \varepsilon_{\rm width}^{-1}\omega_{\rm high}^2\,,\,\, |\omega|< \omega_{\rm low}\rp\}\,,\\
\mc{F}_{\rm int}&:=\lp\{(\omega,m,\Lambda) \text{~admissible~}\colon \Lambda\leq \varepsilon_{\rm width}^{-1}\omega_{\rm high}^2\,,\,\, \omega_{\rm low}<|\omega|< \omega_{\rm high}\rp\}\,,\\
\mc{F}_{\rm high}&:=\mc F_{\rm admiss}\backslash\lp(\mc{F}_{\rm low} \cup \mc{F}_{\rm int}\rp)\,.
\end{align*}

Let us first introduce the notation
\begin{align*}
\swei{\mathfrak{W}}{s}(\omega,m,\Lambda):=\lp({u}^{[s],\,a,\,\omega}_{m\Lambda,\,\mc{I}^+}\rp)'{u}^{[s],\,a,\,\omega}_{m\Lambda,\,\mc{H}^+}-{u}^{[s],\,a,\,\omega}_{m\Lambda,\,\mc{I}^+}\lp({u}^{[s],\,a,\,\omega}_{m\Lambda,\,\mc{H}^+}\rp)'\,,
\end{align*}
where ${u}^{[s],\,a,\,\omega}_{m\Lambda,\,\mc{I}^+}$ and ${u}^{[s],\,a,\,\omega}_{m\Lambda,\,\mc{H}^+}$ are given in Definition~\ref{def:psihor-psiout}. $\swei{\mathfrak{W}}{s}$ is independent of $r$ and obeys explicit bounds in $\mc F_{\rm int}$ proven in \cite{Shlapentokh-Rothman2015,TeixeiradaCosta2019,TdC-thesis}, based on the seminal paper \cite{Whiting1989}. We quote such bounds from the thesis \cite{TdC-thesis}:

\begin{theorem}[Quantitative mode stability] \label{thm:quantitative-mode-stab-in-text} Fix $s\in\frac12\mathbb{Z}$, $M>0$ and some $a_0\in[0,M)$.  Then, assuming $|a|\leq a_0$, we have 
\begin{align*}
\sup_{(\omega,m,\Lambda)\in\mc F_{\rm int},s\leq 0}\lp|\swei{\mathfrak{W}}{s}(\omega,m,\Lambda)\rp|^{-1}\leq B(a_0,\omega_{\rm high},\omega_{\rm low},\varepsilon_{\rm width})\,,\\
\sup_{(\omega,m,\Lambda)\in\mc F_{\rm int},s> 0}|\omega-m\upomega_+|^{-1}\lp|\swei{\mathfrak{W}}{s}(\omega,m,\Lambda)\rp|^{-1}\leq B(a_0,\omega_{\rm high},\omega_{\rm low},\varepsilon_{\rm width})\,.
\end{align*}
\end{theorem}

Next, we review\footnote{The analogous results to Theorems~\ref{thm:ODE-ILED} and \ref{thm:ODE-ILED-fixed-m} in our previous work are \cite[Theorems 6.1 and 6.2]{SRTdC2020}. Not every statement in Theorems~\ref{thm:ODE-ILED} and \ref{thm:ODE-ILED-fixed-m} is explicitly stated in \cite[Theorems 6.1 and 6.2]{SRTdC2020}, but the needed results follow easily from the proofs given for the latter, see \cite[Remark 6.1.1]{SRTdC2020}.} the main results of our frequency analysis in our previous work \cite{SRTdC2020}:
\begin{theorem}[Uniform high and low frequency estimates] \label{thm:ODE-ILED}
Fix $s\in\{0,\pm 1,\pm 2\}$, $k_0\in\{0,\dots,|s|\}$, $M>0$, and $a_0\in[0,M)$. For $k=0,\dots, |s|$, let $\smlambdak{\uppsi}{s}{k}$ be solutions, with outgoing boundary conditions (see Definition~\ref{def:outgoing-bdry-freq-space}), to the transformed system of radial ODEs in Definition~\ref{def:alt-transformed-system-ODE} with respect to inhomogeneities $\smlambdak{\mathfrak G}{s}{k}$ and $\smlambdak{\mathfrak g}{s}{k_0}$ satisfying $\smlambdak{\mathfrak G}{s}{k},\smlambdak{\mathfrak g}{s}{k_0}=O(w)$ as $r^*\to \pm \infty$.

There are parameters $\omega_{\rm low}>0$, $\omega_{\rm high}>0$, $\varepsilon_{\rm width}>0$, $E>0$ and $E_W>0$, depending only on $s$, $M$, $a_0$ and $J_{\rm max}$, such that, for each  $|a|\leq a_0$ and each admissible $(\omega,m,\Lambda)$, there is a parameter $r_{\rm trap}$ and a choice of functions $y_{(k)}$, $\hat{y}_{(k)}$, $\tilde{y}$, $f$,  $h_{(k)}$, $\chi_1$, $\chi_2$, $\chi_3$, for each $0\leq k\leq |s|$ (we drop the subscripts when $k=|s|$), which depend on the frequency triple only in a region $r^*\in[-R^*_{\rm large},R^*_{\rm large}]$ and satisfy the uniform bounds
\begin{gather*}
\begin{gathered}
|y_{(k)}|+|\hat{y}_{(k)}|+|\tilde{y}|+|f|+|f'|+|h_{(k)}|+|\chi_1|+|\chi_2|+|\chi_3|+|r_{\rm trap}|\leq B\,,\\
 (|y'_{(k)}|+|\hat{y}_{(k)}'|+|\tilde y'|+|f'|+|h_{(k)}|)(r^2+|r^*|^{3/2})\leq B\,,\\
 \chi_1,\chi_3=\begin{dcases}
 1,&r^*\geq R^*_{\rm large}\\
 0,&r^*\leq -R^*_{\rm large}
 \end{dcases}\,, \qquad  \chi_2=\begin{dcases}
 0,&r^*\geq R^*_{\rm large}\\
 1,&r^*\leq -R^*_{\rm large}
 \end{dcases}
 \end{gathered}\,,\numberthis\label{eq:current-bounds}\\
 |r_{\rm trap}-r_+|^{-1}\leq B(a_0) \,,\qquad |r_{\rm trap}-3M|+|r_{\rm trap}-3M|^{-1}\leq B(a_0)\,,
\end{gather*}
so that, writing
\begin{align*}
&\lp(\smlambdak{\mathfrak G}{s}{j}+\smlambdak{\mathfrak g}{s}{j}\rp)\cdot(f_j,h_j,y_j,\chi)\cdot \lp(\smlambdak{\uppsi}{s}{j},\lp(\smlambdak{\uppsi}{s}{j}\rp)'\rp) \\
&\quad = \lp\{2(y+\hat{y}+\tilde{y}+f)\Re\lp[\smlambda{\mathfrak G}{s}(\overline{\smlambda{\Psi}{s}})'\rp]+(h+f')\Re\lp[\smlambda{\mathfrak G}{s}\overline{\smlambda{\Psi}{s}}\rp]\rp\}\mathbbm{1}_{\{j=k_0=|s|\}}\\
&\quad\qquad -E\lp(\chi_1(\omega-m\upomega_+)+\chi_2\omega\rp)\Im\lp[\smlambda{\mathfrak G}{s}\overline{\smlambda{\Psi}{s}}\rp]\mathbbm{1}_{\{j=k_0=|s|\}}\\
&\quad\qquad+\mathbbm{1}_{\{j\leq k_0,j<|s|\}} 2(y_{(j)}+\hat{y}_{(j)})\Re\lp[\smlambdak{\mathfrak G}{s}{j}(\overline{\smlambdak{\uppsi}{s}{j}})'\rp] +\mathbbm{1}_{\{j<k_0\leq |s|\}}h_{(j)}\Re\lp[\smlambdak{\mathfrak G}{s}{j}\overline{\smlambdak{\uppsi}{s}{j}}\rp]\\
&\quad\qquad +\mathbbm{1}_{\{j\leq k_0,j<|s|\}}\mathbbm{1}_{\mc F_{\rm high}}\Re\lp[\lp(\smlambdak{\mathfrak G}{s}{j}+\smlambdak{\mathfrak g}{s}{j}\rp)\overline{\smlambdak{\uppsi}{s}{j}}\rp]\\
&\quad\qquad -E\mathbbm{1}_{\{j< k_0\leq |s|\}}\mathbbm{1}_{\mc F_{\rm high}^c}\lp(\chi_2\omega  -E\chi_1(\omega-m\upomega_+)\rp)\Im\lp[\smlambdak{\mathfrak G}{s}{j}\overline{\smlambdak{\uppsi}{s}{j}}\rp]\\
&\quad\qquad-\frac12(-1)^s\sign s \mathbbm{1}_{\{j=k_0=|s|\}} E_W\chi_3(\omega-m\upomega_+)w(r^2+a^2)^{|s|+1/2}\\
&\quad\qquad\qquad \times \lp\{\Im\lp[\frac{\smlambdak{\mathfrak G}{s}{0}}{w}\overline{\lp(\frac{r^2+a^2}{\Delta}\uL\rp)^{2|s|}\lp((r^2+a^2)^{|s|-1/2}\smlambdak{\uppsi}{s}{0}\rp)}\rp]\rp.\\
&\quad\qquad\qquad\qquad\lp.-\Im\lp[\smlambdak{\uppsi}{s}{0}\overline{\lp(\frac{r^2+a^2}{\Delta}\uL\rp)^{2|s|}\lp((r^2+a^2)^{|s|-1/2}\frac{\smlambdak{\mathfrak G}{s}{0}}{w}\rp)}\rp]\rp\}
\,. \numberthis \label{eq:ODE-estimates-inhomogeneity-terms}
\end{align*}
the following estimates hold.

\begin{enumerate}[label=(\roman*)]
\item \label{it:thm-ODE-ILED-withoutTS} {\normalfont Estimates for weak $r\to r_+$ decay.}  If for $s>0$  we can only assume $\smlambdak{\mathfrak G}{s}{k}=O(\Delta)$ as $r^*\to -\infty$, then we have
\begin{align*}
&\int_{r_+}^\infty \lp\{\frac{1}{r^2}|(\smlambda{\Psi}{s})'|^2+\lp(1-\frac{r_{\rm trap}}{r}\rp)^2\lp(\frac{\omega^2}{r^2}+\frac{|\Lambda|}{r^3}\rp)|\smlambda{\Psi}{s}|^2+\frac{1}{r^3}\lp(s^2+\frac{1}{r}\rp)|\smlambda{\Psi}{s}|^2\rp\}dr \\
&\qquad +\int_{r_+}^\infty\lp\{ \frac{1}{r^{5/2}}(|m|+|\omega|)|\smlambda{\Psi}{s}|^2+\sum_{k=0}^{|s|-1}\frac{r_{\rm trap}}{r^3}(|m|^3+|\omega|^3)|\smlambdak{\uppsi}{s}{k}|^2\rp\}dr\\
&\qquad + \sum_{k=0}^{|s|-1}\int_{r_+}^\infty \lp\{\frac{1}{r^2}|(\smlambdak{\uppsi}{s}{k})'|^2+\frac{1}{r^2}\lp(\omega^2+\frac{|\Lambda|+1}{r}\rp)|(\smlambdak{\uppsi}{s}{k})|^2\rp\}dr\\
&\qquad+\sum_{k=0}^{|s|}\lp\{(\omega-m\upomega_+)^2|\smlambdak{\uppsi}{s}{k}|^2\Big|_{r=-\infty}+\omega^2|\smlambdak{\uppsi}{s}{k}|^2\Big|_{r=\infty}\rp\}\\
&\quad\leq B\mathbbm{1}_{\mc F_{\rm int}}\sum_{k=0}^{|s|}\int_{-R^*_{\rm large}}^{R^*_{\rm large}} \lp(|(\smlambdak{\uppsi}{s}{k})'|^2+|\smlambdak{\uppsi}{s}{k}|^2\rp)dr^*  \\
&\quad\qquad+ B\int_{-\infty}^\infty \sum_{k=0}^{|s|}(\smlambdak{\mathfrak G}{s}{k}+\smlambdak{\mathfrak g}{s}{k})\cdot(f_k,h_k,y_k,\chi_k)\cdot \lp(\smlambdak{\uppsi}{s}{k},\lp(\smlambdak{\uppsi}{s}{k}\rp)'\rp) dr^*\\
&\quad\qquad + B\mathbbm{1}_{\{s>0\}}\mathbbm{1}_{\mc F_{\rm low}}\lp(|\smlambdak{\uppsi}{s}{0}|^2\Big|_{r=-\infty}+B\mathbbm{1}_{\mc F_{\rm low}}\sum_{j=0}^{|s|-1}|w^{-1}\smlambdak{\mathfrak G}{s}{j}|^2\Big|_{r=-\infty}\rp)  \,, \numberthis\label{eq:part1-ILED-withoutTS-top}
\end{align*}
and we can improve our estimates for $k<|s|$ to
\begin{align*}
&b(a_0)\int_{r_+}^\infty \lp\{\frac{1}{r^2}|(\smlambdak{\uppsi}{s}{k})''|^2+\frac{1}{r^2}(\omega^2+|\Lambda|+1)|(\smlambdak{\uppsi}{s}{k})'|^2+ \frac{1}{r^2}\lp(\omega^2+\frac{|\Lambda|+1}{r}\rp)|\smlambdak{\uppsi}{s}{k}|^2\rp\}dr\\
&\quad\qquad +\lp\{(\omega-m\upomega_+)^2|\smlambdak{\uppsi}{s}{k}|^2\Big|_{r=r_+}+\omega^2|\smlambdak{\uppsi}{s}{k}|^2\Big|_{r=\infty}\rp\}\\
&\quad\leq B\int_{r_+}^\infty \lp\{\frac{1}{r^2}|(\smlambdak{\uppsi}{s}{k+1})'|^2+\frac{1}{r^{5/2}}(|m|+|\omega|)|\smlambdak{\uppsi}{s}{k+1}|^2+\frac{1}{r^3}|\smlambdak{\uppsi}{s}{k+1}|^2\rp\}dr  \\
&\quad\qquad  +B\sum_{j=0}^{k+1}\mathbbm{1}_{\mc F_{\rm int}}\int_{-R^*_{\rm large}}^{R^*_{\rm large}} \lp(|(\smlambdak{\uppsi}{s}{j})'|^2+|\smlambdak{\uppsi}{s}{j}|^2\rp)dr^*\numberthis\label{eq:part1-ILED-withoutTS-bottom} \\
&\quad\qquad +B\int_{-\infty}^\infty \sum_{j=0}^{k}(\omega^2+m^2+1)\lp(\smlambdak{\mathfrak G}{s}{j}+\smlambdak{\mathfrak g}{s}{j}\rp)\cdot(f_j,h_j,y_j,\chi_j)\cdot\lp(\smlambdak{\uppsi}{s}{j},\lp(\smlambdak{\uppsi}{s}{j}\rp)'\rp) dr^*\\
&\quad\qquad +B\int_{-\infty}^\infty \sum_{j=0}^{k}\lp(\smlambdak{\mathfrak G}{s}{j}+\smlambdak{\mathfrak g}{s}{j}\rp)'\cdot(f_j,h_j,y_j,\chi_j)\cdot\lp(\lp(\smlambdak{\uppsi}{s}{j}\rp)',\lp(\smlambdak{\uppsi}{s}{j}\rp)''\rp) dr^* \\ 
&\quad \qquad + B\mathbbm{1}_{\{s>0\}}\mathbbm{1}_{\mc F_{\rm low}}\lp(|\smlambdak{\uppsi}{s}{0}|^2\Big|_{r=r_+}+\sum_{j=0}^{k-1}\lp(|w^{-1}\smlambdak{\mathfrak G}{s}{j}|^2+|w^{-1}\lp(\smlambdak{\mathfrak G}{s}{j}\rp)'|^2\rp)\Big|_{r=r_+}\rp)\,,
\end{align*}
where we take $E_W=0$ in \eqref{eq:ODE-estimates-inhomogeneity-terms} if $s>0$.

\item \label{it:thm-ODE-ILED-withoutenergy} {\normalfont Estimates for compact $r^*$ support.} If $\smlambdak{\uppsi}{s}{k}$ and $\smlambdak{\mathfrak G}{s}{k}$ are compactly supported in $r^*$, then  \eqref{eq:part1-ILED-withoutTS-top} and  \eqref{eq:part1-ILED-withoutTS-bottom} also hold with the following modifications: we set $E_W=E= 0$ in \eqref{eq:ODE-estimates-inhomogeneity-terms}, we suppress the term on the right hand side which has been localized to $\mc F_{\rm int}$ and the very last line in both estimates.

\item \label{it:thm-ODE-ILED-bdry-term} {\normalfont Boundary term estimates.} If we have the stronger decay $\smlambdak{\mathfrak G}{s}{k}=O(w^{1+(|s|-k)(1+\sign s)})$ as $r^*\to -\infty$, then we have the following boundary term estimate for any $R_{\rm bdry^*}>-\infty$:
\begin{align*}
&\mathbbm{1}_{\mc F_{\rm low}}\lp[\sum_{k=0}^{|s|}\lp[(\omega-m\upomega_+)^2+(|s|-k)^2\frac{(r_+-r_+)^2}{4Mr_+}\rp]|\smlk{\uppsi}{s}{k}|^2\Big|_{r=r_+}\rp]\\
&\quad\leq B\int_{-\infty}^{R^*_{\rm bdry}} \sum_{k=0}^{|s|}\lp(\smlambdak{\mathfrak G}{s}{k}+\smlambdak{\mathfrak g}{s}{k}\rp)\cdot(f_k,h_k,y_k,\chi_k)\cdot \lp(\smlambdak{\uppsi}{s}{k},\lp(\smlambdak{\uppsi}{s}{k}\rp)'\rp) dr^*\\
&\quad \qquad +\sum_{k=0}^{|s|}\int_{R_{\rm bdry}^*}^{R_{\rm bdry}^*+1}\mathbbm{1}_{\mc F_{\rm low}}\lp(|(\smlambdak{\uppsi}{s}{k})'|^2+|\smlambdak{\uppsi}{s}{k}|^2\rp)dr^*
\,, \numberthis\label{eq:part1-bdry-withTS}
\end{align*}
where $E_W\neq 0$. If $s>0$ and we only have  $\smlambdak{\mathfrak G}{s}{k}=O(\Delta)$ as $r^*\to -\infty$, then we take take $E_W=0$; there is an  $R_{\rm bdry}^*>0$ sufficiently large such that if $\omega_{\rm low}$ in the statement above is sufficiently small depending on $R_{\rm bdry}^*$ then
\begin{align*}
&\mathbbm{1}_{\mc F_{\rm low}}\sum_{k=0}^{|s|}R_{\rm bdry}^*\int_{R_{\rm bdry}^*}^{2R_{\rm bdry}^*} \lp\{\frac{1}{r^2}|(\smlambdak{\uppsi}{s}{k})'|^2+\frac{1}{r^2}\lp(\omega^2+\frac{|\Lambda|+1}{r}\rp)|\smlambdak{\uppsi}{s}{k}|^2\rp\}dr\\
&\quad\leq \sum_{k=0}^{|s|} |\smlambdak{\uppsi}{s}{k}|^2\Big|_{r=r_+}\\
&\quad\quad +B\int_{-\infty}^\infty \sum_{k=0}^{|s|}(\smlambdak{\mathfrak G}{s}{k}+\smlambdak{\mathfrak g}{s}{k})\cdot(f_k,h_k,y_k,\chi_k)\cdot \lp(\smlambdak{\uppsi}{s}{k},\lp(\smlambdak{\uppsi}{s}{k}\rp)'\rp) dr^*\,. \numberthis\label{eq:part1-ILED-withoutTS-top-low}
\end{align*}
\end{enumerate}
\end{theorem}

Furthermore the results of \cite{SRTdC2020} allow us to define the trapping parameters $\gamma_\pm(a_0,M)$ appearing in our energy norms, see \eqref{eq:def-zeta}:
\begin{definition}[Trapping parameters] \label{def:trapping-parameters} Fix $M>0$ and $a_0\leq [0,M)$. We set
\begin{align*}
3M-\gamma_-&:=\inf_{a\leq a_0, (\omega,m,\Lambda)\in\mc F_{\rm high},r_{\rm trap}(\omega,m,\Lambda)\neq 0}r_{\rm trap}(\omega,m,\Lambda) - \eta(a_0)\,,\\
3M+\gamma_+&:=\sup_{a\leq a_0, (\omega,m,\Lambda)\in\mc F_{\rm high},r_{\rm trap}(\omega,m,\Lambda)\neq 0}r_{\rm trap}(\omega,m,\Lambda) + \eta(a_0)\,,
\end{align*}
where $\eta(a_0)$ is a fixed continuous function with the properties that $\eta(0)=\lim_{a_0\to M}\eta(a_0)=0$ and $\eta(a_0)>0$ for $a_0\in(0,M)$. 
\end{definition}

With this definition, we have
\begin{align*}
\zeta=\lp(1-\frac{3M}{r}\rp)^2\lp(1-\mathbbm{1}_{[3M-\gamma_-,3M+\gamma_+]}\rp)\leq B\lp(1-\frac{r_{\rm trap}}{r}\rp)^2
\end{align*}
for any admissible $(\omega,m,\Lambda)$. Our lower bound on $r_{\rm trap}$ becomes sharper if we restrict to fixed $m$ solutions:

\begin{theorem}[Fixed $m$ high and low frequency estimates] \label{thm:ODE-ILED-fixed-m} 
Fix $s\in\{0,\pm 1,\pm 2\}$, $M>0$, $a_0\in[0,M)$ and an admissible azimuthal number $m$. There are positiive parameters $R^*_{\rm large}$, $\omega_{\rm high}$, $\omega_{\rm low}$, $\varepsilon_{\rm width}$, $E$ and $E_W$, depending only on $s$, $M$, $a_0$ and $m$, such that, for each  $|a|\leq a_0$ and each admissible $\omega$ and $\Lambda$, there is a parameter $r_{\rm trap}$, which is either zero or verifies $(1+\sqrt{2})M<r_{\rm trap}\leq B(a_0)$, and a choice of functions $y_{(k)}$, $\hat{y}_{(k)}$, $\tilde{y}$, $f$,  $h_{(k)}$, $\chi_1$, $\chi_2$, $\chi_3$, for each $0\leq k\leq |s|$ (we drop the subscripts when $k=|s|$), depending on the  frequencies only in a region $r^*\in[-R^*_{\rm large},R^*_{\rm large}]$ and satisfying the uniform bounds  \eqref{eq:current-bounds} so that the estimates of Theorem~\ref{thm:ODE-ILED} hold with constant $B$ depending also on $m$.
\end{theorem}

%---------------
%Physical space analysis
%---------------

\section{Summing the frequency space estimates}
\label{sec:physical-space-estimates}

The main purpose of this section is to upgrade Theorems~\ref{thm:quantitative-mode-stab-in-text}, \ref{thm:ODE-ILED} and \ref{thm:ODE-ILED-fixed-m} into  a series of integrated in $r$ energy estimates for the transformed system of Definitions~\ref{def:transformed-system} and \ref{def:alt-transformed-system}.  A first result is 

\begin{proposition}[ILED for hyperboloidal cutoff I] \label{prop:ODE-to-PDE-future-int} Fix $s\in\{0,\pm 1,\pm 2\}$, $M>0$, and $a_0\in[0,M)$. Let $\xi=\xi(\tilde t^*)$ be a smooth cutoff which vanishes to the past of $\Sigma_0$ and is equal to 1 in the future of $\Sigma_1$. For $k=0,\dots, |s|$, let $\swei{\upphi}{s}_k$ denote solutions to the homogeneous transformed system of  Definition~\ref{def:transformed-system} with the property that $\xi\swei{\upphi}{s}_k$ verify the conditions in  Definitions~\ref{def:suf-integrability} and \ref{def:outgoing-bdry-phys-space}. Then, we have the integrated local energy decay estimate
\begin{align*}
&\mathbb{I}^{\rm deg}[\Phi](0,\infty) +\int_{\mc{R}_0} \frac{1}{r^2}\lp(|\Phi'|^2+|T\mc{P}_{\rm trap}[\xi\Phi]|^2+r^{-1}|\mathring{\slashed\nabla}^{[s]}\mc{P}_{\rm trap}[\xi\Phi]|^2+r^{-1}\lp(s^2+r^{-1}\rp)|\Phi|^2\rp)dr  d\sigma d\tau\\
&\quad+\sum_{k=0}^{|s|-1}\sum_{J_1+J_2=0}^{|s|-k-1}\int_{\mc{R}_0} \frac{1}{r^2}\lp(|T^{J_1}Z^{J_2}\upphi_k''|^2+|T^{J_1+1}Z^{J_2}\upphi_k'|^2+r^{-1}|T^{J_1}Z^{J_2}\mathring{\slashed\nabla}^{[s]}\upphi_k'|^2+|T^{J_1}Z^{J_2}\upphi_k'|^2\rp)dr  d\sigma d\tau\\
&\quad+\sum_{k=0}^{|s|-1}\sum_{J_1+J_2=0}^{|s|-k-1}\int_{\mc{R}_0} \frac{1}{r^2}\lp(|T^{J_1+1}Z^{J_2}\upphi_k|^2+r^{-1}|T^{J_1}Z^{J_2}\mathring{\slashed\nabla}^{[s]}\upphi_k|^2+r^{-1}|T^{J_1}Z^{J_2}\upphi_k|^2\rp)dr  d\sigma d\tau\\
&\!\!\quad\leq  B(a_0)\sum_{k=0}^{|s|}\overline{\mathbb{E}}^{|s|-k}[\tilde\upphi_k](0)\,. \numberthis \label{eq:ODE-to-PDE-future-int}
\end{align*}
\end{proposition}

Similar methods allow us to prove more general results. For instance, if we do not assume that $\swei{\upphi}{s}_k$ are sufficiently integrable in time, we can nevertheless establish an integrated-in-$r$ estimate:
 
\begin{proposition}[ILED for hyperboloidal cutoff II] \label{prop:ODE-to-PDE-non-future-int} Fix $s\in\{0,\pm 1,\pm 2\}$, $M>0$, $a_0\in[0,M)$, and some $\tau>0$. Let $\xi=\xi(\tilde t^*)$ be a smooth cutoff which vanishes to the past of $\Sigma_0$ and to the future of $\Sigma_\tau$ and is equal to 1 in the time slab between $\Sigma_1$ and $\Sigma_{\tau-1}$. For $k=0,\dots, |s|$, let $\swei{\upphi}{s}_k$ denote solutions to the homogeneous transformed system of  Definition~\ref{def:transformed-system} with the property that $\xi\swei{\upphi}{s}_k$  verify the conditions in  Definitions~\ref{def:suf-integrability} and \ref{def:outgoing-bdry-phys-space}. Then, for some $p>1$, we have the integrated local energy decay estimate
\begin{align*}
&\sum_{J_1+J_2=0}^{3|s|}\mathbb{I}^{\rm deg}[T^{J_1}Z^{J_2}\Phi](0,\tau)\\
&\quad+\sum_{k=0}^{|s|-1}\sum_{J_1+J_2=0}^{4|s|-k-1}\int_{\mc{R}_{(0,\tau)}} \frac{1}{r^2}\lp(|T^{J_1}Z^{J_2}\upphi_k''|^2+|T^{J_1+1}Z^{J_2}\upphi_k'|^2+r^{-1}|T^{J_1}Z^{J_2}\mathring{\slashed\nabla}^{[s]}\upphi_k'|^2+|T^{J_1}Z^{J_2}\upphi_k'|^2\rp)dr  d\sigma d\tau\\
&\quad+\sum_{k=0}^{|s|-1}\sum_{J_1+J_2=0}^{4|s|-k-1}\int_{\mc{R}_{(0,\tau)}} \frac{1}{r^2}\lp(|T^{J_1+1}Z^{J_2}\upphi_k|^2+r^{-1}|T^{J_1}Z^{J_2}\mathring{\slashed\nabla}^{[s]}\upphi_k|^2+r^{-1}|T^{J_1}Z^{J_2}\upphi_k|^2\rp)dr  d\sigma d\tau\\
&\!\!\quad\leq  B(a_0)\sum_{k=0}^{|s|}\lp(\overline{\mathbb{E}}^{4|s|-k}_p[\tilde\upphi_k](\tau)+\overline{\mathbb{E}}_p^{4|s|-k}[\tilde\upphi_k](0)\rp)\,. \numberthis \label{eq:ODE-to-PDE-non-future-int}
\end{align*}
\end{proposition}

Finally, it will also be useful to consider a case where the homogeneous solutions are sufficiently integrable in time but are not necessarily outgoing:

\begin{proposition}[ILED for radial cutoff] \label{prop:ODE-to-PDE-radial-cutoff-ILED} Fix $s\in\mathbb{Z}$, $M>0$ and $a_0\in[0,M)$. Let $\chi=\chi(r^*)$ be a smooth cutoff function which vanishes for $r^*\geq R^*+1$ and is equal to 1 for $r^*\leq R^*$. For $k=0,\dots,|s|$, let $\swei{\upphi}{s}_{k}$  solutions to the homogeneous transformed system in Definition~\ref{def:transformed-system}, with the property that 
\begin{align*} 
\swei{\upphi}{s}_{k,\cutr}:=\lp(w^{-1}\mc L\rp)^{-1}\lp(\chi_R \swei{\upphi}{s}_0\rp)\,,\quad k=0,\dots,|s|\,, \qquad \swei{\Phi}{s}_{\cutr}\equiv\swei{\upphi}{s}_{|s|,\cutr}\,,
\end{align*}
are sufficiently integrable and outgoing solutions to the inhomogeneous transformed system in Definition~\ref{def:transformed-system}, in the sense of Definitions~\ref{def:outgoing-bdry-phys-space} and \ref{def:suf-integrability}. Then, we have the integrated local energy decay estimate 
\begin{align*}
&\int_{\mc{R}} \lp(|\Phi_\cutr'|^2+|T\mc{P}_{\rm trap}[\Phi_\cutr]|^2+|\mathring{\slashed\nabla}^{[s]}\mc{P}_{\rm trap}[\Phi_\cutr]|^2+|T^{\frac12}[\Phi_{\cutr}]|^2+|\Phi_\cutr|^2 \rp)dr  d\sigma d\tau \\
&\qquad + \sum_{k=0}^{|s|-1}\lp(\mathbb{I}[T^{J_1}Z^{J_2}\upphi_{k,\cutr}](-\infty,\infty)+\mathbb{I}[T^{J_1}Z^{J_2}\upphi_{k,\cutr}'](-\infty,\infty)+\int_{\mc R}|Z^{\frac32}\widetilde{\mc{P}}_{\rm trap}[\upphi_{k,\cutr}]|^2 drd\sigma d\tau\rp) \\
&\quad\leq  B(a_0, R^*)\sum_{k=0}^{|s|}\mathbb{I}^{|s|-k}[\upphi_{k,\cutr}\mathbbm{1}_{\{R^*\leq |r^*|\leq R^*+1\}}](-\infty,\infty)\,. \numberthis\label{eq:ODE-to-PDE-radial-cutoff-ILED}
\end{align*}
\end{proposition}

\subsection{An inhomogeneous system from hyperboloidal cutoffs}

Let $0\leq \tau_0\leq \tau_1\leq \dots \tau_N\leq \tau_{N+1}\leq \infty$, with $|\tau_{2i+1}-\tau_{2i}|\leq B$ for $i=0,\dots,N-1$. Take $\xi=\xi(\tilde{t}^*)$, where $\tilde t^*$ is the coordinate identifying leaves of the hyperboloidal folliation introduced in Section~\ref{sec:hyp-folliation}, to be a smooth cutoff function satisfying
\begin{align*}
\supp(\xi)\subset \mc R_{(\tau_0,\tau_{N+1})}\,,\qquad \supp(\nabla\xi)\subset \uplus_{i=0}^{N-1} \mc R_{(\tau_{2i},\tau_{2i+1})}\,.
\end{align*}
This section is concerned with the study of a inhomogeneous version of the transformed system of Definitions~\ref{def:transformed-system} and \ref{def:alt-transformed-system} where the inhomogeneity is caused by the application of the cutoff $\xi=\xi(\tilde{t}^*)$; see Section~\ref{sec:weird-cutoff-system} for the precise definition of the transformed system.  In Sections~\ref{sec:current-errors}, \ref{sec:bulk-terms} and \ref{sec:bdry-terms} we deal with, respectively, errors in Theorem~\ref{thm:ODE-ILED} arising from interaction of the cutoffs with the application of separated current templates in high and low frequency regimes, bulk terms in an intermediate frequency and possible boundary terms from a bounded frequency regime. 

\subsubsection{Defining cutoff transformed variables}
\label{sec:weird-cutoff-system}

In this section, we introduce the relevant inhomogeneous version of the transformed system of Definition~\ref{def:transformed-system}, and prove some elementary estimates on it. 

We denote by $\swei{\upphi}{s}_k$, for $k=0,\dots,|s|$, solutions to the \textit{homogeneous} transformed system of Definition~\ref{def:transformed-system}. For some $k_0\in\{0,\dots, |s|\}$, let us assume that $\xi$ is chosen to ensure
\begin{align*}
\swei{\upphi}{s}_{k_0,\cutt}&:=\xi\swei{\upphi}{s}_{k_0}\,,\\
\mathfrak{h}^{[s]}_{k_0}&:=  -\lp(\mc L \xi \underline{\mc L}\swei{\upphi}{s}_{k_0}+\underline{\mc L}\mc L\xi\swei{\upphi}{s}_{k_0}\rp)\mathbbm{1}_{k_0\leq |s|}\,,\\
\mathfrak{H}^{[s]}_{k_0}&:=\lp([\swei{\mathfrak{R}}{s},\xi]-\mathbbm{1}_{\{k_0\neq |s|\}}\underline{\mc L}\xi \mc L\rp)\swei{\upphi}{s}_{k_0} +\mathbbm{1}_{\{k_0\neq |s|\}} w \underline{\mc L}\xi \swei{\upphi}{s}_{k_0+1}\nonumber\\
&\qquad +\sum_{k=0}^{k_0-1}w\lp(ac_{s,k_0,k}^{\rm id}+ac_{s,k_0,k}^{Z}Z\rp)(\xi\swei{\upphi}{s}_k-\swei{\upphi}{s}_{k,\cutt})\,.\numberthis \label{eq:weird-cutoffs-top-inhom-formula}
\end{align*}
are outgoing and sufficiently integrable, in the sense of Definition~\ref{def:outgoing-bdry-phys-space} and \ref{def:suf-integrability}. If $k_0<|s|$, set
\begin{align}
\swei{\upphi}{s}_{k_0+1,\cutt}:=\xi\swei{\upphi}{s}_{k_0+1}\,. \label{eq:def-top-level-weird-cutoff}
\end{align}
We let $\smlambdak{\uppsi}{s}{k_0}$ and  $\smlambdak{\uppsi}{s}{k_0+1}$ be defined from, respectively, $\xi\swei{\upphi}{s}_{k_0}$ and $\xi\swei{\upphi}{s}_{k_0+1}$; $\smlambdak{\mathfrak G}{s}{k_0}$ and $\smlambdak{\mathfrak g}{s}{k_0}$  be defined from, respectively, $\swei{\mathfrak{H}}{s}_{k_0}$ and  $\swei{\mathfrak{h}}{s}_{k_0}$ as in \eqref{eq:def-sml-uppsi-G}. 

It remains to complete the transformed systems of Definitions~\ref{def:transformed-system} and \ref{def:alt-transformed-system} by defining $\swei{\upphi}{s}_{j,\cutt}$ for $0\leq k<k_0$. For  $(\omega,m,\Lambda)\in\mc{F}_{\rm admiss}$, set
\begin{align*}
\mu(r^*) =\int_0^{r^*}i\lp(\omega-\frac{am}{x^2+a^2}\rp)dx^*
\end{align*}
so that $\mc{L}\mu=0$. With $\smlambdak{\uppsi}{s}{k_0}$ as above, we define for $k=0,\dots, k_0-1$
\begin{align}
\begin{split}
\uppsi_{(k),\,ml}^{[s],\, a\omega}(r^*)&:=e^{-\mu(r^*)}\int_{-\infty}^{r^*} e^{\mu(x^*)} w \uppsi_{(k+1),\,ml}^{[s],\, a\omega}dr^*\,, \quad s< 0\,,\\
\uppsi_{(k),\,ml}^{[s],\, a\omega}(r^*)&:=e^{\mu(r^*)}\int_{r^*}^\infty e^{-\mu(x^*)} w \uppsi_{(k+1),\,ml}^{[s],\, a\omega}dr^*\,, \quad s> 0\,,
\end{split}\label{eq:uppsi-L-inversion-freq-space}
\\
\begin{split}
\mathfrak{G}_{(k),\,ml}^{[s],\, a\omega}(r^*)&:=w e^{-\mu(r^*)}\int_{-\infty}^{r^*} e^{\mu(x^*)} \mathfrak{G}_{(k+1),\,ml}^{[s],\, a\omega}dr^*\,,\quad s< 0\,,\\
\mathfrak{G}_{(k),\,ml}^{[s],\, a\omega}(r^*)&:=w e^{\mu(r^*)}\int_{r^*}^\infty e^{-\mu(x^*)} \mathfrak{G}_{(k+1),\,ml}^{[s],\, a\omega}dr^*\,,\quad s> 0\,.
\end{split}\label{eq:RW-cutoff-freq-inhom-k}
\end{align}
In physical space, we then let
\begin{align}
\upphi_{k,\cutt}^{[s]}&:=\int_{\mathbb{S}^2}\int_{-\infty}^\infty e^{-i\omega t}\uppsi_{(k),\,ml}^{[s],\, a\omega}(r)\sml{S}{s}(\theta)e^{im\phi} d\omega d\sigma\,,\qquad k=0,\dots k_0-1\,, \label{eq:weird-psi-cut}\\
\mathfrak{H}_{k}^{[s]}&:=\frac{1}{\sqrt{2\pi}}\int_{\mathbb{S}^2}\int_{-\infty}^\infty e^{-i\omega t}\mathfrak{G}_{(k),\,ml}^{[s],\, a\omega}(r)\sml{S}{s}(\theta)e^{im\phi} d\omega d\sigma\,,\qquad k=0,\dots k_0\,. \label{eq:weird-G-cut-k}\,,
\end{align}
This system is suitable for frequency analysis:

\begin{lemma}[Backwards cutoffs]\label{lemma:weird-inverses} Fix $s\in\mathbb{Z}$ and $k_0\in\{0,\dots |s|\}$, and let $\swei{\upphi}{s}_{k,\cutt}$, $\swei{\mathfrak H}{s}_k$ and $\swei{\mathfrak h}{s}_k$, for $k=0\,\dots k_0+1$ or, if $k_0=|s|$, $k=0\,\dots |s|$, be as above. Then, if $k_0<|s|$, $\swei{\upphi}{s}_{k,\cutt}$ are outgoing and sufficiently integrable solutions to the transformed system of Definition~\ref{def:alt-transformed-system} with inhomogeneities $\mathfrak{H}_{k}^{[s]}$ and $\mathfrak{h}_{k_0}^{[s]}$; if $k_0=|s|$, $\swei{\upphi}{s}_{k,\cutt}$ are outgoing and sufficiently integrable solutions to the transformed system of Definition~\ref{def:alt-transformed-system} with inhomogeneities $\mathfrak{H}_{k}^{[s]}$. Furthermore, for $k=0,\dots, k_0-1$, $\swei{\upphi}{s}_{k,\cutt}$ satisfy the additional transport equation
\begin{align}
\mc L(\xi\upphi_k^{[s]}-\upphi_{k,\cutt}^{[s]})=\mc L\xi\upphi_k^{[s]}+w(\xi\upphi_{k+1}^{[s]}-\upphi_{k+1,\cutt}^{[s]})\,, \qquad k=0,\dots k_0-1\,. \label{eq:transport-weird-cutoff}
\end{align}
\end{lemma}

\begin{remark}[Decay of the inhomogeneities at $r=r_+$] \label{rmk:decay-inhoms-weird-cutoff} Note that, in the system we have constructed, the inhomogeneities $\swei{\mathfrak{H}}{s}_k$ for $s>0$ do not benefit from improved decay as $r\to r_+$; indeed, we expect only $\smlk{\mathfrak{G}}{s}{k}=O(w)$ as $r\to r_+$. 
\end{remark}

\begin{proof}[Proof of Lemma~\ref{lemma:weird-inverses}] The lemma follows by induction on $k$: since the outgoing and sufficient integrability conditions hold for $k_0$, we seek to show that they hold for $k<k_0$. Thus, assume that, for some $k=0,\dots, k_0-1$, $\smlk{\uppsi}{s}{k+1}$ solves an inhomogeneous radial ODE \eqref{eq:transformed-k-separated} with inhomogeneity $\smlk{\mathfrak G}{s}{k+1}$ and has outgoing boundary conditions; and assume that the corresponding $\swei{\upphi}{s}_{k+1,\cutt}$ is a sufficiently integrable outgoing solution to the inhomogeneous equation \eqref{eq:transformed-k} with inhomogeneity $\swei{\mathfrak H}{s}_{k+1}$ such that $\swei{\mathfrak{H}}{s}_{k+1}=0$ and $\swei{\upphi}{s}_{k+1,\cutt}=\xi\swei{\upphi}{s}_{k+1}$ to the past of $\Sigma_{\tau_0}$.

By the strong decay of $\smlk{\uppsi}{s}{k+1}$ as $r^*\to (\sign s)\infty$, the integrands in \eqref{eq:uppsi-L-inversion-freq-space} are well-defined for each $r^*\in(-\infty,\infty)$. It is also easy to see that the transport equation \eqref{eq:transformed-transport-separated} holds
\begin{align*}
\mc{L}\uppsi_{(k),\,ml}^{[s],\, a\omega}=\mc{L}\mu \uppsi_{(k),\,ml}^{[s],\, a\omega}+ w\uppsi_{(k+1),\,ml}^{[s],\, a\omega} =w\uppsi_{(k+1),\,ml}^{[s],\, a\omega}\,,
\end{align*}
in frequency space for each $r^*\in(-\infty,\infty)$; and similarly for \eqref{eq:transformed-transport-inhom-separated}.  Since 
\begin{align*}
\mu(r)&=  \int^{r^*}_{0}i\lp(\omega-\frac{am}{x^2+a^2}\rp)dx^* 
= i\omega r^*+o(r^*)\,,\qquad \text{as~}r^*\to \infty\,,\\
\mu(r)&= -\int^{0}_{r^*}i\lp(\omega-\frac{am}{x^2+a^2}\rp)dx^* 
= -i(\omega-m\upomega_+) r^*+o(r^*)\,,\qquad \text{as~}r^*\to -\infty\,,
\end{align*}
we deduce that for $s\geq 0$ we have
\begin{align*}
e^{-i\omega r^*}\uppsi_{(k),\,ml}^{[s],\, a\omega}(r^*)&=e^{-i\omega r^*}e^{\mu(r)}\int_{r^*}^{\infty} e^{-\mu(x)}w \uppsi_{(k+1),\,ml}^{[s],\, a\omega} dx^* \\
&= e^{-2i\omega r^*}\int_{r^*}^{\infty} e^{2i\omega x^*} wx^{-1}\lp(e^{-i\omega x^*}x\uppsi_{(k+1),\,ml}^{[s],\, a\omega}\rp)dx^*  = O(r^{-1})\,,
\end{align*}
as $r\to \infty$, and for $r\to r_+$,
\begin{align*}
e^{-i(\omega-m\upomega_+) r^*}\uppsi_{(k),\,ml}^{[s],\, a\omega}(r^*)
&=\int_{r^*}^{\infty} w \lp(e^{i(\omega-m\upomega_+) x^*}\uppsi_{(k+1),\,ml}^{[s],\, a\omega}\rp) dx^* \\
& = \int_{r}^{\infty} (x^2+a^2)^{-1}dx \lp(e^{i(\omega-m\upomega_+) r^*}\uppsi_{(k+1),\,ml}^{[s],\, a\omega}\rp)\Big|_{r^*=-\infty}  +O(r-r_+)\,.
\end{align*}
Similar formulas hold for the case $s<0$, \textit{mutatis mutandis}. These asymptotic formulas are consistent with outgoing boundary conditions, see Definition~\ref{def:outgoing-bdry-freq-space}. In fact, if we assume more on $\uppsi_{(k+1),\,ml}^{[s],\, a\omega}$, as expected for instance if $\mathfrak G_{(k+1),\,ml}^{[s],\, a\omega}$ has stronger $r^*$ decay, then for $s>0$
\begin{align*}
e^{-i\omega r^*}r^{2(|s|-k)}\uppsi_{(k),\,ml}^{[s],\, a\omega}(r^*)&=e^{-i\omega r^*}r^{2(|s|-k)}e^{\mu(r)}\int_{r^*}^{\infty} e^{-\mu(x)}w \uppsi_{(k+1),\,ml}^{[s],\, a\omega} dx^* \\
&= e^{-2i\omega r^*}r^{2(|s|-k)}\int_{r^*}^{\infty} e^{2i\omega x^*} wx^{-2(|s|-k-1)} \lp(e^{-i\omega x^*}x^{2(|s|-k-1)}\uppsi_{(k+1),\,ml}^{[s],\, a\omega}\rp)dx^* \\
& = \frac{1}{(2i\omega)}\lp(e^{-i\omega r^*}r^{2(|s|-k-1)}\uppsi_{(k+1),\,ml}^{[s],\, a\omega}\rp)\Big|_{r^*=\infty}  +O(r^{-1})\,,
\end{align*}
which is consistent with the boundary term relations of Lemma~\ref{lemma:bdry-term-relations}. Turning to the $s<0$ case, we find that we can also recover the boundary term relations in Lemma~\ref{lemma:bdry-term-relations} at  $r=r_+$. Similar argument shows that $\mathfrak{G}_{(k),\,ml}^{[s],\, a\omega}/w$ is a smooth function of $r$ which is regular at $\mc H^+$, and such that 
\begin{align*}
\frac{\mathfrak{G}_{(k),\,ml}^{[s],\, a\omega}}{w}&= O(r^{-(|s|-k)})\,,\quad\text{as $r\to \infty$ if $s>0$}\,,\\
\frac{\mathfrak{G}_{(k),\,ml}^{[s],\, a\omega}}{w}&= O(\Delta^{|s|-k})\,,\quad\text{as $r\to r_+$ if $s<0$}\,.
\end{align*}

Now we turn to the physical space definitions. By the above asymptotics for $\mathfrak{G}_{(k),\,ml}^{[s],\, a\omega}$, we find easily that $\mathfrak{H}_k$ satisfies \eqref{eq:integrability-Hk}.  To check that $\swei{\upphi}{s}_{k,\cutt}$ is sufficiently integrable, we need to work a bit more; for instance, if $s\geq 0$,
\begin{align*}
\int_{-R^*}^{R^*}|\smlk{\uppsi}{s}{k}|^2 dr^* \leq \int_{-R^*}^{R^*}\int_{r^*}^\infty|\smlk{\uppsi}{s}{k+1}|^2(x^*)w(x^*)dx^* dr^*\leq B(R^*) \int_{-R^*}^\infty w|\smlk{\uppsi}{s}{k+1}|^2dr^*
\end{align*} 
and by the sufficient integrability of $\swei{\upphi}{s}_{k+1,\cutt}$ and the fact that it is outgoing, in particular by the estimates in Lemma~\ref{lemma:suf-integrability}, we deduce that $\swei{\upphi}{s}_{k,\cutt}$ is also sufficiently integrable and thus well-defined. The outgoing property for $\swei{\upphi}{s}_{k,\cutt}$ is ensured by the outgoing boundary conditions of $\smlk{\uppsi}{s}{k}$ in frequency space.

The construction of $\swei{\upphi}{s}_k$ and  $\mathfrak{H}_{k}^{[s]}$  also clearly ensures that the physical transport relations \eqref{eq:transformed-transport} and \eqref{eq:transformed-transport-inhom} are verified. It is then easy to deduce that \eqref{eq:transport-weird-cutoff} holds. Thus, in the past of $\Sigma_{\tau_0}$, we compute
\begin{align*}
\mc{L}\lp(\xi\swei{\upphi}{s}_{k}-\swei{\upphi}{s}_{k,\cutt}\rp)=\mc{L}\xi\swei{\upphi}{s}_{k}=0\,.
\end{align*}
Integrating from $\mc H^-$ or $\mc I^-$, for $s<0$ and $s>0$ respectively, and using the (physical space) outgoing condition, we deduce that $\swei{\upphi}{s}_{k,\cutt}=\xi\swei{\upphi}{s}_{k}$ in the past of $\Sigma_{\tau_0}$.   Similarly, the difference between the inhomogeneity $\mathfrak{H}_{k}^{[s]}$ and one obtained by deriving the equation for $\xi\swei{\upphi}{s}_{k}$ must also vanish, and so $\mathfrak{H}_{k}^{[s]}=0$ in the past of  $\Sigma_{\tau_0}$.  
\end{proof}

Let us start by stating a simple lemma which compares the backwards-integrated $\swei{{\upphi}}{s}_{k,\cutt}$ cutoff variables introduced in Lemma~\ref{lemma:weird-inverses} with that obtained by directly multiplying the homogeneous transformed variables by the cutoff: 

\begin{lemma}[Comparison between backwards and direct cutoffs]\label{lemma:weird-cutoffs-diff-fake-real} Fix some $s\in\mathbb{Z}$ and $k_0\in\{0,\dots,|s|\}$ as in Lemma~\ref{lemma:weird-inverses}. Then, for any $k<k_0$ and $\tau\geq \tau_0$, we have the estimates
\begin{align*}
&\sum_{X\in\{\mathrm{id},T,Z,\underline{\mc L}\}}\lp(\int_{\mc R_{(\tau_0,\infty)}}w|X(\xi\upphi_k-\upphi_{k,\cutt})|^2dr^*d\tau d\sigma +\int_{\Sigma_\tau}|X(\xi\upphi_k-\upphi_{k,\cutt})|^2dr d\sigma\rp)\\
&\qquad +\sum_{X\in\{\mathrm{id},T,Z,\underline{\mc L}\}}\lp\{\begin{array}{lr}\int_{\mc H^+_{(\tau_0,\infty)}}|X(\xi\upphi_k-\upphi_{k,\cutt})|^2d\sigma d\tau, &s>0\\
\int_{\mc I^+_{(\tau_0,\infty)}}|X(\xi\upphi_k-\upphi_{k,\cutt})|^2d\sigma d\tau, &s<0
\end{array}\rp\}\leq B\sum_{i=0}^{N-1}\sum_{j=k}^{k_0-1}\mathbb{E}[\tilde\upphi_j](\tau_{2i})\,.\numberthis\label{eq:weird-cutoffs-diff-fake-real-bulk-commuted-final}
\end{align*}
\end{lemma}

\begin{proof} The result follows from estimates for the transport equation \eqref{eq:transport-weird-cutoff} in Lemma~\ref{lemma:weird-inverses}. We stress that we take $k<k_0$ strictly.

\medskip
\noindent \textit{Step 1: $0$th order estimates.} Let $u$ be defined through $L(u)=0$ and normalized so that $u=-\infty$ on $\mc I^-$ and $u=+\infty$ on $\mc H^+$.  From the transport equation \eqref{eq:transport-weird-cutoff} for $s>0$, we obtain
\begin{align*}
&\int_{\mc H^+_{(\tau_{0},\infty)}}|\xi\upphi_k-\upphi_{k,\cutt}|^2d\sigma d\tau\\
&\quad =\int_{\mathbb{S}^2}\int_{\tau_0}^\infty |\xi\upphi_k-\upphi_{k,\cutt}|^2(u=\infty) d\tau d\sigma = \int_{\mathbb{S}^2}\int_{\tau_0}^\infty \lp|\int_{u_0}^{\infty}\lp[\uL\xi \upphi_k + w(\xi \upphi_{k+1}-\upphi_{k+1,\cutt})\rp]du\rp|^2 d\tau d\sigma\\
&\quad \leq  \int_{\mathbb{S}^2}\int_{\tau_0}^\infty\int_{u_0}^{\infty}\lp[\Big|\frac{r^2}{\Delta}\uL\xi \upphi_k\Big|^2 + r^{-4}|\xi \upphi_{k+1}-\upphi_{k+1,\cutt}|^2\rp] \frac{\Delta}{r^2}du d\tau d\sigma\\
&\quad  \leq \int_{\mc R_{(\tau_0,\infty)}}|\uL\xi r^2/\Delta|^2|\tilde \upphi_k|^2r^{-2} drd\sigma d\tau + \int_{\mc R_{(\tau_0,\infty)}} \frac{w}{r^{2}}|\xi \upphi_{k+1}-\upphi_{k+1,\cutt}|^2 dr^* d\tau d\sigma\\
&\quad \leq B\sum_{i=0}^{N-1}\int_{\mc R_{(\tau_{2i},\tau_{2i+1})}}\mathbbm{1}_{\supp(\nabla\xi)}|\tilde \upphi_k|^2r^{-2} drd\sigma d\tau + B\int_{\mc R_{(\tau_0,\infty)}} \frac{w}{r^{2}}|\xi \upphi_{k+1}-\upphi_{k+1,\cutt}|^2 dr^* d\tau d\sigma \,. \numberthis\label{eq:weird-cutoffs-diff-fake-real-hor}
\end{align*}
In the above, we use that integrand on the left hand side vanishes for sufficiently negative $u\leq u_0$, by Lemma~\ref{lemma:weird-inverses}, as well as a Jensen inequality to exchange the absolute value and the integral in $u$. A similar procedure for $s<0$ at $\mc{I}^+$ yields
\begin{align*}
&\int_{\mc I^+_{(\tau_0,\infty)}}|\xi\upphi_k-\upphi_{k,\cutt}|^2d\sigma d\tau\\
&\quad\leq B\sum_{i=0}^{N-1}\int_{\mc R_{(\tau_{2i},\tau_{2i+1})}}w|\tilde\upphi_k|^2 dr^* d\tau d\sigma
+B\int_{\mc R_{(\tau_0,\infty)}}w|\xi\upphi_{k+1}-\upphi_{k+1,\cutt}|^2 dr^* d\tau d\sigma\,. \numberthis\label{eq:weird-cutoffs-diff-fake-real-infty}
\end{align*}
Likewise, by the same procedure we have, for any $\tau\geq \tau_0$,
\begin{align*}
&\int_{\Sigma_\tau}|\xi\upphi_k-\upphi_{k,\cutt}|^2d\sigma dr\\
&\quad\leq B\sum_{i=0}^{N-1}\int_{\mc R_{(\tau_{2i},\tau_{2i+1})}}w|\tilde\upphi_k|^2 dr^* d\tau d\sigma
+B\int_{\mc R_{(\tau_0,\tau)}}w|\xi\upphi_{k+1}-\upphi_{k+1,\cutt}|^2 dr^* d\tau d\sigma\,. \numberthis\label{eq:weird-cutoffs-diff-fake-Sigma}
\end{align*}

Now let us note that, for a smooth $c(r)$, \eqref{eq:transport-weird-cutoff} implies the identity
\begin{align*}
\mc L\lp(c(r)|\xi\upphi_k-\upphi_{k,\cutt}|^2\rp)+\sign s c'(r)|\xi\upphi_k-\upphi_{k,\cutt}|^2 = 2c(r)\lp(\mc L\xi \upphi_k +w(\xi\upphi_{k+1}-\upphi_{k+1,\cutt})\rp)(\xi\upphi_k-\upphi_{k,\cutt})\,.
\end{align*}
Choosing $c=-(r-r_+)/r$ if $s<0$ and $c=-r^{-1}$ if $s>0$, and recalling that $\xi\upphi_k-\upphi_{k,\cutt}=0$ to the past of $\Sigma_{\tau_0}$, we arrive at
\begin{align*}
&\int_{\mc R_{(\tau_0,\infty)}}w|\xi\upphi_k-\upphi_{k,\cutt}|^2dr^*d\tau d\sigma\\
&\quad \leq B
\lp\{\begin{array}{lr}
\int_{\mc H^+_{(\tau_0,\infty)}}|\xi\upphi_k-\upphi_{k,\cutt}|^2d\sigma d\tau, &s>0\\
\int_{\mc I^+_{(\tau_0,\infty)}}|\xi\upphi_k-\upphi_{k,\cutt}|^2d\sigma d\tau, &s<0
\end{array}\rp\}  + \limsup_{\tau\to \infty}\int_{\Sigma_\tau}|\xi\upphi_k-\upphi_{k,\cutt}|^2d\sigma dr \\
&\quad\qquad+B\sum_{i=0}^{N-1}\int_{\mc R_{(\tau_{2i},\tau_{2i+1})}}w|\tilde \upphi_k|^2 dr^* d\tau d\sigma
+B\int_{\mc R_{(\tau_0,\infty)}}w|\xi\upphi_{k+1}-\upphi_{k+1,\cutt}|^2 dr^* d\tau d\sigma\\
&\quad \leq B\sum_{i=0}^{N-1}\int_{\mc R_{(\tau_{2i},\tau_{2i+1})}}w|\tilde \upphi_k|^2 dr^* d\tau d\sigma+ B\int_{\mc R_{(\tau_0,\infty)}}w|\xi\upphi_{k+1}-\upphi_{k+1,\cutt}|^2 dr^* d\tau d\sigma\\
&\quad \leq B\sum_{j=k}^{k_0-1}\sum_{i=0}^{N-1}\int_{\mc R_{(\tau_{2i},\tau_{2i+1})}}w|\tilde \upphi_k|^2 dr^* d\tau d\sigma\,, \numberthis\label{eq:weird-cutoffs-diff-fake-real-bulk}
\end{align*}
where we have invoked estimates \eqref{eq:weird-cutoffs-diff-fake-real-hor},  \eqref{eq:weird-cutoffs-diff-fake-real-infty} and \eqref{eq:weird-cutoffs-diff-fake-Sigma}, and applied induction in $k\leq k_0-1$.  Appealing to the finite in time estimates of Proposition~\ref{prop:finite-in-time-first-order} concludes the case $X=\mathrm{id}$.

\medskip
\noindent \textit{Step 2: $1$st order estimates.} Since $[Z,\mc{L}]=0$ and $[T,\mc{L}]=0$, estimates  \eqref{eq:weird-cutoffs-diff-fake-real-hor}, \eqref{eq:weird-cutoffs-diff-fake-real-infty} and \eqref{eq:weird-cutoffs-diff-fake-real-bulk} hold replacing both $\xi\upphi_k-\upphi_{k,\cutt}$ by $X(\xi\upphi_k-\upphi_{k,\cutt})$ and $\tilde{\upphi}_k$ by $X\tilde\upphi_k$ for $X=Z,T$. To derive $\underline{\mc L}$-commuted estimates, we begin by commuting \eqref{eq:transport-weird-cutoff} with $\underline{\mc L}$,
\begin{align*}
\mc{L}(\underline{\mc L}(\xi\upphi_k-\upphi_{k,\cutt}))&=  L\uL\xi\upphi_k+\mc L\xi \underline{\mc L}\upphi_k+\frac{4arw}{r^2+a^2}\sign s Z(\xi\upphi_k-\upphi_{k,\cutt})\\
&\qquad+\sign s w'(\xi\upphi_{k+1}-\upphi_{k+1,\cutt})+w\mc L(\xi\upphi_{k+1}-\upphi_{k+1,\cutt})\,,
\end{align*}
and then repeat the arguments of the previous step \textit{mutatis mutandis} to conclude the proof.
\end{proof}

Very similar methods to the ones used in Lemma~\ref{lemma:weird-cutoffs-diff-fake-real} can be used to obtain some preliminary estimates on the inhomogeneities of the transformed system introduced in Section~\ref{sec:weird-cutoff-system}:

\begin{lemma}[Estimates on hyperboloidal cutoff inhomogeneities]\label{lemma:weird-cutoffs-inhoms-rough-estimates} Fix some $s\in\mathbb{Z}$ and $k_0\in\{0,\dots,|s|\}$ as in Lemma~\ref{lemma:weird-inverses}. Then, for any $k<k_0$ and $\tau\geq \tau_0$, we have the following estimates. If $k_0<|s|$, then 
\begin{align*}
&\sum_{X\in\{\mathrm{id},T,Z,\p_{r^*}\}}\lp(
\lp\{\begin{array}{lr}
\int_{\mc H^+_{(\tau_0,\infty)}}\lp|w^{-1}X\mathfrak{H}_k\rp|^2 d\tau d\sigma, &s>0\\
\int_{\mc I^+_{(\tau_0,\infty)}}\lp|w^{-1}X\mathfrak{H}_k\rp|^2 d\tau d\sigma, &s<0
\end{array}\rp\}
+\int_{\mc R_{(\tau_0,\infty)}}w\lp|\frac{X\mathfrak{H}_k}{w}\rp|^2 dr^* d\sigma d\tau\rp)\\
&\quad \leq B\sum_{i=0}^{N-1}\lp(\sum_{j=0}^{k_0}\mathbb{E}^1[\tilde\upphi_{j}](\tau_{2i})+\mathbb{E}[\tilde\upphi_{k_0+1}](\tau_{2i}) \rp)\,, \numberthis \label{eq:weird-cutoffs-inhoms-rough-estimates-lower}
\end{align*}
If $k_0=|s|$, then
\begin{align*}
&\sum_{k=0}^{|s|-1}
\lp\{\begin{array}{lr}
\int_{\mc H^+_{(\tau_0,\infty)}}\lp|w^{-1}\mathfrak{H}_k\rp|^2 d\tau d\sigma, &s>0\\
\int_{\mc I^+_{(\tau_0,\infty)}}\lp|w^{-1}\mathfrak{H}_k\rp|^2 d\tau d\sigma, &s<0
\end{array}\rp\}
+\int_{\mc R_{(\tau_0,\infty)}}w\lp|\frac{\mathfrak{H}_k}{w}\rp|^2 dr^* d\sigma d\tau\\
&\quad \leq B\sum_{i=0}^{N-1}\lp(\sum_{j=0}^{|s|}\mathbb{E}[\tilde\upphi_{j}](\tau_{2i})+
\lp\{\begin{array}{lr}
\int_{\mc H^+_{(\tau_{2i},\tau_{2i+1})}}|\Phi|^2 d\tau d\sigma, &s>0\\
\int_{\mc I^+_{(\tau_{2i},\tau_{2i+1})}}|\Phi|^2 d\tau d\sigma, &s<0
\end{array}\rp\}\rp)\,. \numberthis \label{eq:weird-cutoffs-inhoms-rough-estimates-upper}
\end{align*}
and
\begin{align*}
\int_{\mc R_{(\tau_0,\tau_N)}} \lp(|\mathfrak{H}_{|s|}|^2+w\lp|\frac{\mathfrak{H}_{|s|}}{w}-L\xi\uL \Phi-\uL\xi L \Phi\rp|^2\rp) dr^*d\tau d\sigma 
&\leq B\sum_{i=0}^{N-1}\sum_{j=0}^{|s|}\mathbb{E}[\tilde\upphi_{j}](\tau_{2i})\,. \numberthis\label{eq:weird-cutoffs-inhoms-rough-estimates-upper-improved}
\end{align*}
\end{lemma} 

\begin{proof} As before, we first study the case of 0th order estimates for the inhomogeneities and only then look at first derivatives.

\medskip
\noindent \textit{Step 1: 0th order estimates.} Let us begin by noting that, from \eqref{eq:weird-cutoffs-top-inhom-formula}, we easily derive an estimate for $\mathfrak{H}_{k_0}$,
\begin{align*}
\int_{\mc R_{(\tau_0,\infty)}} |\mathfrak{H}_{k_0}|^2 dr^*d\tau d\sigma 
&\leq B\int_{\mc R_{(\tau_0,\infty)}} \lp[|[\mathfrak{R}_{k_0},\xi]\upphi_{k_0}|^2+w^2\sum_{j=0}^{k_0-1}\lp(|\xi\upphi_j-\upphi_{j,\cutt}|^2 + |Z(\xi\upphi_j-\upphi_{j,\cutt})|^2\rp)\rp]\!dr^*d\tau d\sigma\\
&\leq B\sum_{i=0}^{N-1}\mathbb{E}[\tilde\upphi_{k_0}](\tau_{2i})+B\sum_{j=0}^{k_0-1}\sum_{i=0}^{N-1}\mathbb{E}[\tilde\upphi_{j}](\tau_{2i})\leq B\sum_{j=0}^{k_0}\sum_{i=0}^{N-1}\mathbb{E}[\tilde\upphi_{j}](\tau_{2i})\,, 
\end{align*}
where we have concluded using the results of the previous Lemma~\ref{lemma:weird-cutoffs-diff-fake-real}. If $k_0<|s|$ then we may use the alternative form for $[\mathfrak{R}_{k_0},\xi]\upphi_{k_0}$ given in \eqref{eq:weird-cutoffs-top-inhom-formula} to improve the $r$-weights on the left hand side  of the above estimate:
\begin{align*}
\int_{\mc R_{(\tau_0,\infty)}} w\lp|\frac{\mathfrak{H}_{k_0}}{w}\rp|^2 dr^*d\tau d\sigma 
&\leq B\sum_{i=0}^{N-1}\lp(\sum_{j=0}^{k_0}\mathbb{E}[\tilde\upphi_{j}](\tau_{2i})+\int_{\Sigma_{\tau_{2i}}}r^{-2}|\tilde\upphi_{k_0+1}|^2dr d\sigma \rp)\,. \numberthis\label{eq:weird-cutoffs-top-inhom-estimate-lower}
\end{align*}
Similarly, we may also improve the estimate for $k_0=|s|$ if we remove the term $\underline{L}\xi L \Phi$ and $L\xi \uL \Phi$, obtaining \eqref{eq:weird-cutoffs-inhoms-rough-estimates-upper-improved}.

Turning to the lower level inhomogeneities, i.e.\ to $\mathfrak{H}_k$ with $k<k_0$, we will invoke the transport relation \eqref{eq:transport-weird-cutoff} to mimic the proof of \eqref{eq:weird-cutoffs-diff-fake-real-hor}. We will always need to be careful whenever a term involving $\mathfrak{H}$ appears. For instance, by following the steps of the proof of \eqref{eq:weird-cutoffs-diff-fake-real-hor} in the previous lemma, we easily see that for $s>0$
\begin{align*}
&\int_{\mc H^+_{(\tau_0,\infty)}}\lp|\frac{\mathfrak{H}_k}{w}\rp|^2 d\tau d\sigma \leq B \int_{\mc R_{(\tau_0,\infty)}}\frac{r^2}{\Delta}\lp|\mathfrak H_{k+1}\rp|^2 du d\tau d\sigma \,,
\end{align*}
if $k+1<|s|$, but if $k+1=|s|=k_0$ we must revisit the proof more carefully if we wish to have a finite quantity on the right hand side. With the same coordinate $u$ as in the proof of Lemma~\ref{lemma:weird-cutoffs-diff-fake-real}, chosen to satisfy $\uL u=0$, $u=-\infty$ at $\mc I^+$ and $u=\infty$ at $\mc H^+$, we have
\begin{align*}
&\int_{\mc H^+_{(\tau_0,\infty)}}\lp|\frac{\mathfrak{H}_k}{w}\rp|^2(u=\infty) d\tau d\sigma \leq \int_{\tau_0}^\infty \Big|\int_{u_0}^\infty \frac{\mathfrak{H} -\uL\xi L\Phi-L\uL\xi \Phi}{r^{-2}\Delta} \frac{\Delta}{r^2} du\Big|^2 d\tau d\sigma +\int_{\tau_0}^\infty \Big|\int_{u_0}^\infty \uL(L\xi\Phi) du\Big|^2 d\tau d\sigma\\
&\quad\leq B\sum_{i=0}^{N-1}\Big(\sum_{j=0}^{|s|}\mathbb{E}[\tilde\upphi_{j}](\tau_{2i})+\int_{\mc H^+_{(\tau_{2i},\tau_{2i+1})}}|\Phi|^2d\sigma d\tau\Big) \,.
\end{align*}
We proceed similarly for $s<0$ for the hypersurfaces $\Sigma_\tau$ with $\tau\geq \tau_0$. Thus, we arrive at 
\begin{align*}
&\lp\{\begin{array}{lr}
\int_{\mc H^+_{(\tau_0,\infty)}}\lp|\frac{\mathfrak{H}_k}{w}\rp|^2 d\tau d\sigma,\quad &s>0\\
\int_{\mc I^+_{(\tau_0,\infty)}}\lp|\frac{\mathfrak{H}_k}{w}\rp|^2 d\tau d\sigma,\quad &s<0
\end{array}\rp\} 
+\int_{\Sigma_{\tau}}\lp|\frac{\mathfrak{H}_k}{w}\rp|^2 dr d\sigma  
\\
&\quad \leq B \lp\{\begin{array}{lr}
\displaystyle \int_{\mc R_{(\tau_0,\infty)}} \lp\{\begin{array}{lr}
r^{-2}\Delta,\quad &s>0\\
r^2,\quad &s<0
\end{array}\rp\} \lp|\frac{\mathfrak H_{k+1}}{w} \rp|^2 du d\tau d\sigma,\quad &k+1<|s|\\
\displaystyle \sum_{i=0}^{N-1}\sum_{j=0}^{|s|}\Big(\mathbb{E}[\tilde\upphi_{j}](\tau_{2i})+ \lp\{\begin{array}{lr}
\int_{\mc H^+_{(\tau_{2i},\tau_{2i+1})}}|\Phi|^2 d\tau d\sigma,\quad &s>0\\
\int_{\mc I^+_{(\tau_{2i},\tau_{2i+1})}}|\Phi|^2 d\tau d\sigma,\quad &s<0
\end{array}\rp\}\Big), & k+1=k_0=|s|\\
\end{array}\rp\} \,.\numberthis \label{eq:weird-cutoffs-lower-inhom-hor-infty}
\end{align*}

Now, from \eqref{eq:weird-cutoffs-diff-fake-real-hor}, we derive the identity
\begin{align*}
\mc L\lp(c(r)\lp|\frac{\mathfrak{H}_k}{w}\rp|^2\rp)+\sign s c'(r)\lp|\frac{\mathfrak{H}_k}{w}\rp|^2 = 2c(r)\Re\lp[\frac{\mathfrak{H}_k}{w}\overline{\mathfrak H_{k+1}}\rp]\,,
\end{align*}
for a suitably regular $c(r)$. Choosing $c=-\Delta/(r^2+a^2)$ if $s<0$ and $c=-r^{-1}$ if $s>0$, again then by \eqref{eq:weird-cutoffs-lower-inhom-hor-infty}
\begin{align*}
&\int_{\mc{R}_{(\tau_0,\infty)}}w\lp|\frac{\mathfrak H_k}{w}\rp|^2dr^*d\tau d\sigma \\
&\quad\leq 
B\lp\{\begin{array}{lr}
\int_{\mc{H}^+(\tau_0,\infty)}\lp|w^{-1}\mathfrak H_k\rp|^2d\tau d\sigma & s>0\\
\int_{\mc{I}^+(\tau_0,\infty)}\lp|w^{-1}\mathfrak H_k\rp|^2d\tau d\sigma & s<0
\end{array}\rp\}+B\limsup_{\tau\to \infty}\int_{\Sigma_{\tau}} \frac{1}{r}\lp|\frac{\mathfrak H_k}{w}\rp|^2 dr d\sigma +B\int_{\mc{R}_{(\tau_0,\infty)}} \frac{|\mathfrak{H}_{k+1}|^2}{w} dr^*d\tau d\sigma\\
&\quad\leq B\int_{\mc{R}_{(\tau_0,\infty)}} w\lp|\frac{\mathfrak{H}_{k+1}}{w}\rp|^2 dr^*d\tau d\sigma \,,\numberthis\label{eq:transport-estimates-inhom-intermediate}
\end{align*} 
as long as $k+1<|s|$; we can then conclude from \eqref{eq:weird-cutoffs-top-inhom-estimate-lower}. But if $k+1=|s|=k_0$, then we must again be more careful; we have
\begin{align*}
\int_{\mc{R}_{(\tau_0,\infty)}}\frac{\lp|\mathfrak H_k\rp|^2}{w}dr^*d\tau d\sigma 
&\leq 
B\lp\{\begin{array}{lr}
\int_{\mc{H}^+(\tau_0,\infty)}\lp|w^{-1}\mathfrak H_k\rp|^2d\tau d\sigma & s>0\\
\int_{\mc{I}^+(\tau_0,\infty)}\lp|w^{-1}\mathfrak H_k\rp|^2d\tau d\sigma & s<0
\end{array}\rp\}+B\lim_{\tau\to \infty}\int_{\Sigma_{\tau}}\frac{1}{r}\lp|\frac{\mathfrak H_k}{w}\rp|^2 dr d\sigma \\
&\qquad +B\lp|\int_{\mc{R}_{(\tau_0,\tau_{N+1})}} c(r)\Re\lp[\frac{\mathfrak{H}_{|s|-1}}{w}\overline{\mathfrak{H}}\rp] dr^*d\tau d\sigma\rp|\,,
\end{align*} 
and for the last term we note the identity:
\begin{align*}
c(r)\Re\lp[\frac{\mathfrak{H}_{|s|-1}}{w}\overline{\mathfrak{H}}\rp]  &=  c(r)\underline{\mc{L}}\xi\Re\lp[\frac{\mathfrak{H}_{|s|-1}}{w}\overline{\mc L\Phi}\rp]  +c(r)\Re\lp[\frac{\mathfrak{H}_{|s|-1}}{w}\overline{\lp([\mathfrak{R},\xi]-\underline{\mc L}\xi \mc L\rp)\Phi}\rp] \\
&\qquad+wc(r)\Re\lp[\frac{\mathfrak{H}_{|s|-1}}{w}\overline{(a c_{s,k}^{\rm id}+a c_{s,k}^{Z}Z)(\xi\upphi_k-\upphi_{k,\cutt})}\rp] \\
&= \mc L \Re\lp[c(r)\underline{\mc L}\xi\frac{\mathfrak{H}_{|s|-1}}{w}\overline{\Phi}-\frac{1}{2}c(r)(\underline{\mc L}\xi)^2|\Phi|^2\rp] 
- \mc L \lp(c(r)\underline{\mc L}\xi\rp)\Re\lp[\frac{\mathfrak{H}_{|s|-1}}{w}\overline{\Phi}\rp] \\
&\qquad -c(r)\underline{\mc L}\xi\Re\lp[\lp(\mathfrak{H}-\mc L\xi \underline{\mc L}\Phi\rp)\overline{\Phi}\rp]+c(r)\Re\lp[\frac{\mathfrak{H}_{|s|-1}}{w}\overline{\lp([\mathfrak{R},\xi]-\underline{\mc L}\Phi\mc L\rp)\Phi}\rp] \\
&\qquad+ \frac{1}{2}\mc L\lp( c(r)(\underline{\mc L}\xi)^2\rp)|\Phi|^2+c(r)\sum_{k=0}^{|s|-1}w\Re\lp[\frac{\mathfrak{H}_{|s|-1}}{w}\overline{(a c_{s,k}^{\rm id}+a c_{s,k}^{Z}Z)(\xi\upphi_k-\upphi_{k,\cutt})}\rp]\,,
\end{align*}
where we have integrated by parts in any term $\underline{\mc L}\xi\mc L \Phi$ appearing due to $\mathfrak{H}$.  Thus, if $k+1=|s|=k_0$, by an application of Cauchy--Schwarz we obtain
\begin{align*}
&\sum_{k=0}^{|s|-1}\int_{\mc{R}_{(\tau_0,\infty)}}w\lp|\frac{\mathfrak H_k}{w}\rp|^2dr^*d\tau d\sigma \\
&\quad\leq B\lp|\int_{\mc{R}_{(\tau_0,\tau_{N+1})}} c(r)\Re\lp[\frac{\mathfrak{H}_{|s|-1}}{w}\overline{\mathfrak{H}}\rp] dr^*d\tau d\sigma\rp|\\
&\quad\leq B
\lp\{\begin{array}{lr}
\sum_{i=0}^{N-1}\int_{\mc H^+_{(\tau_{2i},\tau_{2i+1})}}|\Phi|^2d\tau d\sigma+\int_{\mc H^+_{(\tau_0,\tau_{N+1})}}\lp|\frac{\mathfrak H_{|s|-1}}{w}\rp|^2d\tau d\sigma, &s>0\\
\sum_{i=0}^{N-1}\int_{\mc I^+_{(\tau_{2i},\tau_{2i+1})}}|\Phi|^2d\tau d\sigma+\int_{\mc I^+_{(\tau_0,\tau_{N+1})}}\lp|\frac{\mathfrak H_{|s|-1}}{w}\rp|^2d\tau d\sigma, &s<0\\
\end{array}\rp\}
\\
&\quad\qquad +B\sum_{k=0}^{|s|}\int_{\mc{R}_{(\tau_0,\tau_{N+1})}}\frac{w}{r}\lp(|\xi\upphi_k-\upphi_{k,\cutt}|^2+|a||Z(\xi\upphi_k-\upphi_{k,\cutt})|^2\rp)dr^*d\tau d\sigma\\
&\quad\leq B\sum_{i=0}^{N-1}\sum_{k=0}^{|s|}\mathbb{E}[\tilde \upphi_k](\tau_{2i})+ B\sum_{i=0}^N\lp\{\begin{array}{lr}
\int_{\mc H^+_{(\tau_{2i},\tau_{2i+1})}}|\Phi|^2 d\tau d\sigma, &s>0\\
\int_{\mc I^+_{(\tau_{2i},\tau_{2i+1})}}|\Phi|^2 d\tau d\sigma, &s<0
\end{array}\rp\}\,,
\end{align*}
after invoking \eqref{eq:weird-cutoffs-lower-inhom-hor-infty}, estimate \eqref{eq:weird-cutoffs-diff-fake-real-bulk} from Lemma~\ref{lemma:weird-cutoffs-diff-fake-real} and the finite-in-time energy estimates of Proposition~\ref{prop:finite-in-time-first-order}. 

\medskip
\noindent \textit{Step 2: 1st order estimates.} The procedure for derivatives of the inhomogeneities is very similar to the previous step. It is useful to note that $[\mc L,T]=[\mc L,Z]$ and so estimates for the Killing derivatives $T$ and $Z$ follow immediately. For $\p_{r^*}$ derivatives, we have to derive an analogue of the transport equation \eqref{eq:transformed-transport-inhom}: for $k<k_0$,
\begin{align*}
\mc L \lp(\frac{\mathfrak H_k}{w}\rp)'= [\mc L,\p_{r^*}]\lp(\frac{\mathfrak H_k}{w}\rp)+\mathfrak{H}_{k+1}' = \frac{2ar}{r^2+a^2}Z\mathfrak{H}_k +\mathfrak{H}_{k+1}'\,,
\end{align*}
and we also have that
\begin{align*}
\frac{\mathfrak{H}_k'}{w}=\lp(\frac{\mathfrak H_k}{w}\rp)'+\frac{w'}{w}\lp(\frac{\mathfrak H_k}{w}\rp)\,.
\end{align*}
Thus, by repeating Step 1 with the above transport equation in mind, in addition to the $Z$-commuted Step 1 estimates, we conclude the proof.
\end{proof}

\subsubsection{Separated current errors}
\label{sec:current-errors}

In this section, we bound the errors arising from employing separated currents to the inhomogeneous transformed system which was constructed in Section~\ref{sec:weird-cutoff-system} by applying hyperboloidal cutoffs  to its homogeneous counterpart. More precisely,  our goal is to estimate the terms involving \eqref{eq:ODE-estimates-inhomogeneity-terms} which feature in the right hand side of \eqref{eq:part1-ILED-withoutTS-top} and \eqref{eq:part1-ILED-withoutTS-bottom} in our Theorem~\ref{thm:ODE-ILED}, first appearing in \cite{SRTdC2020}.

To begin with, we set $k_0=|s|$ in our construction from Section~\ref{sec:weird-cutoff-system} and we analyze the interactions between the hyperboloidal cutoff inhomogeneities and the Killing and virial separated currents which involves solely the top level transformed variable, $\swei{\Phi}{s}$. 

\begin{lemma}[Top level current errors from hyperboloidal cutoffs]\label{lemma:current-errors-wave} Let $s\in\mathbb{Z}$. Let   $\chi_R^\pm$ be smooth cutoffs which are supported, respectively, for $r^*\in(R^*,\infty)$ and  $r^*\in(-\infty,-R^*)$.  We have the identity
\begin{align}
\begin{split}
[\mathfrak{R},\xi]&= \lp(L\uL\xi-a^2w\sin^2\theta TT\xi\rp)+(L\xi-aw\sin^2\theta T\xi)\uL+(\uL\xi-aw\sin^2\theta T\xi)L\\
&\qquad-2aw\frac{\rho^2}{r^2+a^2}T\xi\Phi+2ia w s\cos\theta T\xi. \end{split}\label{eq:inhom-commutator}
\end{align}
Thus, the following estimates hold.
\begin{itemize}
\item For errors due to $T$ Killing energy current, if $R^*$ is sufficiently large
\begin{align*}
&\int_{\mathbb{S}^2}\int_{-\infty}^\infty\int_{-\infty}^{\infty}  \chi_R^+ \Re\lp([\mathfrak{R},\xi]\Phi \overline{T(\xi\Phi)}\rp) dtdr^*d\sigma \leq B\sum_{i=0}^{N-1}\Big\{ \mathbb{E}[\Phi](\tau_{2i}) + 
\int_{\mc I^{+}_{(\tau_{2i},\tau_{2i+1})}}|\Phi|^2d\sigma d\tau\Big\}\,,\numberthis\label{eq:T-Killing-energy-errors-wave-infty}\\
&\int_{\mathbb{S}^2}\int_{-\infty}^\infty\int_{-\infty}^{\infty} \chi_R^- \Re\lp([\mathfrak{R},\xi]\Phi \overline{T(\xi\Phi)}\rp) dtdr^*d\sigma \leq B \sum_{i=0}^{N-1}\overline{\mathbb{E}}[\Phi](\tau_{2i})\,.\numberthis\label{eq:T-Killing-energy-errors-wave-hor}
\end{align*}
\item For errors due to $K$ Killing energy current, if $R^*$ is sufficiently large
\begin{align*}
&\int_{\mathbb{S}^2}\int_{-\infty}^\infty\int_{-\infty}^\infty  \chi_R^+ \Re\lp([\mathfrak{R},\xi]\Phi \overline{K(\xi\Phi)}\rp) dtdr^*d\sigma \leq B  \sum_{i=0}^{N-1}\mathbb{E}_2[\Phi](\tau_{2i})\,, \numberthis\label{eq:K-Killing-energy-errors-wave-infty} \\
&\int_{\mathbb{S}^2}\int_{-\infty}^\infty\int_{-\infty}^\infty  \chi_R^- \Re\lp([\mathfrak{R},\xi]\Phi \overline{K(\xi\Phi)}\rp) dtdr^*d\sigma \leq B\sum_{i=0}^{N-1}\Big\{ \mathbb{E}[\Phi](\tau_{2i}) + \int_{\mc H^{+}_{(\tau_{2i},\tau_{2i+1})}}|\Phi|^2d\sigma d\tau\Big\}\,. \numberthis\label{eq:K-Killing-energy-errors-wave-hor}
\end{align*}

\item For errors due to the virial currents, if $R^*$ is sufficiently large
\begin{align*}
&\lp|\int_{\mathbb{S}^2}\int_{-\infty}^\infty\int_{-\infty}^{\infty} \chi_R^+ \Re\lp([\mathfrak{R},\xi]\Phi \overline{(\xi\Phi)'}\rp) dtdr^*d\sigma \rp|+\lp|\int_{\mathbb{S}^2}\int_{-\infty}^\infty\int_{-\infty}^\infty\frac{1}{r}\chi_R^+ \Re\lp([\mathfrak{R},\xi]\Phi \overline{(\xi\Phi)}\rp) dtdr^*d\sigma \rp|\\
&\quad\leq B \sum_{i=0}^{N-1}\Big\{ \mathbb{E}[\Phi](\tau_{2i}) + 
\int_{\mc I^{+}_{(\tau_{2i},\tau_{2i+1})}}|\Phi|^2d\sigma d\tau\Big\}\,,\numberthis\label{eq:virial-current-errors-wave-infty}\\
&\lp|\int_{\mathbb{S}^2}\int_{-\infty}^\infty\int_{-\infty}^{\infty}  \chi_R^- \Re\lp([\mathfrak{R},\xi]\Phi \overline{(\xi\Phi)'}\rp) dtdr^*d\sigma\rp|+ \lp|\int_{\mathbb{S}^2}\int_{-\infty}^\infty\int_{-\infty}^\infty\frac{\Delta}{r^2}\chi_R^- \Re\lp([\mathfrak{R},\xi]\Phi \overline{(\xi\Phi)}\rp) dtdr^*d\sigma\rp|\\
&\quad\leq B \sum_{i=0}^{N-1} \Big\{\mathbb{E}[\Phi](\tau_{2i}) + 
\int_{\mc H^{+}_{(\tau_{2i},\tau_{2i+1})}}|\Phi|^2d\sigma d\tau\Big\}\,.\numberthis\label{eq:virial-current-errors-wave-hor}
\end{align*}

\item For errors in a bounded $|r^*|$ region, 
\begin{align*}
&\lp|\int_{\mathbb{S}^2}\int_{-R^*}^{R^*}\int_{-\infty}^\infty   \Re\lp([\mathfrak{R},\xi]\Phi \overline{T(\xi\Phi)}\rp) dtdr^*d\sigma\rp|+\lp|\int_{\mathbb{S}^2}\int_{-R^*}^{R^*}\int_{-\infty}^\infty   \Re\lp([\mathfrak{R},\xi]\Phi \overline{K(\xi\Phi)}\rp) dtdr^*d\sigma\rp|\\
&\quad\qquad+\lp|\int_{\mathbb{S}^2}\int_{-R^*}^{R^*}\int_{-\infty}^\infty  \Re\lp([\mathfrak{R},\xi]\Phi\overline{(\xi\Phi)}'\rp) dtdr^*d\sigma \rp|+\lp|\int_{\mathbb{S}^2}\int_{-R^*}^{R^*}\int_{-\infty}^\infty \Re\lp([\mathfrak{R},\xi]\Phi\overline{(\xi\Phi)}\rp) dtdr^*d\sigma \rp|\\
&\quad \leq B(R^*)\sum_{i=0}^{N-1}\mathbb{E}[\Phi](\tau_{2i})\numberthis\label{eq:current-errors-wave-bounded-r}
\end{align*}
\end{itemize}
We can remove integrals over $\mc H^+$ and $\mc I^+$ if we replace $\mathbb{E}$ by, respectively, $\overline{\mathbb{E}}$ and $\mathbb{E}_2$.
\end{lemma}

\begin{proof} Recall the finite-in-time estimates of Proposition~\ref{prop:finite-in-time-first-order}. In light of the properties of $\xi$, a glance at  \eqref{eq:inhom-commutator} shows that the only potentially problematic terms are due to the $\uL\xi L\Phi$ and $L\xi\uL\Phi$ terms in the inhomogeneity. We sketch the argument to obtain \eqref{eq:T-Killing-energy-errors-wave-infty} and \eqref{eq:K-Killing-energy-errors-wave-hor}; the remaining estimates either follow in a similar fashion or are the result of simply applications of Cauchy--Schwarz.

As $r\to \infty$, the term $L\xi\uL\Phi$ can be treated with Cauchy--Schwarz and it is the  term $\uL\xi L\Phi$ requires extra work due to insufficient $r$-decay. Using $2T=L+\uL-2\frac{a}{r^2+a^2}Z$, we write
\begin{align*}
\chi_R^+\uL\xi \Re[\overline{L\Phi} T(\xi\Phi)] = \frac12\chi_R^+ T\xi\uL\xi L|\Phi|^2+\frac12\chi_R^+\xi\uL\xi\Re[L\Phi\overline{\uL\Phi}]-\frac{a\chi_R^+}{r^2+a^2}\xi \uL\xi\Re[Z\Phi\overline{L\Phi}]+\frac12\chi_R^+\xi\uL\xi |L\Phi|^2\,.
\end{align*}
Here, the last two terms have enough $r$-decay to  be controlled by the initial data norm. For the first two terms, we integrate by parts in $L$, obtaining
\begin{align*}
&\frac12L\lp(\chi_R^+ T\xi\uL\xi |\Phi|^2+\chi_R^+\xi\uL\xi\Re[\Phi\overline{\uL\Phi}]\rp)\\
&\qquad-\frac12 L\chi_R^+\lp( T\xi\uL\xi |\Phi|^2+\xi\uL\xi\Re[\Phi\overline{\uL\Phi}]\rp)
-\frac12 \chi_R^+\lp( LT\xi\uL\xi |\Phi|^2+L(\xi\uL\xi)\Re[\Phi\overline{\uL\Phi}]+\xi\uL\xi\Re[\overline\Phi{L\uL\Phi}]\rp).
\end{align*}
We can now use the equation verified by $\Phi$ and integrate by parts to replace $\Re[\overline\Phi{L\uL\Phi}]$ with at most quadratic terms in first angular derivatives; it is easy to check that all these terms have suitable $r$-decay. Thus, the last line eventually leads to terms which are either compactly supported in $r$ or have sufficient $r$-decay, so it can be controlled by the initial data norm.  The first line leads to a contribution along $\mc I^+$, but after applying the finite-in-time energy estimates of Proposition~\ref{prop:finite-in-time-first-order}, only the zeroth order term cannot must remain on the right hand side of the estimates.

As $r\to r_+$, the problematic term is due to $L\xi\uL\Phi$: similarly to before, we have 
\begin{align*}
\chi_R^-L\xi \Re[\overline{\uL\Phi} K(\xi\Phi)] = \frac12\chi_R^- K\xi L\xi \uL|\Phi|^2+\frac12\chi_R^-\xi L\xi\lp\{\Re[L\Phi\overline{\uL\Phi}]+\frac{2\omega_+(r^2-r_+^2)}{r^2+a^2}\Re[Z\Phi\overline{\uL\Phi}]+ |\uL\Phi|^2\rp\}\,.
\end{align*}
Here, the last two terms have enough decay in $r-r_+$ to be controlled by the degenerate initial data norm in the full range $|a|\leq M$. For the first two terms, we integrate by parts in $\uL$, obtaining
\begin{align*}
&\frac12\uL\lp(\chi_R^- K\xi L\xi |\Phi|^2+\chi_R^-\xi L\xi\Re[L\Phi\overline{\Phi}] \rp)\\
&\qquad-\frac12\uL\chi_R^-\lp(K\xi L\xi |\Phi|^2+\xi L\xi\Re[L\Phi\overline{\Phi}] \rp)-
\frac12 \chi_R^-\lp( LK\xi L\xi |\Phi|^2+\uL(\xi L\xi)\Re[\Phi\overline{L\Phi}]+\xi L\xi\Re[\overline\Phi{\uL L\Phi}]\rp).
\end{align*}
By the same procedure as in the $r\to \infty$ case, we see that the first line leads to a contribution along $\mc H^+$ (but note again that the term proportional $|L\Phi|^2$ can be controlled by initial data) whereas the second line eventually yields terms with enough $(r-r_+)$-decay to be controlled by the degenerate initial data norm. 
\end{proof}

We are now ready to estimate the errors introduced by our choices of separated currents with full generality:

\begin{lemma}[Current errors from hyperboloidal cutoffs] \label{lemma:normal-current-errors-nonpeeling} Fix $s\in\{0,\pm 1,\pm 2\}$, $M>0$, $a_0\in[0,M)$ and $k_0\in\{0,\dots,|s|\}$ as in Lemma~\ref{lemma:weird-inverses}. Then, the following estimates hold for all $|a|\leq a_0$. If $k_0=|s|$, then 
\begin{align*}
&\sum_{k=0}^{|s|}\int_{-\infty}^\infty\sum_{ml}\int_{-\infty}^\infty\smlk{\mathfrak G}{s}{k}\cdot(f_k,h_k,y_k,\chi)\cdot \lp(\smlk{\uppsi}{s}{k},\lp(\smlk{\uppsi}{s}{k}\rp)'\rp) dr^*d\omega \\
&\quad \leq B\varepsilon \int_{-\infty}^\infty\sum_{ml} \int_{-\infty}^\infty w\lp(\sum_{k=0}^{|s|}|(\smlk{\uppsi}{s}{k})'|^2+\lp(\omega^2+\frac{|\Lambda|+1}{r}\rp)\sum_{k=0}^{|s|-1}|\smlk{\uppsi}{s}{k}|^2\rp)dr^*d\omega \\
&\quad \qquad + B\varepsilon \int_{-\infty}^\infty\sum_{ml} \int_{-\infty}^\infty w\lp[r^{-1}+\lp(\omega^2+|\Lambda|r^{-1}\rp)\frac{(r-r_{\rm trap})^2}{r^2}\rp]|\sml{\Psi}{s}|^2dr^*d\omega \\
&\qquad\quad + B\varepsilon^{-1}\sum_{i=0}^{N-1}\lp\{\sum_{j=0}^{|s|}\mathbb{E}[\swei{\tilde\upphi}{s}_j](\tau_{2i})+\int_{\mc H^+_{(\tau_{2i},\tau_{2i+1})}}|\swei{\Phi}{s}|^2drd\sigma+\int_{\mc I^+_{(\tau_{2i},\tau_{2i+1})}}|\swei{\Phi}{s}|^2drd\sigma\rp\}\,,\numberthis\label{eq:current-errors-top}
\end{align*}
with the following caveats: if there is any high frequency regime where we take $\chi_2\equiv 1$, then $\mathbb{E}$ must be replaced by $\overline{\mathbb{E}}$ and we can suppress the integral over $\mc H^+$; if there is any high frequency regime where we take $\chi_1\equiv 1$, then $\mathbb{E}$ must be replaced by $\mathbb{E}_2$ and we can suppress the integral over $\mc I^+$.

On the other hand, if $k_0<|s|$, then 
\begin{align*}
&\sum_{k=0}^{k_0}\int_{-\infty}^\infty\sum_{ml}\int_{-\infty}^\infty(\omega^2+m^2+1)\lp(\smlk{\mathfrak G}{s}{k}+\smlk{\mathfrak g}{s}{k}\rp)\cdot(f_k,h_k,y_k,\chi)\cdot \lp(\smlk{\uppsi}{s}{k},\lp(\smlk{\uppsi}{s}{k}\rp)'\rp) dr^*d\omega \\
&\qquad +\sum_{k=0}^{k_0}\int_{-\infty}^\infty\sum_{ml}\int_{-\infty}^\infty\lp(\smlk{\mathfrak G}{s}{k}+\smlk{\mathfrak g}{s}{k}\rp)'\cdot(f_k,h_k,y_k,\chi)\cdot \lp(\lp(\smlk{\uppsi}{s}{k}\rp)',\lp(\smlk{\uppsi}{s}{k}\rp)''\rp) dr^*d\omega \\
&\quad \leq B(R^*)\varepsilon \sum_{k=0}^{k_0}\int_{-\infty}^\infty\sum_{ml} \int_{r_+}^\infty \frac{1}{r^2}\lp(|(\smlk{\uppsi}{s}{k})''|^2+(1+|\Lambda|+\omega^2)|(\smlk{\uppsi}{s}{k})'|^2\rp)dr d\omega \\
&\quad\qquad + B(R^*)\varepsilon \sum_{k=0}^{k_0}\int_{-\infty}^\infty\sum_{ml} \int_{r_+}^\infty \frac{1}{r^2}\lp(\omega^2+\frac{1+|\Lambda|}{r}\rp)|\smlk{\uppsi}{s}{k}|^2dr d\omega \\
&\qquad\quad + B(R^*)\varepsilon^{-1}\sum_{i=0}^{N-1}\lp(\sum_{j=0}^{k_0}\mathbb{E}^1[\swei{\tilde\upphi}{s}_j](\tau_{2i})+\mathbb{E}[\swei{\tilde\upphi}{s}_{k_0+1}](\tau_{2i})\rp)\,.\numberthis\label{eq:current-errors-bottom}
\end{align*}
\end{lemma}

\begin{proof} By Plancherel,  \eqref{eq:ODE-estimates-inhomogeneity-terms}  can be expressed in terms of the variables in our inhomogeneous transformed system from Section~\ref{sec:weird-cutoff-system}: for $k\leq k_0\leq |s|$, we have
\begin{align*}
&\int_{-\infty}^\infty\sum_{ml}\lp(\smlambdak{\mathfrak G}{s}{k},\smlambdak{\mathfrak g}{s}{k}\rp)\cdot(f_k,h_k,y_k,\chi)\cdot \lp(\smlambdak{\uppsi}{s}{k},\lp(\smlambdak{\uppsi}{s}{k}\rp)'\rp) dr^* \\
&\quad = \int_{-\infty}^\infty\lp\{2(y+\hat{y}+\tilde{y}+f)\Re\lp[\mathfrak{H}\overline{\Phi_\cutt}'\rp]dr^*+\int_{-\infty}^\infty(h+f')\Re\lp[{\mathfrak H}\overline{\Phi_\cutt}\rp]\rp\}\mathbbm{1}_{\{k=k_0=|s|\}}dr^*\\
&\quad\qquad +\int_{-\infty}^\infty E\lp(\chi_1\Re\lp[\mathfrak{H}\overline{K\Phi_{\cutt}}\rp]+\chi_2\Re\lp[\mathfrak{H}\overline{T\Phi_{\cutt}}\rp]\rp)\mathbbm{1}_{\{k=k_0=|s|\}}dr^*\\
&\quad\qquad+\int_{-\infty}^\infty\lp\{ 2(y_{(k)}+\hat{y}_{(k)})\Re\lp[\mathfrak{H}_k\overline{\upphi_{k,\cutt}}'\rp]\mathbbm{1}_{\{k\leq k_0\}}+h_{(k)}\Re\lp[\mathfrak{H}_k\overline{\upphi_{k,\cutt}}\rp]\mathbbm{1}_{\{k< k_0\}}\rp\}dr^*\\
&\quad\qquad +\mathbbm{1}_{\{k\leq k_0\}}\int_{-\infty}^\infty\sum_{ml}\mathbbm{1}_{\mc F_{\rm high}}\Re\lp[\lp(\smlambdak{\mathfrak G}{s}{k}+\smlambdak{\mathfrak g}{s}{k}\rp)\overline{\smlambdak{\uppsi}{s}{k}}\rp]dr^*\\
&\quad\qquad -\mathbbm{1}_{\{k<k_0\}}\int_{-\infty}^\infty\sum_{ml}\mathbbm{1}_{\mc F_{\rm high}^c}\lp(\chi_1(\omega-m\upomega_+)\Im\lp[\smlambdak{\mathfrak G}{s}{k}\overline{\smlambdak{\uppsi}{s}{k}}\rp]+\chi_2\omega\Im\lp[\smlambdak{\mathfrak G}{s}{k}\overline{\smlambdak{\uppsi}{s}{k}}\rp]\rp)dr^*\\
&\quad\qquad+\frac12 (-1)^{s}\sign s \int_{-\infty}^\infty (r^2+a^2)^{|s|+1/2} E_W\chi_3 \Re\lp\{K{\upphi}_{0,\cutt}w\lp(\frac{\mc L}{w(r^2+a^2)}\rp)^{2|s|}\lp((r^2+a^2)^{|s|-1/2}\frac{\overline{{\mathfrak{H}_{0}}}}{w}\rp)\rp.\\
&\quad\qquad\qquad\qquad\qquad\qquad\qquad+\lp.
{\mathfrak{H}}_{0}\lp(\frac{\mc L}{w(r^2+a^2)}\rp)^{2|s|}\lp((r^2+a^2)^{|s|-1/2}\overline{K{\upphi}_{0,\cutt}}\rp)\rp\}\mathbbm{1}_{\{k=k_0=|s|\}}dr^*
\,. \numberthis \label{eq:inhomogeneity-terms-Plancherel}
\end{align*}

\medskip
\noindent \textit{Step 1: virial and Killing energy currents errors for $k=k_0<|s|$.} In this step, we consider all but the last term in the inhomogeneity errors given in \eqref{eq:inhomogeneity-terms-Plancherel}, for the case $k=k_0<|s|$. Recall the formula \eqref{eq:weird-cutoffs-top-inhom-formula}  for the inhomogeneity $\mathfrak{H}_{k_0}$:
\begin{align*}
\mathfrak{H}_{k_0}&:=\lp([\mathfrak{R}_{k_0},\xi]-\mathbbm{1}_{\{k_0\neq |s|\}}\underline{\mc L}\xi \mc L\rp){\upphi}_{k_0} +\mathbbm{1}_{\{k_0\neq |s|\}} w \underline{\mc L}\xi {\upphi}_{k_0+1}\nonumber\\
&\qquad +\sum_{j=0}^{k_0-1}w\lp(ac_{s,k_0,j}^{\rm id}+ac_{s,k_0,j}^{Z}Z\rp)(\xi{\upphi}_j-{\upphi}_{j,\cutt})\,.
\end{align*}
For the first term, we can repeat the steps in our proof of Lemma~\ref{lemma:current-errors-wave}, now noting that since $k_0<|s|$ here, the result is strictly easier to show as $\upphi_{k_0}$ enjoys some additional $r$ decay. The second term leads to an error involving the next level transformed variable, which we estimate by $$\sum_{i=0}^{N-1}\int_{\Sigma_{\tau_{2i}}}r^{-2}|\tilde\upphi_{k_0+1}|dr d\sigma\,,$$ 
by the finite in time estimates.
The errors arising from third term are no longer supported solely in the finite time slabs between $\mc{R}_{(\tau_{2i},\tau_{2i+1})}$;  we nevertheless still apply Cauchy--Schwarz and the estimates of Lemma~\ref{lemma:weird-cutoffs-diff-fake-real} for the errors arising. For instance, for the virial current errors,
\begin{align*}
&\int_{\mc R_{(\tau_0,\tau_{N+1})}} w\Re\lp[\lp(ac_{s,k_0,k}^{\rm id}+ac_{s,k_0,k}^{Z}Z\rp)(\xi{\upphi}_j-{\upphi}_{j,\cutt})\lp(\overline{\upphi_{k_0,\cutt}'}+r^{-1}\overline{\upphi_{k_0,\cutt}}\rp)\rp]dr^*d\sigma d\tau \\
&\quad \leq B\varepsilon^{-1}\sum_{X\in\{\mathrm{id},Z\}}\int_{\mc R_{(\tau_0,\tau_{N+1})}} w|X(\xi{\upphi}_j-{\upphi}_{j,\cutt})|^2 dr^*d\sigma d\tau \\
&\quad\qquad + B\varepsilon\int_{\mc R_{(\tau_0,\tau_{N+1})}}w\lp(|\upphi_{k_0,\cutt}'|^2+r^{-1}|\upphi_{k_0,\cutt}|^2\rp) dr^*d\sigma d\tau\\
&\quad \leq B\varepsilon^{-1}\sum_{i=1}^N\sum_{j=0}^{k_0}\mathbb{E}[\tilde\upphi_j](\tau_i) + B\varepsilon\int_{\mc R_{(\tau_0,\tau_{N+1})}}w\lp(|\upphi_{k_0,\cutt}'|^2+r^{-1}|\upphi_{k_0,\cutt}|^2\rp) dr^*d\sigma d\tau\,.
\end{align*}
In the above, the terms with $\varepsilon$ are bulk integrals of the inhomogeneous transformed variables (which we will be able to absorb later). The same procedure can be employed for the Killing energy current errors, since we assume $k_0<|s|$. 

To conclude, we address the errors arising from the high frequencies:
\begin{align*}
&\int_{-\infty}^\infty\sum_{ml}\mathbbm{1}_{\mc F_{\rm high}}\Re\lp[\lp(\smlambdak{\mathfrak G}{s}{k_0}+\smlambdak{\mathfrak g}{s}{k_0}\rp)\overline{\smlambdak{\uppsi}{s}{k_0}}\rp]dr^*\\
&\leq B\varepsilon\int_{-\infty}^\infty\sum_{ml} \mathbbm{1}_{\mc F_{\rm high}} w |\smlambdak{\uppsi}{s}{k_0}|^2 dr^* + B\varepsilon^{-1} \int_{\mc R_{(\tau_0,\tau_{N+1})}} w\lp|\frac{\mathfrak H_{k_0}+\mathfrak h_{k_0}}{w}\rp|^2 dr^* d\sigma d\tau\\
&\quad \leq B\varepsilon^{-1}\sum_{i=1}^N\sum_{j=0}^{k_0}\mathbb{E}[\tilde\upphi_j](\tau_i) +B\varepsilon\int_{\mc R_{(\tau_0,\tau_{N+1})}}w|T\upphi_{k_0,\cutt}|^2 dr^*d\sigma d\tau\,,
\end{align*}
using Lemma~\ref{lemma:weird-cutoffs-inhoms-rough-estimates} and the fact that $\mathfrak{h}_{k_0}$ has no new terms compared to $\mathfrak{H}_{k_0}$ .

Finally, we note that, letting $X\in\{\p_{r^*},T,Z\}$, the entire procedure outlined in the step can be repeated if we replace $\upphi_{j,\cutt}$ by $X\upphi_{j,\cutt}$ and $\tilde\upphi_{j}$ by $X\tilde\upphi_{j}$ for any $j=0,\dots, k_0+1$, $\mathfrak H_k$ by $X\mathfrak{H}_k$ and $X\smlk{\uppsi}{s}{k}$ by $X\smlk{\uppsi}{s}{k}$.  

\medskip
\noindent \textit{Step 2: virial and Killing energy currents errors if $k_0=|s|$.} In this step, we consider all but the last term in the inhomogeneity errors given in \eqref{eq:inhomogeneity-terms-Plancherel}, but now for the case $k=k_0=|s|$. We recall from \eqref{eq:weird-cutoffs-top-inhom-formula}
\begin{align*}
\mathfrak{H}&:=[\mathfrak{R},\xi]{\Phi}+\sum_{j=0}^{|s|-1}w\lp(ac_{s,j}^{\rm id}+ac_{s,j}^{Z}Z\rp)(\xi{\upphi}_j-{\upphi}_{j,\cutt})\,.
\end{align*}
For the first term, the current errors can be dealt with using Lemma~\ref{lemma:current-errors-wave}. For the second term, we may apply Cauchy--Schwarz as in Step 1, for errors arising from the virial currents. The same strategy will work for errors associated to the Killing energy currents which are supported at sufficiently very large $|r^*|\geq R^*$: from Lemma~\ref{lemma:weird-cutoffs-diff-fake-real}, we have 
\begin{align*}
&\int_{\mc R_{(\tau_0,\tau_{N+1})}} w\Re\lp[\lp(ac_{s,j}^{\rm id}+ac_{s,j}^{Z}Z\rp)(\xi{\upphi}_j-{\upphi}_{j,\cutt})\lp(\chi_R^-\overline{K\Phi_{\cutt}}+\chi_R^+\overline{T\Phi_{\cutt}}\rp)\rp] dr^*d\sigma d\tau\\
&\quad\leq B\varepsilon\int_{\mc R_{(\tau_0,\tau_{N+1})}} w\lp(|T\Phi_{\cutt}|^2+\frac{1}{r}|\mathring{\slashed\nabla}\Phi_{\cutt}|^2\rp)dr^*d\sigma d\tau\\
&\quad\qquad+B\varepsilon^{-1}\sum_{j=0}^{|s|-1}\sum_{X\in\{\mathrm{id},Z\}}\int_{\mc R_{(\tau_0,\tau_{N+1})}} w|X(\xi{\upphi}_j-{\upphi}_{j,\cutt})|^2 dr^*d\sigma d\tau\\
&\quad\leq B\varepsilon^{-1}\sum_{i=0}^{N-1}\sum_{k=0}^{|s|}\mathbb{E}[\tilde\upphi_k](\tau_{2i})+B\varepsilon\int_{\mc R_{(\tau_0,\tau_{N+1})}} w(1-\chi_R^+-\chi_R^-)\lp(|T\Phi_{\cutt}|^2+\frac{1}{r}|\mathring{\slashed\nabla}\Phi_{\cutt}|^2\rp)dr^*d\sigma d\tau\,,
\end{align*}
where $\chi_{R}^\pm$ is a cutoff which is supported, respectively, for large $r^*\geq R^*$ and for very negative $r^*\leq -R^*$. Notice that the term with $\epsilon$ represents a bulk term which we can hope to absorb at the end, but only due to its degeneration for bounded $|r^*|$: for energy current errors localized in $|r^*|\leq R^*$,  one should pursue a different strategy. We note the identity
\begin{align*}
&w\xi\Re\lp[\lp(ac_{s,j}^{\rm id}+ac_{s,j}^{Z}Z\rp)(\xi{\upphi}_j-{\upphi}_{j,\cutt})\overline{T(\xi\Phi)}\rp]\\
&\quad = \Re\lp[\lp(ac_{s,j}^{\rm id}+ac_{s,j}^{Z}Z\rp)(\xi{\upphi}_j-{\upphi}_{j,\cutt})T(\mc L\upphi_{|s|-1,\cutt})\rp]\\
&\quad =\mc L\lp(\chi\Re\lp[\lp(ac_{s,j}^{\rm id}+ac_{s,j}^{Z}Z\rp)(\xi{\upphi}_j-{\upphi}_{j,\cutt})\overline{T\upphi_{|s|-1,\cutt}}\rp]\rp) \\
&\quad\qquad +\sign s \Re\lp[\lp(a(\chi c_{s,j}^{\rm id})'+a(\chi  c_{s,j}^{Z})'Z\rp)(\xi{\upphi}_j-{\upphi}_{j,\cutt})\overline{T\upphi_{|s|-1,\cutt}}\rp]\\
&\qquad\quad -\chi \Re\lp[\lp(ac_{s,j}^{\rm id}+ac_{s,j}^{Z}Z\rp)\lp[\mc L\xi \upphi_j+w(\xi\upphi_{j+1}-\upphi_{j+1,\cutt})\rp]\overline{T\upphi_{|s|-1,\cutt}}\rp]\,,
\end{align*}
where $\chi=\chi_R^++\chi_R^-$. Therefore, by the estimates of Lemma~\ref{lemma:weird-cutoffs-diff-fake-real} we have for $s>0$
\begin{align*}
&\sum_{j=0}^{|s|-1}\int_{\mc R_{(\tau_0,\tau_{N+1})}} w\chi \Re\lp[\lp(ac_{s,j}^{\rm id}+ac_{s,j}^{Z}Z\rp)(\xi{\upphi}_j-{\upphi}_{j,\cutt})\overline{T(\xi\Phi)}\rp]dr^*d\sigma d\tau \\
&\quad \leq B\varepsilon\int_{\mc R_{(\tau_0,\tau_{N+1})}}w |T\upphi_{|s|-1,\cutt}|^2dr^* d\sigma d\tau\\
&\quad\qquad  +B\varepsilon^{-1}\sum_{j=0}^{|s|-1}\sum_{X\in\{\mathrm{id},Z\}}\Big(\int_{\mc H^+_{(\tau_0,\tau_{N+1})}} |X(\xi{\upphi}_j-{\upphi}_{j,\cutt})|^2d\sigma d\tau+\int_{\mc R_{(\tau_0,\tau_{N+1})}}w|X(\xi{\upphi}_j-{\upphi}_{j,\cutt})|^2dr^* d\sigma \Big)\\
&\quad\qquad +B\varepsilon^{-1}\sum_{j=0}^{|s|-1}\int_{\mc R_{(\tau_0,\tau_{N+1})}}w|\xi{\upphi}_{j+1}-{\upphi}_{j+1,\cutt}|^2dr^* d\sigma \\
&\quad \leq B\varepsilon\int_{-\infty}^\infty  w |T\upphi_{|s|-1,\cutt}|^2dr^* d\sigma +
B\varepsilon^{-1}\sum_{k=0}^{|s|}\sum_{i=1}^N\mathbb{E}[\tilde\upphi_k](\tau_i)\,,
\end{align*}
for some $\varepsilon>0$ and, in the case $s<0$ we can argue similarly to obtain
\begin{align*}
&\sum_{j=0}^{|s|-1}\int_{\mc R_{(\tau_0,\tau_{N+1})}} w\Re\lp[\lp(ac_{s,j}^{\rm id}+ac_{s,j}^{Z}Z\rp)(\xi{\upphi}_j-{\upphi}_{j,\cutt})\overline{T(\xi\Phi)}\rp]dr^*d\sigma d\tau \\
&\quad \leq B\varepsilon\int_{-\infty}^\infty  w |T\upphi_{|s|-1,\cutt}|^2dr^* d\sigma  +
B\varepsilon^{-1}\sum_{k=0}^{|s|}\sum_{i=1}^N\mathbb{E}[\tilde\upphi_k](\tau_i)\,.
\end{align*}
The same holds with $T$ replaced by the vector field $K$, though we note that in the latter case we may use the support of $\chi$ to improve the $r$-weights in our final estimate.

\medskip
\noindent \textit{Step 3: Teukolsky--Starobinsky energy current errors in the case $k=k_0=|s|$.}  We now turn to the last term in \eqref{eq:inhomogeneity-terms-Plancherel}. Let us first note that 
\begin{align*}
\lp(\frac{\mc L}{w(r^2+a^2)}\rp)^{|s|}\lp((r^2+a^2)^{|s|-1/2}K{\upphi}_{0,\cutt}\rp)&=(r^2+a^2)^{-1/2}\sum_{k=0}^{|s|}p_{s,|s|,k}(r)K{\upphi}_{k,\cutt}\,,\\
\lp(\frac{\mc L}{w(r^2+a^2)}\rp)^{|s|}\Big((r^2+a^2)^{|s|-1/2}\frac{\mathfrak{H}_{0}}{w}\Big)&=\frac{(r^2+a^2)^{-1/2}}{w}\sum_{k=0}^{|s|}p_{s,|s|,k}(r)\mathfrak{H}_{k}\,.
\end{align*}
where our notation is that $p_{s,k,j}(r)$ are polynomials of degree $|s|-k$ in $r$ which can be explicitly computed. By integrating in the  $\mc L$ direction $|s|$ times, we see that
\begin{align*}
&(r^2+a^2)^{|s|+1/2}{\mathfrak{H}}_{0}\lp(\frac{\mc L}{w(r^2+a^2)}\rp)^{2|s|}\lp((r^2+a^2)^{|s|-1/2}\overline{K{\upphi}_{0,\cutt}}\rp)\\
&\quad=\sum_{j=1}^{|s|}(-1)^{j-1} \mc L\Big\{
\Big((r^2+a^2)^{|s|-j+1/2}\sum_{k=0}^{j-1}p_{s,j-1,k}(r)\frac{{\mathfrak{H}}_{k}}{w}\Big)\\
&\quad\qquad\qquad\qquad\qquad\qquad\times\lp(\frac{\mc L}{w(r^2+a^2)}\rp)^{|s|-j}\Big((r^2+a^2)^{-1/2}\sum_{k=0}^{|s|}p_{s,|s|,k}(r)\overline{K{\upphi}_{k,\cutt}}\Big)
\Big\}\\
&\quad\qquad+(-1)^{|s|}\Big(\sum_{k=0}^{|s|}p_{s,|s|,k}(r){\mathfrak{H}}_{k}\Big)\Big(\sum_{k=0}^{|s|}p_{s,|s|,k}(r)\overline{K{\upphi}_{k,\cutt}}\Big)\,.\numberthis\label{eq:TS-errors-intermediate}
\end{align*}
For the sake of concreteness, we can compute some of the polynomials involved in this formulas:
\begin{gather*}
p_{s,|s|,|s|}=1,\quad p_{s,1,0}= -(2|s|-1)r\sign s,\quad
p_{s,2,1}=- 4(|s|-1)r\sign s,\quad p_{s,2,0}=(2|s|-1)(a^2+2|s|r^2).
\end{gather*}
We may also expand the last term in \eqref{eq:TS-errors-intermediate} as 
\begin{align*}
&\Big(\sum_{k=0}^{|s|}p_{s,|s|,k}(r){\mathfrak{H}}_{k}\Big)\Big(\sum_{k=0}^{|s|}p_{s,|s|,k}(r)\overline{K{\upphi}_{k,\cutt}}\Big)\\
&\quad = 2\mathfrak{H}\overline{K\Phi_{\cutt}} +\overline{K\mc L\upphi_{|s|-1,\cutt}}\sum_{k=0}^{|s|-1}p_{s,|s|,k}\frac{\mathfrak{H}_k}{w} + \sum_{k=0}^{|s|}\sum_{j=0}^{|s|-1}p_{s,|s|,k}p_{s,|s|,j}{\mathfrak{H}}_{k}\overline{K{\upphi}_{j,\cutt}}\\
&\quad =\Big[\sum_{k=0}^{|s|}\sum_{j=0}^{|s|-1}p_{s,|s|,k}p_{s,|s|,j}{\mathfrak{H}}_{k}\overline{K{\upphi}_{j,\cutt}}- \overline{K\upphi_{|s|-1,\cutt}}\sum_{k=0}^{|s|-1} \Big(\sign s p_{s,|s|,k}'+w^{-1}p_{s,|s|,k-1}\Big)\frac{\mathfrak{H}_k}{w}\Big]\\
&\quad\qquad +\mc L\Big(\overline{K\upphi_{|s|-1,\cutt}}\sum_{k=0}^{|s|-1}p_{s,|s|,k}\frac{\mathfrak{H}_k}{w}\Big)+\mathfrak{H}\Big(\overline{K\Phi_{\cutt}}-\overline{K\upphi_{|s|-1,\cutt}}p_{s,|s|,|s|-1}\Big)\,,
\end{align*}
and the last term above as
\begin{align*}
\mathfrak{H}\overline{K\upphi_{|s|-1,\cutt}}p_{s,|s|,|s|-1} &= (\mathfrak{H}-L\xi\uL\Phi)\overline{K\upphi_{|s|-1,\cutt}}p_{s,|s|,|s|-1} +\uL\lp(L\xi \Phi \overline{K\upphi_{|s|-1,\cutt}}p_{s,|s|,|s|-1} \rp) \\
&\qquad- \uL(L\xi p_{s,|s|,|s|-1}) \Phi \overline{K\upphi_{|s|-1,\cutt}}-L\xi p_{s,|s|,|s|-1}w\Phi\overline{\Phi}\,.
\end{align*}
The above formulas also hold replacing  $({\mathfrak{H}}_{0}, \overline{K{\upphi}_{0,\cutt}})$ by $(w K{\upphi}_{0,\cutt}, \overline{{\mathfrak H}_{0}}/w)$, and replacing $K$ by $T$.  

From the above manipulations, we see that the last term in \eqref{eq:inhomogeneity-terms-Plancherel} leads to: (i) boundary terms, (ii) bulk terms with compact support in $r^*$ (when $\mc L$ hits $\chi_3$), and (iii) bulk terms supported away from $r=\infty$ (as $\chi_3$ is itself compactly supported in $r$). It is clear that the terms (i) and (ii) are only finite at our desired level of regularity if we restrict ourselves to $|s|\leq 2$. If this condition is met then we may appeal to Lemmas~\ref{lemma:weird-cutoffs-diff-fake-real} and \ref{lemma:weird-cutoffs-inhoms-rough-estimates} to control the boundary terms (i). For the bulk terms (ii) and (iii), we invoke Lemma~\ref{lemma:current-errors-wave} to deal with any errors of the form $\mathfrak{H}\overline{K\Phi_{\cutt}}$ and Cauchy--Schwarz together with the Lemmas~\ref{lemma:weird-cutoffs-diff-fake-real} and \ref{lemma:weird-cutoffs-inhoms-rough-estimates}, as in Steps 2 and 3, for all other terms.

\medskip
\noindent \textit{Step 4: virial and Killing current errors in the case $k<k_0$.}
Now, we turn to the errors arising from the inhomogeneities $\mathfrak{H}_k$ for $k<k_0$. These are supported in the entire $\mc R_{(\tau_0,\infty)}$. We treat these errors by applications of Cauchy--Schwarz, producing bulk terms in $\tilde\upphi_{k,\cutt}$ with smallness $\varepsilon$, and terms involving $\mathfrak{H}_k$ which we bound through Lemma~\ref{lemma:weird-cutoffs-inhoms-rough-estimates}.
\end{proof}

\subsubsection{Bulk terms in the intermediate frequency range}
\label{sec:bulk-terms}

Recall that for solutions to the radial ODEs of Definitions~\ref{def:transformed-system-ODE} or \ref{def:alt-transformed-system-ODE} supported in  $(\omega,m,\Lambda)\in\mc{F}_{\rm int}$, our Theorems~\ref{thm:ODE-ILED} and \ref{thm:ODE-ILED-fixed-m} fail to provide control of an integrated energy in a compact range of $r^*$. The goal of this subsection is to provide exactly this missing result:
\begin{proposition}[ILED for Teukolsky cutoff in $\mc{F}_{\rm int}$] \label{prop:ODE-ILED-compact-r-mode-stab} Fix $s\in\mathbb{Z}$, $k_0\in\{0,\dots,|s|\}$, $M>0$, and assume $(\omega,m,l)\in\mc{F}_{\rm int}$. For $k=0,\dots, |s|$, define $\uppsi_{(k),\,ml}^{[s],\,a,\,\omega}$ from $\swei{\upphi}{s}_{k,\,\cutt}$ as in Section~\ref{sec:weird-cutoff-system}. Then, for any given $R^*>0$, we have the following bulk estimates for $s\leq 0$ and, if $|a|=M$, $s>0$. If $k<|s|$, then 
\begin{align*}
&\sum_{k=0}^{|s|}\int_{\mc{F}_{\rm int}}\sum_{(m,l)\in\mc{F}_{\rm int}} (\omega^2+m^2+1)^{|s|-k-1}\int_{-R^*}^{R^*} \lp(\big|\uppsi_{(k),\,ml}^{[s],\,a,\,\omega}\big|^2+\big|\big(\uppsi_{(k),\,ml}^{[s],\,a,\,\omega}\big)'\big|^2+\big|\big(\uppsi_{(k),\,ml}^{[s],\,a,\,\omega}\big)''\big|^2\rp)dr^*\\
&\quad\leq B(R^*,\omega_{\rm high},\omega_{\rm low},\varepsilon_{\rm width}) \sum_{i=0}^{N-1} \sum_{k=0}^{|s|} {\mathbb{E}}[\swei{\tilde\upphi}{s}_{k}](\tau_{2i})\,.\numberthis\label{eq:ODE-ILED-compact-r-mode-stab}
\end{align*}
If $|a|\leq a_0\in[0,M)$ and $s=+1,+2$, then the estimate also holds, but the constant $B$ depends on $a_0$.
\end{proposition}

To begin, it will be convenient to introduce an auxiliary inhomogeneous transformed system:
\begin{lemma}[Forwards cutoffs]\label{lemma:teukolsky-cutoff-inhom-k} Fix $s\in\mathbb{Z}$, and $k_0\in\{0,\dots,|s|\}$. Let $\swei{\upphi}{s}_k$, for $k=0,\dots,|s|$ be  solutions to the homogeneous transformed system of Definition~\ref{def:transformed-system}. Let 
$$\swei{\upphi}{s}_{0,\cut}:=\xi\swei{\upphi}{s}_{0}$$ 
Define
\begin{align}
\swei{\upphi}{s}_{k,\cut}:=\lp(w^{-1}\mc L\rp)^{k-k_0}(\xi\swei{\upphi}{s}_{k_0})=\sum_{j=k_0}^{k-1}d_j(t^*,r)\swei{\dbtilde{\upphi}}{s}_{j}\,,\label{eq:teukolsky-cutoff-hierarchy}
\end{align} 
where $\tilde c_j$ are given in \eqref{eq:rescalings} for explicit $d_j(t^*,r)$ supported in $\supp(\nabla \xi)$ and satisfying $d_j(t^*,r)=O(1)$ and  $\mc{\uL}d_j(t^*,r)=O(w)$ as $|r^*|\to \infty$. Furthermore,  $\swei{\upphi}{s}_{k,\cut}$ satisfies \eqref{eq:transformed-k} with 
\begin{align}
\begin{split}
\swei{\breve{\mathfrak{H}}}{s}_{k}&=\lp([\mathfrak{R},\xi]-\mathbbm{1}_{\{k<|s|\}}\uL\xi L\rp)\swei{\upphi}{s}_k+\mathbbm{1}_{\{k<|s|\}} w\uL\xi\swei{\upphi}{s}_{k+1}+(|s|-k)\frac{w'}{w}L\xi +wd_{s,k,k}^{\rm id}\swei{\upphi}{s}_{k}\\
&\qquad+ w\sum_{j=k_0}^{k-1} \lp( \frac{r^2+a^2}{\Delta}d_{s,k,j}^{\uL} \uL  +  \frac{1}{r^2}d_{s,k,j}^{Z} Z+d_{s,k,j}^{\rm id}\rp)\swei{\tilde\upphi}{s}_j, \quad s\leq 0\,,\\
\swei{\breve{\mathfrak{H}}}{s}_{k}&=\lp([\mathfrak{R},\xi]-\mathbbm{1}_{\{k<|s|\}}L\xi\uL\rp)\swei{\upphi}{s}_k+\mathbbm{1}_{\{k<|s|\}} w L\xi\swei{\upphi}{s}_{k+1}- (|s|-k)\frac{w'}{w}\uL\xi \swei{\upphi}{s}_k+wd_{s,k,k}^{\rm id}\swei{\upphi}{s}_{k}\\
&\qquad+ w \sum_{j=k_0}^{k-1} \lp((r^2+a^2) d_{s,k,j}^{L} L  + d_{s,k,j}^{Z} Z+ d_{s,k,j}^{\rm id}\rp)\swei{\dbtilde \upphi}{s}_j,\quad s\geq 0\,,
\end{split}\label{eq:teukolsky-cutoff-inhom-k}
\end{align}
where the first term can be given explicitly by \eqref{eq:inhom-commutator}. The functions  $d_{s,k,j}^X$ are explicit functions of $r$, $\theta$ and $\xi=\xi(t^*,r)$, which are supported in $\supp(\nabla \xi)$, and such that $d_{s,k,j}^X=O(1)$ and $\underline{\mc L}d_{s,k,j}^{\underline{\mc L}}=O(w)$ as $|r^*|\to \infty$, in the $L^2_{\mathbb{R}_t\times\mathbb{S}^2}$ sense. 

If $\xi$ is chosen so that $\xi\swei{\upphi}{s}_0$ is verify the conditions of Definitions~\ref{def:suf-integrability} and \ref{def:outgoing-bdry-phys-space} then, by  Remark~\ref{rmk:peeling-fixed-tau}, for each $k=0,\dots, |s|$,  $\swei{\upphi}{s}_{k,\cut}$  are outgoing and sufficiently integrable solutions of the transformed system of Definition~\ref{def:alt-transformed-system} with inhomogeneity $\swei{\breve{\mathfrak{H}}}{s}_k$.
\end{lemma}
\begin{proof} Let us introduce functions $\tilde d_{s,k,j}^X$ such that 
\begin{align*}
\breve{\mathfrak{H}}_k=\sum_{j=0}^{k}\sum_{X\in \{\underline{\mc L},\Phi,{\rm id}\}}\tilde d_{s,k,j}^X\cdot X\dbtilde \upphi_j
\end{align*}
For $k=0$, by direct inspection of \eqref{eq:transformed-R-k}, we easily compute
\begin{gather}
\begin{gathered}
\tilde d_{s,0,0}^{\rm id}=\lp(\frac{L\uL\xi}{w}-a^2\sin^2\theta TT\xi +2ias\cos\theta T\xi -\sign s\lp(|s|\frac{w'}{w} -\frac{\tilde c_0'}{\tilde c_0}\rp)\frac{\mc L \xi}{w}\rp)w\tilde c_0\,,\\
\tilde d_{s,0,1}^{\rm  id}=\lp(\underline{\mc L} \xi-a^2w\sin^2\theta T\xi\rp)w\tilde c_0\,,\quad
\tilde d_{s,0,0}^{\underline{\mc L}}=\lp(\frac{\mc L\xi}{w}-a^2\sin^2\theta T\xi\rp)w\tilde c_0\,,\quad
\tilde d_{s,0,0}^{\Phi}= -\frac{2aw \rho^2}{r^2+a^2}T\xi\tilde c_0\,.
\end{gathered}\label{eq:seed-d-function}
\end{gather}
To derive recursive formulas for the $\tilde d_{s,k,j}^X$, we notice
\begin{align*}
\mc L\lp(\frac{d_{s,k,j}^{\underline{\mc L}}}{w}\underline{\mc L}\upphi_j\rp) &= \mc L\lp(\frac{d_{s,k,j}^{\underline{\mc L}}}{w}\rp)\underline{\mc L}\upphi_j +\frac{d_{s,k,j}^{\underline{\mc L}}}{w}[\mc L,\underline{\mc L}]\upphi_j +\frac{d_{s,k,j}^{\underline{\mc L}}}{w}\underline{\mc L}(w\upphi_{j+1})\\
&=\mc L\lp(\frac{d_{s,k,j}^{\underline{\mc L}}}{w}\rp)\underline{\mc L}\upphi_j +\frac{2ar w}{r^2+a^2}\sign s\frac{d_{s,k,j}^{\underline{\mc L}}}{w}\Phi\upphi_j +d_{s,k,j}^{\underline{\mc L}}\underline{\mc L}\upphi_{j+1}+\sign s\frac{w'}{w}d_{s,k,j}^{\underline{\mc L}}\upphi_{j+1}\\
&=\mc L\lp(\frac{d_{s,k,j}^{\underline{\mc L}}}{w}\rp)\tilde c_j\underline{\mc L}\dbtilde\upphi_j+\sign s \tilde c_j'\mc L\lp(\frac{d_{s,k,j}^{\underline{\mc L}}}{w}\rp)\dbtilde\upphi_j +\frac{2ar }{r^2+a^2}\sign s d_{s,k,j}^{\underline{\mc L}}\tilde c_j\Phi\dbtilde\upphi_j \\
&\qquad+d_{s,k,j}^{\underline{\mc L}}\tilde c_{j+1}\underline{\mc L}\dbtilde\upphi_{j+1}+d_{s,k,j}^{\underline{\mc L}}\sign s \lp(\frac{\tilde c_{j+1}'}{\tilde c_{j+1}}+\frac{w'}{w}\rp)\tilde c_{j+1}\dbtilde\upphi_{j+1}\,,
\end{align*}
so that we have
\begin{align*}
\begin{dcases}
\tilde d_{s,k,j}^{ \rm id}&=\tilde d_{s,k-1,j-1}^{ \rm id}+\mc L(w^{-1}\tilde d_{s,k-1,j}^{ \rm id})+\sign s\lp(\frac{w'}{w}+\frac{c_{j}'}{c_j}\rp)d_{s,k-1,j-1}^{\underline{\mc L}}\,,\\
\tilde d_{s,k,j}^{\Phi}&=\tilde d_{s,k-1,j-1}^{\Phi}+\mc L(w^{-1}\tilde d_{s,k-1,j}^{\Phi})+\sign s\frac{2ar	}{r^2+a^2}d_{s,k-1,j}^{\underline{\mc L}}\,,\\
\tilde d_{s,k,j}^{\underline{\mc L}}&=\tilde d_{s,k-1,j-1}^{\underline{\mc L}}+\mc L(w^{-1}\tilde d_{s,k-1,j}^{\underline{\mc L}})\,.
\end{dcases}
\end{align*} 
From~\eqref{eq:seed-d-function}, it is now easy to obtain the asymptotics of the non-tilded, $d_{s,k,j}^X$, functions in the statement. 
\end{proof}

\begin{remark}[Comparison between two types of hyperboloidal cutoffs] \label{rmk:comparison-cutoff-types} Lemma~\ref{lemma:teukolsky-cutoff-inhom-k} shows that, when the cutoff $\xi$ is applied to the Teukolsky variable, weights consistent with peeling norms on the initial data will naturally emerge  as necessary to handle the inhomogeneity errors. In contrast, if the cutoff is applied to the transformed variable of level $k_0$ and the transport equations \eqref{eq:transformed-transport} are integrated to produce lower level transformed quantities, as in Lemma~\ref{lemma:weird-inverses} above, we are able to have weaker weights on the initial data, see Section~\ref{sec:current-errors}. This is precisely why we chose to work with the system of Section~\ref{sec:weird-cutoff-system} as our primary inhomogeneous system arising from hyperboloidal cutoffs.
\end{remark}

Since $\swei{\mathfrak{W}}{s}\neq 0$ in $\mc F_{\rm int}$ (see Theorem~\ref{thm:quantitative-mode-stab-in-text}),  we can, in principle, define Green's formulas for the corresponding inhomogeneous transformed system of radial ODEs in Lemma~\ref{lemma:transformed-system-ODE}. Formally, we expect
\begin{align}
\begin{split}
&\frac{1}{\swei{\mathfrak{W}}{s}}{\uppsi}^{[s],\,a,\,\omega}_{(k),\, ml,\,\mc{I}^+}(r^*)\int_{-\infty}^{r^*} \frac{(x^2+a^2)^{|s|}}{\Delta^{\frac{|s|}{2}(1+\sign s)}}\lp(\Delta^{\frac{s}{2}}{u}^{[s],\,a,\,\omega}_{ml,\,\mc{H}^+}\rp)(x^*)  \breve{\mathfrak{G}}^{[s],\,a,\,\omega}_{(0),\,ml}(x^*) dx^*\\
&\qquad +\frac{1}{\swei{\mathfrak{W}}{s}}{\uppsi}^{[s],\,a,\,\omega}_{(k),\,ml,\,\mc{H}^+}(r^*)\int_{r^*}^{\infty} \frac{(x^2+a^2)^{|s|}}{\Delta^{\frac{|s|}{2}(1+\sign s)}}\lp(\Delta^{\frac{s}{2}}{u}^{[s],\,a,\,\omega}_{ml,\,\mc{I}^+}\rp)(x^*) \breve{\mathfrak{G}}^{[s],\,a,\,\omega}_{(0),\,ml}(x^*) dx^*\,,
\end{split}\label{eq:cutoff-mode-stability-formal-Greens}
\end{align}
to yield a solution to \eqref{eq:transformed-k-separated} with outgoing boundary conditions. However, the integrals involved in \eqref{eq:cutoff-mode-stability-formal-Greens} are certainly not \textit{a priori} well-defined. We will find it convenient to introduce a regularization in the time frequency $\omega$: we let\footnote{Strictly speaking, we have not yet discussed how the radial ODE and its boundary conditions change if we allow $\omega\in \mathbb{C}$. However, the only point where the extension is nontrivial to the complex plane is nontrivial is our discussion of spin-weighted spheroidal harmonics: if the spheroidal parameter $\nu$ is complex then in general we cannot expect to identify $\Lambda_{ml}^{[s],\nu}$ without ambiguity in the labeling $l$. Nevertheless, if $\Im\nu$ is sufficiently small, it is easy to see that we can; see for instance the thesis \cite{TdC-thesis}.} $\omega_\epsilon:=\omega+i\epsilon$ for some $\epsilon>0$.  By adapting the arguments of the case $s=0$ in \cite[Theorem 1.9]{Shlapentokh-Rothman2015}, we obtain:
\begin{lemma}[Integral bounds  for Green's formula I]  \label{lemma:cutoff-mode-stability-integral-bounds} Fix $s\in\mathbb{Z}$ and assume  $(\omega,m,l)\in\mc{F}_{\rm int}$. Then, for any fixed $R^*<\infty$, we have the following uniform-in-$\epsilon$ estimates: for any $|a|\leq M$, we have
\begin{align}
\begin{split}
&\int_{\mc{F}_{\rm int}}\sum_{(m,l)\in\mc{F}_{\rm int}}\lp|\int_{-R^*}^{\infty} \frac{(r^2+a^2)^{|s|}}{\Delta^{\frac{|s|}{2}(1+\sign s)}}\lp(\Delta^{\frac{s}{2}}{u}^{[s],\,a,\,\omega_\epsilon}_{ml,\,\mc{I}^+}\rp)(r^*)  \breve{\mathfrak{G}}^{[s],\,a,\,\omega_\epsilon}_{(0),\,ml}(r^*) dr^*\rp|^2d\omega \\
&\quad\leq B(R^*,\omega_{\rm high},\omega_{\rm low},\varepsilon_{\rm width}) \sum_{i=0}^{N-1} \sum_{k=0}^{|s|}  {\mathbb{E}}[\swei{\tilde \upphi}{s}_{k}](\tau_{2i}) \,;
\end{split} \label{eq:cutoff-mode-stability-integral-bounds-I+}
\end{align}
and for either $a=M$ or, if $s\leq 0$, for $|a|< M$, we have
\begin{align}
&\int_{\mc{F}_{\rm int}}\sum_{(m,l)\in\mc{F}_{\rm int}}\lp(\mathbbm{1}_{\{s\leq 0\}}+
|\omega-m\omega_+|^{4|s|}\rp)\lp|\int_{-\infty}^{R^*}  \frac{(r^2+a^2)^{|s|}}{\Delta^{\frac{|s|}{2}(1+\sign s)}}\lp(\Delta^{\frac{s}{2}}{u}^{[s],\,a,\,\omega_\epsilon}_{ml,\,\mc{H}^+}\rp)(r^*)  \breve{\mathfrak{G}}^{[s],\, a,\,\omega_\epsilon}_{(0),\,ml}(r^*) dr^* \rp|^2d\omega\nonumber\\
&\quad\leq  B(R^*,\omega_{\rm high},\omega_{\rm low},\varepsilon_{\rm width}) \sum_{i=0}^{N-1} \sum_{k=0}^{|s|}  {\mathbb{E}}[\swei{\tilde\upphi}{s}_{k}](\tau_{2i}) \,. \label{eq:cutoff-mode-stability-integral-bounds-H+}
\end{align}  
\end{lemma}
\begin{proof} To ease the notation, we drop all the frequency parameters from the sub and superscripts, leaving only a superscript $\epsilon$ to indicate, for frequency space functions, that $\omega$ has been replaced by $\omega_\epsilon$  and, for physical space functions, that they have been multiplied by the regularization factor $e^{-\epsilon t}$.

We examine the integrals reaching $\pm \infty$ separately. \textit{A priori}, the integrands have too strong $r$-weights to allow for a direct application of Cauchy--Schwarz. We combine two strategies to circumvent this issue: (i) invoking  the definition of $\swei{\breve{\mathfrak{G}}}{s}_{(0)}$ and change from $(t,r^*)$ integration to integration over coordinates $(u,v)$ given by $v=\frac12(t+r^*)$ and $u=\frac12 (t-r^*)$ so as to uncover hidden cancellations, as in \cite[Lemma 3.3]{Shlapentokh-Rothman2015}, and (ii) integrate by parts in a way that improves the $r$-weights.

\medskip
\noindent\textit{Step 1: estimate \eqref{eq:cutoff-mode-stability-integral-bounds-I+}}. The term which we seek to estimate has the form
\begin{align}
&\lp|\int_{-R^*}^{\infty} \frac{(r^2+a^2)^{|s|}}{\Delta^{\frac{|s|}{2}(1+\sign s)}}\lp(\Delta^{\frac{s}{2}}\sweie{u}{s}_{\mc{I}^+}\rp)  \sweie{\breve{\mathfrak{G}}}{s}_{(0)} dr^*\rp|\leq \lp|\int_{-R^*}^{\infty} e^{i\omega r^*}\!c(r^*)\sweie{\breve{\mathfrak{G}}}{s}_{(0)} dr^*\rp|
+\lp|\int_{-R^*}^{\infty} \frac{e^{i\omega r^*}}{r^*} c_{\infty}\sweie{\breve{\mathfrak{G}}}{s}_{(0)} dr^*\rp|\,, \label{eq:cutoff-mode-stability-intermediate-1}
\end{align} 
for some $c(r^*)=\sum_{j=0}^{2|s|}c_j (r^*)^{2|s|-j}$, where the $c_j$, $j=0,\dots 2|s|$, and $c_\infty$ arise from the asymptotic expansion of $\sweie{u}{s}_{\mc{I}^+}$ as $r^*\to \infty$ (see \cite[Section 4.3.2]{SRTdC2020}) and depend on $R^*$, the frequency parameters and the spin $s$. For instance, $c_j=0$ for $j\leq 2|s|-1$ if $s\geq 0$ but this is not necessarily the case if $s<0$. From \eqref{eq:teukolsky-cutoff-inhom-k} with $k=0$, it is clear that  the second term in \eqref{eq:cutoff-mode-stability-intermediate-1}, in $L^2_{\omega\in\mc F_{\rm int}}l^2_{(m,l)\in\mc F_{\rm int}}$, can be controlled by initial data for ${\upphi}_{(0)}$ by a direct application of Cauchy--Schwarz:
\begin{align*}
\int_{\omega\in \mc{F}_{\rm int}}\sum_{(m,l)\in\mc{F}_{\rm int}}\lp|\int_{-R^*}^{\infty} \frac{e^{i\omega r^*}}{r^*} \sweie{\breve{\mathfrak{G}}}{s}_{(0)} dr^*\rp|^2
\leq \lp(\int_{-R^*}^{\infty}\frac{1}{(r^*)^2}dr^*\rp)\int_{\omega\in \mc{F}_{\rm int}}\sum_{(m,l)\in\mc{F}_{\rm int}}\lp(\int_{-R^*}^{\infty} \lp|\sweie{\breve{\mathfrak{G}}}{s}_{(0)}\rp|^2 dr^*\rp)d\omega\,,
\end{align*}
where we can now apply Plancherel, see Lemma~\ref{lemma:Plancherel}. The first term in \eqref{eq:cutoff-mode-stability-intermediate-1} does not have enough $r$-decay for the same strategy to work, so we focus on it for the remainder of the proof. 

By the definition of $\sweie{\breve{\mathfrak{G}}}{s}_{(0)}$, we have
\begin{align*}
&\int_{\omega\in \mc{F}_{\rm int}}\sum_{(m,l)\in\mc{F}_{\rm int}}\lp|\int_{R^*}^{\infty} e^{i\omega r^*}  c(x^*)\sweie{\breve{\mathfrak{G}}}{s}_{(0)} dr^*\rp|^2 d\omega
\\
&\quad= \int_{\omega\in \mc{F}_{\rm int}}\sum_{(m,l)\in\mc{F}_{\rm int}} \lp|\int_{\mathbb{S}^2}\int_{-\infty}^\infty\int_{R^*}^\infty e^{i\omega_\epsilon (t+r^*)}\sweie{\breve{\mathfrak{H}}}{s}_{(0)} c(r^*) S_{ml}^{[s],\,(a\omega_\epsilon)}(\theta)e^{-im\phi} dr^* dt d\sigma\rp|^2 d\omega\,.
\end{align*}
Now, we switch from $(t,r^*)$ variables to $(u,v)$ variables, where $v=\frac12(t+r^*)$ and $u=\frac12 (t-r^*)$. Since  $R^*>\infty$ is fixed, the integrand is compactly supported in the variable $u$, and supported in $v\in (A,\infty)$ where $A=A(R^*)$ is finite see Figure~\ref{fig:cutoff-support}.  Thus,
\begin{align}
\begin{split}
&\int_{\mc{F}_{\rm int}}\sum_{(m,l)\in\mc{F}_{\rm int}}\lp|\int_{-R^*}^{\infty} e^{i\omega r^*}  c(r^*)\sweie{\breve{\mathfrak{G}}}{s}_{(0)}  dr^*\rp|^2 d\omega \\
&\quad \leq \int_{\omega\in \mc{F}_{\rm int}}\sum_{(m,l)\in\mc{F}_{\rm int}} \int_{-\infty}^\infty\lp|\int_{\mathbb{S}^2}\int_{A}^\infty e^{2i\omega_\epsilon v}\sweie{\breve{\mathfrak{H}}}{s}_{0}  c(r^*) S_{ml}^{[s],\,a\omega_\epsilon}(\theta)e^{-im\phi} dv d\sigma\rp|^2 du d\omega\,,
\end{split} \label{eq:cutoff-mode-stability-intermediate-2}
\end{align}

\begin{figure}[htbp]
\centering
\includegraphics[scale=1]{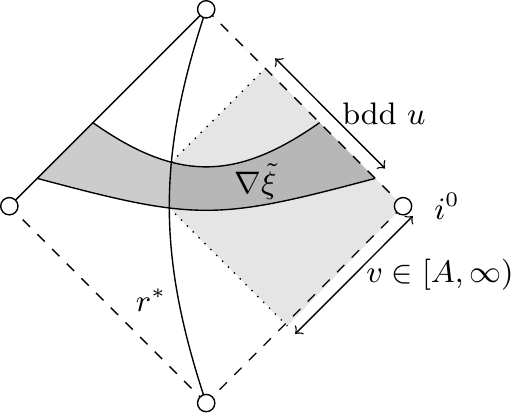}
\caption{A choice of hyperboloidal cuttof and its derivatives' support in spacetime.}
\label{fig:cutoff-support}
\end{figure}

For the term $j= 2|s|$ in $c$,  we can now conclude. From the orthonormality of the $S_{ml}^{[s],\,a\omega_\epsilon}(\theta)e^{-im\phi}$,
\begin{align*}
\int_{\mc{F}_{\rm int}}\sum_{(m,l)\in\mc{F}_{\rm int}}\lp|\int_{-R^*}^{\infty} e^{i\omega r^*} \sweie{\breve{\mathfrak{G}}}{s}_{(0)}  dr^*\rp|^2 d\omega 
&\leq B(\mc{F}_{\rm int}) \int_{\mathbb{S}^2}\int_{\mc{F}_{\rm int}}\int_{-\infty}^\infty\lp|\int_{A}^\infty e^{2i\omega_\epsilon v}\sweie{\breve{\mathfrak{H}}}{s}_{0}   dv \rp|^2 du d\omega d\sigma\\
& \leq B(R^*,\mc{F}_{\rm int}) \int_{\mathbb{S}^2}\int_{\mc{F}_{\rm int}}\int_{-\infty}^\infty\int_{-\infty}^\infty\lp|\swei{\breve{\mathfrak{H}}}{s}_{0}\rp|^2 dtdr^*d\sigma\,.
\end{align*}
Here,  \eqref{eq:inhom-commutator} guarantees that this can be controlled by initial data for $\upphi_{0}$ and, if $s\neq 0$, initial data for $\upphi_{1}$.  If $s\geq 0$, this concludes the proof; thus, from now on we assume that $s<0$.

To gain better $r$-weights,  we integrate \eqref{eq:cutoff-mode-stability-intermediate-2} by parts in the $L$ direction, using
\begin{align*}
L\lp[e^{2i\omega_\epsilon v}\rp] =4i\omega_\epsilon\,,
\end{align*}
and taking advantage of the fact that $0\not\in\mc F_{\rm int}$. Our regularization $\epsilon$ ensures that we do not pick up any boundary terms at $v=\infty$ in the integration; to avoid boundary terms at $v=A$ will also be convenient to note that the integral is unchanged if we multiply the integrand by a cutoff $0\leq \chi_{R}\leq 1$ which vanishes for $r^*\leq -R^*-1$ and is equal to 1 for $r^*\geq -R^*$. It is easy to see that, for the terms $j=2|s|-1$ and $j=2|s|-2$ in $c$, this is enough:
\begin{align*}
&|4i\omega_\epsilon|^2\int_{\omega\in \mc{F}_{\rm int}}\sum_{(m,l)\in\mc{F}_{\rm int}} \int_{-\infty}^\infty\lp|\int_{\mathbb{S}^2}\int_{A}^\infty L(e^{2i\omega_\epsilon v})\frac{\sweie{\breve{\mathfrak{H}}}{s}_{0}}{w} w (r^*)^2 S_{ml}^{[s],\,a\omega_\epsilon}(\theta)e^{-im\phi} dv d\sigma\rp|^2 du d\omega\\
&\quad\leq \int_{\mathbb{S}^2}\int_{\mc{F}_{\rm int}}\int_{-\infty}^\infty\lp|\int_{A}^\infty e^{2i\omega_\epsilon v}\sweie{\breve{\mathfrak{H}}}{s}_{1}  dv \rp|^2 du d\omega d\sigma +\int_{\mathbb{S}^2}\int_{\mc{F}_{\rm int}}\int_{-\infty}^\infty\lp|\int_{A}^\infty e^{2i\omega_\epsilon v}\sweie{\breve{\mathfrak{H}}}{s}_{0} (r^*)^2 \chi_{R}' dv \rp|^2 du d\omega d\sigma \\
&\quad\qquad+\int_{\mathbb{S}^2}\int_{\mc{F}_{\rm int}}\int_{-\infty}^\infty\lp|\int_{A}^\infty e^{2i\omega_\epsilon v}\frac{\sweie{\breve{\mathfrak{H}}}{s}_{0}}{w}\lp((r^*)^2w\rp)'   dv \rp|^2 du d\omega d\sigma\,,
\end{align*}
where we have used the transport relation~\eqref{eq:transformed-transport-inhom} for the inhomogeneities and the orthonormality of $S_{ml}^{[s],\,(a\omega_\epsilon)}(\theta)e^{-im\phi}$. By \eqref{eq:inhom-commutator} and Lemma~\ref{lemma:teukolsky-cutoff-inhom-k}, the first term and the last term can now be controlled by initial data for $\upphi_{0}$ and $\upphi_{1}$. The middle term term is compactly supported in both $u$ (as before) and $v$ (by the support properties of $\chi_R'$), so it is trivially controlled by the initial data for $\upphi_{0}$ and  $\upphi_{1}$, as its $r$-weights are not relevant.  

The above concludes the proof for $s=-1$. For $s\leq -2$, we simply need to iterate the integration by parts procedure $|s|$ times.

\medskip
\noindent\textit{Step 2: estimate \eqref{eq:cutoff-mode-stability-integral-bounds-H+}}. This step follows in a very similar manner. We seek to estimate
\begin{align}
\begin{split}
&\lp|\int_{-\infty}^{R^*} \frac{(r^2+a^2)^{|s|}}{\Delta^{\frac{|s|}{2}(1+\sign s)}}\lp(\Delta^{\frac{s}{2}}\sweie{u}{s}_{\mc{H}^+}\rp)  \sweie{\breve{\mathfrak{G}}}{s}_{(0)} dr^*\rp|
\\
&\leq \lp|\int_{-\infty}^{R^*}e^{-i(\omega_\epsilon -m\upomega_+)r^*} c(r)\sweie{\breve{\mathfrak{G}}}{s}_{(0)}dr^*\rp|
+\lp|\int_{-\infty}^{R^*}(r-r_+) c_{\infty}e^{-i(\omega_\epsilon -m\upomega_+) r^*} \sweie{\breve{\mathfrak{G}}}{s}_{(0)} dr^*\rp|\,, 
\end{split}\label{eq:cutoff-mode-stability-intermediate-3}
\end{align} 
for some $c(r^*)=\sum_{j=0}^{2|s|}c_j (r-r_+)^{j-2|s|}$, where the $c_j$, $j=0,\dots 2|s|$, and $c_\infty$ arise from the asymptotic expansion of $\sweie{u}{s}_{\mc{H}^+}$ at $r=r_+$ (see \cite[Section 4.3.2]{SRTdC2020}) and depend on $R^*$, the frequency parameters, the spin $s$, and the specific angular momentum $a$. For instance, if $s\leq 0$, then $c_j=0$ for $j\leq 2|s|-1$; if $s>0$ and $|a|<M$, then $c_j=0$ for $0\leq j< |s|$, but for $|a|=M$ generically $c_j\neq 0$. From \eqref{eq:teukolsky-cutoff-inhom-k} with $k=0$, it is clear that  the second term in \eqref{eq:cutoff-mode-stability-intermediate-3}, in $L^2_{\omega\in\mc F_{\rm int}}l^2_{(m,l)\in\mc F_{\rm int}}$, can be controlled by initial data for ${\upphi}_{(0)}$ by a direct application of Cauchy--Schwarz:
\begin{align*}
&\int_{\omega\in \mc{F}_{\rm int}}\sum_{(m,l)\in\mc{F}_{\rm int}}\lp|\int_{-\infty}^{R^*}(r-r_+) c_{\infty}e^{-i(\omega_\epsilon -m\upomega_+) r^*} \sweie{\breve{\mathfrak{G}}}{s}_{(0)} dr^*\rp|^2d\omega\\
&\quad\leq \int_{-\infty}^{R^*}(r-r_+)^2dr^*\int_{ \mc{F}_{\rm int}}\sum_{(m,l)\in\mc{F}_{\rm int}}\lp(\int_{-\infty}^{R^*} \lp|\sweie{\breve{\mathfrak{G}}}{s}_{(0)}\rp|^2 dr^*\rp)d\omega\,,
\end{align*}
where we can now apply Plancherel, see Lemma~\ref{lemma:Plancherel}. The first term in \eqref{eq:cutoff-mode-stability-intermediate-3} does not have enough $(r-r_+)$-decay for the same strategy to work, so we focus on it for the remainder of the proof. 

By the definition of $\sweie{\breve{\mathfrak{G}}}{s}_{(0)}$, we have
\begin{align*}
&\int_{\omega\in \mc{F}_{\rm int}}\sum_{(m,l)\in\mc{F}_{\rm int}}\lp|\int_{-\infty}^{R^*}e^{-i(\omega_\epsilon-m\upomega_+) r^*} c(r)\sweie{\breve{\mathfrak{G}}}{s}_{(0)}dr^*\rp|^2 d\omega
\\
&\quad= \int_{\omega\in \mc{F}_{\rm int}}\sum_{(m,l)\in\mc{F}_{\rm int}} \lp|\int_{\mathbb{S}^2}\int_{-\infty}^\infty\int_{-\infty}^{R^*} e^{i\omega_\epsilon(t-r^*)}
\sweie{\breve{\mathfrak{H}}}{s}_{0}  c(r)S_{ml}^{[s],\,(a\omega_\epsilon)}(\theta)e^{-im(\phi-\upomega_+r^*)} dr^* dt d\sigma\rp|^2 d\omega\,.
\end{align*}
As before, we switch from $(t,r^*)$ to $(u,v)$ as defined above. Since $R^*<\infty$ is fixed, the integrand is compactly supported in $v$ and supported in $v\in (-A,\infty)$ for some finite $A(R^*)$, symmetrically to the situation portrayed in Figure~\ref{fig:cutoff-support}.  Thus,
\begin{align*}
&\int_{\omega\in \mc{F}_{\rm int}}\sum_{(m,l)\in\mc{F}_{\rm int}}\lp|\int_{-\infty}^{R^*}e^{-i(\omega_\epsilon -m\upomega_+) r^*} c(r)\sweie{\breve{\mathfrak{H}}}{s}_{(0)}dr^*\rp|^2 d\omega
\\
&\quad \leq \int_{\omega\in \mc{F}_{\rm int}}\sum_{(m,l)\in\mc{F}_{\rm int}} \int_{-\infty}^\infty\lp|\int_{\mathbb{S}^2}\int_{-A}^{\infty} e^{2i\omega_\epsilon u+2im\upomega_+r^*} \sweie{\breve{\mathfrak{H}}}{s}_{0} c(r) S_{ml}^{[s],\,(a\omega_\epsilon)}(\theta)e^{-i(m\phi+m\upomega_+r^*)} du  d\sigma\rp|^2 dvd\omega\,.
\end{align*}

For the term $j=2|s|$, we can conclude as before. From the orthonormality of the angular functions,
\begin{align}
\begin{split}
&\int_{\omega\in \mc{F}_{\rm int}}\sum_{(m,l)\in\mc{F}_{\rm int}}\lp|\int_{-\infty}^{R^*}e^{-i(\omega_\epsilon -m\upomega_+) r^*} \sweie{\breve{\mathfrak{H}}}{s}_{0}dr^*\rp|^2 d\omega\\
&\quad\leq B(\mc F_{\rm int}) \int_{\mathbb{S}^2}\int_{\mc F_{\rm int}}\int_{-\infty}^\infty \lp|\int_{-A}^{\infty} e^{2i\omega_\epsilon u+im\upomega_+r^*} \sweie{\breve{\mathfrak{H}}}{s}_{0}  du  \rp|^2 dvd\omega d\sigma\\
&\quad\leq B(R^*,\mc F_{\rm int}) \int_{\mathbb{S}^2}\int_{\mc F_{\rm int}}\int_{-\infty}^\infty \int_{\mathbb{S}^2}\int_{-A}^{\infty}  \lp|\sweie{\breve{\mathfrak{H}}}{s}_{0}  du  \rp|^2 dvd\omega d\sigma\,,
\end{split}
\label{eq:cutoff-mode-stability-intermediate-4}
\end{align}
which, from \eqref{eq:inhom-commutator}, can clearly be controlled by initial data for ${\upphi}_{0}$ and, if $s\neq 0$, ${\upphi}_{1}$. If $s\leq 0$, this concludes the proof; thus from now on assume $s>0$ and $|a|=M$.

To gain better $r$-weights,  we can integrate \eqref{eq:cutoff-mode-stability-intermediate-2} by parts in the $\uL$ direction, using
\begin{align*}
\uL\lp[e^{2i\omega_\epsilon u+2im\upomega_+r^*}\rp] =4i(\omega_\epsilon-m\upomega_+)\,.
\end{align*}
Our regularization $\epsilon$ ensures that we do not pick up any boundary terms at $u=\infty$ in the integration if $|a|=M$; to avoid boundary terms at $u=A$ it will also be convenient to note that the integral is unchanged if we multiply the integrand by a cutoff $0\leq \chi_{R}\leq 1$ which vanishes for $r^*\geq R^*+1$ and is equal to 1 for $r^*\leq R^*$. It is easy to see that, for the terms $j=2|s|-1$ and $j=2|s|-2$ in $c$, this is enough:
\begin{align*}
&|4(\omega_\epsilon-m\upomega_+)|^2\int_{\omega\in \mc{F}_{\rm int}}\sum_{(m,l)\in\mc{F}_{\rm int}} \int_{-\infty}^\infty\lp|\int_{\mathbb{S}^2}\int_{A}^\infty \uL e^{2i\omega_\epsilon u}\sweie{\breve{\mathfrak{H}}}{s}_{0} \Delta^{-1} S_{ml}^{[s],\,(a\omega_\epsilon)}(\theta)e^{-im(\phi-\upomega_+)} du d\sigma\rp|^2 dv d\omega\\
&\quad\leq \int_{\mathbb{S}^2}\int_{\mc{F}_{\rm int}}\int_{-\infty}^\infty\lp|\int_{A}^\infty e^{2i\omega_\epsilon u}\sweie{\breve{\mathfrak{H}}}{s}_{1}e^{im\upomega_+r^*}   du \rp|^2 dv d\omega d\sigma \\
&\quad\qquad+\int_{\mathbb{S}^2}\int_{\mc{F}_{\rm int}}\int_{-\infty}^\infty\lp|\int_{A}^\infty e^{2i\omega_\epsilon u}\sweie{\breve{\mathfrak{H}}}{s}_{0} \Delta^{-1}\chi_{R}' e^{im\upomega_+r^*} du \rp|^2 dv d\omega d\sigma \\
&\quad\qquad+\int_{\mathbb{S}^2}\int_{\mc{F}_{\rm int}}\int_{-\infty}^\infty\lp|\int_{A}^\infty e^{2i\omega_\epsilon v}4r \sweie{\breve{\mathfrak{H}}}{s}_{0}  dv \rp|^2 du d\omega d\sigma\,,
\end{align*}
where we have used the transport relation \eqref{eq:transformed-transport-inhom} and the orthonormality of $S_{ml}^{[s],\,(a\omega_\epsilon)}(\theta)e^{-im\phi}$ again. From Lemma \ref{lemma:teukolsky-cutoff-inhom-k}, we can now conclude.
\end{proof}

In the case $s>0$ and $|a|<M$, we need a different approach. The first integral in \eqref{eq:cutoff-mode-stability-formal-Greens} is not well-defined, since even with our regularization one cannot, in general, expect the inhomogeneity to have sufficient decay as $r\to r_+$.  Thus, we will instead pursue a representation
\begin{align}
\uppsi_{(k),\, ml}^{[s],\,a,\,\omega} = \uppsi_{G, (k),\, ml}^{[s],\,a,\,\omega} +\uppsi_{R, (k),\, ml}^{[s],\,a,\,\omega}\,, \label{eq:decomposition-R-G}
\end{align}
where $\uppsi_{R, (0),\, ml}^{[s],\,a,\,\omega}$ verifies an equation,
\begin{align}
\lp(\uppsi_{R, (0),\, ml}^{[s],\,a,\,\omega}\rp)''+(\omega^2- \swei{\mc V}{s}_0)\uppsi_{R, (0),\, ml}^{[s],\,a,\,\omega} = \tilde{\mathfrak{G}}^{[s],\, a,\,\omega}_{(0),\,ml}:= \breve{\mathfrak{G}}^{[s],\, a,\,\omega}_{(0),\,ml}- \lp(\uppsi_{G, (0),\, ml}^{[s],\,a,\,\omega}\rp)''-(\omega^2- \swei{\mc V}{s}_0)\uppsi_{G, (0),\, ml}^{[s],\,a,\,\omega} \,, \label{eq:radial-ODE-psi0R}
\end{align}
whose  right hand side has improved decay as $r\to r_+$.

\begin{lemma}[Integral bounds  for Green's formula II]  \label{lemma:cutoff-mode-stability-integral-bounds-improved} Fix $s\in\mathbb{Z}_{\geq 1}$, and $a_0\in[0,M)$.  Assume $|a|\leq a_0$ and $(\omega,m,l)\in\mc{F}_{\rm int}$. Then, for any fixed $R^*<\infty$, we can find functions $\uppsi^{[s],\, a,\,\omega_\epsilon}_{G, (k),\,ml}(r^*)$ supported for $r^*\leq -R^*$ such that
\begin{align*}
&\sum_{k=0}^{|s|}  \int_{\mc{F}_{\rm int}}\sum_{(m,l)\in\mc{F}_{\rm int}}\int_{-R^*}^{R^*} \lp(|\uppsi^{[s],\, a,\,\omega_\epsilon}_{G, (k),\,ml}|^2+|(\uppsi^{[s],\, a,\,\omega_\epsilon}_{G, (k),\,ml})'|^2+|(\uppsi^{[s],\, a,\,\omega_\epsilon}_{G, (k),\,ml})''|^2\rp) dr^* \\
&\quad \leq B(R^*,a_0,\omega_{\rm high},\omega_{\rm low},\varepsilon_{\rm width}) \sum_{i=0}^{N-1} \sum_{k=0}^{|s|}  {\mathbb{E}}[\swei{\tilde\upphi}{s}_{k}](\tau_{2i})\numberthis\label{eq:cutoff-mode-stability-integral-bounds-H+-improved-subtraction}
\end{align*}
for $k=0,\dots, |s|$ and such that we have the following uniform-in-$\epsilon$ estimate:
\begin{align*}
&\int_{\mc{F}_{\rm int}}\sum_{(m,l)\in\mc{F}_{\rm int}}|\omega-m\upomega_+|^2\lp|\int_{-\infty}^{R^*}  \frac{(r^2+a^2)^{|s|}}{\Delta^{\frac{|s|}{2}(1+\sign s)}}\lp(\Delta^{\frac{s}{2}}{u}^{[s],\,a,\,\omega_\epsilon}_{ml,\,\mc{H}^+}\rp)(r^*)  \times \rp.\\
&\quad\qquad\qquad\qquad\qquad\qquad\qquad \lp.\times \lp[\breve{\mathfrak{G}}^{[s],\, a,\,\omega_\epsilon}_{(0),\,ml}(r^*)-\lp(\uppsi^{[s],\, a,\,\omega_\epsilon}_{G, (0),\,ml}(r^*)\rp)''-(\omega^2- \swei{\mc V}{s}_0)\uppsi^{[s],\, a,\,\omega_\epsilon}_{G, (0),\,ml}(r^*)\rp] dr^* \rp|^2d\omega\\
&\quad\leq  B(R^*,a_0,\omega_{\rm high},\omega_{\rm low},\varepsilon_{\rm width}) \sum_{i=0}^{N-1} \sum_{k=0}^{|s|}  {\mathbb{E}}[\swei{\tilde\upphi}{s}_{k}](\tau_{2i}) \,. \numberthis\label{eq:cutoff-mode-stability-integral-bounds-H+-improved}
\end{align*}
\end{lemma}
\begin{proof} For brevity, we use the notation introduced in \eqref{eq:radial-ODE-psi0R} and drop most sub and superscripts. In view of the radial Teukolsky--Starobinsky identities of Proposition~\ref{prop:TS-radial-constant-identities}, we have the identity 
\begin{align*}
&\frac{(r^2+a^2)^{s}}{\Delta^{s}}\Delta^{\frac{s}{2}}{u}^{[s],\,\epsilon}_{\mc{H}^+} \tilde{\mathfrak{G}}^{[s],\,\epsilon}_{(0)} \mathfrak{C}_s^{(2)}\lp\{\begin{array}{lr}
(r_+-r_-)^{s} &|a|<M\\
1&|a|=M
\end{array}\rp\}\\
&\quad=(r^2+a^2)^{s-1/2}\frac{\tilde{\mathfrak{G}}^{[s],\,\epsilon}_{(0)}}{w}w(r^2+a^2)\lp(\frac{L}{w(r^2+a^2)}\rp)^{2s}\lp((r^2+a^2)^{s-1/2}{\uppsi}^{[-s],\,\epsilon}_{(0),\,\mc{H}^+}\rp)\,.
\end{align*}
Note that the constant $\mathfrak{C}_s^{(2)}$ is the source of the $\omega-m\upomega_+$ weights in \eqref{eq:cutoff-mode-stability-integral-bounds-H+-improved}. Thus, we have reduced the proof of \eqref{eq:cutoff-mode-stability-integral-bounds-H+-improved} to estimating the integral  
\begin{align*}
&\int_{-\infty}^{R^*} (r^2+a^2)^{s-1/2}\frac{\tilde{\mathfrak{G}}_{(0)}^{[+s],\,\epsilon}}{w}w(r^2+a^2)\lp(\frac{L}{w(r^2+a^2)}\rp)^{2s}\lp((r^2+a^2)^{s-1/2}\uppsi_{(0),\,\mc H^+}^{[-s],\,\epsilon}\rp)dr^* \numberthis\label{eq:cutoff-mode-stab-improved-intermediate}
\end{align*}
in $L^2_{\omega\in\mc F_{\rm int}}l^2_{ml\in\mc F_{\rm int}}$, where the integration by parts structure is evident. 

\medskip
\noindent\textit{Step 1: the case $s=1$.} As a first example, for $s=1$, integration by parts in \eqref{eq:cutoff-mode-stab-improved-intermediate} without the tilde leads to the identity
\begin{align}
\begin{split}
&\int_{-\infty}^{R^*} \chi_R (r^2+a^2)^{1/2}\frac{\mathfrak{G}_{(0)}^{[+1],\,\epsilon}}{w}w(r^2+a^2)\lp(\frac{L}{w(r^2+a^2)}\rp)^{2}\lp((r^2+a^2)^{1/2}\uppsi_{(0),\,\mc H^+}^{[-1],\,\epsilon}\rp)dr^* \\
&\quad =\int_{-\infty}^{R^*}\frac{d}{dr^*}\lp[(r^2+a^2)^{-1/2}\lp(r\uppsi_{(0),\,\mc H^+}^{[-1],\,\epsilon}+\uppsi_{(1),\,\mc H^+}^{[-1],\,\epsilon}\rp)\frac{\mathfrak{G}_{(0),\,ml}^{[+1],\,\epsilon}}{w}\chi_R\rp]dr^*\\
&\quad\qquad +\int_{-\infty}^{R^*}
\lp(\uppsi_{(1),\,\mc H^+}^{[-1],\,\epsilon}-r\uppsi_{(0),\,\mc H^+}^{[-1],\,\epsilon}\rp)\lp(\mathfrak{G}_{(1)}^{[+1],\,\epsilon}-r\mathfrak{G}_{(0)}^{[+1],\,\epsilon}\rp) \chi_R dr^*\\
&\quad\qquad +\int_{-\infty}^{R^*}\chi_R'\lp[(r^2+a^2)^{-1/2}\lp(r\uppsi_{(0),\,\mc H^+}^{[-1],\,\epsilon}+\uppsi_{(1),\,\mc H^+}^{[-1],\,\epsilon}\rp)\frac{\mathfrak{G}_{(0)}^{[+1],\,\epsilon}}{w}\rp]dr^*\,,
\end{split}\label{eq:cutoff-mode-stability-intermediate-5}
\end{align}
where $\chi_R$ is a bump function which is introduced to ensure there are no boundary terms at $r^*=R^*$. In the above, the boundary term at $r^*=-\infty$ vanishes due to the regularization $\epsilon$. The bulk terms with compact $r^*$ support (last line) can be controlled by Cauchy--Schwarz. Finally, for the bulk terms without compact $r^*$ support (middle line), we apply the same arguments as in Step 2 of the proof of Lemma~\ref{lemma:cutoff-mode-stability-integral-bounds}: from the outgoing boundary conditions of $\uppsi^{[+s],\,\epsilon}_{(k)}$, it is easy to see that these integrals can each be estimated similarly to \eqref{eq:cutoff-mode-stability-intermediate-3}, where $c(r)$ is a constant. Thus, we arrive at \eqref{eq:cutoff-mode-stability-integral-bounds-H+-improved} with the choice $\uppsi_{G, (k),\, ml}^{[s],\,a,\,\omega}\equiv 0$.

\medskip
\noindent\textit{Step 2: the cases $s\geq 2$.} In the case $s\geq 2$, we integrate by parts in \eqref{eq:cutoff-mode-stab-improved-intermediate} in a similar fashion to that of the previous step to distribute the derivatives between the homogeneous solution and the inhomogeneity. Again, we produce three types of terms: (i) bulk terms with compact $r^*$ support which (as long as enough integrations by parts have been performed) can be controlled by Cauchy--Schwarz; (ii) bulk terms without compact $r^*$ support of the form
\begin{align*}
\int_{-\infty}^{R^*}\lp(\sum_{k=0}^{|s|}p_{s,|s|,k}(r)\tilde{\mathfrak{G}}^{[+s],\,\epsilon}_{(k)}\rp)\lp(\sum_{k=0}^{|s|}p_{-s,|s|,k}(r)\uppsi^{[+s],\,\epsilon}_{(k)}\rp)dr^*\,,
\end{align*}
where $p_{\pm s,|s|,k}$ are polynomials in $r$ of degree $|s|-k$; and (iii) boundary terms at $r^*=-\infty$.  However, it is not hard to see that for $s\geq 2$ the structure of the boundary terms is already very different from the $s=1$ case: as $\frac{L}{w(r^2+a^2)}$ derivatives will act on $\uppsi_{(k),\, \mc H^+}^{[-s],\,\epsilon}$ in the boundary terms, these may in fact be infinite unless the inhomogeneities benefit from more decay.  Thus, we must construct a suitable $\uppsi_{G, (k),\, ml}^{[s,\,a,\,\omega]}\not\equiv 0$.

Our construction is motivated by the asymptotic analysis of \cite[Section 4.3.2]{SRTdC2020}. Since $\swei{\breve{\mathfrak{H}}}{s}_0/w$ is smooth up to $\mc H^+$, from Lemma~\ref{lemma:teukolsky-cutoff-inhom-k}, we deduce that
\begin{align*}
\frac{\mathfrak{G}_{(0)}^{[s]}}{w}=(r-r_+)^{\xi}\lp(\sum_{j=0}^{|s|-2}G_{(0),j}^{[s],r_+,+}(r-r_+)^j + O((r-r_+)^{|s|+1})\rp) 
\end{align*}
for some coefficients $G_{(0),k,j}\in\mathbb{C}$.  We let $\uppsi_{G, (0),\, ml}^{[s,\,a,\,\omega]}$ be defined by 
\begin{align*}
\uppsi_{G, (0)}^{[s]}(r^*) := \chi(r^*) \sum_{j=1}^{|s|-1} b_{G,(0),j}^{[s],r_+,+}(r-r_+)^{k+\xi}\,,
\end{align*}
where $\chi$ is a cutoff function which is supported close to $r=r_+$, and where $b_{G,(0),j}^{[s],r_+,+}$ are obtained recursively through
\begin{align*}
\lp[(\xi-s+j)(\xi-s+j-1)+ (\xi-s+j)f_0 +g_0\rp]b_{G,(0),j}^{[s],r_+,+} = G_{(0),j-1}^{[s],r_+,+}-\sum_{i=1}^{j-1}\lp((\xi-s+j)f_{j-i} +g_{j-i}\rp)b_{G,(0),i}^{[s],r_+,+}\,,
\end{align*}
with $f_0$ and $g_0$ explicit functions of $(\omega,m,l)$, $s$ and $(a,M)$ given in \cite[Lemma 4.3.1]{SRTdC2020}. We then define
\begin{align*}
\uppsi_{G, (k)}^{[s]}(r^*) :=\lp(w^{-1}\uL\rp)^k \uppsi_{G, (0)}^{[s]}(r^*)\,, 
\end{align*}
for $k=0,\dots,|s|$; hence we easily obtain
\begin{align*}
\int_{-R^*}^{R^*}\lp(|\uppsi_{G, (k)}^{[s]}|^2+|(\uppsi_{G, (k)}^{[s]})'|^2+|(\uppsi_{G, (k)}^{[s]})''|^2\rp) \leq B(R^*, \mc F_{\rm int}) \sum_{j=0}^{|s|-2}\lp|\frac{\mathfrak{G}_{(j)}^{[s]}}{w}\rp|^2(-\infty) \,,
\end{align*}
and \eqref{eq:cutoff-mode-stability-integral-bounds-H+-improved-subtraction} follows by Plancherel, Lemma~\ref{lemma:teukolsky-cutoff-inhom-k}, and the fact that the transport equation \eqref{eq:transformed-transport} allows us to control overbarred energies for $k<|s|$ by non-overbarred energies for all $k=0,\dots |s|$.

Repeating \cite[Lemmas 4.3.2, 4.3.3]{SRTdC2020} \textit{mutatis mutandis}, we find that 
\begin{align*}
\tilde{\mathfrak{G}}^{[s]}_{(k)}(r^*):= w\lp(\frac{\uL}{w}\rp)^k\lp[\frac{{\mathfrak{G}}^{[s]}_{(0)}(r^*)}{w}-\frac{\lp(\uppsi^{[s]}_{G, (0)}(r^*)\rp)''-(\omega^2- \swei{\mc V}{s}_0)\uppsi^{[s]}_{G, (0)}(r^*)}{w}\rp]  =O((r-r_+)^{|s|-k}) \numberthis\label{eq:tilde-G-decay}
\end{align*}
as $r\to r_+$. This decay is just enough to make the boundary terms (iii) finite, but adding the regularization $\omega\mapsto\omega_\epsilon$ will ensure that they in fact vanish.  The bulk integrals (ii) can then be decomposed into two terms
\begin{align*}
\int_{-\infty}^{R^*}\lp(\sum_{k=0}^{|s|}p_{s,|s|,k}(r){\mathfrak{G}}^{[+s],\,\epsilon}_{(k)}\rp)\lp(\sum_{k=0}^{|s|}p_{-s,|s|,k}(r)\uppsi^{[+s],\,\epsilon}_{(k)}\rp)dr^* \\
+\int_{-\infty}^{R^*}\lp[\sum_{k=0}^{|s|}p_{s,|s|,k}(r)\lp(\tilde{\mathfrak{G}}^{[+s],\,\epsilon}_{(k)}-{\mathfrak{G}}^{[+s],\,\epsilon}_{(k)}\rp)\rp]\lp(\sum_{k=0}^{|s|}p_{-s,|s|,k}(r)\uppsi^{[+s],\,\epsilon}_{(k)}\rp)dr^*\,,
\end{align*}
which we may estimate separately. For the first term, we can proceed as in the step above or Step 2 in Lemma~\ref{lemma:cutoff-mode-stability-integral-bounds}. For the second term, we note that the difference $\tilde{\mathfrak{G}}^{[+s]}_{(k)}-{\mathfrak{G}}^{[+s]}_{(k)}$ is entirely due to $\uppsi^{[+s]}_{G, (k)}$, which we have already estimated.
\end{proof}

Theorem~\ref{thm:quantitative-mode-stab-in-text} and Lemmas~\ref{lemma:cutoff-mode-stability-integral-bounds} and \ref{lemma:cutoff-mode-stability-integral-bounds-improved} are the heart of the proof of Proposition~\ref{prop:ODE-ILED-compact-r-mode-stab}. Indeed,   Proposition~\ref{prop:ODE-ILED-compact-r-mode-stab} now follows easily:

\begin{proof}[Proof of Proposition~\ref{prop:ODE-ILED-compact-r-mode-stab}]
Let us first assume $s\leq 0$ and $|a|=M$. By the assumptions, the $\epsilon$-dependent Green's formula 
\begin{align*}
\begin{split}
u&=\frac{1}{\swei{\mathfrak{W}}{s}(\omega_\epsilon,m,l)} {u}^{[s],\,a,\,\omega_\epsilon}_{ml,\,\mc{I}^+}(r^*)\int_{-\infty}^{r^*} \frac{(x^2+a^2)^{|s|}}{\Delta^{\frac{|s|}{2}(1+\sign s)}}\lp(\Delta^{\frac{s}{2}}{u}^{[s],\,a,\,\omega_\epsilon}_{ml,\,\mc{H}^+}\rp)(x^*)  \breve{\mathfrak{G}}^{[s],\,a,\,\omega_\epsilon}_{(0),\,ml}(x^*) dx^*\\
&\qquad +\frac{1}{\swei{\mathfrak{W}}{s}(\omega_\epsilon,m,l)}{u}^{[s],\,a,\,\omega_\epsilon}_{ml,\,\mc{H}^+}(r^*)\int_{r^*}^{\infty} \frac{(x^2+a^2)^{|s|}}{\Delta^{\frac{|s|}{2}(1+\sign s)}}\lp(\Delta^{\frac{s}{2}}{u}^{[s],\,a,\,\omega_\epsilon}_{ml,\,\mc{I}^+}\rp)(x^*) \breve{\mathfrak{G}}^{[s],\,a,\,\omega_\epsilon}_{(0),\,ml}(x^*) dx^*\,,\\
&\quad= \frac{1}{\swei{\mathfrak{W}}{s}(\omega_\epsilon,m,l)} {u}^{[s],\,a,\,\omega_\epsilon}_{ml,\,\mc{I}^+}(r^*)\int_{-\infty}^{r^*}{u}^{[s],\,a,\,\omega_\epsilon}_{ml,\,\mc{H}^+}(x^*)  {H}^{[s],\,a,\,\omega_\epsilon}_{(0),\,ml}(x^*) dx^*\\
&\qquad\quad +\frac{1}{\swei{\mathfrak{W}}{s}(\omega_\epsilon,m,l)}{u}^{[s],\,a,\,\omega_\epsilon}_{ml,\,\mc{H}^+}(r^*)\int_{r^*}^{\infty} {u}^{[s],\,a,\,\omega_\epsilon}_{ml,\,\mc{I}^+}(x^*) {H}^{[s],\,a,\,\omega_\epsilon}_{(0),\,ml}(x^*) dx^*\,,
\end{split}
\end{align*} 
is well-defined and (uniquely) specifies a solution to the inhomogeneous radial ODE~\eqref{eq:radial-ODE-u} with $\omega=\omega_\epsilon$, which has outgoing boundary conditions. This follows from standard ODE theory directly, as in the $s=0$ case worked out in \cite[Proposition 3.1]{Shlapentokh-Rothman2015}. Furthermore, taking derivatives of the expression above gives
\begin{align*}
\lp(\frac{d}{dr^*}\rp)^Nu&=\frac{1}{\swei{\mathfrak{W}}{s}(\omega_\epsilon,m,l)} \lp(\frac{d}{dr^*}\rp)^N\lp({u}^{[s],\,a,\,\omega_\epsilon}_{ml,\,\mc{I}^+}\rp)(r^*)\int_{-\infty}^{r^*} \frac{(x^2+a^2)^{|s|}}{\Delta^{\frac{|s|}{2}(1+\sign s)}}\lp(\Delta^{\frac{s}{2}}{u}^{[s],\,a,\,\omega_\epsilon}_{ml,\,\mc{H}^+}\rp)(x^*)  \breve{\mathfrak{G}}^{[s],\,a,\,\omega_\epsilon}_{(0),\,ml}(x^*) dx^*\\
& +\frac{1}{\swei{\mathfrak{W}}{s}(\omega_\epsilon,m,l)}\lp(\frac{d}{dr^*}\rp)^N\lp({u}^{[s],\,a,\,\omega_\epsilon}_{ml,\,\mc{H}^+}\rp)(r^*)\int_{r^*}^{\infty} \frac{(x^2+a^2)^{|s|}}{\Delta^{\frac{|s|}{2}(1+\sign s)}}\lp(\Delta^{\frac{s}{2}}{u}^{[s],\,a,\,\omega_\epsilon}_{ml,\,\mc{I}^+}\rp)(x^*) \breve{\mathfrak{G}}^{[s],\,a,\,\omega_\epsilon}_{(0),\,ml}(x^*) dx^*\,.
\end{align*} 
Lemma~\ref{lemma:cutoff-mode-stability-integral-bounds} and Theorem~\ref{thm:quantitative-mode-stab-in-text} ensures that for $r^*$ finite, each of the above integrals, weighted by $(\mathfrak{W}^{[s]})^{-1}$, converges to their $\epsilon=0$ counterparts in $L^2_{\omega\in\mc F_{\rm int}}l^2_{(m,l)\in\mc F_{\rm int}}$. The same is true for ${u}^{[s],\,a,\,\omega_\epsilon}_{ml,\,\mc{I}^+}$ and ${u}^{[s],\,a,\,\omega_\epsilon}_{ml,\,\mc{H}^+}$, as well as their derivatives. Thus, the above formulas yield a representation in $L^2_{\omega\in\mc F_{\rm int}}l^2_{(m,l)\in\mc F_{\rm int}}$ for $\tilde\uppsi_{(0)}$ and its $r^*$ derivatives; and, by differentiation, for any $k=0,\dots, |s|$ they also yield the representation \eqref{eq:cutoff-mode-stability-formal-Greens}  for $\tilde\uppsi_{(k)}$ in $L^2_{\omega\in\mc F_{\rm int}}l^2_{(m,l)\in\mc F_{\rm int}}$ and a similar one for its $r^*$ derivatives. The statement of the proposition then easily follows for $s\leq 0$ and for $|a|=M$.
 
On the other hand, if $|a|<M$ and $s>0$, we have an additional step in the proof: we take \eqref{eq:decomposition-R-G} and apply the Green's formula only to the second term, considering \eqref{eq:radial-ODE-psi0R}. To be precise, we note that for $\epsilon>0$, the formula 
\begin{align*}
\uppsi_{(k),\, ml}^{[s],\,a,\,\omega_\epsilon} &= \uppsi_{G, (0),\, ml}^{[s],\,a,\,\omega_\epsilon} \\
&\qquad +
\frac{{\uppsi}^{[s],\,a,\,\omega_\epsilon}_{(k),\, ml,\,\mc{I}^+}(r^*)}{\swei{\mathfrak{W}}{s}(\omega_\epsilon,m,l)}\int_{-\infty}^{r^*} \frac{(x^2+a^2)^{|s|}}{\Delta^{\frac{|s|}{2}(1+\sign s)}}\lp(\Delta^{\frac{s}{2}}{u}^{[s],\,a,\,\omega_\epsilon}_{ml,\,\mc{H}^+}\rp)(x^*) \tilde{\mathfrak{G}}^{[s],\,a,\,\omega_\epsilon}_{(0),\,ml}(x^*) dx^*\\
&\qquad +\frac{1}{\swei{\mathfrak{W}}{s}(\omega_\epsilon,m,l)}{\uppsi}^{[s],\,a,\,\omega_\epsilon}_{(k),\,ml,\,\mc{H}^+}(r^*)\int_{r^*}^{\infty} \frac{(x^2+a^2)^{|s|}}{\Delta^{\frac{|s|}{2}(1+\sign s)}}\lp(\Delta^{\frac{s}{2}}{u}^{[s],\,a,\,\omega_\epsilon}_{ml,\,\mc{I}^+}\rp)(x^*) \tilde{\mathfrak{G}}^{[s],\,a,\,\omega_\epsilon}_{(0),\,ml}(x^*) dx^*\,,
\end{align*}
is well-defined by \eqref{eq:tilde-G-decay}, and (uniquely) specifies a solution to the inhomogeneous radial ODE~\eqref{eq:transformed-k-separated} with $\omega=\omega_\epsilon$, which has outgoing boundary conditions. A similar representation can be obtained for its derivatives. Lemma~\ref{lemma:cutoff-mode-stability-integral-bounds-improved} and Theorem~\ref{thm:quantitative-mode-stab-in-text} ensures that for $r^*$ finite, each of the three terms in the above identity converges to their $\epsilon=0$ counterparts in $L^2_{\omega\in\mc F_{\rm int}}l^2_{(m,l)\in\mc F_{\rm int}}$. Thus, replacing $\omega_\epsilon$ by $\omega$ we obtain a representation for $\uppsi_{(k),\, ml}^{[s],\,a,\,\omega}$, and its $r^*$ derivatives, in $L^2_{\omega\in\mc F_{\rm int}}l^2_{(m,l)\in\mc F_{\rm int}}$. From these formulas, the conclusion of the proposition follows easily.
\end{proof}

\subsubsection{Boundary terms at \texorpdfstring{$r=r_+$}{the horizon} in the bounded frequency range} 
\label{sec:bdry-terms}

In this section, we estimate the boundary terms at $r=r_+$ which are generated from performing integrated estimates on the inhomogeneous transformed system of Section~\ref{sec:weird-cutoff-system}. More precisely, the our goal is to estimate the very last lines of \eqref{eq:part1-ILED-withoutTS-top} and \eqref{eq:part1-ILED-withoutTS-bottom} in our Theorem~\ref{thm:ODE-ILED}, which first appeared in \cite{SRTdC2020}.

\begin{lemma}[Boundary terms at horizon I] \label{lemma:weird-inverses-bdry-term-rp} Fix $a_0\in[0,M)$, $s\in\{\pm 1,\pm 2\}$ and $k_0\in\{0,\dots,|s|\}$ as in Lemma~\ref{lemma:weird-inverses}. Then, 
\begin{align*}
& \int_{\mc F_{\rm low}}\sum_{(m,l)\in\mc F_{\rm low}}(\omega^2+m^2+1)^{4|s|-k_0}|\smlambdak{\uppsi}{+s}{0}|^2 d\omega\\
&\qquad+ \int_{\mc F_{\rm high}^c}\sum_{(m,l)\in\mc F_{\rm high}^c}(\omega^2+m^2+1)^{4|s|-k_0-1}\sum_{j=0}^{k_0-1}\lp(|w^{-1}\smlambdak{\mathfrak G}{+s}{j}|^2+|w^{-1}(\smlambdak{\mathfrak G}{s}{j})'|^2\rp)\Big|_{r=-\infty}d\omega\\
 &\quad\leq B(a_0)\sum_{i=0}^{N-1}\sum_{k=0}^{|s|}\overline{\mathbb{E}}^{4|s|-k-2}[\swei{\tilde\upphi}{+s}_k](\tau_{2i})
\,. \numberthis\label{eq:weird-inverses-bdry-term-rp}
\end{align*}
However, if $\xi$ and $\swei{\upphi}{s}_k$ are such that $\swei{\breve{\mathfrak{H}}}{s}_{0}$, defined in \eqref{eq:teukolsky-cutoff-inhom-k}, is compactly supported in $r^*$ for every $k\in\{0,\dots,|s|\}$, then we have estimates at a lower level of regularity: 
\begin{align*}
\int_{\mc F_{\rm low}}\sum_{(m,l)\in\mc F_{\rm low}}|\smlambdak{\uppsi}{+s}{0}|^2 d\omega\leq B(a_0)\sum_{i=0}^{N-1}\sum_{k=0}^{|s|}{\mathbb{E}}[\swei{\tilde\upphi}{+s}_k](\tau_{2i})\,.\numberthis\label{eq:weird-inverses-bdry-term-rp-best}
\end{align*}
\end{lemma}
\begin{proof}  Let us first note that multiplying by $(\omega^2+m^2+1)$ has no effect as these frequencies are all bounded in $\mc F_{\rm high}^c$. To control the boundary terms due to the inhomogeneities, we can apply Plancherel together with the estimates of Lemma~\ref{lemma:weird-cutoffs-inhoms-rough-estimates}. Thus, it remains to estimate the boundary term of $\smlk{\uppsi}{+s}{0}$. 

We introduce the notation $\smlk{\breve \uppsi}{+s}{k}$ for the radial function obtained from  $\swei{\upphi}{+s}_{0,\cut}:=\xi \swei{\upphi}{+ s}_{0}$ according to \eqref{eq:def-sml-uppsi-G}. We have
\begin{align*}
&\int_{\mc F_{\rm high}^c}\sum_{(m,l)\in \mc F_{\rm high}^c}|\smlk{ \uppsi}{+s}{0}|^2(-\infty)d\omega\\
&\quad\leq \int_{\mc F_{\rm high}^c}\sum_{(m,l)\in \mc F_{\rm high}^c}|\smlk{\breve \uppsi}{+s}{0}|^2(-\infty)d\omega + \int_{\mc H^+_{(\tau_0,\infty)}}|\xi\swei{\upphi}{+ s}_{0}-\swei{\upphi}{+s}_{0,\cutt}|^2d\tau d\sigma\\
&\quad\leq B\int_{\mc F_{\rm high}^c}\sum_{(m,l)\in \mc F_{\rm high}^c}|\smlk{\breve \uppsi}{+s}{0}|^2(-\infty)d\omega + B\sum_{i=1}^N\sum_{k=0}^{|s|}\mathbb{E}[\swei{\tilde\upphi}{s}_k](\tau_i)\,,
\end{align*}
by the estimates of Lemma~\ref{lemma:weird-cutoffs-diff-fake-real}. 

If $\breve{\mathfrak{H}}_0$ is compactly supported in $r^*$, then the remaining inhomogeneities $\breve{\mathfrak{H}}_k$ obtained by the forward cutoffs of Lemma~\ref{lemma:teukolsky-cutoff-inhom-k} are themselves compactly supported in $r^*$. Thus,  $\upphi_{k,\cut}:= (w^{-1}\uL)^k\breve{\upphi}_0$ and $\breve{\mathfrak{H}}_k$ are outgoing and sufficiently integrable solutions of the transformed system of Definition~\ref{def:transformed-system}, where moreover  $\breve{\mathfrak{H}}_k$ has strong decay at $r=r_+$. By Theorem~\ref{thm:ODE-ILED}\ref{it:thm-ODE-ILED-bdry-term}, \eqref{eq:part1-bdry-withTS} holds. Summing in $(\omega,m,l)$ and using the estimates of Lemma~\ref{lemma:normal-current-errors-nonpeeling}, we have
\begin{align*}
&\int_{\mc F_{\rm low}}\sum_{(m,l)\in \mc F_{\rm low}}|\smlk{\breve \uppsi}{+s}{0}|^2(-\infty)d\omega \\
&\quad\leq B(a_0)\sum_{i=0}^{N-1}\lp(\sum_{k=0}^{|s|}\mathbb{E}[\tilde\upphi_k](\tau_{2i}) +\int_{\mc H^+_{(\tau_{2i},\tau_{2i+1})}}|\Phi|d\tau d\sigma\rp)\\ &\qquad +B(a_0)\sum_{k=0}^{|s|}\int_{\mc F_{\rm low}}\sum_{(m,l)\in \mc F_{\rm low}}\int_{R^*_{\rm bdry}}^{R^*_{\rm bdry}+1}\lp(|\p_{r^*}(w^{-1}\uL)^k(\smlk{\breve\uppsi}{s}{0})|^2+|(w^{-1}\uL)^k(\smlk{\breve\uppsi}{s}{0})|^2\rp)dr^*d\omega\,.
\end{align*}
We then appeal to the estimates of Lemma~\ref{lemma:weird-cutoffs-diff-fake-real} to replace the transformed variables in last term by $\smlk{\uppsi}{s}{k}$; then,  \eqref{eq:part1-ILED-withoutTS-top-low} in Theorem~\ref{thm:ODE-ILED}\ref{it:thm-ODE-ILED-bdry-term} applies, so combining this with Lemma~\ref{lemma:normal-current-errors-nonpeeling}, we obtain
\begin{align*}
&\int_{\mc F_{\rm low}}\sum_{(m,l)\in \mc F_{\rm high}^c}|\smlk{\breve \uppsi}{+s}{0}|^2(-\infty)d\omega \\
&\quad\leq B(a_0)\sum_{i=0}^{N-1}\lp(\sum_{k=0}^{|s|}{\mathbb{E}}[\tilde\upphi_k](\tau_{2i}) +\int_{\mc H^+_{(\tau_{2i},\tau_{2i+1})}}|\Phi|d\tau d\sigma\rp)\\ &\quad\qquad +B(a_0)\sum_{k=0}^{|s|}\int_{\mc F_{\rm low}}\sum_{(m,l)\in \mc F_{\rm low}}\int_{R^*_{\rm bdry}}^{R^*_{\rm bdry}+1}\lp(|(\smlk{\uppsi}{s}{k})'|^2+|\smlk{\uppsi}{s}{k}|^2\rp)dr^*d\omega\\
&\qquad \leq B(a_0)\sum_{i=0}^{N-1}\lp(\sum_{k=0}^{|s|}{\mathbb{E}}[\tilde\upphi_k](\tau_{2i}) +\int_{\mc H^+_{(\tau_{2i},\tau_{2i+1})}}|\Phi|d\tau d\sigma\rp) + \frac{B(a_0)}{R_{\rm bdry}^*}\int_{\mc F_{\rm low}}\sum_{(m,l)\in \mc F_{\rm low}}|\smlk{\breve \uppsi}{+s}{0}|^2(-\infty)d\omega\,,
\end{align*}
where we now take $R_{\rm bdry}^*$ to be sufficiently large (thus finally fixing the frequency space partition parameters) to absorb the last term into the left hand side.

If we do not make assumptions on the support of $\upphi_{k}$,  then in general we do not have enough decay in the inhomogeneities to justify applying the Teukolsky--Starobinsky energy current to control the remaining boundary term, see Remark~\ref{rmk:decay-inhoms-weird-cutoff}. Thus, we try to apply the Teukolsky--Starobinsky identities in a roundabout way. We denote by $\mathfrak{S}_\theta\colon\mathscr S^{[s]}_\infty(\mc R)\to \mathscr S^{[-s]}_\infty(\mc R)$ the a map which acts by $\mathfrak{S}_\theta\colon f(t,r,\theta,\phi)\mapsto f(t,r,\pi-\theta,\phi)$. From our outgoing solution $\swei{\upphi}{+s}_0$ to the homogeneous transformed PDE \eqref{eq:transformed-k} with $k=0$ and spin $+s$, we can produce, by (the proof of) Proposition~\ref{prop:TS-radial-constant-identities}, an outgoing solution to the homogeneous transformed PDE \eqref{eq:transformed-k} with $k=0$ and spin $-s$, given by
\begin{align*}
\swei{\upphi}{-s}_0:=\mathfrak{S}_{\theta}(r^2+a^2)^{1/2-s}\Delta^s\lp(\frac{r^2+a^2}{\Delta}\uL\rp)^{2s}\lp((r^2+a^2)^{s-1/2}\swei{\upphi}{+s}_0\rp)\,.
\end{align*}
Specifically, we can compute
\begin{align*}
\swei{\tilde\upphi}{-1}_0&=\frac{r^2+a^2}{\Delta}\mathfrak{S}_{\theta}\uL\swei{\tilde\upphi}{+1}_1+\frac{a^2}{(r^2+a^2)^{3/2}}\mathfrak{S}_{\theta}\swei{\tilde\upphi}{+1}_0\\
\swei{\tilde\upphi}{-2}_0&=\mathfrak{S}_{\theta}\lp(\frac{r^2+a^2}{\Delta}\uL\rp)^2\swei{\tilde\upphi}{+2}_2-\frac{5r}{r^2+a^2}\lp(\frac{r^2+a^2}{\Delta}\uL\rp)\mathfrak{S}_{\theta}\swei{\tilde\upphi}{+2}_2+\frac{10}{r^2+a^2}\mathfrak{S}_{\theta}\swei{\tilde\upphi}{+2}_2\numberthis \label{eq:fake-minus-variables}\\
&\qquad-\frac{8r}{(r^2+a^2)^{3/2}}\mathfrak{S}_{\theta}\swei{\tilde\upphi}{+2}_1+\frac{a^2(2r^2+5a^2)}{(r^2+a^2)^3}\mathfrak{S}_{\theta}\swei{\tilde\upphi}{+2}_0\,.
\end{align*}
Let $\swei{\upphi}{- s}_{0,\cut}=\xi\swei{\upphi}{- s}_{0}$; this function solves an \textit{inhomogeneous} transformed PDE \eqref{eq:transformed-k} with $k=0$, spin $-s$, and inhomogeneities
\begin{align*}
\swei{\mathfrak{H}}{-s}_0=([\mathfrak{R}^{[-s]},\xi]-\underline{\mc L}\xi)\mc L\swei{\upphi}{- s}_{0}+w\underline{\mc L}\xi\swei{\upphi}{- s}_{1}\,,\qquad \swei{\mathfrak{h}}{-s}_0= -\underline{\mc L}\mc L\xi\swei{\upphi}{- s}_{0}-\mc L\xi\underline{\mc L}\swei{\upphi}{- s}_{0}\,.
\end{align*}
Since, by assumption, $\xi\swei{\upphi}{+s}_{0}$ satisfies the conditions in Definition~\ref{def:suf-integrability} and \ref{def:outgoing-bdry-phys-space}, so will $\swei{\upphi}{- s}_{0,\cut}$; thus $\swei{\mathfrak{H}}{-s}_0$ verifies the support condition in Definition~\ref{def:outgoing-bdry-phys-space} and the estimates of Section~\ref{sec:toolbox-physical-space} yield that condition \eqref{eq:integrability-Hk} is verified. Let $\smlk{\breve \uppsi}{-s}{0}$ and $\sml{\breve \Psi}{-s}$ be the radial functions defined from $\swei{\upphi}{- s}_{0,\cut}$ and $(w^{-1}L)^{|s|}\swei{\upphi}{- s}_{0,\cut}$, respectively, via \eqref{eq:def-sml-uppsi-G}, and $\sml{\Psi}{-s}$ (without the accent) be defined from $\swei{\Psi}{-s}_\cutt=\xi\swei{\Psi}{-s}$ in a similar fashion. Since the decay of $\swei{\mathfrak{H}}{-s}_0$ and $\swei{\mathfrak{h}}{-s}_0$ is strong enough for \eqref{eq:part1-bdry-withTS} in Theorem~\ref{thm:ODE-ILED} and the precise boundary relations in Lemma~\ref{lemma:bdry-term-relations} to hold, we apply these results and, after summing in $(\omega,m,l)$, we have 
\begin{align*}
&\int_{\mc F_{\rm low}}\sum_{(m,l)\in \mc F_{\rm low}}\frac{\mathfrak{C}_s^{(2)}}{\mathfrak{C}_s}\lp|\smlk{\breve \uppsi}{-s}{0}\rp|^2(-\infty)d\omega\\
&\quad \leq B(a_0)\int_{\mc F_{\rm low}}\sum_{(m,l)\in \mc F_{\rm low}}\frac{\mathfrak{C}_s^{(9)}}{\mathfrak{C}_s}\lp|\sml{\breve \Psi}{-s}\rp|^2(-\infty)d\omega\\
&\quad \leq B(a_0) \int_{\mc F_{\rm low}}\sum_{(m,l)\in \mc F_{\rm low}}\lp|\sml{\Psi}{-s}\rp|^2(-\infty)d\omega \\
&\qquad + B(a_0)\int_{\mc H^+_{(\tau_0,\infty)}} |(w^{-1}L)^{|s|}(\xi\swei{\upphi}{-s}_0)-\xi\swei{\Phi}{-s}|^2d\sigma d\tau\\
&\quad \leq B(a_0)\sum_{i=0}^{N-1}\sum_{k=0}^{|s|-1}\overline{\mathbb{E}}[\swei{\tilde\upphi}{-s}_k](\tau_{2i}) \leq B(a_0)\sum_{i=0}^{N-1}\sum_{k=0}^{|s|}\overline{\mathbb{E}}^{|s|+k}[\swei{\tilde\upphi}{+s}_k](\tau_{2i})\,,
\end{align*}
where we have used Lemma~\ref{lemma:normal-current-errors-nonpeeling} and Proposition~\ref{prop:ODE-ILED-compact-r-mode-stab} to control the right hand side of \eqref{eq:part1-bdry-withTS}, and the identities \eqref{eq:fake-minus-variables} to conclude.

To conclude, we need to relate the boundary terms of $\smlk{\breve \uppsi}{\pm s}{0}$. To this end, we introduce yet another set of functions:
\begin{align*}
\swei{\hat \upphi}{+s}_{0,\cut}:=\mathfrak{S}_{\theta}(r^2+a^2)^{1/2-s}\Delta^s\lp(\frac{r^2+a^2}{\Delta}L\rp)^{2s}\lp((r^2+a^2)^{s-1/2}\swei{\upphi}{-s}_{0,\cut}\rp)
\end{align*}
 is an outgoing solution to an inhomogeneous transformed PDE \eqref{eq:transformed-k} and we let $\smlk{\hat \uppsi}{s}{0}$ be its radial part defined from $\swei{\hat \uppsi}{s}_{0}$ via \eqref{eq:def-sml-uppsi-G}. Since $\mathfrak{C}_s$ and $\mathfrak{C}_s^{(2)}$ do not depend on the sign of $s$, we have the identity
\begin{align*}
\smlk{\breve\uppsi}{+s}{0}(-\infty)&= \lp(\smlk{\breve \uppsi}{+s}{0}-\frac{1}{\mathfrak{C}_s}\smlk{\hat\uppsi}{+s}{0}\rp)(-\infty)+\frac{1}{\mathfrak{C}_s}\smlk{\hat\uppsi}{+s}{0}(-\infty)\\
&=\frac{1}{\mathfrak{C}_s}\lp(\mathfrak{C}_s\smlk{\breve \uppsi}{+s}{0}-\smlk{\hat\uppsi}{+s}{0}\rp)(-\infty)+\frac{\mathfrak{C}_s^{(2)}}{\mathfrak{C}_s}\smlk{\breve \uppsi}{-s}{0}(-\infty)\,, 
\end{align*}
where we used that, as stated above, the spin $-s$ inhomogeneity decays sufficiently fast as $r\to r_+$ for the boundary relations at $r=r_+$ of Proposition~\ref{prop:TS-radial-constant-identities} to hold. As the second term has already been estimated in $\mc F_{\rm high}^c$, to conclude we just compute
\begin{align*}
&\int_{-\infty}^\infty\sum_{ml}\lp|\mathfrak{C}_s\smlk{\breve \uppsi}{+s}{0}-\smlk{\hat\uppsi}{+s}{0}\rp|^2(-\infty)d\omega\\
&\quad = \int_{\mc H^+_{(0,\infty)}}(r^2+a^2)^{1/2-s}\Delta^s\lp(\frac{r^2+a^2}{\Delta}L\rp)^{2s}\lp\{(r^2+a^2)^{s-1/2}\rp. \times\\
&\quad\qquad\qquad \qquad \lp. \times\lp[(r^2+a^2)^{1/2-s}\Delta^s\lp(\frac{r^2+a^2}{\Delta}\uL\rp)^{2s},\xi\rp]\lp((r^2+a^2)^{s-1/2}\swei{\upphi}{+s}_0\rp)\rp\} dtd\sigma\\
&\quad \leq B\sum_{i=0}^{N-1}\sum_{k=0}^{|s|}\overline{\mathbb{E}}^{2|s|+k-1}[\swei{\tilde \upphi}{+s}_k](\tau_{2i})\,,
\end{align*}
by using the finite-in-time estimates.
\end{proof}

\subsubsection{Putting everything together}
\label{sec:hyp-cutoffs-conclusion}

We are finally ready to prove Propositions~\ref{prop:ODE-to-PDE-future-int} and \ref{prop:ODE-to-PDE-non-future-int}: 

\begin{proof}[Proof of Propositions~\ref{prop:ODE-to-PDE-future-int} and \ref{prop:ODE-to-PDE-non-future-int}] Recall that, if $\upphi_{k}$ denote solutions to the homogeneous transformed system for $k=0,\dots,|s|$, then $\upphi_{k,\cutt}$, $\mathfrak{H}_k$ and $\mathfrak{h}_k$ denote the solution and inhomgeneities, respectively, to the inhomogeneous transformed system introduced in Section~\ref{sec:weird-cutoff-system} for some choice of $k_0\in\{0,\dots,|s|\}$. By Lemma~\ref{lemma:weird-inverses}, this system is outgoing and sufficiently integrable, so by Lemma~\ref{lemma:reduction-classical-odes} it gives rise to a frequency localized system, for $\uppsi_{k}$, $\mathfrak{G}_k$ and $\mathfrak{g}_k$, which satisfies the conditions of Theorem~\ref{thm:ODE-ILED} from our previous \cite{SRTdC2020}: in particular,  \eqref{eq:part1-ILED-withoutTS-top} and \eqref{eq:part1-ILED-withoutTS-bottom} hold, hence so do their summed-in-$(\omega,m,l)$ versions.
 
\medskip
\noindent\textit{Case 1: $\breve{\mathfrak{H}}_{0}$ is compactly supported in $r^*$.} Let us assume, for the moment, that for $k=0,\dots,|s|$ the homogeneous solutions $\upphi_k$ and $\xi$ are such that $\breve{\mathfrak{H}}_{0}$ from \eqref{eq:teukolsky-cutoff-inhom-k} have compact $r^*$ support. We begin by setting $k_0=|s|$ in the construction in Section~\ref{sec:weird-cutoff-system}, in particular in Lemma~\ref{lemma:weird-inverses}.  We combine the summed-in-$(\omega,m,l)$ \eqref{eq:part1-ILED-withoutTS-top} with Proposition~\ref{prop:ODE-ILED-compact-r-mode-stab}, to control the bulk terms in a compact $r^*$ range on the right hand side of the former estimate. For the errors in \eqref{eq:part1-ILED-withoutTS-top} involving boundary terms at $r=r_+$ when $s>0$, we employ \eqref{eq:weird-inverses-bdry-term-rp-best} from Lemma~\ref{lemma:weird-inverses-bdry-term-rp}. Then we invoke \eqref{eq:current-errors-top} in Lemma~\ref{lemma:normal-current-errors-nonpeeling}, to control the separated current errors; as long as $|a|\leq a_0<M$, choosing $\varepsilon$ sufficiently small there depending on $a_0$, we have
\begin{align*}
&\int_{\mc{R}_{(\tau_0,\tau_{N+1})}} \frac{1}{r^2}\lp(|\Phi_\cutt'|^2+|T\mc{P}_{\rm trap}[\Phi_{\cutt}]|^2+\frac{1}{r}|\mathring{\slashed\nabla}^{[s]}\mc{P}_{\rm trap}[\Phi_{\cutt}]|^2+\frac{\lp(s^2+r^{-1}\rp)}{r}|\Phi_\cutt|^2\rp)dr  d\sigma d\tau\\
&\qquad +\int_{\mc{R}_{(\tau_0,\tau_{N+1})}}\lp(\frac{1}{r^{\frac 52}}|T^{\frac12}\Phi_{\cutt}|^2+\frac{1}{r^3}|Z^{\frac12}\Phi_{\cutt}|^2\rp)dr  d\sigma d\tau\\
&\quad\leq B(a_0)\sum_{i=0}^{N-1}\Big(\sum_{k=0}^{|s|}{\mathbb{E}}[\tilde\upphi_{k}](\tau_{2i})+\int_{\mc H^+_{(\tau_{2i},\tau_{2i+1})}}|\Phi|^2d\tau d\sigma +\int_{\mc I^+_{(\tau_{2i},\tau_{2i+1})}}|\Phi|^2d\tau d\sigma \Big)\,,\numberthis\label{eq:hyp-conclusion-1}\\
&\sum_{k=0}^{|s|-1}\int_{\mc{R}_{(\tau_0,\tau_{N+1})}} \frac{1}{r^3}\lp(r|\upphi_{k,\cutt}'|^2+r|T\upphi_{k,\cutt}|^2+|\mathring{\slashed\nabla}^{[s]}\upphi_{k,\cutt}|^2+|\upphi_{k,\cutt}|^2+|Z^{\frac32}\upphi_{k,\cutt}|^2+r^{\frac12}|T^{\frac32}\upphi_{k,\cutt}|^2\rp)dr  d\sigma d\tau\\
&\quad\leq B(a_0)\sum_{i=0}^{N-1}\Big(\sum_{k=0}^{|s|}{\mathbb{E}}[\tilde\upphi_{k}](\tau_{2i}) +\int_{\mc H^+_{(\tau_{2i},\tau_{2i+1})}}|\Phi|^2d\tau d\sigma +\int_{\mc I^+_{(\tau_{2i},\tau_{2i+1})}}|\Phi|^2d\tau d\sigma \Big)\,,\numberthis\label{eq:hyp-conclusion-2}
\end{align*}
where $\mathbb{E}$ must be replaced by $\overline{\mathbb{E}}$ or $\mathbb{E}_2$ depending on our choices in Theorem~\ref{thm:ODE-ILED}. In \eqref{eq:hyp-conclusion-1}, we can drop the scissors on the left hand side by using finite-in-time energy estimates, as in Proposition~\ref{prop:finite-in-time-first-order}, in every term except where $\Phi_{\cutt}$ is acted on by the microlocal operators $\mc P_{\rm trap}$, $T^{\frac12}$ and $Z^{\frac12}$. In \eqref{eq:hyp-conclusion-1}, we may also drop the scissors, in all but the terms where $\upphi_{k,\cutt}$ is acted on by $T^{\frac32}$ and $Z^{\frac32}$, by using a combination of Lemma~\ref{lemma:weird-cutoffs-diff-fake-real} and finite-in-time estimates. 

Note that, for the $k<|s|$ quantities, estimate \eqref{eq:hyp-conclusion-2} is not at the right level of regularity. Indeed, in view of our reliance on transport estimates in Lemmas~\ref{lemma:weird-cutoffs-diff-fake-real} and \ref{lemma:weird-cutoffs-inhoms-rough-estimates}, which crucially go into the proof of the latter, we cannot expect the argument above not to lose derivatives. To gain them back, we iterate the above procedure by considering a new inhomogeneous transformed system now produced by taking $k_0=|s|-1$, then $k_0=|s|-2$, etc. Thus, fix some $k_0<|s|$ and define an outgoing and sufficiently integrable transformed system via the construction in Section~\ref{sec:weird-cutoff-system}, in particular in Lemma~\ref{lemma:weird-inverses}. As before, we combine the summed-in-$(\omega,m,l)$ \eqref{eq:part1-ILED-withoutTS-top} from Theorem~\ref{thm:ODE-ILED} with Proposition~\ref{prop:ODE-ILED-compact-r-mode-stab}, to control the bulk terms in a compact $r^*$ range on the right hand side of the former estimate, and \eqref{eq:weird-inverses-bdry-term-rp-best} from Lemma~\ref{lemma:weird-inverses-bdry-term-rp} for the boundary terms at $r=r_+$ when $s>0$. We combine this estimate with \eqref{eq:current-errors-bottom} in Lemma~\ref{lemma:normal-current-errors-nonpeeling} for the current errors, taking $\varepsilon>0$ sufficiently small so that the first two terms in the latter inequality may be absorbed by the the summed-in-$(\omega,m,l)$ left hand side of  \eqref{eq:part1-ILED-withoutTS-bottom}, and by Lemma~\ref{lemma:Plancherel} obtain
\begin{align*}
&\int_{\mc{R}_{(\tau_0,\tau_{N+1})}} \frac{1}{r^2}\lp(|\upphi_{k_0,\cutt}''|^2+|\upphi_{k_0,\cutt}'|^2+r^{-1}|\mathring{\slashed\nabla}^{[s]}\upphi_{k_0,\cutt}'|^2+|\upphi_{k_0,\cutt}'|^2\rp)dr  d\sigma d\tau\\
&\qquad+\int_{\mc{R}_{(\tau_0,\tau_{N+1})}} \frac{1}{r^2}\lp(|\upphi_{k_0,\cutt}|^2+r^{-1}|\mathring{\slashed\nabla}^{[s]}\upphi_{k_0,\cutt}|^2+r^{-1}|\upphi_{k_0,\cutt}|^2\rp)dr \\
&\quad\leq B(a_0)\sum_{i=0}^{N-1}\Big(\sum_{j=0}^{k_0}\mathbb{E}^{1}[\tilde\upphi_{j}](\tau_i)+\mathbb{E}[\tilde\upphi_{k_0+1}](\tau_{2i})\Big) \\
&\quad\qquad + B(a_0)\int_{\mc{R}_{(\tau_0,\tau_{N+1})}}\frac{1}{r^3}\lp(r|\upphi_{k_0+1,\cutt}'|^2+r^{\frac12}|T^{1/2}\upphi_{k_0+1,\cutt}|^2+|Z^{1/2}\upphi_{k_0+1,\cutt}|^2+|\upphi_{k_0+1,\cutt} |^2 \rp)dr d\sigma d\tau\,,
\end{align*}
with the same caveats on the $\mathbb{E}$ on the right hand side as before. We may drop the scissors in the previous estimate by appealing to the finite in time energy estimates, as in Proposition~\ref{prop:finite-in-time-first-order}. Then, we commute with $T$ and $Z$ to get
\begin{align*}
&\sum_{J_1+J_2=0}^{|s|-k_0-1}\int_{\mc{R}_{(\tau_0,\tau_{N+1})}} \frac{1}{r^2}\lp(|T^{J_1}Z^{J_2}\upphi_{k_0}''|^2+|T^{J_1+1}Z^{J_2}\upphi_{k_0}'|^2+r^{-1}|T^{J_1}Z^{J_2}\mathring{\slashed\nabla}^{[s]}\upphi_{k_0}'|^2+|T^{J_1}Z^{J_2}\upphi_{k_0}'|^2\rp)dr  d\sigma d\tau\\
&\quad+\sum_{J_1+J_2=0}^{|s|-k_0-1}\int_{\mc{R}_{(\tau_0,\tau_{N+1})}} \frac{1}{r^2}\lp(|T^{J_1+1}Z^{J_2}\upphi_{k_0}|^2+r^{-1}|T^{J_1}Z^{J_2}\mathring{\slashed\nabla}^{[s]}\upphi_{k_0}|^2+r^{-1}|T^{J_1}Z^{J_2}\upphi_{k_0}|^2\rp)dr \\
%&\quad\leq B\sum_{i=0}^N\Big(\sum_{j=0}^{k_0}\mathbb{E}^{|s|-k_0}[\tilde\upphi_{j}](\tau_i)+\mathbb{E}^{|s|-k_0-1}[\tilde\upphi_{k_0+1}](\tau_i)\Big) +B\sum_{J_1+J_2=0}^{|s|-k_0-1}\mathbb{I}^{\rm deg}[T^{J_1}Z^{J_2}\upphi_{k_0+1}](\tau_0,\tau_{N+1})\\
&\quad\leq B(a_0)\sum_{i=0}^{N-1}\sum_{j=0}^{k_0+1}\mathbb{E}^{|s|-j}[\tilde\upphi_{j}](\tau_{2i}) \\
&\qquad\quad + B(a_0)\sum_{J_1+J_2=0}^{|s|-k_0-1}\int_{\mc{R}_{(\tau_0,\tau_{N+1})}}\frac{1}{r^2}\lp(|T^{J_1}Z^{J_2}\upphi_{k_0+1,\cutt}'|^2+r^{-1}|T^{J_1}Z^{J_2}\upphi_{k_0+1,\cutt} |^2 \rp)dr d\sigma d\tau\\
&\qquad\quad + B(a_0)\sum_{J_1+J_2=0}^{|s|-k_0-1}\int_{\mc{R}_{(\tau_0,\tau_{N+1})}}\frac{1}{r^2}\lp(r^{-\frac12}|T^{J_1+1/2}Z^{J_2}\upphi_{k_0+1,\cutt}'|^2+r^{-1}|T^{J_1}Z^{J_2+1/2}\upphi_{k_0+1,\cutt} |^2 \rp)dr d\sigma d\tau\,.
\end{align*}
Iterating the estimate for different $k_0$, we get
\begin{align*}
&\sum_{k=0}^{|s|-1}\sum_{J_1+J_2=0}^{|s|-k-1}\int_{\mc{R}_{(\tau_0,\tau_{N+1})}} \frac{1}{r^2}\lp(|T^{J_1}Z^{J_2}\upphi_{k}''|^2+|T^{J_1+1}Z^{J_2}\upphi_k'|^2+r^{-1}|T^{J_1}Z^{J_2}\mathring{\slashed\nabla}^{[s]}\upphi_k'|^2+|T^{J_1}Z^{J_2}\upphi_k'|^2\rp)dr  d\sigma d\tau\\
&\quad+\sum_{J_1+J_2=0}^{|s|-k-1}\int_{\mc{R}_{(\tau_0,\tau_{N+1})}} \frac{1}{r^2}\lp(|T^{J_1+1}Z^{J_2}\upphi_k|^2+r^{-1}|T^{J_1}Z^{J_2}\mathring{\slashed\nabla}^{[s]}\upphi_k|^2+r^{-1}|T^{J_1}Z^{J_2}\upphi_k|^2\rp)dr \\
&\quad\leq B(a_0)\sum_{i=0}^{N-1}\sum_{j=0}^{|s|}\mathbb{E}^{|s|-k}[\tilde\upphi_{j}](\tau_{2i})\\
&\quad\qquad + B(a_0)\int_{\mc{R}_{(\tau_0,\tau_{N+1})}}\frac{1}{r^3}\lp(r|\Phi_{\cutt}'|^2+r^{\frac12}|T^{1/2}\Phi_{\cutt}|^2+|Z^{1/2}\Phi_{\cutt}|^2+|\Phi_{\cutt} |^2 \rp)dr d\sigma d\tau\,.\numberthis\label{eq:hyp-conclusion-3}
\end{align*}

\medskip
\noindent\textit{Case 2: no support assumptions.} If we cannot assume that $\breve{\mathfrak{H}}_{0}$ has compact $r^*$ support, then the procedure above still applies with one difference: to control the boundary terms in \eqref{eq:part1-ILED-withoutTS-bottom} we must replace \eqref{eq:weird-inverses-bdry-term-rp-best} with the worse estimates \eqref{eq:weird-inverses-bdry-term-rp} from Lemma~\ref{lemma:weird-inverses-bdry-term-rp}. To keep the estimates having the same regularity on the left and right hand side, we commute the estimate obtained by combining \eqref{eq:part1-ILED-withoutTS-top}, Proposition~\ref{prop:ODE-ILED-compact-r-mode-stab}, and \eqref{eq:current-errors-bottom} with $T$ and $Z$ a few more times. 

\medskip
\noindent\textit{Conclusion.} For Proposition~\ref{prop:ODE-to-PDE-future-int}, we consider a cutoff $\xi$ where, in the notation of the start of the section,  $N=1$, $\tau_0=0$, $\tau_1=1$, and $\tau_{2}=\infty$. By Lemma~\ref{lemma:reduction-argument} we can assume $\mathfrak{\upphi}_{k}\big|_{\mc R_{(0,1)}}$ is compactly supported in $r^*$ for every $k=0,\dots, |s|$. Thus, we are the setting of Case 1 applies and furthermore we have
\begin{align*}
\int_{\mc H^+_{(0,1)}}|\Phi|^2d\tau d\sigma +\int_{\mc I^+_{(0,1)}}|\Phi|^2d\tau d\sigma = 0\,.\numberthis\label{eq:summing-ODE-ILED-intermediate-3}
\end{align*}
In this case, we should choose $\overline{\mathbb E}$ in lieu of $\mathbb{E}$ on the right hand side of \eqref{eq:hyp-conclusion-1}, \eqref{eq:hyp-conclusion-2} and \eqref{eq:hyp-conclusion-3}.

For Proposition~\ref{prop:ODE-to-PDE-non-future-int}, take  $N=2$, $\tau_0=0$, $\tau_1=1$, and $\tau_{2}=\tau-1$ and $\tau_2=\tau$ for some $\tau>0$. In this case, we cannot assume $\mathfrak{\upphi}_{k}\big|_{\mc R_{(0,\tau)}}$ is compactly supported in $r^*$, and so though we still have \eqref{eq:summing-ODE-ILED-intermediate-3}, the same is not true for the integrals over $\mc H^+_{(\tau-1,\tau)}$ and $\mc I^+_{(\tau-1,\tau)}$. Nevertheless, in the first case we may appeal to the finite in time energy estimates of Proposition~\ref{prop:finite-in-time-first-order} to obtain 
\begin{align*}
\int_{\mc H^+_{(\tau-1,\tau)}}|\Phi|^2d\tau d\sigma \leq B\overline{\mathbb{E}}[\Phi](\tau)\,;
\end{align*}
in the second case we can use the fundamental theorem of calculus and Jensen's inequality to obtain
\begin{align*}
\int_{\mc I^+_{(\tau-1,\tau)}}|\Phi|^2d\tau d\sigma \leq \int_{\mc I^+_{(\tau-1,\tau)}}\lp|\int_{v_0}^\infty r^p L\Phi r^p dv\rp|^2d\tau d\sigma \leq B\int_{\mc R_{(\tau-1,\tau)}}r^p|L\Phi|^2 dr d\tau d\sigma \leq \mathbb{E}_p[\Phi](\tau)\,,
\end{align*}
for some $p>1$.
\end{proof}

\subsection{An inhomogeneous system from radial cutoffs}
\label{sec:cutoffs-radial} 
 
Finally, we prove Proposition~\ref{prop:ODE-to-PDE-radial-cutoff-ILED}. 

\begin{proof}[Proof of Proposition~\ref{prop:ODE-to-PDE-radial-cutoff-ILED}]  Recall that, by assumption, $\upphi_k$ denote homogeneous solutions to the transformed system of Definition~\ref{def:transformed-system} which are sufficiently integrable. By cutting them off in the radial direction, we obtain new transformed variables
\begin{align*}
\swei{\upphi}{s}_{k,\cutr}=\lp(w^{-1}\mc L\rp)^{-1}\lp(\chi_R \swei{\upphi}{s}_0\rp)\,, \numberthis\label{eq:forward-radial-cutoffs}
\end{align*}
which are outgoing, as per Definition~\ref{def:outgoing-bdry-phys-space}. 

The proof  is entirely analogous to those of Propositions~\ref{prop:ODE-to-PDE-future-int} and \ref{prop:ODE-to-PDE-non-future-int} for hyperboloidal cutoffs: we take the estimate in Theorem~\ref{thm:ODE-ILED}(iii) and bound the error terms, in this case only due to inhomogeneity errors, by initial data.  In the context of the aforementioned propositions, such terms were difficult to treat because one had to be careful with the $r$ weights placed on the initial energy norms, see Section~\ref{sec:current-errors} and Remark~\ref{rmk:comparison-cutoff-types}. In the present setting, since the functions involved are all compactly supported in $r^*$, the arguments used can be vastly simplified: in particular, considering forwards cutoffs in the style of \eqref{eq:forward-radial-cutoffs} are enough to conclude.
\end{proof}
 
\section{Proof of the main theorem}
\label{sec:proof-main-thm}

In this section, we establish our main result, namely integrated local energy decay, energy boundedness, as well as energy and pointwise decay, for solutions which are sufficiently integrable in the entire region to the future of the data on $\Sigma_0$.

\subsection{A first estimate and the notion of future integrability}

We begin by stating a first estimate which follows easily from our work in Section~\ref{sec:physical-space-estimates}:
\begin{lemma}[A general estimate] \label{lemma:ILED-for-non-future-int}
For any $p\in(1,2)$, we have
\begin{align*}
&\sum_{k=0}^{|s|}\lp(\overline{\mathbb{I}}_p^{\mathrm{deg},J+|s|-k}[\swei{\tilde\upphi}{s}_k](0,\tau)+\overline{\mathbb{E}}_p^{J+|s|-k}\swei{\tilde\upphi}{s}_k](\tau)\rp) \leq  B\sum_{k=0}^{|s|}\lp(\overline{\mathbb{E}}_p^{J+|s|-k}[\swei{\tilde\upphi}{s}_k](\tau)+\overline{\mathbb{E}}_p^{J+|s|-k}[\swei{\tilde\upphi}{s}_k](0)\rp) \,, \numberthis \label{eq:ILED-non-futurint-deg}\\
&\sum_{k=0}^{|s|}\lp(\overline{\mathbb{I}}_2^{J+|s|-k}[\swei{\tilde\upphi}{s}_k](0,\tau)+\overline{\mathbb{E}}_2^{J+|s|-k+1}[\swei{\tilde\upphi}{s}_k](\tau)\rp) \\
&\quad \leq  B\sum_{k=0}^{|s|}\lp(\overline{\mathbb{E}}_p^{J+|s|-k+1}[\swei{\tilde\upphi}{s}_k](\tau)+\overline{\mathbb{E}}_p^{J+|s|-k+1}[\swei{\tilde\upphi}{s}_k](0)\rp) \numberthis\label{eq:ILED-non-futurint-nondeg}, 
\end{align*}
for $J\geq 3|s|$. 
\end{lemma}
\begin{proof}
The statement follows from combining the higher order estimates of Proposition~\ref{prop:higher-order-bulk-flux} with the integrated estimates of Proposition~\ref{prop:ODE-to-PDE-non-future-int}.
\end{proof}

Recalling the estimates of Propositions~\ref{prop:bulk-flux-large-r} and \ref{prop:higher-order-bulk-flux}, which do not ensure an a priori almost energy boundedness statement, it is easy to see that the the estimate in Lemma~\ref{lemma:ILED-for-non-future-int} is not ideal: the lack of smallness in either of the aforementioned statements implies that we cannot use them to close neither an integrated local energy statement nor an energy boundedness statement when $a_0<M$ is not assumed to be small. As a replacement for an a priori energy boundedness statement, let us introduce the notion of future integrability:

\begin{definition}[Future integrability] \label{def:future-int} Take $\xi(t^*,r)$ to be a cutoff function which vanishes identically in the past of $\Sigma_0$, and is identically 1 in the future of $\Sigma_1$. 
We say that, for $k=0,\dots, |s|$, $\swei{\upphi}{s}_k$ are future integrable solutions to the homogeneous transformed system of Definition~\ref{def:transformed-system} if $\xi \swei{\upphi}{s}_k$ are outgoing and sufficiently integrable, in the sense of Definitions~\ref{def:outgoing-bdry-phys-space} and \ref{def:suf-integrability}.
\end{definition}
An energy boundedness statement would, of course, imply that all solutions arising from suitably regular data verify this condition:

\begin{lemma}[From finite flux to future integrability] \label{lemma:future-int-finite-flux} Fix $s\in\mathbb{Z}_{\geq 2}$, $M>0$, $|a|<M$. For each $k\in\{0,\dots,|s|\}$, let $\swei{\upphi}{s}_k$ be solutions to the homogeneous transformed system of Definition~\ref{def:transformed-system} arising from smooth, compactly supported initial data. If $\swei{\tilde\upphi}{s}_k$ have finite energy flux on each $\Sigma_\tau$, i.e.
\begin{align*}
\sup_{\tau>0}\overline{\mathbb{E}}_p^J[\swei{\tilde\upphi}{s}_k](\tau)<\infty\,, 
\end{align*}
for all $k\in\{0,\dots, |s|\}$, all $J\geq 0$ and some $p\in(1,2)$, then $\swei{\upphi}{s}_k$ are future integrable.
\end{lemma}
\begin{proof}
By Lemma~\ref{lemma:suf-integrability}, since
\begin{align*}
\sup_{\tau\leq 0}\overline{\mathbb{E}}^J[\xi\swei{\tilde\upphi}{s}_k\mathbbm{1}_{[-R^*,R^*]}](\tau)=\overline{\mathbb{E}}^J[\xi\swei{\tilde\upphi}{s}_k\mathbbm{1}_{[-R^*,R^*]}](0)<\infty
\end{align*}
if $|a|<M$, we just need to show that 
\begin{align*}
\overline{\mathbb{I}}^J[\xi\swei{\tilde\upphi}{s}_k\mathbbm{1}_{[-R^*,R^*]}](0,\infty)<\infty
\end{align*}
for each $J$ and $R^*$. Since we have
\begin{align*}
\sum_{k=0}^{|s|}\overline{\mathbb{I}}^J[\xi\swei{\tilde\upphi}{s}_k\mathbbm{1}_{[-R^*,R^*]}](0,\tau)&\leq \sum_{k=0}^{|s|}\lp(\overline{\mathbb{I}}^J[\swei{\tilde\upphi}{s}_k\mathbbm{1}_{[-R^*,R^*]}](0,\infty)+ \overline{\mathbb{E}}^{J+1}[\swei{\tilde\upphi}{s}_k](\tau) + \overline{\mathbb{E}}^{J+1}[\swei{\tilde\upphi}{s}_k](0)\rp)\\
&\leq \sum_{k=0}^{|s|}\lp(\overline{\mathbb{E}}^{J+1}_p[\swei{\tilde\upphi}{s}_k](\tau) + \overline{\mathbb{E}}^{J+1}_p[\swei{\tilde\upphi}{s}_k](0)\rp)\,,
\end{align*}
for sufficiently large $J$ and $p\in (1,2)$ by Lemma~\ref{lemma:ILED-for-non-future-int}, the conclusion follows.
\end{proof}

\subsection{Integrated local energy decay for future integrable solutions}

In this section, we establish integrated local energy decay, or ILED for short, of the transformed system through hyperboloidal foliations, assuming future integrability.
\begin{theorem}[ILED for future integrable solutions] \label{thm:ILED-futurint} Fix $s\in\{0,\pm 1,\pm 2\}$, $a_0\in[0,M)$, and assume that, for each $k=0,\dots, |s|$, $\swei{\tilde\upphi}{s}_k$ are future integrable solutions of  the homogeneous transformed system of Definition~\ref{def:transformed-system}. For any $p\in[0,2)$ and $\delta\in(0,1)$, and any  $|a|\leq a_0$, we have the following ILED estimates:
\begin{align*}
&\int_{\mc{R}_{(0,\infty)}} \lp(\frac{1}{r^{1+\delta}}|(\swei{\Phi}{s})'|^2+\frac{1}{r^{1+\delta}}|T\mc{P}_{\rm trap}[\xi\swei{\Phi}{s}]|^2+r^{-3}|\mathring{\slashed\nabla}^{[s]}\mc{P}_{\rm trap}[\xi\swei{\Phi}{s}]|^2+r^{-3}\lp(s^2+r^{-\delta}\rp)|\swei{\Phi}{s}|^2\rp)dr  d\sigma d\tau\\
&\qquad+\overline{\mathbb{I}}^{\mathrm{deg}}_{-\delta,p}[\swei{\Phi}{s}](0,\infty)+\sum_{k=0}^{|s|-1}\overline{\mathbb{I}}^{\mathrm{deg},|s|-k}_{-\delta,p}[\swei{\tilde\upphi}{s}_k](0,\infty) \\
&\quad\leq  B(a_0,\delta,p)\overline{\mathbb{E}}_{p}[\swei{\Phi}{s}](0)+B(a_0,\delta)\sum_{k=0}^{|s|-1}\overline{\mathbb{E}}^{|s|-k}_p[\swei{\tilde\upphi}{s}_k](0)\,; \numberthis\label{eq:ILED-futurint-first-order}
\end{align*}
 and, for $J\geq 1$, 
\begin{align*}
&\overline{\mathbb{I}}^{\mathrm{deg},J}_{-\delta,p}[\swei{\Phi}{s}](0,\infty)+\sum_{k=0}^{|s|-1}\overline{\mathbb{I}}^{\mathrm{deg},J+|s|-k}_{-\delta,p}[\swei{\tilde\upphi}{s}_k](0,\infty) \\
&\quad\leq  B(a_0,\delta,p)\overline{\mathbb{E}}_{p}^J[\swei{\Phi}{s}](0)+B(a_0,\delta)\sum_{k=0}^{|s|-1}\overline{\mathbb{E}}^{J+|s|-k}_{p}[\swei{\tilde\upphi}{s}_k](0)\,, \numberthis\label{eq:ILED-futurint-higher-order}
\end{align*}
in the last case  excluding $p=0$. In the above, the estimates remain valid for $p=2$ if $s\leq 0$.
\end{theorem}

\begin{proof} Let us consider the case  of \eqref{eq:ILED-futurint-first-order} first. It is very easy to show that the restriction of the bulk integrals in Theorem~\ref{thm:ILED-futurint} to the regions close to $\mc I^+$ and $\mc H^+$ are controlled by initial data; this follows from our preliminary estimates in Section~\ref{sec:toolbox-physical-space}. Indeed, recall the large $|r^*|$ estimates of Propositions~\ref{prop:bulk-flux-large-r} and \ref{prop:higher-order-bulk-flux}:
\begin{align*} 
&b(\delta)\sum_{k=0}^{|s|}\mathbb{I}_{-\delta,p}^{|s|-k}[\swei{\tilde\upphi}{s}_k\mathbbm{1}_{\{|r^*|\geq 2R^*\}}](0,\tau) \\
&\quad\leq B\sum_{k=0}^{|s|}\lp(\mathbb{E}^{|s|-k}[\swei{\tilde\upphi}{s}_j\mathbbm{1}_{\{R^*\leq|r^*|\leq 2R^*\}}](\tau)+\mathbb{E}_p[\swei{\tilde\upphi}{s}_k](0)+\mathbb{I}^{|s|-k}[\swei{\tilde\upphi}{s}_k\mathbbm{1}_{\{ R^*\leq |r^*|\leq 2R^*\}}](0,\tau)\rp)\,,\numberthis \label{eq:ILED-large-r-intermediate}
\end{align*}
for any $\tau>0$ and sufficiently large $R^*>0$, and $p$ as constrained in the statements of the aforementioned propositions. To obtain our large $|r^*|$ result, we want to remove the first term from the right hand side of \eqref{eq:ILED-large-r-intermediate}. Lemma~\ref{lemma:pigeonhole-energy-decay} produces a sequence $\{\tau_n\}_{n=1}^\infty$ such that $\tau_n\to \infty$ which can help us: 
choosing $\tau=\tau_n$ in \eqref{eq:ILED-large-r-intermediate} and taking the limit $n\to \infty$, we obtain
\begin{align*} 
\sum_{k=0}^{|s|}\mathbb{I}_{-\delta,p}^{\mathrm{deg},{|s|-k}}[\swei{\tilde\upphi}{s}_k](0,\infty) 
&\leq B(\delta)\sum_{k=0}^{|s|}\lp(\mathbb{E}_p^{|s|-k}[\swei{\tilde\upphi}{s}_k](0)+\mathbb{I}^{\mathrm{deg},|s|-k}[\swei{\tilde\upphi}{s}_k\mathbbm{1}_{\{ |r^*|\leq R^*\}}](0,\infty)\rp)\,,\numberthis\label{eq:future-int-ILED-large-r-basic}
\end{align*}
for sufficiently large $R^*$.

The difficulty in establishing the ILED estimate of Theorem~\ref{thm:ILED-futurint} arises in the bounded $r$ region, due to the occurrence of trapping. The results from our previous paper \cite{SRTdC2020}, as embodied in physical space for future integrable solutions by Proposition~\ref{prop:ODE-to-PDE-future-int}, provide this key result for subextremal Kerr:
\begin{align*}
&\int_{\mc{R}_0\cap \{|r^*|\leq R^*\}} \lp(|(\swei{\Phi}{s})'|^2+|T\mc{P}_{\rm trap}[\xi\swei{\Phi}{s}]|^2+|\mathring{\slashed\nabla}^{[s]}\mc{P}_{\rm trap}[\xi\swei{\Phi}{s}]|^2+|\swei{\Phi}{s}|^2\rp)dr  d\sigma d\tau\\
&\qquad+\sum_{k=0}^{|s|}\sum_{J=0}^{|s|-k}\mathbb{I}^{\mathrm{deg},J}[\swei{\tilde\upphi}{s}_k\mathbbm{1}_{\{|r^*|\leq R^*\}}](0,\infty) \leq B(a_0,R^*)\sum_{k=0}^{|s|}\sum_{J=0}^{|s|-k} 
\overline{\mathbb{E}}^J[\swei{\tilde\upphi}{s}_k](0)\,. \numberthis\label{eq:future-int-ILED-bdd-r-basic}
\end{align*}
Thus, fixing $R^*$ sufficiently large that \eqref{eq:future-int-ILED-large-r-basic} holds and combining it with \eqref{eq:future-int-ILED-bdd-r-basic}, we arrive at \eqref{eq:ILED-futurint-first-order}.

Finally, for estimate \eqref{eq:ILED-futurint-first-order}, we proceed similarly. We can first appeal to Proposition~\ref{prop:higher-order-bulk-flux} for large $|r^*|$ and choose a dyadic sequence  $\{\tau_n\}_{n=0}^\infty$ such that $\tau_n\to \infty$ to arrive at
\begin{align*}
&\sum_{k=0}^{|s|}\sum_{J=0}^{J_{\rm max}+|s|-k}\overline{\mathbb{I}}_{-\delta,p}^{\mathrm{deg},J}[\swei{\tilde\upphi}{s}_k](0,\infty) \\
&\quad\leq B(\delta,J_{\rm max})\sum_{k=0}^{|s|}\sum_{J=0}^{J_{\rm max}+|s|-k}\lp(\overline{\mathbb{E}}_p^J[\swei{\tilde\upphi}{s}_k](0)+\overline{\mathbb{I}}^{\mathrm{deg},J}[\swei{\tilde\upphi}{s}_k\mathbbm{1}_{\{ |r^*|\leq R^*\}}](0,\infty)\rp)\,, \numberthis\label{eq:future-int-ILED-large-r-higher}
\end{align*}
for sufficiently large. Then, fixing $R^*$ and combining with a $T$ and $Z$ commuted version of Proposition~\ref{prop:ODE-to-PDE-future-int}, by our higher order estimates of Proposition~\ref{prop:higher-order-bulk-flux}  the result follows.
\end{proof}

\subsection{Energy flux boundedness for future integrable solutions}

In this section, we establish boundedness of the energy fluxes of the transformed system through hyperboloidal foliations, through $\mc H^+$ and through $I^+$, assuming future integrability. 

\begin{theorem}[Energy boundedness for future integrable solutions] \label{thm:energy-bddness-futurint} Fix $s\in\{0,\pm 1,\pm 2\}$, $a_0\in[0,M)$, and assume that, for each $k=0,\dots, |s|$, $\swei{\tilde\upphi}{s}_k$ are future integrable solutions of  the transformed system of Definition~\ref{def:transformed-system}. For any $\tau>0$, integer $J\geq 0$, $p_1\in[0,2)$ excluding $p=0$ if $J\geq 1$, and  $|a|\leq a_0$, we have the following estimates for the energy fluxes: 
\begin{align*}
&\sum_{k=0}^{|s|-1}\overline{\mathbb{E}}^{J+|s|-k}_{p}[\tilde\upphi_k](\tau)+ \sum_{k=0}^{|s|}\overline{\mathbb{E}}^{J}_{p}[\Phi](\tau) +\mathbb{E}_{\mc I^+,p}^{J}[\Phi](0,\tau) + \overline{\mathbb{E}}^{J}_{\mc H^+}[\Phi](0,\tau)\\
&\qquad+\sum_{k=0}^{|s|-1}\lp(\mathbbm{1}_{\{s<0\}}\mathbb{E}_{\mc I^+,p}^{|s|-k}[\tilde\upphi_{k}](0,\tau)+ \mathbbm{1}_{\{s>0\}}\overline{\mathbb{E}}^{|s|-k}_{\mc H^+}[\tilde\upphi_k](0,\tau)\rp) \\
&\quad\leq B(a_0,J,p)\sum_{k=0}^{|s|}\overline{\mathbb{E}}^{J+|s|-k}_p[\tilde\upphi_k](0)\,.\numberthis\label{eq:energy-bddness-futurint}
\end{align*}
The estimate remains valid for $p=2$ if $s\leq 0$.
\end{theorem}

\begin{remark}[Energy boundedness for Schwarzschild vs.\ Kerr]
Let us note that, for $a=0$,  Theorem~\ref{thm:energy-bddness-futurint} can be easily derived from the conservation and coercivity of the Killing energy associated to the $T$ vector field. In contrast, for $|a|\leq M$, because this energy is not coercive even when $s=0$, we must rely on a more intricate argument based on the results of our previous paper \cite{SRTdC2020}. As before, the occurrence of trapping means that the  main difficulty in establishing energy flux boundedness arises in the bounded $r$ region. There, we rely crucially on the results from previous paper \cite{SRTdC2020}, as embodied in physical space by Proposition~\ref{prop:ODE-to-PDE-radial-cutoff-ILED}.
\end{remark}

\begin{proof}[Proof of Theorem~\ref{thm:energy-bddness-futurint}] 
With the results of the previous section, it is very easy to show that the portions of the fluxes which are very close to $\mc I^+$ and, if $|a|<M$, $\mc H^+$, are bounded by data, provided that the bounded-$r^*$ flux is as well: by Proposition~\ref{prop:bulk-flux-large-r} and the higher order estimates of Proposition~\ref{prop:higher-order-bulk-flux}, 
\begin{align*}
&\sum_{k=0}^{|s|}\overline{\mathbb{E}}^{J+|s|-k}[\tilde\upphi_k\mathbbm{1}_{|r^*|\geq 2R^*}](\tau) +\mathbb{E}_{\mc I^+,p}^{J}[\Phi](0,\tau) + \overline{\mathbb{E}}^{J}_{\mc H^+}[\Phi](0,\tau)\\
&\qquad+\sum_{k=0}^{|s|-1}\lp(\mathbbm{1}_{\{s<0\}}\mathbb{E}_{\mc I^+,p}^{|s|-k}[\tilde\upphi_{k}](0,\tau)+ \mathbbm{1}_{\{s>0\}}\overline{\mathbb{E}}^{|s|-k}_{\mc H^+}[\tilde\upphi_k](0,\tau)\rp)\\
&\quad \leq B\sum_{k=0}^{|s|}\overline{\mathbb{E}}^{J+|s|-k}[\tilde\upphi_k\mathbbm{1}_{|r^*|\geq R^*}](0)  +B\sum_{k=0}^{|s|}\mathbb{E}^{J+|s|-k}[\tilde\upphi_k\mathbbm{1}_{|r^*|\leq R^*}](\tau)\,.
\end{align*}
for any $\tau>0$, any integer $J\geq 0$ and sufficiently large $R^*$, and $p$ with the constraints in the aforementioned propositions. We may fix $R^*$ by the requirement that the above estimate holds.

The ensuing steps lead to control over the outstanding bounded $r^*$ flux. We note that it is enough to look at $J=0$, as the case of $J>0$ follows by induction from the higher order estimates of Proposition~\ref{prop:higher-order-bulk-flux}. Our proof is a simple adaptation of that in \cite{Dafermos2016b} for $s=0$, and follows several steps. First, we extend our solutions of the homogeneous transformed system to the past of the initial slice $\Sigma_0$ so as to obtain sufficiently integrable (i.e.\ past and future integrable simultaneously) solutions. It is these extensions which we cutoff radially to isolate the bounded $|r^*|$ region. Second, we split up our solutions to the inhomogenous transformed system into a sum of $N$ solutions each with small support in frequency space (and thus small range for $r_{\rm trap}(\omega,m,l)$), and apply \eqref{eq:1st-order-estimate-for-energy-bddness} from Proposition~\ref{prop:estimates-assuming-ILED}. This estimate requires us to choose an initial $\tau_0$; the penultimate step shows that we can make a choice such that the final estimate no longer involves our intermediate $N$ solution pieces. These steps ensure control over the energy for the $k=|s|$ bounded $|r^*|$ flux; in the last step we use this fact to obtain bounds on the fluxes of the lower level transformed variables.

\medskip
\noindent \textit{Step 1: extension}. Recall that $\upphi_k$, for $k=0,\dots,|s|$, denote solutions to the homogeneous transformed system of Definition~\ref{def:transformed-system} arising from suitably regular data on $\Sigma_0$. Thus, by the well-posedness statement of Proposition~\ref{prop:well-posedness}, $\upphi_k$ are spin-weighted functions in $\mc R_0$. However, appealing to Lemma~\ref{lemma:reduction-argument}, we can assume without loss of generality, that $\tilde\upphi_k=0$ on $\mc I^+_{(0,1)}$ and $\mc H^+_{(0,1)}$. The extension arguments in (the proof of) Lemma~\ref{prop:scattering-construction}, then ensure that we can extend $\upphi_k$, for $k=0,\dots,|s|$, to $\mc R$ in such a way that
\begin{align*}
\sum_{k=0}^{|s|}\sum_{J=0}^{|s|-k}\mathbb{E}^J[\tilde\upphi_k](0^-)\leq B\sum_{k=0}^{|s|}\sum_{J=0}^{|s|-k}\mathbb{E}^J[\tilde\upphi_k](0)\,,
\end{align*}
where $\Sigma_{0^-}$ denotes the image of $\Sigma_0$ under the map $t\mapsto -t$, in Boyer-Lindquist coordinates, and by abuse of notation we denote the extension of $\tilde \upphi_k$ by the same symbol. With this extension, the conclusion of Proposition~\ref{prop:ODE-to-PDE-radial-cutoff-ILED} becomes
\begin{align*}
&\int_{\mc{R}} \lp(|\Phi_\cutr'|^2+|T\mc{P}_{\rm trap}[\xi\Phi_\cutr]|^2+|\mathring{\slashed\nabla}^{[s]}\mc{P}_{\rm trap}[\xi\Phi_\cutr]|^2+|\Phi_\cutr|^2\rp)dr  d\sigma d\tau + \sum_{k=0}^{|s|}\mathbb{I}^{\mathrm{deg},|s|-k}[\upphi_{k,\cutr}](-\infty,\infty) \\
&\quad\leq  B(a_0)\sum_{k=0}^{|s|}\mathbb{E}^{|s|-k}[\upphi_{k,\cutr}](0)+B(a_0)\sum_{k=0}^{|s|}{\mathbb{E}^{|s|-k}}[\upphi_{k,\cutr}](0^-)\leq B(a_0)\sum_{k=0}^{|s|}{\mathbb{E}^{|s|-k}}[\upphi_{k,\cutr}](0)\,, \numberthis\label{eq:energy-bddness-intermediate-1}
\end{align*}
where $\upphi_{k,\cutr}$ is a version of $\upphi_k$ which has been cutoff to be supported only in $r^*\in [-R^*-1,R^*+1]$; we direct the reader to Section~\ref{sec:cutoffs-radial} for the precise meaning of the subscript $\cutr$. 

\medskip
\noindent \textit{Step 2: frequency space splitting}. For some $\epsilon>0$ to be specified, let us consider $N=\lceil\epsilon^{-1}(s^++s^-)\rceil$ intervals in the radial coordinate given by
\begin{align*}
I_i=[3M -s^-+(i-1)\epsilon,3M-s^-+i\epsilon]\,, \,\,i=1,\dots, N\,.
\end{align*}
We consider a partition of frequency space as follows: $\mc F_{\rm admiss}=\uplus_{i=0}^N \mc C_i$, where
\begin{align*}
\mc{C}_0:=\{(\omega,m,l)\colon r_{\rm trap}=0\}\,; \qquad \mc{C}_i:=\{(\omega,m,l)\colon r_{\rm trap}\in I_i\}\,,\,\, i=1,\dots,N\,,
\end{align*}
and define $\upphi_{k,\cutr}^{(i)}:=\mc P_{\mc C_i}[\upphi_{k,\cutr}]$, $k=0,\dots,|s|$, where $\mc P_{\mc C_i}$ is a microlocal operator which denotes the projection, in frequency space, to the set $\mc C_i$. Thus, $\upphi_{k,\cutr}^{(i)}$ is not well localized in physical space (by the previous step, it is a function in $\mc R$ instead of just $\mc R_0$), but it is very precisely localized in physical space.

An easy computation shows that the vector field $T+\frac{2Mar}{(r^2+a^2)^2}Z$ is timelike in $\mc R\backslash\mc H^+$. If $\epsilon$ is sufficiently small, we can choose functions $\chi_i$ whose gradient is supported in $I_i^c$ and such that $T+\upomega_+\chi_i(r) Z$ is also timelike; we choose $\epsilon$ in this way. Then, we apply estimate \eqref{eq:1st-order-estimate-for-energy-bddness} from Proposition~\ref{prop:estimates-assuming-ILED} to each $\Phi_{\cutr}^{(i)}$. Since  $\chi_i'$ and $\mathfrak{H}_k^{(i)}$ are supported away from the $r$-region where $\Phi_{\cutr}^{(i)}$ experiences trapping, we get
\begin{align*}
\mathbb{E}[\Phi_{\cutr}](\tau) &\leq B\sum_{i=1}^N \mathbb{E}[\Phi_{\cutr}^{(i)}](\tau)\\
&\leq B\sum_{j=0}^{|s|}\sum_{i=1}^N\mathbb{E}[\upphi_{j,\cutr}^{(i)}](\tau_0) +B\sum_{i=1}^N \int_{\mc R}\lp(\sum_{j=0}^{|s|-1}|Z^{3/2}\widetilde{\mc P}_{\rm trap}[\upphi_{j,\cutr}^{(i)}]|^2+|Z^{1/2}\Phi_{\cutr}^{(i)}|^2\rp) drd\sigma d\tau \\
&\qquad +B\sum_{i=1}^N \int_{\mc R_{(-\infty,\infty)}}\lp(|(\Phi_{\cutr}^{(i)})'|^2+|T\mc{P}_{\rm trap}[\Phi_{\cutr}^{(i)}]|^2+|\mathring{\slashed\nabla}^{[s]}\mc{P}_{\rm trap}[\Phi_{\cutr}^{(i)}]|^2\rp)drd\sigma d\tau\,.
\end{align*}
(Note that we have passed to frequency space to split the $Z$ derivatives unevenly between $\Phi_{\cutr}^{(i)}$ and $\upphi_{j,\cutr}^{(i)}$ in the last term of \eqref{eq:1st-order-estimate-for-energy-bddness}.) Now, we use Plancherel and use the disjointness of the frequency supports of $\upphi_{j,\cutr}$ to sum in $i$:
\begin{align*}
&\sum_{i=1}^N \int_{\mc R}\Big(|(\Phi_{\cutr}^{(i)})'|^2+|T\mc{P}_{\rm trap}[\Phi_{\cutr}^{(i)}]|^2+|\mathring{\slashed\nabla}^{[s]}\mc{P}_{\rm trap}[\Phi_{\cutr}^{(i)}]|^2+\sum_{j=0}^{|s|-1}|Z^{3/2}\widetilde{\mc P}_{\rm trap}[\upphi_{j,\cutr}^{(i)}]|^2\Big)drd\sigma d\tau\\
&\quad =  \int_{\mc R}\Big(|Z^{1/2}\Phi_{\cutr}|^2+|\Phi_{\cutr}'|^2+|T\mc{P}_{\rm trap}[\Phi_{\cutr}]|^2+|\mathring{\slashed\nabla}^{[s]}\mc{P}_{\rm trap}[\Phi_{\cutr}]|^2+\sum_{j=0}^{|s|- 1}|Z^{3/2}\widetilde{\mc P}_{\rm trap}[\upphi_{j,\cutr}]|^2\Big)drd\sigma d\tau\,.
\end{align*}
In this identity, it is very important that $\upphi_{k,\cutr}$ and $\upphi_{k,\cutr}^{(i)}$ live in the entire $\mc R$ so that we may apply Plancherel. Finally, appealing to \eqref{eq:energy-bddness-intermediate-1}, we obtain
\begin{align*}
\mathbb{E}[\Phi_{\cutr}](\tau) &\leq  B\sum_{j=0}^{|s|}\sum_{i=1}^N\mathbb{E}[\upphi_{j,\cutr}^{(i)}](\tau_0) +B\sum_{k=0}^{|s|}\mathbb{E}^{|s|-k}[\tilde\upphi_{k}](0) \numberthis\label{eq:energy-bddness-intermediate-2}
\end{align*}

\medskip
\noindent \textit{Step 3: a good choice of $\tau_0$}. In the previous step, we have used the form of the two last terms on the right hand side of \eqref{eq:1st-order-estimate-for-energy-bddness} to be able to, once the estimate was applied to $\upphi_{k,\cutr}^{(i)}$, sum their contributions in $i$ and recast them in terms of $\upphi_{k,\cutr}$. For the first term in \eqref{eq:energy-bddness-intermediate-2}, our strategy to remove the $i$-dependence is to make a suitable choice of initial slice $\Sigma_{\tau_0}$, similarly to the first step in the proof of Theorem~\ref{thm:ILED-futurint}.

Recall that each $\upphi_{j,\cutr}^{(i)}$ is sufficiently integrable and compact in $r^*$; hence making use of Lemma~\ref{lemma:pigeonhole-energy-decay}, we can produce a dyadic sequence $\{\tau_n^{(i)}\}_{n=1}^\infty$ such that $\tau_n^{(i)}\to -\infty$ as $n\to \infty$ and so that, 
% \begin{align*}
% \int_{-\infty}^\infty \mathbb{E}[\upphi_{j,\cutr}^{(i)}](\tau)d\tau<\infty\implies \mathbb{E}[\upphi_{j,\cutr}^{(i)}](\tau_n^{(i)})\leq\frac{C_i}{|\tau_n^{(i)}|}\,,
% \end{align*}
%by the pigeonhole principle, for some constant $C_i=C_i(\upphi_{j,\cutr}^{(i)})$ and a dyadic sequence $\{\tau_n^{(i)}\}_{n=1}^\infty$ such that $\tau_n^{(i)}\to -\infty$ as $n\to \infty$. Thus, 
choosing $\tau_0=\tau_n^{(i)}$ in \eqref{eq:energy-bddness-intermediate-2} and then taking $n\to\infty$, we obtain
\begin{align*}
\mathbb{E}[\Phi_{\cutr}](\tau) &\leq B\sum_{k=0}^{|s|}\mathbb{E}^{|s|-k}[\tilde\upphi_{k}](0) \numberthis\label{eq:energy-bddness-intermediate-3}
\end{align*}

\medskip

\noindent \textit{Step 4: the lower level variables}. First note that, by a similar argument to the previous two steps, estimate \eqref{eq:1st-order-estimate-for-energy-bddness} from Proposition~\ref{prop:estimates-assuming-ILED} yields that for $k<|s|$
\begin{align*}
\mathbb{E}[\upphi_{k,\cutr}](\tau) &\leq  B\sum_{j=0}^{|s|}\sum_{i=1}^N\mathbb{E}^1[\upphi_{j,\cutr}^{(i)}](\tau_0) +B\sum_{k=0}^{|s|}\mathbb{E}^{|s|-k+1}[\tilde\upphi_{k}](0)\,,
\end{align*}
so choosing $\tau_0$ as in Step 3 gives us 
\begin{align*}
\mathbb{E}[\upphi_{k,\cutr}](\tau) &\leq B\sum_{k=0}^{|s|}\mathbb{E}^{|s|-k+1}[\tilde\upphi_{k}](0)<\infty\,. \numberthis\label{eq:finite-energy-lower-levels}
\end{align*}

In estimate \eqref{eq:finite-energy-lower-levels}, we have lost one derivative. To retrieve it, we pursue a more delicate argument.  By combining the higher order estimate \eqref{eq:higher-order-energy-bddness} of Proposition~\ref{prop:higher-order-bulk-flux} with Theorem~\ref{thm:ILED-futurint}, we see that 
\begin{align*}
\sum_{k=0}^{|s|-1}\mathbb{E}^{|s|-k}[\tilde\upphi_{k,\cutr}](\tau)\leq  B\mathbb{E}[\Phi_{\cutr}](\tau)+ B\sum_{k=0}^{|s|-1}\sum_{J_1+J_2=0}^{|s|-k-1}\mathbb{E}[T^{J_1}Z^{J_2}\tilde\upphi_{k,\cutr}'](\tau)+ B\sum_{k=0}^{|s|}\mathbb{E}^{|s|-k}[\tilde\upphi_{k}](0)\,.
\end{align*}
Now let $\xi(\tilde t^*)$ be a smooth cutoff is equal to one for $\tilde t^*=\tau$ and vanishes for $\tilde t^*\leq \tau-1$ and $\tilde t^*\geq \tau+1$. We use the fact that $T$ and $Z$ are Killing as well as the estimate \eqref{eq:phys-space-Killing-multiplier-consequence-1-radial-commuted} to deduce
\begin{align*}
&\sum_{k=0}^{|s|-1}\mathbb{E}^{|s|-k}[\tilde\upphi_{k,\cutr}](\tau)= \sum_{k=0}^{|s|-1}\mathbb{E}^{|s|-k}[\xi\tilde\upphi_{k,\cutr}](\tau)\\
&\quad\leq  
B\sum_{k=0}^{|s|-1}\sum_{J_1+J_2=0}^{|s|-k-1}\Big(\mathbb{I}^{\rm deg}[T^{J_1}Z^{J_2}\tilde\upphi_{k+1,\cutr}](\tau-1,\tau+1)+\varepsilon^{-1}\sum_{J=0}^1\mathbb{I}^{\rm deg,J}[T^{J_1}Z^{J_2}\tilde\upphi_{k}](\tau-1,\tau+1)\Big)\\
&\quad\qquad+
B\varepsilon\sum_{k=0}^{|s|-1}\sum_{J_1+J_2=1}^{|s|-k-1}\mathbb{I}^{1}[T^{J_1}Z^{J_2}\tilde\upphi_{k}](\tau-1,\tau+1)+ B\sum_{k=0}^{|s|}\mathbb{E}^{|s|-k}[\tilde\upphi_{k}](0)\\
&\quad\leq  
B\sum_{k=0}^{|s|-1}\sum_{J_1+J_2=0}^{|s|-k-1}\Big(\mathbb{I}^{\rm deg}[T^{J_1}Z^{J_2}\tilde\upphi_{k+1,\cutr}](0,\tau+1)+\varepsilon^{-1}\sum_{J=0}^1\mathbb{I}^{\rm deg,J}[T^{J_1}Z^{J_2}\tilde\upphi_{k}](0,\tau+1)\Big)\\
&\quad\qquad+
B\varepsilon\sum_{k=0}^{|s|-1}\sum_{J_1+J_2=1}^{|s|-k-1}\mathbb{I}^{1}[T^{J_1}Z^{J_2}\tilde\upphi_{k}](\tau-1,\tau+1)+ B\sum_{k=0}^{|s|}\mathbb{E}^{|s|-k}[\tilde\upphi_{k}](0)\\
&\quad\leq B\varepsilon\sum_{k=0}^{|s|-1}\mathbb{E}^{|s|-k}[\tilde\upphi_{k,\cutr}](\tau)+
B\varepsilon^{-1}\sum_{k=0}^{|s|}\mathbb{E}^{|s|-k}[\tilde\upphi_{k}](0)\,,
\end{align*}
where we have used the results of Theorem~\ref{thm:ILED-futurint} once more to replace the  ILED bulk terms by data terms, and large $|r^*|$ estimates as in Proposition~\ref{prop:bulk-flux-large-r}. Taking a supremum in $\tau\in\mathbb R$, we find that for sufficiently small $\varepsilon$, the first term above, which is finite by \eqref{eq:finite-energy-lower-levels}, can be absorbed into the left hand side. This concludes the proof. 
\end{proof}

\subsection{Energy and pointwise decay of future integrable solutions}

From Theorems~\ref{thm:ILED-futurint} and \ref{thm:energy-bddness-futurint}, by (now) standard arguments introduced in \cite{Dafermos2010a} and expanded in \cite{Moschidis2016}, we have:

\begin{theorem}[Energy and pointwise decay] \label{thm:decay} Fix $s\in\{0, \pm 1,\pm 2\}$, $a_0\in[0,M)$, $\eta\in(0,1)$, and assume that, for each $k=0,\dots, |s|$, $\swei{\tilde\upphi}{s}_k$ are future integrable solutions of  the homogeneous transformed system of Definition~\ref{def:transformed-system}. Then, we have that for $\tau>1$, $\eta\in(0,1)$ and $|a|\leq a_0$, the energy flux of the system through $\Sigma_\tau$ decays as a function of $\tau$
\begin{align*}
\sum_{k=0}^{|s|}\mathbb{E}^{|s|-k}[\swei{\tilde\upphi}{s}_k](\tau)\leq \frac{B(a_0,\delta,\eta)\overline{\mathbb{D}}_{2,2}}{\tau^{2-[\eta]}}\,,
\end{align*}
and moreover the solutions themselves decay pointwise in $\tau$:
\begin{align*}
\sup_{\Sigma_\tau}\sup_{0\leq k\leq |s|} |\swei{\tilde\upphi}{s}_k|\leq \frac{B(a_0,\delta,\eta)\sqrt{\overline{\mathbb{D}}_{2,4}}}{\tau^{1/2-[\eta]/2}}\,.
\end{align*}
In the estimates above, $[\eta]=0$ if $s\leq 0$ but $[\eta]=\eta$ if $s>0$; the constants $\overline{\mathbb{D}}_{2,2}$ and $\overline{\mathbb{D}}_{2,4}$ represent higher order norms on the initial data given by
\begin{align*}
\overline{\mathbb{D}}_{2-[\eta],2}&:=
\sum_{k=0}^{|s|}\lp(\sum_{J_1+J_2=0}^1\overline{\mathbb{E}}_{2-[\eta]}^{|s|-k}[T^{J_1}Z^{J_2}\tilde\upphi_k](0)+\sum_{J_1+J_2=2}\overline{\mathbb{E}}_{0}^{|s|-k}[T^{J_1}Z^{J_2}\tilde\upphi_k](0)\rp)\,,\\
\overline{\mathbb{D}}_{2-[\eta],4}&:=\sum_{|\beta|\leq 2}\sum_{\Gamma\in\{T,X^*,\chi_R\Omega_i\}}\sum_{k=0}^{|s|}\lp(\sum_{J_1+J_2=0}^1\overline{\mathbb{E}}_{2-[\eta]}^{|s|-k}[\Gamma^\beta T^{J_1}Z^{J_2}\tilde\upphi_k](0)+\sum_{J_1+J_2=2}\overline{\mathbb{E}}_{0}^{|s|-k}[\Gamma^\beta T^{J_1}Z^{J_2}\tilde\upphi_k](0)\rp)\,,
\end{align*}
where $\Omega_i$ denotes the usual three angular operators and $\chi$ a cutoff which is supported at large $r$.
\end{theorem}

\begin{proof} Recall estimate \eqref{eq:ILED-for-decay} from Proposition~\ref{prop:higher-order-bulk-flux} with $J=0$. By Theorems~\ref{thm:ILED-futurint} and \ref{thm:energy-bddness-futurint}, and the fact that $T$ and $Z$ are Killing, we have for any $0\leq \tau_A\leq \tau_B$, some $R^*$ sufficiently large which we fix, and $p\in[0,2)$,
\begin{align*}
&\sum_{J_1+J_2=0}^J\sum_{k=0}^{|s|}\overline{\mathbb{I}}_p^{|s|-k}[T^{J_1}Z^{J_2}\tilde\upphi_k](\tau_A,\tau_B)\\
&\quad\leq B\sum_{k=0}^{|s|}\lp(\sum_{J_1+J_2=0}^{J}\overline{\mathbb{E}}_p^{|s|-k}[T^{J_1}Z^{J_2}\tilde\upphi_k](\tau_A)+\sum_{J_1+J_2=J+1}\overline{\mathbb{E}}^{|s|-k}[T^{J_1}Z^{J_2}\tilde\upphi_k\mathbbm{1}_{[-R^*,R^*]}](\tau_A)\rp)\,,\\
&\sum_{J_1+J_2=0}^J\sum_{k=0}^{|s|}\overline{\mathbb{E}}_p^{|s|-k}[T^{J_1}Z^{J_2}\tilde\upphi_k](\tau_B)\leq B\sum_{J_1+J_2=0}^{J}\sum_{k=0}^{|s|}\overline{\mathbb{E}}_p^{|s|-k}[T^{J_1}Z^{J_2}\tilde\upphi_k](\tau_A)\,,
\numberthis\label{eq:corollary-starting-point}
\end{align*}
and we can include $p=2$ in the case $s\leq 0$. 
We can now essentially repeat the proofs given by \cite{Dafermos2010a} and \cite{Moschidis2016} for the case $s=0$. It will be convenient to use the notation $[\eta]$ to denote $[\eta]=0$ if $s\leq 0$ and $[\eta]=\eta\in(0,1)$ for $s>0$.

\smallskip
\noindent\textit{Step 1: taking $p=2-[\eta]$ for $s<0$}. 
From the definitions of the norms in Section~\ref{sec:template-energy-norms}, we have 
\begin{align*}
&\sum_{J_1+J_2=0}^1\sum_{k=0}^{|s|}\int_{\tau_0}^\infty\overline{\mathbb{E}}_{1-[\eta]}^{|s|-k}[T^{J_1}Z^{J_2}\tilde\upphi_k](\tau)d\tau \\
&\quad\leq B\sum_{J_1+J_2=0}^1\sum_{k=0}^{|s|}\overline{\mathbb{I}}_{2-[\eta]}^{|s|-k}[T^{J_1}Z^{J_2}\tilde\upphi_k](\tau_0,\infty)\\
&\quad\leq B\sum_{k=0}^{|s|}\lp(\sum_{J_1+J_2=0}^1\overline{\mathbb{E}}_{1-[\eta]}^{|s|-k}[T^{J_1}Z^{J_2}\tilde\upphi_k](\tau_0)+\sum_{J_1+J_2=2}\overline{\mathbb{E}}_0^{|s|-k}[T^{J_1}Z^{J_2}\tilde\upphi_k](\tau_0)\rp)\,,
\end{align*}
using \eqref{eq:corollary-starting-point} in the last line. By the pigeonhole principle, we can find a dyadic sequence $\{\tau_n\}_{n=1}^\infty$ with $\tau_n\to \infty$, say $\tau_n=2^n \tau_0$, such that 
\begin{align*}
&\sum_{J_1+J_2=0}^1\sum_{k=0}^{|s|}\overline{\mathbb{E}}_{1-[\eta]}^{|s|-k}[T^{J_1}Z^{J_2}\tilde\upphi_k](\tau_n) \\
&\quad \leq  \frac{B}{\tau_n}\sum_{k=0}^{|s|}\lp(\sum_{J_1+J_2=0}^1\overline{\mathbb{E}}_{2-[\eta]}^{|s|-k}[T^{J_1}Z^{J_2}\tilde\upphi_k](\tau_0)+\sum_{J_1+J_2=2}\overline{\mathbb{E}}_0^{|s|-k}[T^{J_1}Z^{J_2}\tilde\upphi_k](\tau_0)\rp)\,. \numberthis\label{eq:corollary-p2}
\end{align*}

\smallskip
\noindent\textit{Step 2: taking $p=1$}.  We repeat the steps above, now for the time slab $(\tau_A,\tau_B)=(\tau_n,\tau_{n+1})$, and commuting with $T$ and $Z$ exactly once. We have 
\begin{align*}
&\int_{\tau_n}^{\tau_{n+1}}\sum_{J_1+J_2=1}\sum_{k=0}^{|s|}\overline{\mathbb{E}}_{-[\eta]}^{|s|-k}[T^{J_1}Z^{J_2}\tilde\upphi_k](\tau)d\tau \\
&\quad \leq B\sum_{J_1+J_2=1}\sum_{k=0}^{|s|}\overline{\mathbb{I}}_{1-[\eta]}^{|s|-k}[T^{J_1}Z^{J_2}\tilde\upphi_k](\tau_n,\tau_{n+1})\\
&\quad\leq B\sum_{k=0}^{|s|}\lp(\sum_{J_1+J_2=1}\overline{\mathbb{E}}_{1-[\eta]}^{|s|-k}[T^{J_1}Z^{J_2}\tilde\upphi_k](\tau_n)+\sum_{J_1+J_2=2}\overline{\mathbb{E}}_0^{|s|-k}[T^{J_1}Z^{J_2}\tilde\upphi_k](\tau_n)\rp)\\
%&\quad\leq \sum_{k=0}^{|s|}\frac{B}{\tau_n}\lp(\sum_{J_1+J_2=1}\overline{\mathbb{E}}_{2-[\eta]}^{|s|-k}[T^{J_1}Z^{J_2}\tilde\upphi_k](\tau_0)+\sum_{J_1+J_2=2}\overline{\mathbb{E}}_0[T^2\tilde\upphi_k](\tau_0)\rp)+B\sum_{J_1+J_2=2}\overline{\mathbb{E}}_0^{|s|-k}[T^{J_1}Z^{J_2}\tilde\upphi_k](\tau_0)\\
&\quad\leq \frac{B}{\tau_n}\sum_{k=0}^{|s|}\sum_{J_1+J_2=1}\overline{\mathbb{E}}_{2-[\eta]}^{|s|-k}[T^{J_1}Z^{J_2}\tilde\upphi_k](\tau_0)+B\sum_{J_1+J_2=2}\sum_{k=0}^{|s|}\overline{\mathbb{E}}_0^{|s|-k}[T^{J_1}Z^{J_2}\tilde\upphi_k](\tau_0)\,,
\end{align*}
using \eqref{eq:corollary-p2} for the first term on the second line plus energy boundedness, as in \eqref{eq:corollary-starting-point}, for the second term. By the mean value theorem, we deduce that for some possibly new dyadic sequence $\overline\tau_n\in[\tau_n-\frac{\tau_{n+1}-\tau_n}{4},\tau_n+\frac{\tau_{n+1}-\tau_n}{4}]$, we have
\begin{align*}
&\sum_{J_1+J_2=1}\sum_{k=0}^{|s|}\overline{\mathbb{E}}_{-[\eta]}^{|s|-k}[T^{J_1}Z^{J_2}\tilde\upphi_k](\overline\tau_n)\\
&\quad\leq \frac{B}{\overline\tau_n-\tau_n}\lp( \frac{1}{\tau_n}\sum_{J_1+J_2=1}\sum_{k=0}^{|s|}\overline{\mathbb{E}}_{2-[\eta]}^{|s|-k}[T^{J_1}Z^{J_2}\tilde\upphi_k](\tau_0)+\sum_{J_1+J_2=2}\sum_{k=0}^{|s|}\overline{\mathbb{E}}_0^{|s|-k}[T^{J_1}Z^{J_2}\tilde\upphi_k](\tau_0)\rp)\\
%&\leq \frac{B}{\overline \tau_n^2}\sum_{k=0}^{|s|}\overline{\mathbb{E}}_{2-[\eta]}^{|s|-k}[T\tilde\upphi_k](\tau_0)+\frac{B}{\overline \tau_n}\sum_{k=0}^{|s|}\overline{\mathbb{E}}_0^{|s|-k}[T^2\tilde\upphi_k](\tau_0)\\
&\quad\leq \frac{B}{\overline \tau_n}\sum_{k=0}^{|s|}\lp(\sum_{J_1+J_2=1}\overline{\mathbb{E}}_{2-[\eta]}^{|s|-k}[T^{J_1}Z^{J_2}\tilde\upphi_k](\tau_0)+\sum_{J_1+J_2=2}\overline{\mathbb{E}}_0^{|s|-k}[T^{J_1}Z^{J_2}\tilde\upphi_k](\tau_0)\rp) \,.\numberthis\label{eq:corollary-p1-T}
\end{align*}

We repeat the same steps once more with the non-commuted version of $\overline{\mathbb{E}}_0$, now for $(\tau_A,\tau_B)=(\overline\tau_n,\overline\tau_{n+1})$:
\begin{align*}
\int_{\overline\tau_n}^{\overline\tau_{n+1}}\sum_{k=0}^{|s|}\overline{\mathbb{E}}_{-[\eta]}^{|s|-k}[\tilde\upphi_k](\tau)d\tau &\leq B\sum_{k=0}^{|s|}\overline{\mathbb{I}}_{1-[\eta]}^{|s|-k}[\tilde\upphi_k](\overline\tau_n,\infty)\\
&\leq B\sum_{k=0}^{|s|}\lp(\overline{\mathbb{E}}_{1-[\eta]}^{|s|-k}[\tilde\upphi_k](\overline\tau_n)+\sum_{J_1+J_2=1}\overline{\mathbb{E}}_{-[\eta]}^{|s|-k}[T^{J_1}Z^{J_2}\tilde\upphi_k](\overline\tau_n)\rp)\\
&\leq B\sum_{k=0}^{|s|}\lp(\overline{\mathbb{E}}_{1-[\eta]}^{|s|-k}[\tilde\upphi_k](\tau_n)+B\sum_{J_1+J_2=1}\overline{\mathbb{E}}_{-[\eta]}[T^{J_1}Z^{J_2}\tilde\upphi_k](\overline\tau_n)\rp)\\
&\leq \frac{B}{\overline\tau_n}\sum_{k=0}^{|s|}\lp(\sum_{J_1+J_2=0}^1\overline{\mathbb{E}}_{2-[\eta]}^{|s|-k}[T^{J_1}Z^{J_2}\tilde\upphi_k](\tau_0)+\sum_{J_1+J_2=2}\overline{\mathbb{E}}_0^{|s|-k}[T^{J_1}Z^{J_2}\tilde\upphi_k](\tau_0)\rp)\,,
\end{align*}
where we have used energy boundedness, as given in \eqref{eq:corollary-starting-point}, to pass to the second line and our previous estimates \eqref{eq:corollary-p2} and \eqref{eq:corollary-p1-T} to obtain the conclusion. As before, we deduce that for some possibly new dyadic sequence $\overline{\overline{\tau}}_n\in[\overline\tau_n-\frac{\overline\tau_{n+1}-\overline\tau_n}{4},\overline\tau_n+\frac{\overline\tau_{n+1}-\overline\tau_n}{4}]$, we have
\begin{align*}
\sum_{k=0}^{|s|}\overline{\mathbb{E}}_{-[\eta]}^{|s|-k}[\tilde\upphi_k](\overline{\overline{\tau}}_n)&\leq \frac{B}{\overline{\overline{\tau}}_n^2}\sum_{k=0}^{|s|}\lp(\sum_{J_1+J_2=0}^1\overline{\mathbb{E}}_{2-[\eta]}^{|s|-k}[T^{J_1}Z^{J_2}\tilde\upphi_k](\tau_0)+\sum_{J_1+J_2=2}\overline{\mathbb{E}}_0^{|s|-k}[T^{J_1}Z^{J_2}\tilde\upphi_k](\tau_0)\rp)\,. \numberthis\label{eq:corollary-p1-id}
\end{align*}

\smallskip
\noindent\textit{Step 3: energy decay}. In the case $s\leq 0$, when $[\eta]=\eta\in(0,1)$, we interpolate
\begin{align*}
\sum_{k=0}^{|s|}\overline{\mathbb{E}}_{0}^{|s|-k}[\tilde\upphi_k](\overline{\overline{\tau}}_n)&\leq \lp(\sum_{k=0}^{|s|}\overline{\mathbb{E}}_{-\eta}^{|s|-k}[\tilde\upphi_k](\overline{\overline{\tau}}_n) \rp)^{1-\eta/2}\lp(\sum_{k=0}^{|s|}\overline{\mathbb{E}}_{2-\eta}^{|s|-k}[\tilde\upphi_k](\overline{\overline{\tau}}_n)\rp)^{\eta/2}\\
&\leq \frac{B}{\overline{\overline{\tau}}_n^{2-\eta}} \lp(\sum_{k=0}^{|s|}\lp(\sum_{J_1+J_2=0}^1\overline{\mathbb{E}}_{2-\eta}^{|s|-k}[T^{J_1}Z^{J_2}\tilde\upphi_k](\tau_0)+\sum_{J_1+J_2=2}\overline{\mathbb{E}}_0^{|s|-k}[T^{J_1}Z^{J_2}\tilde\upphi_k](\tau_0)\rp)\rp)^{1-\eta/2}\\
&\quad\qquad \times \lp(\sum_{k=0}^{|s|}\overline{\mathbb{E}}_{2-\eta}^{|s|-k}[\tilde\upphi_k](\tau_0)\rp)^{-\eta/2}\\
&\quad\leq \frac{B}{\overline{\overline{\tau}}_n^{2-\eta}} \sum_{k=0}^{|s|}\lp(\sum_{J_1+J_2=0}^1\overline{\mathbb{E}}_{2-\eta}^{|s|-k}[T^{J_1}Z^{J_2}\tilde\upphi_k](\tau_0)+\sum_{J_1+J_2=2}\overline{\mathbb{E}}_0^{|s|-k}[T^{J_1}Z^{J_2}\tilde\upphi_k](\tau_0)\rp)
\,,
\end{align*}
and use the energy boundedness with $p=2-\eta$ stated in \eqref{eq:corollary-starting-point}. In the case $s\leq 0$, where $[\eta]=0$, we do not need this additional step. To conclude, we note that since we also have energy boundedness for $\overline{\mathbb{E}}_0$ by \eqref{eq:corollary-starting-point}, 
we deduce 
\begin{align*}
\sum_{k=0}^{|s|}\overline{\mathbb{E}}_0^{|s|-k}[\tilde\upphi_k](\overline{\overline{\tau_{n+1}}}) &\leq \frac{B}{\overline{\overline{\tau}}_{n}^2}\sum_{k=0}^{|s|}\lp(\sum_{J_1+J_2=1}\overline{\mathbb{E}}_{2-[\eta]}^{|s|-k}[T^{J_1}Z^{J_2}\tilde\upphi_k](\tau_0)+\sum_{J_1+J_2=2}\overline{\mathbb{E}}_0^{|s|-k}[T^{J_1}Z^{J_2}\tilde\upphi_k](\tau_0)\rp)\\
&\leq \frac{4B}{\overline{\overline{\tau}}_{n+1}^2}\sum_{k=0}^{|s|}\lp(\sum_{J_1+J_2=1}\overline{\mathbb{E}}_{2-[\eta]}^{|s|-k}[T^{J_1}Z^{J_2}\tilde\upphi_k](\tau_0)+\sum_{J_1+J_2=2}\overline{\mathbb{E}}_0^{|s|-k}[T^{J_1}Z^{J_2}\tilde\upphi_k](\tau_0)\rp)
\end{align*}
and so in fact the decay holds for any $\tau>\tau_0>0$:
\begin{align*}
\sum_{k=0}^{|s|}\overline{\mathbb{E}}_0^{|s|-k}[T^j\tilde\upphi_k](\tau) 
&\leq \frac{B}{\tau^2}\sum_{k=0}^{|s|}\lp(\sum_{J_1+J_2=1}\overline{\mathbb{E}}_{2-[\eta]}^{|s|-k}[T^{J_1}Z^{J_2}\tilde\upphi_k](\tau_0)+\sum_{J_1+J_2=2}\overline{\mathbb{E}}_0^{|s|-k}[T^{J_1}Z^{J_2}\tilde\upphi_k](\tau_0)\rp)\,.
\end{align*}

\smallskip
\noindent\textit{Step 4: pointwise decay decay}. Finally, we can obtain pointwise decay through Sobolev inequalities, see for instance \cite[Sections 5.3.7 and 5.3.8]{Dafermos2008}.
\end{proof}

\subsection{Future integrability of all solutions arising from regular data}

In this section, we show that
\begin{theorem}[Future integrability] \label{thm:cont-arg} Fix $s\in\mathbb{Z}_{\leq 2}$ and $M>0$. For all $a<M$, let $\swei{\upphi}{s}_k$ for each $k\in\{0,\dots,|s|\}$ be solutions to the homogeneous transformed system of Definition~\ref{def:transformed-system} arising from smooth, compactly supported initial data. Then, $\swei{\upphi}{s}_k$ are future integrable in the sense of Definition~\ref{def:future-int}.
\end{theorem} 
Theorem~\ref{thm:cont-arg} renders the results of our previous sections, namely Theorems~\ref{thm:ILED-futurint}, \ref{thm:energy-bddness-futurint} and \ref{thm:decay}, unconditional, thus completing the proof of our main results, Theorem~\ref{thm:main} and Corollary~\ref{cor:main}.

We will prove Theorem~\ref{thm:cont-arg} by a reduction the the case of fixed azimuthal modes. If $\swei{\upalpha}{s}$ is a homogeneous solution to the Teukolsky equation \eqref{eq:teukolsky-alpha} arising from suitably regular data as in the above theorem, then it can be decomposed into its azimuthal modes: $\swei{\upphi}{s}_0=\sum_{m\in\mathbb{Z}}\swei{\upphi}{s}_{0,m}$ where each $\swei{\upphi}{s}_{0,m}$ is defined by $Z\swei{\upphi}{s}_{0,m}=im\swei{\upphi}{s}_{0,m}$. The decomposition is preserved by the transport relations generating the transformed system of in Lemma~\ref{lemma:transformed-system}: $\swei{\upphi}{s}_{k,m}:=(w\mc L)^k\swei{\upphi}{s}_{0,m}$ verifies $Z\swei{\upphi}{s}_{k,m}=im\swei{\upphi}{s}_{k,m}$, and one has $\swei{\upphi}{s}_k=\sum_{m\in\mathbb{Z}}\swei{\upphi}{s}_{k,m}$. Moreover, the decomposition is orthogonal; thus, it is enough to prove Theorem~\ref{thm:cont-arg} under the additional assumption that the solutions are supported on a single $m\in\mathbb{Z}$. In light of our characterization of future integrable solutions in Lemma~\ref{lemma:future-int-finite-flux}, Theorem~\ref{thm:cont-arg} follows from the following Proposition:

\begin{proposition}[Future integrability for fixed azimuthal modes]\label{prop:continuity-arg-m} Fix $s\in\mathbb{Z}_\leq 2$, $M>0$, and $m\in\mathbb{Z}$. Let $\mc A_m\subset(-M,M)$ denote the set of all $a\in(-M,M)$ such that all solutions to the homogeneous transformed system of Definition~\ref{def:transformed-system} with black hole parameters $(a,M)$ arising from smooth, compactly supported initial data, which are furthermore supported on the azimuthal number $m$, verify the condition 
\begin{align}
\sup_{\tau>0}\overline{\mathbb{E}}_p^J[\swei{\dbtilde\upphi}{s}_k](\tau)<\infty\,, \label{eq:cont-arg-finite-flux-cond}
\end{align}
for all $k\in\{0,\dots, |s|$ and $J\geq 0$ and some $p\in(1,2)$. Then, we have $\mc A_m=(-M,M)$, i.e.\ the set $\mc A_m$ is 
\begin{enumerate}[label=(\roman*)]
\item non-empty: there is an $a\in(-M,M)$ such that $a\in\mc A_m$;
\item closed: if $\{a_n\}_{n=0}^\infty$ with $a_n\in\mc A_m$ verifies $a:=\lim_{n\to\infty} a_n\in(-M,M)$ then $a\in\mc A$;
\item open: for each $\mathring{a}\in\mc A_m$, there exists an $\epsilon>0$ such that $|a-\mathring{a}|<\epsilon$ implies $a\in\mc A_m$.
\end{enumerate}
\end{proposition}

As we have argued before, in view of condition~\eqref{eq:cont-arg-finite-flux-cond} in Proposition~\ref{prop:continuity-arg-m}, it is clear that if one can close an energy boundedness statement for all solutions of the transformed system for a given black hole rotation parameter $a$ arising from suitably regular data, then $a\in\mc A_m$.  Thus, in light of  the literature preceding this work, it is easy to establish non-emptiness of $\mc A_m$:
\begin{proof}[Proof of Proposition~\ref{prop:continuity-arg-m}(i)] It follows from work of Dafermos, Holzegel and Rodnianski \cite{Dafermos2016a} that  $0\in\mc A_m$, hence the latter set is non-empty; indeed the later \cite{Dafermos2017} shows that there is some sufficiently small $\epsilon>0$ such that $[0,\epsilon)\subset\mc A_m$.
\end{proof}

In turn, closedness of $\mc A_m$ follows immediately from the results of the previous two sections and the continuity of the solutions on $a$:
\begin{proof}[Proof of Proposition~\ref{prop:continuity-arg-m}(ii)]
For some $|a|\leq a_0$, let $\tilde\upphi_k$ denote the solutions, for given smooth, compactly supported initial data, to the transformed system with black hole parameter $a$, which is further supported on a single azimuthal number $m$. Take a sequence $\{a_n\}_{n=0}^\infty$ of elements $a_n\in\mc A_m$ such that $\lim_{n\to\infty} a_n=a$ and define ${\tilde\upphi}_{k,n}$ to be the (future integrable) solutions to the system with the same initial data as ${\tilde\upphi}_k$ and supported on the same azimuthal number $m$. Then, by the  well-posedness statement in Proposition~\ref{prop:well-posedness}, for any $p\in[0,2)$ and $J\geq 0$
\begin{align*}
\sum_{k=0}^{|s|}\overline{\mathbb{E}}_p^{J+|s|-k}[{\tilde\upphi}_k](\tau)&=\lim_{n\to\infty}\sum_{k=0}^{|s|}\overline{\mathbb{E}}_p^{J+|s|-k}[{\tilde\upphi}_{k,n}](\tau)
\end{align*}
By Theorem~\ref{thm:energy-bddness-futurint}, we can bound the right hand side by data: for any $\tau\geq 0$
\begin{align*}
\sum_{k=0}^{|s|}\overline{\mathbb{E}}_p^{J+|s|-k}[{\tilde\upphi}_k](\tau)&\leq B(a_0,J_{\rm max})\sum_{k=0}^{|s|}\overline{\mathbb{E}}_p^{J+|s|-k}[{\tilde\upphi}_{k,n}](0)=B(a_0,J)\sum_{k=0}^{|s|}\overline{\mathbb{E}}_p^{J+|s|-k}[{\tilde\upphi}_{k}](0)<\infty\,,
\end{align*}
since the initial data for the solutions with $a_n$ is the same as for the solution with $a$. Hence, ${\tilde\upphi}_k$ verifies the necessary finite flux condition for all $\tau\geq 0$. Since the initial data was chosen arbitrarily in the class of smooth, compactly supported initial data, we conclude that $a\in\mc A_m$. Thus, $\mc A_m$ is closed.
\end{proof}

Finally, we are left with establishing openness of $\mc A_m$, which is significantly more involved. The main obstacle is the absence of an a priori energy boundedness statement for $a\not\in\mc A_m$: estimate~\eqref{eq:bulk-flux-estimate-large-r} in Proposition~\ref{prop:bulk-flux-large-r} shows that the flux on $\Sigma_\tau$ is controlled by a bulk term at the same differentiability level but the latter, though it can be itself be controlled by the data and the flux on $\Sigma_\tau$, comes without smallness for general subextremal $a$. In the context of fixed $m$ solutions, we can improve this estimate considerably: in the following lemma,  we gain an arbitrary number of derivatives as we show we can control the flux of $J+1$ derivatives of the solutions by a lower order bulk term.

\begin{lemma}[Energy boundedness up to lower order bulk error] \label{lemma:derivative-gain} Fix $M>0$, $s\neq 0$ and $a_0\in[0,M)$. Then, for all $|a|\leq a_0$, $p\in(1,2)$ and $J\geq 3|s|+1$, we have
\begin{align*}
&b(\delta)\sum_{k=0}^{|s|}\overline{\mathbb{I}}_{-\delta,p}^{J+|s|-k-1}[\tilde\upphi_k](0,\tau) + \sup_{\tau'\in[0,\tau]} \sum_{k=0}^{|s|}\overline{\mathbb{E}}_p^{J+|s|-k}[\tilde\upphi_k](\tau') 
\\
&\quad \leq B(a_0,m,J)\sum_{k=0}^{|s|}\overline{\mathbb{E}}_p^{J+|s|-k}[\tilde\upphi_k](0) + B(a_0,m,J)\sum_{k=0}^{|s|}\sum_{J_1=0}^{4|s|-k}\int_0^\tau\int_{\Sigma_{\tau'}} r^{-3}|T^{J_1}\tilde\upphi_k|^2dr d\sigma d\tau'\,. \numberthis\label{eq:derivative-gain}
\end{align*}
\end{lemma}
\begin{proof} 
Recall from Proposition~\ref{prop:estimates-assuming-ILED}, Lemma~\ref{lemma:ILED-for-non-future-int}, and Proposition~\ref{prop:higher-order-bulk-flux}, that for some $p\in(1,2)$ we have
\begin{align*}
\sum_{k=0}^{|s|}\overline{\mathbb{E}}_{p}^{J+|s|-k}[{\tilde\upphi}_k](\tau)
&\leq B(a_0,m)\sum_{J_1=0}^J\int_0^\tau\int_{\Sigma_{\tau'}} r^{-3}|T^{J_1}\Phi|^2 drd\sigma d\tau'+B(a_0,m)\sum_{k=0}^{|s|}\overline{\mathbb{E}}_p^{J+|s|-k}[{\tilde\upphi}_k](0)
\\
&\qquad+B(a_0,m)\sum_{k=0}^{|s|-1}\sum_{J_1=0}^{J+|s|-k-1}\int_0^\tau\int_{\Sigma_{\tau'}} r^{-3}\lp(|T^{J_1}{\tilde\upphi}_k|^2+|T^{J_1}{\tilde\upphi}_k'|^2\rp)dr d\sigma d\tau'\,,
\end{align*}
for $J\geq 3|s|$. Assume $s\neq 0$; by the transport estimate \eqref{eq:transformed-transport-tilde}, the last term can be estimated by
\begin{align*}
&\sum_{k=0}^{|s|-1}\sum_{J_1=0}^{J+|s|-k-1}\int_0^\tau\int_{\Sigma_{\tau'}} r^{-3}\lp(|T^{J_1}{\tilde\upphi}_k|^2+|T^{J_1}{\tilde\upphi}_k'|^2\rp)dr d\sigma d\tau'\\
&\leq B(m)\sum_{k=0}^{|s|-1}\sum_{J_1=0}^{J+|s|-k-2}\int_0^\tau\int_{\Sigma_{\tau'}} r^{-3}\Big(\sum_{J_2=0}^2|T^{J_1}\p_{r^*}^{J_2}{\tilde\upphi}_k|^2+|T^{J_1}\tilde\upphi_{k+1}|^2+|T^{J_1}\tilde\upphi_{k+1}'|^2\Big)dr d\sigma d\tau'\\
& \leq B(a_0,m)\sum_{k=0}^{|s|-1}\lp(\overline{\mathbb{E}}_{p}^{J+|s|-k-1}[{\tilde\upphi}_k](\tau)+\overline{\mathbb{E}}_{p}^{J+|s|-k-1}[{\tilde\upphi}_k](0)\rp)+ B(a_0,m)\sum_{J_1=0}^{J-1}\int_0^\tau\int_{\Sigma_{\tau'}} r^{-3}|T^{J_1}\Phi|^2 drd\sigma d\tau'\\
&\qquad +B(a_0,m)\sum_{k=0}^{|s|-1}\sum_{J_1=0}^{4|s|-k-1}\int_0^\tau\int_{\Sigma_{\tau'}} r^{-3}(|T^{J_1}\tilde\upphi_{k}|^2+|T^{J_1}\tilde\upphi_{k}'|^2) drd\sigma d\tau'\,.
\end{align*}
Hence, we deduce that 
\begin{align*}
&\sum_{k=0}^{|s|}\overline{\mathbb{E}}_{p}^{J+|s|-k}[{\tilde\upphi}_k](\tau)\\
&\quad\leq B(a_0,m)\int_0^\tau\int_{\Sigma_{\tau'}} r^{-3}\Big\{\sum_{J_1=0}^{J}|T^{J_1}\Phi|^2 +B(a_0,m)\sum_{k=0}^{|s|-1}\sum_{J_1=0}^{4|s|-k-1}(|T^{J_1}\tilde\upphi_{k}|^2+|T^{J_1}\tilde\upphi_{k}'|^2)\Big\}drd\sigma d\tau'\\
&\quad\quad+B(a_0,m)\sum_{k=0}^{|s|}\overline{\mathbb{E}}_p^{J+|s|-k}[{\tilde\upphi}_k](0)\,.
\end{align*}
To gain some smallness in the bulk $T^{J_1}\Phi$ term, we can use the fact that only half of a $T$ derivative is lost at trapping:
\begin{align*}
\int_{\mc R_{(0,\tau)}} r^{-2}|T^{J_1}\Phi|^2 
&\leq \int_{\mc R_{(0,\tau)}} r^{-2}|T^{-1/2}T^{1/2}[T^{J_1}(\xi\Phi)]|^2  + B\mathbb{E}^{J_1-1}[\Phi](\tau)+B\mathbb{E}^{J_1-1}[\Phi](0)\\
&\leq B\int_{-\infty}^\infty \sum_{ml}\lp(\int_{|\omega|\geq \omega_0}\frac{\omega}{\omega_0}|\omega^{J_1} \Psi|^2d\omega +\int_{|\omega|\leq \omega_0}\omega_0^{2}|\omega^{J_1-1}\Psi|^2d\omega \rp)dr^*\\
&\qquad+ B\overline{\mathbb{E}}^{J_1-1}_p[\Phi](\tau)+B\overline{\mathbb{E}}_p^{J_1-1}[\Phi](0)\\
&\leq \frac{B}{\omega_{0}}\mathbb{I}^{\mathrm{deg}}[T^{J_1}\Phi](0,\tau)  + B\omega_0^{2}\mathbb{I}^{\mathrm{deg},J_1-1}[\Phi](\tau)+B\overline{\mathbb{E}}_p^{J_1-1}[\Phi](0)\\
&\leq \frac{B}{\omega_{0}}\mathbb{E}[T^{J_1}\Phi](\tau)+ B\omega_0^2\overline{\mathbb{E}}_p^{J_1-1}[\Phi](\tau)+B\overline{\mathbb{E}}_p^{J_1-1}[\Phi](0)\,,
\end{align*}
using the fact that $\mathbb{I}^{\mathrm{deg}}[T\Phi](0,\tau)$ includes a bulk term in $|T^{3/2}(\xi\Phi)|^2$. 
(Note that this additional argument is actually not necessary in the case $s=0$, where one can make sure the error term in estimate \eqref{eq:estimate-for-openness-first-order} in Proposition~\ref{prop:estimates-assuming-ILED} involving a $T$ derivative is supported outside the region where trapping may occur, see \cite[Lemma 11.2.1]{Dafermos2016b}.) Taking $\omega_{0}$ sufficiently large and iterating concludes the proof.
\end{proof}

It is worth noting that the zeroth order bulk error term in $T^{4|s|-k}\tilde\upphi_k$ appearing in the right hand side of \eqref{eq:derivative-gain} in Lemma~\ref{lemma:derivative-gain} is, of course, controlled by initial data whenever $a\in \mc A_m$ by the integrated local energy decay estimate of Theorem~\ref{thm:ILED-futurint}. Thus, we establish openness of $\mc A_m$ by showing that if $a$ is sufficiently close to an element of $\mc A_m$ then it inherits this property from it:

\begin{proof}[Proof of Proposition~\ref{prop:continuity-arg-m}(iii)] Take $\mathring{a}\in\mc A_m$, and choose $a_0<M$ such that $|\mathring{a}|<a_0$ strictly. Now take $a\in\mathbb R$ with $\epsilon:=|a-\mathring{a}|$ to be fixed later, but sufficiently small so that $|a|\leq a_0$. 

Let $\upphi_{k}$ be solutions of the homogeneous transformed system of Definition~\ref{def:transformed-system} with black hole parameter $a$ arising from suitably regular data. By Lemma~\ref{lemma:derivative-gain}, we just need to show that for $p\in(1,2)$
\begin{align*}
\sum_{k=0}^{|s|}\sum_{J=0}^{4|s|-k}\int_{\mc R_{(0,\tau)}} r^{-3}|T^{J}\tilde\upphi_k|^2dr d\sigma d\tau' \leq B\sum_{k=0}^{|s|}\overline{\mathbb{I}}_{p}^{\mathrm{deg},4|s|-k+1}[\tilde\upphi_k](0,\tau) \numberthis\label{eq:openness-intermediate}
\end{align*}
is controlled by initial data as long as $\epsilon$ is sufficiently small. Indeed, recall that for $\upphi_{k,\circ}$ solutions of the homogeneous transformed system of Definition~\ref{def:transformed-system} with black hole parameter $\mathring{a}$ arising from the same data, \eqref{eq:openness-intermediate} \textit{is} controlled by the data by the results of the previous sections.  

\noindent \textit{Step 1: interpolated system.} For $\epsilon$ sufficiently small, we seek to a  introduce hyperbolic differential operator which allows us to interpolate between the  transformed systems with black hole parameters $a$ and $\mathring{a}$. Let
\begin{align*}
\hat{\mathfrak{R}}_{a,M,0}:=\frac{r^2+a^2}{\Delta}\lp[\mathfrak{R}_0-\sign s\frac{w'}{w}|s|\mc L\rp]\,,
\end{align*}
where the rescaling ensures that the operator extends regularly to $\mc H^+$. Let  $0\leq \chi_\tau\leq 1$ be a cutoff which is 0 in the past of $\Sigma_{\tau-\delta_0}$ and identically 1 to the future of $\Sigma_\tau$, for some sufficiently small $\delta_0$ to be fixed later; recall that these hypersurfaces are independent of the black hole rotation parameter, see Section~\ref{sec:hyp-folliation}. Then we set
\begin{align*}
\hat{\mathfrak{R}}_{0,\tau}&:=\chi_\tau \hat{\mathfrak{R}}_{\mathring{a},M,0}+(1-\chi_\tau)\hat{\mathfrak{R}}_{a,M,0}\,,
\end{align*}
to be the differential operator interpolating between $\hat{\mathfrak{R}}_{\mathring{a},M,0}$ and $\hat{\mathfrak{R}}_{{a},M,0}$. We assume $\epsilon$ is small enough that $\hat{\mathfrak{R}}_{0,\tau}$ is hyperbolic. Then, we define the interpolating $\upphi_{0,\tau}$ to be the solution of $\hat{\mathfrak{R}}_{0,\tau}[\upphi_{0,\tau}]=0$ which arises from the initial data posed for $\upphi_{0}$. Then, we have
\begin{align*}
\frac{r^2-2M r+\mathring{a}^2}{r^2+\mathring{a}^2}\hat{\mathfrak{R}}_{\mathring{a},M,0}[\upphi_{0,\tau}]&=\mathfrak{H}_0\,,\qquad\text{with} \quad
\mathfrak{H}_0:=\frac{r^2-2M r+\mathring{a}^2}{r^2+\mathring{a}^2}(\mathfrak{R}_{\mathring{a},M,0}-\mathfrak{R}_{0,\tau})[\upphi_{k,\tau}]\,,
\end{align*}
and the transport equations \eqref{eq:transformed-transport} and \eqref{eq:transformed-transport-inhom-separated} then, by Lemma~\ref{lemma:transformed-system}, yield a complete system of the form of Definition~\ref{def:transformed-system}, where there and in what follows all derivatives are with respect to the spacetime with rotation parameter $\mathring{a}$. Thus, crucially, we can interpret $\upphi_{k,\tau}$ as \textit{future-integrable} solutions to an \textit{inhomogeneous} version of the transformed system  of Definition~\ref{def:transformed-system}. 

To conclude the step, we expand on the form of the inhomogeneity. We have
\begin{align*}
\mathfrak{H}_k=(a-\mathring{a})w\sum_{j=0}^k\lp(\sum_{J_1+J_2+J_3=1}^{2+k-j}c_{s,k,j}^{(J_1,J_2,J_3)}(r,\theta,\tilde t^*)(r^{-1}\mathring{\slashed\nabla})^{J_1}T^{J_2}(X^*)^{J_3}\tilde\upphi_{j,\tau}
+r^{-2+k-j} c_{s,k,j}^0(r,\tilde t^*)\tilde \upphi_{j,\tau}\rp)\,,
\end{align*}
where $c_{s,k,j}^{(J_1,J_2,J_3)}$ and $c_{s,k,j}^0$ are regular functions in $r\in[r_+,\infty)$. In particular, if $|s|\leq 2$, then
\begin{align*}
\int_{\mc R_{(0,\tau)}} \frac{|{\mathfrak{H}}_k|^2}{w}dr^*d\sigma d\tau
&\leq 
\epsilon B(\delta_0) \sum_{j=0}^k\int_{\mc R_{(0,\tau-\delta_0)}} r^{-2}\Big(\sum_{1\leq |i|\leq 2+k-j}|(r^{-1}\mathring{\slashed\nabla})^{i_1}T^{i_2}(X^*)^{i_3} \upphi_{j,\tau}|^2+s^2|\tilde \upphi_{j,\tau}|^2\Big)dr d\sigma d\tau \\
&\leq \epsilon B(\delta_0) \sum_{j=0}^{k}\overline{\mathbb{I}}^{1+k-j}_{1}[\tilde\upphi_{j,\tau}](0,\tau) \,,
\end{align*}
and moreover, setting $\delta=2-p\in(0,1)$, we have
\begin{align*}
&\int_{\mc R_{(0,\tau-\delta_0)}} r^{p-1}\frac{|\tilde{\mathfrak{H}}_k|^2}{w}dr^*d\sigma d\tau\\
&\quad\leq 
\epsilon B(\delta_0) \sum_{j=0}^k\int_{\mc R_{(0,\tau-\delta_0)}} \Big(r^{-1-(2-p)}\sum_{1\leq |i|\leq 2+k-j}|(r^{-1}\mathring{\slashed\nabla})^{i_1}T^{i_2}(X^*)^{i_3}\tilde \upphi_{j,\tau}|^2+s^2r^{p-3}|\tilde \upphi_{j,\tau}|^2\Big)dr d\sigma d\tau \\
&\quad\leq \epsilon B(\delta_0,\delta) \sum_{j=0}^{k}\overline{\mathbb{I}}^{1+k-j}_{-\delta,p}[\tilde\upphi_{j,\tau}](0,\tau-\delta_0)\,.
\end{align*}
All of the norms are taken with respect to the spacetime with black hole parameter $\mathring{a}$.

\medskip

\noindent \textit{Step 2: integrated estimates.} With norms taken with respect to the spacetime with $\mathring{a}$, we can appeal to a version of Proposition~\ref{prop:ODE-to-PDE-future-int} for inhomogeneous transformed systems to control a degenerate bulk for $\tilde\upphi_k$ with non-optimal weights as $r^*\to \pm \infty$. Indeed, by repeating the proof of this proposition, we have 
\begin{align*}
&\mathbb{I}^{\rm deg}[\Phi_\tau](0,\infty) + \sum_{k=0}^{|s|-1}\sum_{J_1+J_2=0}^{|s|-k}\int_{\mc{R}_0} \frac{1}{r^2}\lp(|T^{J_1}Z^{J_2}\upphi_{k,\tau}''|^2+|T^{J_1+1}Z^{J_2}\upphi_{k,\tau}'|^2+r^{-1}|T^{J_1}Z^{J_2}\mathring{\slashed\nabla}^{[s]}\upphi_{k,\tau}'|^2\rp)dr  d\sigma d\tau\\
&\qquad + \sum_{k=0}^{|s|-1}\sum_{J_1+J_2=0}^{|s|-k}\int_{\mc{R}_0} \frac{1}{r^2}\lp(|T^{J_1+1}Z^{J_2}\upphi_{k,\tau}|^2+\frac{1}{r}|T^{J_1}Z^{J_2}\mathring{\slashed\nabla}^{[s]}\upphi_{k,\tau}|^2\rp)dr  d\sigma d\tau\\
&\qquad+\sum_{k=0}^{|s|-1}\sum_{J_1+J_2=0}^{|s|-k}\int_{\mc{R}_0} \frac{1}{r^2}\lp(\frac{1}{r}|T^{J_1}Z^{J_2}\upphi_{k,\tau}|^2+|T^{J_1}Z^{J_2}\upphi_{k,\tau}'|^2\rp)dr  d\sigma d\tau\\
&\quad\leq  B(a_0)\overline{\mathbb{E}}^{|s|-k}[\tilde\upphi_{k,\tau}](0) + \epsilon B(a_0)\int_{\mc R_{(0,\infty)}}(1-\zeta) \lp(|T\Phi|^2+|\mathring{\slashed\nabla}^{[s]}\Phi_\tau|\rp)dr d\sigma d\tau \\
&\quad\qquad+B(a_0)\int_{\mc R_{(0,\infty)}} \frac{|{\mathfrak{H}}|^2}{\epsilon w}dr^*d\sigma d\tau+B(a_0)\sum_{k=0}^{|s|-1}\sum_{J_1+J_2=0}^{|s|-k-1}\sum_{X\in\{\mathrm{id},X^*\}}\int_{\mc R_{(0,\infty)}} \frac{|T^{J_1}Z^{J_2}X{\mathfrak{H}}|^2}{\epsilon w}dr^*d\sigma d\tau\,.
\end{align*}
To obtain the last two terms, we have applied Cauchy--Schwarz to all terms arising from the interaction of frequency-localized currents with the inhomogeneities $\mathfrak{H}_k$; then for $k<|s|$ we have used the transport estimates \eqref{eq:transport-estimates-inhom-intermediate} to control $\mathfrak{H}_k$ terms by $\mathfrak{H}$ terms. (Notice that we lose several derivatives in the process, but the nature of the statement we are seeking to prove does not require us to be optimal.) Then, an easy extension of Proposition~\ref{prop:higher-order-bulk-flux} yields that fixing a  sufficiently large $R^*>0$,
\begin{align*}
&\sum_{k=0}^{|s|}\overline{\mathbb{I}}_{-\delta,p}^{\mathrm{deg},J+|s|-k}[\tilde\upphi_{k,\tau}](0,\infty)\\
&\quad\leq B(\delta)\sum_{J_1+J_2=0}^{J}\int_0^{\infty} \int_{\Sigma_{\tau'}} r^{p+1}|T^{J_1}Z^{J_2}\mathfrak{H}|^2 dr d\sigma d\tau' \\
&\quad\qquad +B(\delta)\sum_{k=0}^{|s|-1}\sum_{J_1+J_2=0}^{J+|s|-k-1}\sum_{X\in\{\mathrm{id},X^*\}}\int_0^{\infty} \int_{\Sigma_{\tau'}} r^{p+1}|T^{J_1}Z^{J_2}X\tilde{\mathfrak{H}}_k|^2 dr d\sigma d\tau'\\
&\quad\qquad +B(\delta)\sum_{k=0}^{|s|}\overline{\mathbb{I}}^{\mathrm{deg},J+|s|-k}[\tilde\upphi_{k,\tau}\mathbbm{1}_{[-R^*,R^*]}](0,\infty) + B(\delta)\sum_{k=0}^{|s|}\overline{\mathbb{E}}_p^{J+|s|-k}[\tilde\upphi_{k,\tau}](0)\\
&\quad\leq  B(a_0,\delta)\sum_{J_1+J_2=0}^{J+|s|-k}\int_0^{\tau} \int_{\Sigma_{\tau'}} r^{p+1}|T^{J_1}Z^{J_2}\mathfrak{H}|^2 dr d\sigma d\tau' + B(a_0,\delta)\sum_{k=0}^{|s|}\overline{\mathbb{E}}_p^{J+|s|-k}[\tilde\upphi_{k,\tau}](0) \\
&\quad\qquad +B(a_0,\delta)\sum_{k=0}^{|s|-1}\sum_{J_1+J_2=0}^{J+|s|-k-1}\sum_{X\in\{\mathrm{id},X^*\}}\int_0^{\tau} \int_{\Sigma_{\tau'}} \epsilon^{-1}r^{p+1}|T^{J_1}Z^{J_2}X\tilde{\mathfrak{H}}_k|^2 dr d\sigma d\tau'\\
&\quad\qquad + \epsilon B(a_0,\delta)\sum_{J_1+J_2=0}^{J}\int_{\mc R_{(0,\infty)}}(1-\zeta) \lp(|T^{1+J_1}Z^{J_2}\Phi_\tau|^2+|T^{J_1}Z^{J_2}\mathring{\slashed\nabla}^{[s]}\Phi_\tau|\rp)dr d\sigma d\tau\\\
&\quad\leq \epsilon B(\delta_0,\delta) \sum_{k=0}^{|s|-1}\overline{\mathbb{I}}^{J+|s|-k+1}_{-\delta,p}[\tilde\upphi_{k,\tau}](0,\tau) +  \epsilon B(\delta_0,\delta) \overline{\mathbb{I}}^{J+1}_{-\delta,p}[\Phi_{\tau}](\tau,\infty)+ B(a_0,\delta)\sum_{k=0}^{|s|}\overline{\mathbb{E}}_p^{J+|s|-k}[\tilde\upphi_{k,\tau}](0)\,,
\end{align*}
where we have made use of the estimates on the inhomogeneity which we have obtained in the previous step.

Since $\upphi_{k,\tau}=\upphi_k$ in the past of $\Sigma_{\tau-\delta_0}$, $\upphi_{k,\tau}=\upphi_{k,\circ}$ in the future of $\Sigma_\tau$, and $\upphi_k$, $\upphi_{k,\tau}$ and $\upphi_{k,\circ}$ all share the same initial data, we deduce
\begin{align*}
&\sum_{k=0}^{|s|}\overline{\mathbb{I}}_{-\delta,p}^{\mathrm{deg},J+|s|-k}[\tilde\upphi_{k}](0,\tau)\\
&\quad\leq \epsilon B(\delta_0,\delta)\Big( \sum_{k=0}^{|s|}\overline{\mathbb{I}}^{J+|s|-k+1}_{-\delta,p}[\tilde\upphi_{k}](0,\tau-\delta_0) +  \overline{\mathbb{I}}^{J+1}_{-\delta,p}[\Phi_{\circ}](\tau,\infty) \Big)\\
&\quad\qquad + B(a_0,\delta)\sum_{k=0}^{|s|}\lp[\lp(\delta_0+\epsilon B(\delta_0)\rp)\sup_{\tau'\in[0,\tau]}\overline{\mathbb{E}}_p^{J+|s|-k+1}[\tilde\upphi_{k}](\tau)+\overline{\mathbb{E}}_p^{J+|s|-k+1}[\tilde\upphi_{k}](0)\rp]\\
&\quad\leq \epsilon B(\delta_0,\delta)\sum_{k=0}^{|s|}\overline{\mathbb{I}}^{J+|s|-k+1}_{-\delta,p}[\tilde\upphi_{k}](0,\tau) \\
&\quad\qquad + B(a_0,\delta)\sum_{k=0}^{|s|}\lp[\lp(\delta_0+\epsilon B(\delta_0)\rp)\sup_{\tau'\in[0,\tau]}\overline{\mathbb{E}}_p^{J+|s|-k+1}[\tilde\upphi_{k}](\tau)+\overline{\mathbb{E}}_p^{J+|s|-k+1}[\tilde\upphi_{k}](0)\rp]\,,\numberthis~\label{eq:cont-arg-end}
\end{align*}
where we have used finite in time estimates as well as the fact that $\upphi_{k,\circ}$ is future integrable with respect to $\mathring{a}$ and thus verifies the integrated local energy decay and energy boundedness estimates of Theorems~\ref{thm:ILED-futurint} and \ref{thm:energy-bddness-futurint}. 

\medskip
\noindent \textit{Step 3: conclusion.}  Notice that the right hand side of \eqref{eq:cont-arg-end} is higher order in derivatives than the left hand side. Thus, to conclude we need to appeal to our derivative gain estimate \eqref{eq:derivative-gain} in Lemma~\ref{lemma:derivative-gain}; using \eqref{eq:openness-intermediate}, we have 
\begin{align*}
&\sup_{\tau'\in[0,\tau]}\sum_{k=0}^{|s|}\overline{\mathbb E}_p^{J+|s|-k}[\tilde\upphi_k](\tau) + \sum_{k=0}^{|s|}\overline{\mathbb I}_p^{J+|s|-k-1}[\tilde\upphi_k](0,\tau) \\
&\quad\leq B(a_0,m,\delta)\sum_{k=0}^{|s|}\overline{\mathbb I}_p^{\mathrm{deg}, 4|s|-k+1}[\tilde\upphi_k](0,\tau) +B(a_0,m,\delta)\overline{\mathbb E}_p^{J+|s|-k}[\tilde\upphi_k](0)\\
&\quad \leq  B(a_0,\delta,m)\lp(\delta_0+\epsilon B(\delta_0)\rp)\sum_{k=0}^{|s|}\lp\{\sup_{\tau'\in[0,\tau]}\overline{\mathbb{E}}_p^{4|s|-k+2}[\tilde\upphi_{k}](\tau)+\overline{\mathbb{I}}^{4|s|-k+2}_{-\delta,p}[\tilde\upphi_{k}](0,\tau)\rp\} \\
&\quad\qquad+B(a_0,\delta, m)\overline{\mathbb E}_p^{J+|s|-k}[\tilde\upphi_k](0)\,.
\end{align*}
Choosing $\delta_0$ sufficiently small and fixing it, we see that for sufficiently small $\epsilon=|a-\mathring{a}|$ we can absorb the first term on the right hand side into the left hand side of the above inequality as long as $J\geq 3|s|+2$. This concludes the proof. 
\end{proof}

\bibliographystyle{../../../halpha-abbrv-rita}
{\small \bibliography{../unpub,../../../library}}

\end{document}